%% file: thesis-arxiv.tex
\documentclass[10pt, reqno]{amsart}
\usepackage[top=1in,bottom=1in,left=1in,right=1in]{geometry}
\usepackage[dvipsnames]{xcolor}
\usepackage[hidelinks=true,pdfencoding=auto, psdextra]{hyperref}
\usepackage{cleveref}
\usepackage{amsthm}
\usepackage{thmtools}
\declaretheorem[numberwithin=section]{theorem}
\declaretheorem[sibling=theorem]{lemma, proposition, corollary, claim, goal, conjecture}

\declaretheorem[style=remark, sibling=theorem]{remark}

\newcommand{\sidenote}{\footnote}
\usepackage{etoolbox}
\makeatletter
\patchcmd{\@maketitle}{\global\topskip42\p@\relax}
  {\global\topskip42\p@\relax \vspace*{-1ex}}
  {}{}
\makeatother
\usepackage[backend=biber,style=numeric]{biblatex}
\crefname{equation}{}{}
\crefrangeformat{equation}{(#3#1#4)--(#5#2#6)}
\newcommand{\step}[1]{\smallskip\noindent\emph{#1}}
\usepackage{import}
\usepackage{xifthen}
\usepackage{pdfpages}
\usepackage{transparent}
\usepackage{makecell}
\usepackage{subcaption}

\usepackage{amsmath}
\usepackage{amssymb}
\usepackage{xifthen}
\usepackage{faktor}
\usepackage{tikz-cd}
\usepackage{siunitx}
\usepackage{cancel}
\usepackage{mathtools}
\usepackage{slashed}

\newcommand{\dletter}{\mathrm{d}}
\newcommand{\dd}{\mathop{}\!\dletter}

\newcommand{\mdd}[2][]{
\ifthenelse{\equal{#1}{}}
    {\def\temp {\dletter}}
    {\def\temp {\dletter^{{#1}}}}
\mathop{}\!\temp #2}
\newcommand{\mpd}[2][]{
\ifthenelse{\equal{#1}{}}
    {\def\temp {\partial}}
    {\def\temp {\partial^{{#1}}}}
\temp #2}

\newcommand{\dv}[3][]{\frac{\mdd[#1]{#2}}{\mdd{#3}\ifthenelse{\equal{#1}{}}
    {\def\temp {}}
    {\def\temp {^{#1}}}\temp}}
\newcommand{\pdv}[3][]{\frac{\mpd[#1]{#2}}{\mpd{#3}\ifthenelse{\equal{#1}{}}
    {\def\temp {}}
    {\def\temp {^{#1}}}\temp}}

\newcommand{\Grad}{\nabla}

\DeclarePairedDelimiter{\norm}{\lVert}{\rVert}
\DeclarePairedDelimiter{\abs}{\lvert}{\rvert}
\DeclarePairedDelimiter{\floor}{\lfloor}{\rfloor}

\DeclarePairedDelimiter{\set}{\{}{\}}

\DeclarePairedDelimiterX{\braket}[2]{\langle}{\rangle}{#1 \delimsize\vert #2}

\usepackage{bm}

\renewcommand{\vec}[1]{\if#1\relax\bm{#1}\else\mathbf{#1}\fi}

\newcommand{\conj}[1]{\overline{#1}}

\DeclareMathOperator{\SO}{SO}

\DeclareMathOperator{\Supp}{supp}

\DeclareMathOperator{\Ric}{Ric}

\newcommand{\card}{\#}

\newcommand{\Z}{\mathbf{Z}}

\newcommand{\R}{\mathbf{R}}

\newcommand{\side}[2]{\text{\textup{(#1 of }}\cref{#2}\text{\textup{)}}}

\newcommand{\genlegendre}[4]{%
  \genfrac{(}{)}{}{#1}{#3}{#4}%
  \if\relax\detokenize{#2}\relax\else_{\!#2}\fi
}

\makeatletter
\renewcommand*\env@matrix[1][*\c@MaxMatrixCols c]{%
  \hskip -\arraycolsep
  \let\@ifnextchar\new@ifnextchar
  \array{#1}}
\makeatother

\makeatletter
\newcommand*\owedge{\mathpalette\@owedge\relax}
\newcommand*\@owedge[1]{%
  \mathbin{%
    \ooalign{%
      $#1\m@th\bigcirc$\cr
      \hidewidth$#1\m@th\wedge$\hidewidth\cr
    }%
  }%
}
\makeatother

\let\oldfrac\frac% Store \frac
\makeatletter
\newcommand{\groupit}[1]{#1}% To group...
\newcommand{\nogroupit}[1]{#1}% ...or not to group
\renewcommand{\frac}[2]{%
  \setbox\z@\hbox{$#1$}% Store numerator
  \setbox\tw@\hbox{$#2$}% Store denominator
  \ifdim\wd\z@>1em \let\groupornot@i\groupit\else\let\groupornot@i\nogroupit\fi% Measure numerator
  \ifdim\wd\tw@>1em \let\groupornot@ii\groupit\else\let\groupornot@ii\nogroupit\fi% Measure denominator
  \mathchoice
  {\oldfrac{#1}{#2}}% display style
  {\oldfrac{#1}{#2}} % text style
    {\groupornot@i{#1}/\groupornot@ii{#2}}% script style
    {\groupornot@i{#1}/\groupornot@ii{#2}}% script-script style
}
\makeatother

\def\XXint#1#2#3{{\setbox0=\hbox{$#1{#2#3}{\int}$ }
\vcenter{\hbox{$#2#3$ }}\kern-.6\wd0}}

\newcommand{\Rc}{R_{\bullet}}
\numberwithin{equation}{section}
\author{Onyx Gautam}
\date{\today}
\title{Late-time tails and mass inflation for the spherically symmetric Einstein--Maxwell--scalar field system}
\begin{document}

\maketitle
\begin{abstract}
We establish a decay result in the black hole exterior region of spherically
symmetric solutions to the Einstein--Maxwell--scalar field system arising from
compactly supported admissible data. Our result allows for large initial data,
and it is the first decay statement for higher order derivatives of the scalar
field.

Solutions to this model generically develop a singularity in the black hole
interior. Indeed, Luk--Oh \cite{luk-oh-scc1,luk-oh-scc2} identify a generic class
of initial data that produces \(C^2\)-future-inextendible solutions. However,
they leave open the question of mass inflation: does the Hawking mass become
identically infinite at the Cauchy horizon? By work of
Luk--Oh--Shlapentokh-Rothman \cite{luk2022scattering}, our decay result implies
mass inflation for sufficiently regular solutions in the generic class
considered by Luk--Oh \cite{luk-oh-scc1,luk-oh-scc2}.

Together with the methods and results of Luk--Oh \cite{luk-oh-tails}, our
estimates imply a late-time tails result for the scalar field. This result
provides another proof of generic mass inflation, through a result of Dafermos
\cite{dafermos05}. Another application of our late-time tails result, due to Van
de Moortel \cite{maxime-spacelike-null}, is the global construction of two-ended
black holes that contain null and spacelike singularities.
\end{abstract}
\tableofcontents
\section{Introduction}
\label{sec:orgc849934}
The Einstein--Maxwell--scalar field system models gravity in the presence of an
electromagnetic field and a scalar field that do not interact except through the
underlying geometry. A solution to this system is a quadruple
\((\mathcal{M},g,\varphi{},F)\) consisting of a \(4\)-dimensional manifold
\(\mathcal{M}\) with Lorentzian metric \(g\), a real-valued scalar field
\(\varphi{}\), and a \(2\)-form \(F\) (the Maxwell field) that satisfy the
following equations:
\begin{equation}\label{einstein-maxwell}
\begin{cases}
\Ric(g) - \frac{1}{2}gR(g) = 2(T^{\mathrm{(sf)}} + T^{\text{(em)}}), \\
T^{\text{(sf)}}_{\alpha{}\beta{}} = \partial{}_\alpha{}\varphi{}\partial{}_\beta{}\varphi{} - \frac{1}{2}g_{\alpha{}\beta{}}\partial{}^\mu{}\varphi{}\partial{}_\mu{}\varphi{}, \\
T^{(\text{em})}_{\alpha{}\beta{}} = F_\alpha{}^\nu{}F_{\beta{}\nu{}} - \frac{1}{4}g_{\alpha{}\beta{}}F^{\mu{}\nu{}}F_{\mu{}\nu{}}, \\
\Box{}_g\varphi{} = 0, \qquad \dd{}F = 0, \qquad \mathrm{div}_g F = 0.
\end{cases}
\end{equation}
Here \(\Ric(g)\) and \(R(g)\) are the Ricci curvature tensor and Ricci scalar
curvature of the metric \(g\), and \(\Box_g\) and \(\mathrm{div}_g\) are the
Laplace--Beltrami operator and divergence operator associated to \(g\). We
consider \cref{einstein-maxwell} in spherical symmetry (see \cref{sec:emsf}), when
\(\mathcal{M}\) admits an \(\SO(3)\)-action that preserves \(g\), \(\varphi{}\),
and \(F\).

The explicit family of Reissner--Nordström spacetimes are spherically symmetric
solutions to \cref{einstein-maxwell} with vanishing scalar field: in local
coordinates, the Maxwell field is \(2r^{-2}\mathbf{e}^2
\dd{}t\wedge{}\!\dd{}r\), and the metric takes the form
\begin{equation}\label{reissner-nordstrom-metric}
g_{M,\mathbf{e}} = -\Bigl(1 - \frac{2M}{r} + \frac{\mathbf{e}^2}{r^2}\Bigr)\dd{}t^2 + \Bigl(1 - \frac{2M}{r} + \frac{\mathbf{e}^2}{r^2}\Bigr)^{-1}\dd{}r^2 + r^2g_{\mathbf{S}^2},
\end{equation}
where \(M\) and \(\mathbf{e}\) are real-valued constants to be thought of as
mass and charge, respectively, and \(g_{\mathbf{S}^2}\) is the round metric on
the sphere of radius \(1\). Special cases of the Reissner--Nordström metrics
include the Schwarzschild metrics (\(\mathbf{e} = 0\)) and the flat Minkowski
metric (\(M = \mathbf{e} = 0\)). A Reissner--Nordström spacetime with \(0 <
\abs{\mathbf{e}}<M\) is said to be subextremal. We restrict our attention to
solutions of \cref{einstein-maxwell} that settle down to a subextremal
Reissner--Nordström spacetime (namely the so-called future-admissible solutions,
which are discussed in \cref{sec:admissible-data}).
\subsection{Statement of main theorem}
\label{sec:org5c00429}
\begin{theorem}[Main theorem, rough version]
Solutions to \eqref{einstein-maxwell} arising from smooth and compactly supported
future-admissible spherically symmetric Cauchy data satisfy the following
estimate on the event horizon:
\begin{equation}
\abs{(v\underline{\partial{}}_v)^k\varphi{}}|_{\mathcal{H}} \lesssim_{\epsilon{},k,\varphi{}} v^{-1+\epsilon{}}\text{ for } \epsilon{}>0\text{ and }k\ge 0.
\end{equation}
Here \(v\) is an outgoing null coordinate normalized appropriately on a curve of
constant area-radius \(r\) (see \cref{sec:coordinate-systems}) and
\(\underline{\partial{}}_v\) is the associated coordinate derivative in
\((r,v)\) coordinates.
\label{main-theorem-rough}
\end{theorem}
A precise version of \cref{main-theorem-rough} is given in
\cref{sec:reduction-to-characteristic}. For a version of \cref{main-theorem-rough}
stated in an Eddington--Finkelstein-type gauge (where the outgoing null
coordinate is normalized on the event horizon), see
\cref{main-theorem-double-null}.
\begin{remark}[Applications]
We already mention that our main theorem implies mass inflation (see
\cref{sec:mass-inflation}). Moreover, the estimates used to prove
\cref{main-theorem-rough}, when combined with the methods and results of Luk--Oh
\cite{luk-oh-tails}, imply a late-time tails result (see
\cref{sec:intro:late-time-tails}), which gives another proof of mass inflation (see
\cref{mass-inflation-alternate}) and has a further application to the co-existence
of spacelike and null singularities in the interior of dynamical spherically
symmetric black holes (see \cref{sec:coexistence}).
\end{remark}
\begin{remark}[Coordinate derivatives associated to \(v\)]
The vector field \(\underline{\partial{}}_v\) referenced in \cref{main-theorem-rough} is
tangent to surfaces of constant \(r\). It is timelike in the exterior region and
is generically not tangent to the event horizon (which is not a surface of
constant \(r\) in evolution). On the other hand, \(\partial_{\bar{v}}\) is null and
always tangent to the event horizon. On Reissner--Nordström, \(\underline{\partial{}}_v\)
is the stationary Killing field, which coincides with \(\partial_{\bar{v}}\) on the event
horizon.
\end{remark}
\begin{remark}[Large data]
Our theorem does not have a smallness assumption, so we can handle large data.
In the proof it is crucial to observe certain reductive structure in the error
terms arising from commutation (see the discussion in
\cref{sec:structure-of-proof}).
\end{remark}
\begin{remark}[Decay in a finite-\(r\) region]
Our proof of \cref{main-theorem-rough} in fact shows that, for arbitrarily large \(R > 0\), we have
\(\abs{(v\underline{\partial{}}_v)^k\varphi{}}\lesssim_{\epsilon{},k,\varphi{},R}
v^{-1 + \epsilon{}}\) in a finite-\(r\) region \(\set{r\le R}\) contained in the
exterior of the black hole. This is
because the scaling vector field \(S\) (see \cref{sec:commutator-vector-fields}) is
equal to \(v\underline{\partial{}}_v\) in this region. We do not expect the same
result to hold if \(\underline{\partial{}}_v\) is replaced by \(\partial{}_v\).
\end{remark}
\subsection{Applications and related work}
\label{sec:org67d0f92}
\subsubsection{Strong cosmic censorship}
\label{sec:org15eecaa}
The Reissner--Nordström spacetimes are geodesically incomplete. In physical
terms, the fate of an observer who enters the black hole region is determined
only for finite time. The geodesic incompleteness of the Schwarzschild
spacetimes is due to a singularity in the black hole interior (at
``\(\set{r=0}\)'' in the coordinates of \cref{reissner-nordstrom-metric}) that
destroys incoming observers. Sbierski \cite{sbierski-schwarzschild} rigorously
established the strength of this singularity by showing that Schwarzschild is
\(C^0\)-inextendible, namely that it does not isometrically embed as a proper
subset of Lorentzian manifold with continuous metric. In contrast, the
subextremal Reissner--Nordström spacetimes admit infinitely many inequivalent
extensions across the Cauchy horizon as smooth solutions to
\cref{einstein-maxwell}. Physically speaking, an observer falling through a
Reissner--Nordström black hole continues unimpeded, but Einstein's equations do
not uniquely specify how they continue. Penrose conjectured that examples such
as Reissner--Nordström are unstable to perturbations, so that they would not
appear in nature. We now state a modern formulation of this conjecture.
\begin{conjecture}[Strong cosmic censorship]
A generic asymptotically flat solution to the Einstein equations is future
inextendible as a suitably regular Lorentzian manifold.
\label{strong-cosmic-censorship}
\end{conjecture}
The original motivation behind \cref{strong-cosmic-censorship} was that an observer
crossing the Cauchy horizon of a Reissner--Nordström black hole would
observe the outside universe shifted infinitely to the blue
\cite[p.~222]{osti_4125134}.\sidenote{Contrast this to the redshift effect by which an observer outside the
event horizon receives signals from an observer crossing the event horizon as
infinitely shifted to the red.} At the level of linear perturbations,
\cite{osti_4461850,mcnamara78,Chandrasekhar1982OnCT} argued that the blueshift
effect should manifest by exponentially amplifying gravitational disturbances
from the black hole exterior, and indeed \cite{Franzen_2016,Luk_2017} verified that
generic solutions to the linear wave equation on Reissner--Nordström remain
bounded but are not in \(H^{1}_{\text{loc}}\) near the Cauchy horizon. See
\cite{Luk:2015pay,Hintz_2017,Franzen_2020,ma2023precise} for extensions of
these results to subextremal Kerr and \cite{Hintz_2017_2} for the case of
positive cosmological constant. The instability of the subextremal Kerr Cauchy
horizon under the linearized Einstein vacuum equations has been established by
\cite{Sbierski_2023} and \cite{Ma_2023}. See
\cite{Van_de_Moortel_2018,kehle2022strong} for results on
\cref{strong-cosmic-censorship} in the context of the spherically symmetric
Einstein--Maxwell--Klein Gordon system, and \cite{rossetti2024strong} for the
case of positive cosmological constant.

We now survey the main results related to \cref{strong-cosmic-censorship} for the
nonlinear and spherically symmetric matter model \eqref{einstein-maxwell}, which provide a nearly complete
understanding of the problem. The strongest expectation would be that
perturbations of Reissner--Nordström are \(C^0\)-inextendible, like
Schwarzschild. In fact this expectation is false, but the conjecture holds in the higher
regularity class of \(C^2\).
\begin{theorem}[Strong cosmic censorship fails in \(C^0\); Dafermos \cite{dafermos03}, Dafermos--Rodnianski \cite{dafermos-rodnianski-price}]
Future-admissible\sidenote{The condition of future-admissibility ensures that the future development of the
Cauchy data has the same Penrose diagram as Reissner--Nordström (see
\cite[Def.~3.1]{luk-oh-scc1} for a precise definition).} solutions to \cref{einstein-maxwell} with non-vanishing charge are
\(C^0\)-extendible.
\end{theorem}
Dafermos--Luk \cite{dafermos-2017-inter-dynam} showed an analogous result outside
symmetry, namely that perturbations of subextremal Kerr (as solutions to the
Einstein vacuum equations) remain \(C^0\)-extendible.
\begin{theorem}[Strong cosmic censorship holds in \(C^2\); Luk--Oh \cite{luk-oh-scc1,luk-oh-scc2}]
There is a class \(\mathcal{G}\) of future-admissible Cauchy data for
\cref{einstein-maxwell} such that
\begin{enumerate}
\item \(\mathcal{G}\) is generic, namely open (in a weighted \(C^1\) topology) and dense (in a weighted \(C^\infty\) topology),
\item solutions arising from data in \(\mathcal{G}\) are \(C^2\)-future-inextendible.
\end{enumerate}
\label{luk-oh-scc}
\end{theorem}
Sbierski \cite{sbierski-holonomy} showed that small data solutions in the generic
class \(\mathcal{G}\) of Luk--Oh are moreover \(C^{0,1}\)-inextendible.
\subsubsection{Mass inflation}
\label{sec:org22749a9}
\label{sec:mass-inflation}
The first attempt to understand the instability of the Cauchy horizon in
Reissner--Nordström for a nonlinear model was due to Hiscock \cite{MR617171}, who
considered an explicit solution to the spherically symmetric Einstein--null dust
system (a simpler analogue of \cref{einstein-maxwell}) with one incoming dust. In
the context of this work, the metric remains continuous at the Cauchy horizon,
but its Christoffel symbols blow up in a parallelly propagated frame. In their
seminal works \cite{israel-poisson,PhysRevD.41.1796}, Poisson--Israel considered
a null dust solution with an additional, outgoing, null dust. They showed that,
generically, the Hawking mass becomes infinite at the Cauchy horizon, and they
named this scenario \emph{mass inflation}. Because the Hawking mass bounds the
Kretschmann scalar from below \cite{Kommemi2013TheGS} (namely the full trace
\(\mathrm{Riem}_{\alpha{}\beta{}\gamma{}\delta{}}\mathrm{Riem}^{\alpha{}\beta{}\gamma{}\delta{}}\)
of the Riemann curvature tensor), mass inflation immediately implies the \(C^2\)
formulation of strong cosmic censorship.\sidenote{From the form of the Hawking mass, mass inflation also implies
inextendibility in the class of \(C^1\) spherically symmetric Lorentzian
manifolds.} Many analytic and numerical
studies \cite{MR1118553,10.2307/52772,Brady_1995} have corroborated the
phenomenon of mass inflation. See also \cite{MR4227173} for a study of mass
inflation for the Einstein--Maxwell--Klein Gordon system and
\cite{Costa_2017,rossetti2024strong} for studies in the case of positive
cosmological constant.

We now return to our discussion of \cref{einstein-maxwell}. Mathematical results on
mass inflation for the Einstein--Maxwell--scalar field system require upper
bounds and generic lower bounds for the decay of the scalar field along the
horizon. Dafermos \cite{dafermos05} showed that mass inflation follows from the
pointwise upper bound of Dafermos--Rodnianski \cite{dafermos-rodnianski-price}
and a yet unproven pointwise lower bound (see \cref{tails-theorem}).
Luk--Oh--Shlapentokh-Rothman showed that mass inflation follows from the
integrated lower bound of \cite{luk-oh-scc1,luk-oh-scc2} and yet unproven
integrated upper bounds (see \cref{main-theorem-double-null}). We state these
results in the following two theorems. In these theorems, \(\bar{v}\) is an
(appropriately normalized) Eddington--Finkelstein type coordinate, and
\(\partial_{\bar{v}}\) is the associated coordinate derivative in double null
coordinates.
\begin{theorem}[Best known upper and lower bounds along the horizon]
The following estimates hold for future-admissible solutions to \cref{einstein-maxwell}:
\begin{enumerate}
\item (Pointwise upper bound; Dafermos--Rodnianski
   \cite{dafermos-rodnianski-price}) We have
\begin{equation}\label{price-law}
\abs{\varphi}|_{\mathcal{H}} + \abs{\partial{}_{\bar{v}}\varphi{}}|_{\mathcal{H}}\lesssim_{\epsilon{},\varphi{}} \bar{v}^{-3+\epsilon{}}.
\end{equation}
\item (\(L^2\) lower bound; Luk--Oh \cite{luk-oh-scc1,luk-oh-scc2}) Solutions in the generic class \(\mathcal{G}\) of \cref{luk-oh-scc} satisfy
\begin{equation}\label{L2-lower-bound}
  \int_{\mathcal{H}} \bar{v}^p(\partial{}_{\bar{v}}\varphi{})^2 \dd{}\bar{v} = \infty \text{ for }p > 7.
\end{equation}
\end{enumerate}
\end{theorem}
\begin{theorem}[Best known (conditional) results on mass inflation]
\label{mass-inflation-conditional}
Mass inflation holds for future-admissible solutions to \eqref{einstein-maxwell} satisfying either of the following conditions:
\begin{enumerate}
\item (Dafermos \cite{dafermos05}) the upper bound \eqref{price-law} and a pointwise lower bound:
\begin{equation}\label{dafermos-mass-inflation}
\liminf_{\bar{v}\to \infty} \bar{v}^p\abs{\partial{}_{\bar{v}}\varphi{}}|_{\mathcal{H}}(\bar{v}) > 0\,\, \text{for some } p < 9\text{;}
\end{equation}
\item (Luk--Oh--Shlapentokh-Rothman \cite{luk2022scattering}) or the lower bound \eqref{L2-lower-bound} and \(L^2\) upper bounds:
\begin{equation}\label{luk-oh-sr-mass-inflation}
\int_{\mathcal{H}} \bar{v}^4(\partial{}_{\bar{v}}\varphi{})^2\dd{}\bar{v} < \infty \quad \text{and}\quad \int_{\mathcal{H}}\bar{v}^8(\partial{}_{\bar{v}}^k\varphi{})^2\dd{}\bar{v}<\infty \,\, \text{for some }k\ge 2.
\end{equation}
\end{enumerate}
\end{theorem}
\begin{remark}
See \cref{mass-inflation-from-main-theorem} for an explanation of how to extract
 statement (2) from the results of \cite{luk2022scattering} (which in particular states the
 second criterion in \cref{luk-oh-sr-mass-inflation} only for \(k = 2\)). Observe
 that the first criterion in \cref{luk-oh-sr-mass-inflation} holds by
 \cref{price-law}.
\end{remark}
We now give an application of our main theorem to mass inflation. From our main
theorem (\cref{main-theorem-rough}), we can obtain a decay result for higher
derivatives tangent to the event horizon.
\begin{corollary}[Translation of main theorem into an Eddington--Finkelstein-type gauge]
In the setting of \cref{main-theorem-rough}, we have
\begin{equation}\label{intro:main-thm-2}
\abs{(\bar{v}\partial{}_{\bar{v}})^k\varphi{}}|_{\mathcal{H}}\lesssim_{\epsilon{},k,\varphi{}} \bar{v}^{-1+\epsilon{}}\text{ for } \epsilon{}>0\text{ and }0\le k\le 4,
\end{equation}
where \(\bar{v}\) is an Eddington--Finkelstein-type coordinate normalized
appropriately on the event horizon (see \cref{sec:coordinate-systems}) and \(\partial_{\bar{v}}\)
is the coordinate derivative associated to \(\bar{v}\) in double null coordinates
\label{main-theorem-double-null}
\end{corollary}
A precise version of \cref{main-theorem-double-null} is
given in \cref{sec:reduction-to-characteristic}.
\begin{remark}[Range of \(k\) in \cref{main-theorem-double-null}]
The range of \(k\) for which \cref{main-theorem-double-null} holds could be
increased given the analogue of \cref{main-theorem-rough} with a faster decay rate
(as proven in \cref{tails-theorem}). Note that \(k = 4\) is the minimal \(k\) for
which the decay rate \(\bar{v}^{-1 + \epsilon{}}\) for
\((\bar{v}\partial_{\bar{v}})^k\varphi{}|_{\mathcal{H}}\) implies the second
criterion in \cref{luk-oh-sr-mass-inflation}. However, we do not expect that \(k\)
can be taken arbitrarily large in \cref{main-theorem-double-null}, in view of the
expected inverse polynomial tail of \(\partial_{\bar{v}}r|_{\mathcal{H}}\) (a
consequence of the tail for \(\abs{\partial_{\bar{v}}\varphi{}}|_{\mathcal{H}}\)
demonstrated in \cref{tails-theorem}). That is, we do not expect each null
derivative along the horizon to gain a power of decay. See
\cref{sec:Stovdv-translation} for further discussion.
\label{range-of-k}
\end{remark}
Together with the already known \cref{price-law}, \cref{main-theorem-double-null}
implies the previously unknown estimate \cref{luk-oh-sr-mass-inflation} (with \(k =
4\)), which implies mass inflation.
\begin{corollary}[Generic mass inflation]
Mass inflation holds for solutions to \cref{einstein-maxwell} arising from
compactly supported Cauchy data in the generic class of Luk--Oh (see
\cref{luk-oh-scc}).
\end{corollary}
\subsubsection{Late-time tails}
\label{sec:org42fc504}
\label{sec:intro:late-time-tails} After implementing the spacetime elliptic
estimates of Luk--Oh \cite[Sec.~5.3]{luk-oh-tails}, the estimates used to prove
\cref{main-theorem-rough} are enough to satisfy the assumptions on the spacetime metric
and the scalar field in \cite[Sec.~2]{luk-oh-tails} and therefore prove the following
sharp decay result.
\begin{theorem}[Sharp Price's law result]
Solutions to \cref{einstein-maxwell} arising from smooth and compactly supported
future-admissible spherically symmetric Cauchy data satisfy the following
estimate on the event horizon:
\begin{equation}
\abs{(v\underline{\partial{}}_v)^k\varphi{}|_{\mathcal{H}} - C_k\mathfrak{L}[\varphi{}]v^{-3}} \lesssim v^{-3-\delta{}}\text{ for }k\ge 0,
\end{equation}
where \(\delta{} > 0\) is a small constant, \(C_k\neq{}0\) are explicit non-zero constants,
\(\mathfrak{L}[\varphi{}]\) is a dynamical constant that is non-zero exactly for
solutions arising from the generic class of data constructed by Luk--Oh (see
\cref{luk-oh-scc}), and the coordinate \(v\) and coordinate derivative
\(\underline{\partial{}}_v\) are as in \cref{main-theorem-rough}. Moreover, we have
\begin{equation}\label{tails-theorem-null}
\abs{(\bar{v}\partial{}_{\bar{v}})^k\varphi{}|_{\mathcal{H}} - C_k\mathfrak{L}[\varphi{}]\bar{v}^{-3}} \lesssim \bar{v}^{-3-\delta{}}\text{ for }0\le k\le 2,
\end{equation}
where the Eddington--Finkelstein-type coordinate and coordinate derivative
\(\bar{v}\) and \(\partial{}_{\bar{v}}\) are as in \cref{main-theorem-double-null}.
\label{tails-theorem}
\end{theorem}
For a precise version of \cref{tails-theorem}, see \cref{tails-theorem-precise}.
\begin{remark}[Alternative proof of mass inflation]
Observe that \cref{tails-theorem} recovers the previously known Price's law result
\cref{price-law} of Dafermos--Rodnianski \cite{dafermos-rodnianski-price} and
implies the previously unknown \cref{dafermos-mass-inflation}, and these estimates
together imply mass inflation by the result of Dafermos \cite{dafermos05} (see
\cref{mass-inflation-conditional}).
\label{mass-inflation-alternate}
\end{remark}
\begin{remark}[Range of \(k\)]
As in \cref{main-theorem-double-null}, the range of \(k\) in \cref{tails-theorem-null}
could be increased (up to \(k = 7\)), but we do not expect that \(k\) can be taken arbitrarily
large (see \cref{range-of-k}).
\end{remark}
There is a vast literature on Price's law results for linear waves on
asymptotically flat spacetimes
\cite{Donninger_2011,donninger2011pointwise,tataru,MR2921169,MR3725885,ANGELOPOULOS2023108939,MR4365146,luk-oh-tails},
including subextremal Reissner--Nordström and Kerr. We also mention
\cite{Ma_2023,millet2023optimal} for Price's law results for the Teukolsky
equation on subextremal Kerr. The analogue of \cref{tails-theorem} for linear waves
on subextremal Reissner--Nordström was first established in \cite{MR3725885}. In
that setting, \(\mathfrak{L}\) is a linear form that vanishes on solutions
arising from a codimension one class of data. See \cite{rutgers-price} for a
late-time tails result for \eqref{einstein-maxwell} for small data solutions with
non-vanishing Newman--Penrose constant (note in particular that the compactly
supported solutions we consider have vanishing Newman--Penrose constant).
\subsubsection{Black holes with spacelike and null singularities}
\label{sec:orgaf603fd}
\label{sec:coexistence}
\begin{figure}[htbp]
\vspace{-3ex}
    \def\svgwidth{0.75\columnwidth}
    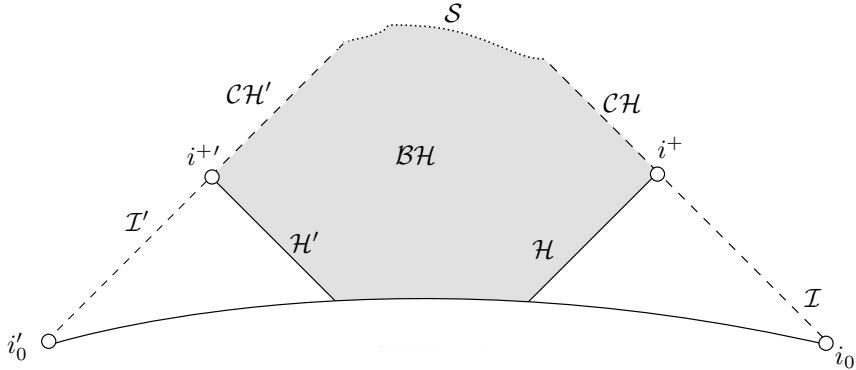

\caption{The a priori Penrose diagram for solutions to \cref{einstein-maxwell}. The achronal singular set \(\mathcal{S}\), to which the area-radius function \(r\) extends continuously to \(0\), may be empty, as it is in Reissner--Nordström. This characterization is due to \cite{Dafermos_2014,dafermos05,Kommemi2013TheGS}. In particular, Dafermos \cite{Dafermos_2014} proved that \(\mathcal{S}\) can be non-empty only for large perturbations of Reissner--Nordström.}
\label{fig:two-ended-penrose}
\end{figure}
The main result of \cite{maxime-spacelike-null} is a local result
\cite[Thm.~2.1]{maxime-spacelike-null} in the black hole interior (near the
junction between the Cauchy horizon and the spacelike singularity) providing
precise asymptotics for geometric quantities and for the scalar field. This
quantitative local result is combined with a gluing argument (based on
\cite{kehle2022gravitational}) to construct global two-ended spacetimes that
contain both null and spacelike singularities. We summarize this construction in
the following theorem, whose proof requires the decay estimates
\(\partial{}_{\bar{v}}\varphi{}|_{\mathcal{H}} = C\bar{v}^{-4} +
o(\bar{v}^{-4})\) and \(\partial{}_{\bar{v}}^2\varphi{}|_{\mathcal{H}} =
O(\bar{v}^{-5})\) that are provided by \cref{tails-theorem}.
\begin{theorem}[Van de Moortel, {\cite[Thm.~2.6(ii)]{maxime-spacelike-null}}]
There is a large class of two-ended asymptotically flat black hole spacetimes
solving \cref{einstein-maxwell} (namely those arising from an open subset of the
generic class of future-admissible Cauchy data constructed by Luk--Oh in
\cref{luk-oh-scc}) with the Penrose diagram in \cref{fig:two-ended-penrose}.

Both the Cauchy horizon \(\mathcal{CH}\) (to which \(r\) extends continuously as
a strictly positive function) and the achronal singular set \(\mathcal{S}\) (to
which \(r\) extends continuously to \(0\)) are non-empty, and \(\mathcal{S}\) is
spacelike in a neighbourhood of the junction between \(\mathcal{S}\) and the
Cauchy horizon \(\mathcal{CH}\) (in the topology of the Penrose diagram).
Moreover, there are quantitative asymptotics for geometric quantities and
estimates for the scalar field towards \(\mathcal{CH}\), as well as quantitative
Kasner asymptotics towards \(\mathcal{S}\).
\end{theorem}
\newpage
\subsection{Ideas of the proof}
\label{sec:org9c4fe1c}
\label{sec:structure-of-proof}
\subsubsection{Use of a scaling vector field}
\label{sec:org297ef69}
Our strategy to prove \cref{main-theorem-rough} involves
a scaling vector field commutator \(S\) that equals
\(v\underline{\partial{}}_v\) along the horizon (see
\cref{sec:intro-commutator-vector-fields} for the definition of \(S\)). We briefly recall
the merits of such a vector field in other settings and outline the construction
of the vector field in our setting.

On Minkowski, the scaling vector field takes the form \(S_{\text{m}} =
u\partial_u + v\partial_v = t\partial_t + r\partial_r\). It is well known that
\(S_{\text{m}}\) is a useful vector field commutator. It is conformally Killing,
and \([S_{\text{m}},\Box_{\text{m}}] = 2\Box_{\text{m}}\), which means that any
decay estimates for a solution \(\varphi{}\) to \(\Box_{\text{m}}\varphi{} = 0\)
also hold for \(S_{\text{m}}\varphi{}\). Due to the \(t\)-weight in
\(S_{\text{m}}\), one can hope to prove better decay in time for \(\varphi{}\)
given estimates for \(S_{\text{m}}\varphi{}\). Indeed, control of
\(S_{\text{m}}\varphi{}\) together with \(\Gamma{}\varphi{}\) for \(\Gamma{}\) a
Killing vector field of Minkowski leads to decay by Klainerman's Sobolev-type
estimates. The use of a scaling vector field commutator on Minkowski goes back
to Klainerman \cite{MR784477} and Klainerman--Sideris \cite{MR1374174}, and its
use on black hole spacetimes appears in \cite{Luk_2010,MR2921169,tataru}.
\subsubsection{Ingredients of the proof}
\label{sec:orga23eba1}
\label{sec:proof-ingredients}
By the construction of the scaling vector field \(S\), the statement of
\cref{main-theorem-rough} is equivalent to the following estimate along the event
horizon:
\begin{equation}\label{S-phi-estimate}
\abs{S^k\varphi{}}|_{\mathcal{H}} \lesssim v^{-1+\epsilon{}}.
\end{equation}
The estimate \cref{S-phi-estimate} follows from a hierarchy of \(r^p\)-weighted
energy estimates (introduced in \cite{rp-method}) for \(S^k\varphi{}\) with
\(p\in (0,2)\). To perform energy estimates for \(S^k\varphi{}\), we must
control the coupling between the geometry and the scalar field arising from the
commutator \([\Box{},S^k]\).

Our proof uses an inductive argument, where we control derivatives of the
solution assuming control only of lower-order derivatives. In particular, we do
not use a bootstrap argument, which allows us to handle large data solutions.
The key ingredients that let us close the induction are:
\begin{enumerate}
\item Redshift effect for a subextremal black hole, as manifested by a uniform
lower bound for \(\varpi{} - \mathbf{e}^2/r\) along the horizon at late
times.
\item Monotonicity of \(\varpi{}-\mathbf{e}^2/r\) that propagates the lower bound in (1).
\item Energy decay and subsequent pointwise decay derived from \(r^p\)-weighted
energy estimates.
\item Reductive structure in the error terms arising from commutation.
\item Identification of weak and strong decay estimates required for geometric
quantities.
\end{enumerate}
The quantity \(\varpi{}\) we refer to is the renormalized Hawking mass, \(\mathbf{e}\)
is the charge parameter, and \(r\) is the area radius function, which appears in the
form of the metric in double null coordinates:
\begin{equation}
g = -\Omega{}^2\dd{}u\dd{}v + r^2g_{\mathbf{S}^2}.
\end{equation}
For further discussion of these quantities, see \cref{sec:emsf}.

Ingredients (1) and (2) are well-known and were used in
\cite{dafermos-rodnianski-price}. The techniques in (3) were introduced for
linear waves in \cite{rp-method} (and expanded in scope by \cite{moschidis-rp}),
but their implementation in our nonlinear spherically symmetric problem
(including commutation with weighted vector fields) is novel.

The main innovation of this work, which allows us to handle large data, is
ingredient (4). To capture what we call the reductive structure in (4), we
introduce two unweighted vector field commutators: an ingoing null vector field
\(U\) (``the global redshift vector field''), and a vector field \(V\) that is
outgoing null for large \(r\). We order products of the vector fields \(U\),
\(V\), and \(S\), which we call \(\Gamma^\alpha{}\) for a multi-index
\(\alpha{}\), in such a way that \([\Box{},\Gamma{}^\alpha{}]\) includes only
lower order terms and terms that are either small or have a good sign due to the
redshift effect (see Step 2 in \cref{sec:main-theorem-proof} as well as
\cref{wave-comm-formula,wave-comm-formula-UVS}). See
\cref{sec:intro-commutator-vector-fields} for an informal discussion of the
``reductive structure,'' and see Step 2 of \cref{sec:main-theorem-proof} for a more
precise discussion.

Ingredient (5) is important because in order to establish energy estimates at
order \(\alpha{}\) in (4), one needs control of geometric quantities of order
\(\alpha{}\) with weights in \(r\) (what we call ``weak'' estimates) as well as
control of geometric quantities of order \(< \alpha{}\) with weights in both
\(r\) and a time parameter \(\tau{}\) (what we call ``strong'' estimates). In
particular, one must first establish the weak geometric
estimates, then establish energy estimates, and then establish the strong
geometric estimates.
\subsubsection{Commutator vector fields}
\label{sec:org7ac66e0}
\label{sec:intro-commutator-vector-fields}
We use three commutator vector fields, which we call \(U\), \(V\), and \(S\).
They are defined as follows:
\begin{equation*}
U \coloneqq{} \frac{1}{(-\partial{}_ur)}\partial{}_u,\qquad V\coloneqq{} \chi{}_{r\lesssim \Rc}(r)\underline{\partial{}}_v + (1-\chi{}_{r\lesssim \Rc}(r))\overline{\partial}_r, \qquad S\coloneqq{} \chi{}_{r\lesssim \Rc}(r)v\underline{\partial{}}_v + (1-\chi{}_{r\lesssim \Rc}(r))(u\overline{\partial}_u + r\overline{\partial}_r).
\end{equation*}
Here \(\chi_{r\lesssim \Rc}(r)\) is a cutoff function supported in \(\set{r\le \Rc}\)
(for a large parameter \(\Rc > 0\) to be chosen in the course of the proof),
\((\underline{\partial{}}_r,\underline{\partial{}}_v)\) are the coordinate
derivatives in the \((r,v)\) coordinates, and
\((\overline{\partial}_u,\overline{\partial}_r)\) are the coordinate derivatives
in the Bondi--Sachs \((u,r)\) coordinates.

We now informally explain the ``reductive structure'' identified in ingredient (4)
of \cref{sec:proof-ingredients} in terms of the commutation formulas for our
commutator vector fields (see Step 2 in \cref{sec:main-theorem-proof} for a more
careful discussion). For this discussion we ignore the coupling between the
geometry and the scalar field.
\begin{enumerate}
\item (The commutator vector field \(U\))
\label{sec:org50957cd}
We use the redshift vector field \(U\) as a global commutator (see
\cite{MR3098640} for such a use on Schwarzschild). As is well-known, one can
commute with the redshift vector field even though \([\Box{},U]\) contains a
\(\partial{}U\) term, because this term comes with a good sign (which is a
manifestation of the redshift effect). We can commute with the redshift vector
field globally even in our nonlinear setting because in spherical symmetry, the
good sign that makes the redshift vector field a useful commutator near the
horizon is \emph{propagated} globally by the equations.
\item (The commutator vector field \(V\))
\label{sec:org950f222}
On Schwarzschild, our construction of \(V\) specializes to the timelike Killing
vector field \(\partial_t\) in a region of finite \(r\), and to
\(\frac{1}{1-2M/r}\partial_v\) in a region of large \(r\). The vector fields
\(\underline{\partial{}}_v\) and \(\overline{\partial}_r\) satisfy the commutation formula
\begin{equation}\label{intro:dv-commutation}
[\Box{},\underline{\partial{}}_v] = O(r^{-1})\Box{} + O(r^{-1})\partial{}U + O(r^{-2})\partial{},
\end{equation}
and
\begin{equation}
[\Box{},\overline{\partial}_r] = O(r^{-1})\Box{} + O(r^{-2})\overline{\partial}_r^2 + O(r^{-2})\partial{}.
\end{equation}
One now computes the commutation formula for \(V\):
\begin{equation}
[\Box{},V] = O(r^{-1})\Box{} + \mathbf{1}_{r\ge \Rc}O(r^{-2})\partial{}V +  \mathbf{1}_{r\le \Rc}\partial{}U + O(r^{-2})\partial{}
\end{equation}
The structure of the second order terms is crucial. Since \(V\) is only equal to
\(\underline{\partial{}}_v\) in a finite-\(r\) region, one does not worry about
the borderline \(r^{-1}\)-weight on the \(\partial{}U\) term in
\cref{intro:dv-commutation}; indeed, one can absorb this term using the bulk energy
associated to \(U\), which has already been controlled. The \(\partial{}V\) term
looks problematic, although it has a good \(r^{-2}\) weight, because a
\(\partial{}V\) term also appears on the left side of the energy estimate
associated to \(V\). However, the \(\partial{}V\) error term is supported in a
region of large \(r\), and the \(r^{-2}\) weight decays faster than the
\(r^{-1 + \epsilon{}}\) weight in the bulk term on the left, so this term can
be absorbed to the left of the energy estimate associated to \(V\).
\begin{remark}[Alternative choices for \(V\)]
Observe that \(V\) is gauge-invariant for large \(r\). It might seem natural to
construct \(V\) such that it is equal to the gauge-invariant Kodama vector field
\(T = \frac{1-\mu{}}{\partial_vr}\partial_v +
\frac{1-\mu{}}{(-\partial_ur)}\partial_u\) near the horizon, rather than to our
choice of \(\underline{\partial{}}_v = \frac{\partial_vr}{1-\mu{}}T\). In fact,
\([\Box{},T]\) contains a \(T^2\) term, which means that such a construction of
\(V\) would produce a \(\partial{}V\) term in \([\Box{},V]\) that is supported
near the horizon, and so cannot be absorbed as above.
\end{remark}
\item (The commutator vector field \(S\))
\label{sec:org80a2822}
On Schwarzschild, our construction of \(S\) specializes to \(vT\) in a finite
\(r\)-region-and to \(u\overline{\partial}_u + r\overline{\partial}_r\) in a
region of large \(r\), where \((\overline{\partial}_u,\overline{\partial}_r)\)
are the coordinate derivatives in the Bondi--Sachs \((u,r)\) coordinates. This
vector field is not conformally Killing, but we schematically have
\begin{equation}
[\Box{},S] = (2+O(r^{-1}))\Box{} + O(r^{-1+\epsilon{}})\overline{\partial}_r^2 +  \mathbf{1}_{r\le \Rc}\partial{}U + O(r^{-2+\epsilon{}})\partial{}.
\end{equation}
The last two terms are associated to energies lower in the hierarchy than \(S\),
and can be dealt with as before. The \(\overline{\partial}_r^2\) term has a
\(r^{-1+\epsilon{}}\) weight (which is too weak to be absorbed by the bulk
energy, which has an \(r^{-1-\epsilon{}}\) weight), but we can write
\(\overline{\partial}_r\psi{} =
r^{-1}\overline{\partial}_r(r\overline{\partial}_r\psi{}) -
r^{-1}\overline{\partial}_r\psi{}\). The term left to control, then, is
\(O(r^{-2+\epsilon{}})\overline{\partial}_r(r\overline{\partial}_r\psi{})\),
which can be written in the form
\(O(r^{-2+\epsilon{}})\overline{\partial}_r(rV\psi{})\) for large \(r\). This
term can then be estimated by the \(r^p\) bulk term associated to \(V\psi{}\), which
has already been controlled.
\end{enumerate}
\subsubsection{Gauge choice}
\label{sec:orgcccd6fb}
In \cref{main-theorem-rough}, we use a double null gauge \((u,v)\) in which
\(\frac{(-\partial_ur)}{1-\mu{}} = 1\) at null infinity and the quantity
\(\kappa{}\coloneqq{} \frac{\partial{}_vr}{1-\mu{}}\) is \(1\) along a curve of
constant \(r\), namely \(\set{r=r_{\mathcal{H}}}\), where \(r_{\mathcal{H}} =
\sup_{\mathcal{H}}r\) is the limiting value of the area radius function along
the event horizon. On dynamical spacetimes, our gauge differs from the
Eddington--Finkelstein-type \((u,\bar{v})\) gauge that is more common in the
literature (see for example \cite{dafermos-rodnianski-price,luk-oh-scc2}), in
which \(\frac{\partial{}_{\bar{v}}r}{1-\mu{}} = 1\) along the horizon.

The \((u,v)\) gauge is adapted to our commutator vector fields \(V\) and \(S\),
in the sense that \(V\kappa{}\) and \(S\kappa{}\) vanish on \(\set{r =
r_{\mathcal{H}}}\) (because in this region \(V\) and \(S\) are tangent to curves
of constant \(r\)). This enables an argument in which one integrates the
transport equation for \(S\kappa{}\) to \(\set{r=r_{\mathcal{H}}}\) and obtains
boundedness and decay for \(S\kappa{}\) globally (see Step 4 of
\cref{sec:main-theorem-proof}). On the other hand, the quantities
\(\bar{S}^n\bar{\kappa{}}\) (where the barred quantities are those associated to
the \((u,\bar{v})\) gauge) appear to grow along the horizon for sufficiently
large \(n\).
\subsubsection{Outline of the proof}
\label{sec:orgcefe292}
\label{sec:main-theorem-proof} We now describe the main estimates we establish, in
terms of the following nine schematic quantities (which we define precisely in
\cref{sec:norms}), which are parameterized by a multi-index \(\alpha{}\) (see
\cref{sec:multi-index-notation} for an explanation of our multi-index notation):
\begin{description}
\item[{\(\mathfrak{D}_N\)}] a gauge invariant \(r\)-weighted \(L^\infty\) norm on
characteristic initial data that controls the scalar field and its first
\(N+1\) derivatives.
\item[{\(\mathcal{D}_\alpha{}\)}] an \(r\)-weighted \(L^\infty\) norm on characteristic initial
data that controls \(\Gamma{}^{\le \alpha{}}\varphi{}\) and its first derivatives.
\item[{\(\mathfrak{b}_\alpha{}\) and \(\mathfrak{g}_\alpha{}\)}] ``weak'' geometric quantities
controlling \(\Gamma^{\le \alpha{}}\sigma{}\) with \(r\)-weights, where
\(\sigma{}\) is one of \(\varpi{}\),
\(\gamma{}\coloneqq{}\frac{\partial_ur}{1-\mu{}}\), or
\(\kappa{}\coloneqq{}\frac{\partial{}_vr}{1-\mu{}}\). The quantities
\(\mathfrak{b}_\alpha{}\) are shown to be bounded in the course of the proof,
while the quantities \(\mathfrak{g}_\alpha{}\) are only shown to grow slowly
in \(r\). Both \(\mathfrak{b}_\alpha{}\) and \(\mathfrak{g}_\alpha{}\) must be
controlled before the order-\(\alpha{}\) energy is controlled.
\item[{\(\mathcal{E}_\alpha{}\)}] an \(L^2\) energy norm that is non-degenerate at the
horizon and measures boundedness of a quantity \(\norm{r\partial{}\Gamma^{\le
  \alpha{}}\varphi{}}_{L^2(C)}\) on null curves \(C\).
\item[{\(\mathcal{E}_{\alpha{},p}\)}] an \(r^p\)-weighted \(L^2\) energy norm that controls
\(\mathcal{E}_\alpha{}\) and captures decay in \(\tau{}\) (see \cref{tau-def}) of the energy along
a suitable foliation. The quantity \(\mathcal{E}_{\alpha{},q}\) is stronger than the
quantity \(\mathcal{E}_{\alpha{},p}\) when \(p < q\).
\item[{\(\mathcal{P}_{\alpha{},p}\)}] an \(L^\infty\) norm that controls \(\Gamma^{\le \alpha{}}\varphi{}\) with
weights in both \(r\) and \(\tau{}\) (where the \(\tau{}\) weights capture the
decay that follows from \(r^p\)-weighted energy estimates), as well as
derivatives \(U\Gamma^{\le \alpha{}}\varphi{}\) and \(V\Gamma{}^{\le
  \alpha{}}\varphi{}\) with weights in \(r\). This norm includes a stronger
\(r\)-weight on outgoing null derivatives \(V\Gamma^{\le \alpha{}}\varphi{}\) than on ingoing null
derivatives. The quantity \(\mathcal{P}_{\alpha{},q}\) is stronger than the
quantity \(\mathcal{P}_{\alpha{},p}\) when \(p < q\).
\item[{\(\mathfrak{B}_\alpha{}\) and \(\mathfrak{G}_\alpha{}\)}] geometric quantities  stronger
than \(\mathfrak{b}_\alpha{}\) and \(\mathfrak{g}_\alpha{}\), respectively, that also
include \(\tau{}\)-weights on geometric quantities. The quantities
\(\mathfrak{B}_\alpha{}\) are shown to be bounded in the course of the proof,
while the quantities \(\mathfrak{G}_\alpha{}\) are only shown to grow slowly
in \(r\). Both \(\mathfrak{B}_\alpha{}\) and \(\mathfrak{G}_\alpha{}\) are controlled after
the order-\(\alpha{}\) energy is controlled.
\end{description}

In Steps 1--5, we will show that each quantity of order \(\alpha{}\) above is bounded
by lower order quantities and quantities of order \(\alpha{}\) that are higher
up on the list. We remark that the proofs of Steps 1--4 are independent from one
another.

\step{Step 0: Reduction to a characteristic problem
(\cref{sec:reduction-to-characteristic}).} We first restrict the region of interest
from the future of a Cauchy surface to a characteristic rectangle, which reduces
the proof of \cref{main-theorem-rough} to the proof of \cref{main-theorem}. This is
possible since \eqref{einstein-maxwell} is globally well posed by
\cite{Kommemi2013TheGS} and the data we consider is compactly supported.
Moreover, the subextremalilty of the black hole region and monotonicity
properties of \cref{einstein-maxwell} allow us to assume that the redshift
parameter \(\varpi{} - \mathbf{e}^2/r\) is uniformly positive in our
characteristic rectangle. In this rectangle we work with a future-normalized
double null gauge (see \cref{sec:coordinate-systems}).

\step{Step 1: Boundedness of initial data (\cref{sec:boundedness-initial-data}).} In this
step we show
\begin{equation}
\mathcal{D}_\alpha{}\le C(\mathfrak{b}_{<\alpha{}},r^{-1}\mathfrak{g}_{<\alpha{}})\mathfrak{D}_{\abs{\alpha{}}},
\end{equation}
which controls the gauge dependent data quantity \(\mathcal{D}_\alpha{}\) by the gauge invariant
norm \(\mathfrak{D}_{\abs{\alpha{}}}\).

\step{Step 2: Energy boundedness (\cref{sec:energy-estimates}).} The conclusion
of this step is an energy boundedness statement
\begin{equation}\label{intro:energy-boundedness}
\mathcal{E}_\alpha{}\le C(\eta{}_0,\mathfrak{b}_\alpha{},r^{-s}\mathfrak{g}_\alpha{},\mathcal{P}_{<\alpha{},1+\eta{}_0})[\mathcal{D}_\alpha{} + \mathcal{E}_{<\alpha{},4s}]
\end{equation}
for small \(s > 0\) and a small parameter \(\eta_0 > 0\). The constant on the right blows up as
\(\eta_0\) goes to \(0\).

The main difficulty in establishing \cref{intro:energy-boundedness} is to handle
the error terms arising from commutation. To illustrate how these terms are
treated, we will explain how to obtain \cref{intro:energy-boundedness} in the case
\(\abs{\alpha{}} = 1\). We specialize our discussion to an exterior region of a
Schwarzschild spacetime of mass \(M\), in the standard double null coordinates
\(u = t - r_\ast{}\) and \(v = t + r_\ast{}\), where \(r_\ast{} = r +
2M\log(r-2M)\). We use the following vector field multipliers to derive an
energy boundedness and integrated local energy decay estimate:
\begin{align}
T &=  \partial{}_u + \partial{}_v & \text{(Kodama vector field)}, \\
X &= f(r)(\partial{}_u - \partial{}_v) &\text{(Morawetz vector field)}, \\
Y &= \frac{\chi{}_{\mathcal{H}}(r)}{1-2M/r}\partial{}_u & \text{(Redshift vector field)}, \\
Z &=  e^{-M(u-v)}\partial{}_v &\text{(Irregular vector field)}.
\end{align}
Here \(\chi_{\mathcal{H}}(r)\) in the definition of \(Y\) is a cutoff function
localized to the horizon. The vector fields \(T\), \(X\), and \(Y\) are used in
a standard way to derive an energy estimate, while the irregular
vector field \(Z\) is used as in \cite{luk-oh-scc2} to generate a good bulk term
that helps control a certain quartic error term.

Let \(\varphi{}\) be a solution of the wave equation \(\Box{}\varphi{} = 0\),
and let \(\psi{}\) be a general function. Consider a spacetime region
\(\mathcal{R}\) with past boundary \(\Sigma_1\) and future boundary \(\Sigma_2\). We
derive an energy boundedness and integrated local energy decay estimate of the form
\begin{equation}\label{intro:energy-estimate}
E[\psi{},\Sigma{}_2] + E_{\text{bulk}}[\psi{},\mathcal{R}] \lesssim_{\eta{}_0} E[\psi{},\Sigma{}_1] +\underbrace{ \iint_{\mathcal{R}} W\psi{}\Box{}\psi{}\cdot r^2\Bigl(1-\frac{2M}{r}\Bigr)\dd{}u\dd{}v}_{\coloneqq{}\mathcal{E}[\psi{}]} +\cdots{}
\end{equation}
for an energy quantity \(E[\psi{},\Sigma{}]\) non-degenerate at the horizon whose integrand has the schematic
form \(r^2(\partial{}\psi{})^2\), a bulk term \(E_{\text{bulk}}[\psi{},\Sigma{}]\) with the control
\begin{equation}\label{intro:E-bulk}
\iint_{\mathcal{R}} r^{-1+\eta{}_0} (\partial{}\psi{})^2 r^2\Bigl(1-\frac{2M}{r}\Bigr)\dd{}u\dd{}v\lesssim E_{\text{bulk}}[\psi{},\mathcal{R}],
\end{equation}
where \(\eta_0 > 0\) is a small parameter, and a vector field multiplier \(W =
c_TT + c_XX + c_YY\) for \(c_T,c_X,c_Y > 0\) and \(f(r) > 0\) in the definition
of \(X\). We have written \(\cdots{}\) in \cref{intro:energy-estimate} to denote
error terms that we ignore in this discussion, since most of them can be dealt
with as in \cite{luk-oh-scc2}.

Observe that the error term \(\mathcal{E}[\varphi{}]\) vanishes (since \(\Box{}\varphi{} =
0\)), so an energy boundedness statement and an integrated local energy decay
statement for \(\varphi{}\) follows immediately from \cref{intro:energy-estimate}.
We now explain how to control the error term \(\mathcal{E}[\Gamma{}\varphi{}]\)
for \(\Gamma{}\in \set{U,V,S}\), in order to close the energy estimate
\cref{intro:energy-estimate} for \(\psi{} = \Gamma{}\varphi{}\) and establish
\cref{intro:energy-boundedness} in the case \(\abs{\alpha{}} = 1\).

\step{Step 2a: Commutation with \(U\).} Note that the \(\partial_u\) coefficient on the vector
field \(W\) appearing in \eqref{intro:energy-estimate} is positive, because the
same is true for each of the vector fields \(T\), \(X\), and \(Y\). We now explain
how this sign plays a role in our commutation scheme. Let \(U =
\frac{1}{1-2M/r}\partial_u\) be the global redshift vector field. We have
\begin{equation}\label{intro:U-comm}
\Box{}U\varphi{} = -f_U U^2\varphi{} + O(r^{-2})\partial{}\varphi{}.
\end{equation}
for a function \(f_U\) satisfying \(f_U > 0\) and \(f_U = O(r^{-2})\). On
Schwarzschild, the value of \(f_U\) is \(2M/r^2\); in the
general problem, \(f_U = 2r^{-2}(\varpi{}-\mathbf{e}^2/r)\), which is
bounded below by a positive multiple of \(r^{-2}\) (see Step 0). We now explain
how to establish energy estimates for \(U\varphi{}\) given energy estimates for
\(\varphi{}\) (and control of geometric quantities of order \(U\)). Write \(W =
W_UU + W_v\partial{}_v\), so that, by \cref{intro:U-comm}, the first error term in
energy estimate \cref{intro:energy-estimate} applied to \(\psi{} = U\varphi{}\) is
\begin{equation}
\mathcal{E}[U\varphi{}] = \iint_{\mathcal{R}} WU\varphi{}\Box{}U\varphi{} =  -\iint_{\mathcal{R}} W_Uf_U(U^2\varphi{})^2 + \iint_{\mathcal{R}} O(r^{-2})\partial{}_vU\varphi{}U^2\varphi{} +  \iint_{\mathcal{R}}O(r^{-2})\partial{}\varphi{}\partial{}U\varphi{},
\end{equation}
where we have omitted the volume form \(r^2(1-2M/r)\dd{}u\dd{}v\). The first term on
the right has a good sign and can be neglected. One can use the wave equation to
rewrite the second term into the form of the third term. The third term can be
handled with Young's inequality and the bulk term in the already established
energy estimate for \(\varphi{}\).

\step{Step 2b: Commutation with \(V\).} The vector field \(V\) is \(\overline{\partial}_r\) in
Bondi--Sachs \((u,r)\) coordinates (outgoing null) for large \(r\), and timelike
for small \(r\). Assuming that we have controlled geometric quantities of order
\(V\), we have a commutation formula
\begin{equation}\label{intro:V-comm}
\Box{}V\varphi{} = \mathbf{1}_{r\ge \Rc}O(r^{-2})\partial{}V\varphi{} + O_{\Rc}(r^{-2})[\partial{}\varphi{} + \partial{}U\varphi{}],
\end{equation}
where \(\Rc > 0\) is a parameter that can be chosen large. The term in \(\mathcal{E}[V\varphi{}]\) arising from
the second term on the right of \cref{intro:V-comm} can be dealt with as in Step
2a, using Young's inequality and the bulk terms associated to \(\varphi{}\) and
\(U\varphi{}\), which are considered lower order than \(V\varphi{}\) (and so
have already been controlled by the time one seeks control of \(V\varphi{}\)). The first
term on the right of \cref{intro:V-comm} is not small, so the error term it
produces in \(\mathcal{E}[V\varphi{}]\), namely
\begin{equation}\label{intro:V-error}
\iint_{\mathcal{R}\cap \set{r\ge \Rc}} O(r^{-2})(\partial{}V\varphi{})^2,
\end{equation}
is potentially problematic. However, the integral in \cref{intro:V-error} is taken
only over a large-\(r\) region, and the \(r\)-weight in the integrand decays
faster than in the bulk energy (see \cref{intro:E-bulk}). Thus the error term
\cref{intro:V-error} can be bounded by a small multiple of
\(E_{\text{bulk}}[V\varphi{},\mathcal{R}]\) by choosing \(\Rc\) large, and it
can be absorbed into the left side of the energy estimate
\cref{intro:energy-estimate}.

\step{Step 2c: Commutation with \(S\).} At this stage of the argument we have performed
energy estimates---including \(r^p\)-weighted estimates---for \(\varphi{}\),
\(U\varphi{}\), and \(V\varphi{}\), and we have obtained control of geometric
quantities associated to \(S\) in terms of lower order energies. The vector
field \(S\) is \(u\overline{\partial}_u + r\overline{\partial}_r\) in
Bondi--Sachs \((u,r)\) coordinates for large \(r\) and
\(v\underline{\partial{}}_v\) in \((v,r)\) coordinates near the horizon. We have a commutation formula
\begin{equation}\label{intro:S-comm}
\Box{}S\varphi{} = \mathbf{1}_{r\ge R}O(r^{-2+s})\partial{}_v(rV\varphi{}) + O(r^{-2+s})[\partial{}\varphi{} + \partial{}U\varphi{} + \partial{}V\varphi{}]
\end{equation}
for a large \(R > 0\) and small \(s > 0\). The term in \(\mathcal{E}[S\varphi{}]\) arising from the
first term on the right of \cref{intro:S-comm} is controlled (using Young's
inequality) by the \(r^p\)-weighted bulk quantity (with \(p = \eta{}_0\))
associated to the lower order term \(V\varphi{}\), as well as a small multiple
of the bulk energy \(E_{\text{bulk}}[S\varphi{},\mathcal{R}]\). We treat the
terms in \(\mathcal{E}[S\varphi{}]\) arising from the second term on the right
of \cref{intro:S-comm} as in Steps 2a and 2b.

\step{Step 3: Energy decay (\cref{sec:energy-estimates}).} The conclusion of this step is
an energy decay statement
\begin{equation}\label{intro:energy-decay}
\mathcal{E}_{\alpha{},2-\eta{}_0-C_\alpha{}s} \le C(\mathfrak{b}_{\alpha{}},r^{-s}\mathfrak{g}_{\alpha{}},\mathcal{P}_{\alpha{},0},\mathcal{P}_{<\alpha{},1+\eta{}_0})\mathcal{D}_\alpha{},
\end{equation}
for \(\eta_0 > 0\) as in Step 2, \(s > 0\) small (depending on \(\alpha{}\)), and explicit
constants \(C_\alpha{}\) depending on \(\alpha{}\). We adapt the method of
\(r^p\)-weighted energy estimates due to Dafermos--Rodnianski \cite{rp-method} to
establish energy decay along a suitable foliation.

We now explain why the maximum value of \(p < 2\) that we can take when
controlling \(\Gamma{}^\alpha{}\varphi{}\) depends on \(\alpha{}\) (namely \(p
\le 2-\eta_0-C_\alpha{}s\)). The \(r^p\) estimate for \(\psi{}\) contains an error
term of the form
\begin{equation}
\iint_{\set{r\ge R}} r^{p+3}\abs{\Box{}\psi{}}^2\dd{}u\dd{}v
\end{equation}
for a large \(R > 0\). When \(\psi{} = S\varphi{}\), this produces terms of the form
\begin{equation}
\iint_{\set{r\ge R}} r^{p-1}(S\varpi{})^2(\partial{}_v\varphi{})^2\dd{}u\dd{}v \quad \text{and}\quad \iint_{\set{r\ge R}} r^{p-1}(rS(-\gamma{}))^2(\partial{}_v(rV\varphi{}))^2\dd{}u\dd{}v.
\end{equation}
At this stage of the argument, the terms \(S\varpi{}\) and \(rS(-\gamma{})\) can only be shown
to grow slowly in \(r\), at a rate \(r^s\) (see Steps 5ab). Thus the terms above become
\begin{equation}
\iint_{\set{r\ge R}} r^{p-1+2s}[(\partial{}_v\varphi{})^2 + (\partial{}_v(rV\varphi{})^2)]\dd{}u\dd{}v.
\end{equation}
When \(p\le 2-\eta{}_0-2s\), the first term can be controlled by the Morawetz estimate
corresponding to \(\varphi{}\), and the second term is controlled by the bulk
term in the \(r^{p + 2s}\)-energy associated to \(V\varphi{}\). In general, control of the
\(r^p\)-energy associated to \(\Gamma{}^\alpha{}\varphi{}\) requires \(p\le 2-\eta_0-2s\) and control of
the \(r^{p + 2s}\) energy associated to \(\Gamma^{<\alpha{}}\varphi{}\). An induction argument shows
that, assuming the \(r^{2-\eta_0}\)-energy of \(\varphi{}\) is controlled, one can control
the \(r^p\)-energy of \(\Gamma^\alpha{}\varphi{}\) for \(p = 2-\eta_0-C_\alpha{}s\) for a constant \(C_\alpha{}\)
depending on \(\alpha{}\).

\step{Step 4: Pointwise estimates (\cref{sec:pointwise-estimates}).} In this step we
control a pointwise norm of the scalar field by data and the \(r^p\)-weighted
energy norm:
\begin{equation}
\mathcal{P}_{\alpha{},p}\le C(\mathfrak{b}_\alpha{},r^{-s}\mathfrak{g}_\alpha{})(\mathcal{D}_\alpha{} + \mathcal{E}_{\alpha{},p})
\end{equation}
for small \(s > 0\) (depending on \(\alpha{}\)). The redshift effect plays a role in
this section.

\step{Step 5: Estimates for geometric quantities (\cref{sec:geometric-estimates}).} This
is the most technical step of the proof. We show that
\begin{align}
\mathfrak{b}_\alpha{} &\le C(\alpha{},\mathfrak{b}_{<\alpha{}},\mathfrak{B}_{<\alpha{}},r^{-s}\mathfrak{G}_{<\alpha{}},\mathcal{E}_{<\alpha{},1},\mathcal{P}_{<\alpha{},1+\eta{}_0}), \label{intro:b-alpha} \\
\mathfrak{B}_\alpha{} &\le C(\alpha{},\mathfrak{b}_{\alpha{}},\mathfrak{B}_{<\alpha{}},r^{-s}\mathfrak{G}_{<\alpha{}},\mathcal{E}_{<\alpha{},1},\mathcal{P}_{<\alpha{},1}), \label{intro:B-alpha} \\
r^{-s}\mathfrak{g}_\alpha{} &\le C(\alpha{},\mathfrak{b}_{\le \alpha{}},r^{-s}\mathfrak{G}_{<\alpha{}},\mathcal{E}_{<\alpha{},1},\mathcal{P}_{<\alpha{},2-s+\eta{}_0}), \label{intro:g-alpha} \\
r^{-s}\mathfrak{G}_\alpha{} &\le C(\alpha{},\mathfrak{B}_{\le \alpha{}},r^{-s}\mathfrak{g}_{\le \alpha{}},r^{-s}\mathfrak{G}_{<\alpha{}},\mathcal{E}_{<\alpha{},1},\mathcal{P}_{<\alpha{},2-s+\eta{}_0}). \label{intro:G-alpha}
\end{align}
See \cref{sec:schematic-geom-quantities} for the definitions of the schematic
geometric quantities. In this section we do not discuss the estimates for
\(\mathfrak{B}_\alpha{}\) and \(\mathfrak{G}_\alpha{}\), as they are simpler
than the estimates for \(\mathfrak{b}_\alpha{}\) and \(\mathfrak{g}_\alpha{}\).

\step{Step 5a: Estimates for \(S\varpi{}\).} Write \(\varpi{}\) for the renormalized Hawking mass
(see \cref{renormalized-hawking-mass}). We explain why \(S\varpi{}\) is in
\(\mathfrak{g}_S\) and not \(\mathfrak{b}_S\), namely why we can only show that
\(S\varpi{}\) grows slowly in \(r\) (and not that it is bounded). We have
\begin{equation}
\abs{S\varpi{}}\lesssim v\abs{V\varpi{}} + \tau{}\abs{U\varpi{}},
\end{equation}
where \(\tau{}=u\) for large \(r\) and \(\tau{} \sim v\) for small \(r\). When \(r\le \tau{}/2\),
the transport equations for \(\varpi{}\) in the \(u\)- and \(v\)-directions yield
\begin{equation}
v\abs{V\varpi{}} + \tau{}\abs{U\varpi{}}\lesssim r^2\tau{}(D\varphi{})^2,
\end{equation}
where we have written \(D\) to stand for \(U\) or \(V\). Control of the
pointwise norm of \(D\varphi{}\) (which comes from Step 4 and ultimately from
the \(r^p\)-estimates of Step 3) gives just under \(3/2\) powers of decay,
namely \(\mathcal{P}_{2-s}^2[D\varphi{}] \le \mathcal{P}_{<\alpha{},2-s}^2\)
controls \(r\tau^{2-s}\abs{D\varphi{}}^2\) and \(r^2\tau^{1-s-\eta_0}\). By
interpolation, \(\mathcal{P}_{<\alpha{},2-s}^2\) controls
\(r^{2-s-\eta_0}\tau{}\). This shows that \(\abs{S\varpi{}}\lesssim r^{s +
\eta_0}\mathcal{P}_{<\alpha{},2-s}\).

\step{Step 5b: Estimates for \(S(-\gamma{})\).} Write \((-\gamma{}) \coloneqq{} \frac{\partial_ur}{1-2m/r}\) (see
\cref{sec:sph-symmetric-equations} for the definition of the Hawking mass \(m\)).
We now explain why \(rS(-\gamma{})\) is in \(\mathfrak{g}_S\) and not
\(\mathfrak{b}_S\), namely why we can only show that \(rS(-\gamma{})\) grows slowly in
\(r\) (and not that it is bounded) before controlling the energy associated to
\(S\varphi{}\).

As in Step 5a, we have
\begin{equation}
\abs{S\log (-\gamma{})}\lesssim v\abs{V\log (-\gamma{})} + \tau{}\abs{U\log (-\gamma{})}.
\end{equation}
The second term can be controlled by the estimate \(r\tau{}\abs{U\log (-\gamma{})}\lesssim 1\)
contained in the boundedness of the lower order quantity \(\mathfrak{B}_U\). For
the first term, we use the \(\partial{}_v\)-transport equation for \(\log (-\gamma{})\) to find
that, in the region \(r\gg 1\),
\begin{equation}
v\abs{V\log (-\gamma{})}\lesssim rv(\partial{}_v\varphi{})^2.
\end{equation}
In the region \(r\le \tau{}/2\), we have \(\tau{}\sim v\), so this term can be handled exactly as
in Step 5a. In the region \(r\ge \tau{}/2\), we have \(r\sim v\), so this term
is handled using the fact that \(\partial_v\)-derivatives decay with two powers of
\(r\) (namely \(r^2\abs{\partial_v\varphi{}}\le \mathcal{P}_{0,0}\)).

\step{Step 5c: Estimates for \(\Gamma^\alpha{}\kappa{}\).} Write \(\kappa{} \coloneqq{}
\frac{\partial{}_vr}{1-2m/r}\) (see \cref{sec:sph-symmetric-equations} for the
definition of the Hawking mass \(m\)). We will explain how to show that
\(S\kappa{}\) is bounded (given control of lower order energies and geometric quantities).

The gauge condition on the coordinate \(v\) fixes
\(\kappa{}|_{\set{r=r_{\mathcal{H}}}} = 1\), where \(r_{\mathcal{H}} \coloneqq{}
\sup_{\mathcal{H}}r\) is the limiting value of the area radius along the
horizon. We will explain how to obtain a boundedness estimate for \(S\log
\kappa{}\) from control of lower order quantities. Assume that the area radius
\(r\) of the point at which we want to estimate \(S\log \kappa{}\) satisfies
\(r\ge r_{\mathcal{H}}\) (the case \(r\le r_{\mathcal{H}}\) is similar).
Ignoring error terms arising from commuting \(S\) past \(U\), we use the
transport equation \(U\log \kappa{} = -r(U\varphi{})^2\) to obtain
\begin{equation}\label{intro:geom-kappa}
\abs{US\log \kappa{}} \lesssim r\abs{U\varphi{}}\abs{US\varphi{}}\lesssim rv\abs{U\varphi{}}(\abs{UV\varphi{}} + \abs{UU\varphi{}}),
\end{equation}
where we have used a statement \(\abs{S\psi{}}\lesssim v(\abs{U\psi{}} + \abs{V\psi{}})\). Integrating
\cref{intro:geom-kappa} to the constant-\(r\) curve \(\set{r=r_{\mathcal{H}}}\) (noting that \(U = -\underline{\partial{}}_r\) in \((r,v)\)
coordinates and \(S\log \kappa{}|_{\set{r=r_{\mathcal{H}}}} = 0\)), one obtains
\begin{equation}\label{intro:S-kappa}
\abs{S\log \kappa{}}(r)\lesssim v\int_{r_{\mathcal{H}}}^r r'\abs{U\varphi{}}\abs{UD\varphi{}}\dd{}r',
\end{equation}
where we write \(D\) for either \(U\) or \(V\). One establishes
\begin{equation}\label{intro:U-energy-estimate}
\int_{r_{\mathcal{H}}}^r r'(U\psi{})^2\dd{}r'\lesssim v^{-1}\mathcal{E}_1[\psi{}],
\end{equation}
where \(\mathcal{E}_1\) is an \(r^p\)-weighted energy norm with \(p = 1\) (see
\cref{u-energy-estimate}). The method is to split the integral into the regions
where \(\set{r\le v/2}\) and \(\set{r\ge v/2}\). For the first term one uses the
decay captured by the \(r^p\)-weighted energy quantity. The second term is
integrated over a region where \(r\sim v\), so one can use the \(r^2\)-weight in
the energy quantity \(E[\psi{}]\) controlled in Step 2, which is one power
stronger than the \(r\)-weight in \cref{intro:U-energy-estimate}. Now the integral
in \cref{intro:S-kappa} can be controlled using Cauchy--Schwarz and
\cref{intro:U-energy-estimate}, which leads to
\begin{equation}
\abs{S\log \kappa{}}(r)\lesssim \mathcal{E}_1[\varphi{}] + \mathcal{E}_1[D\varphi{}].
\end{equation}
In this way one obtains boundedness of \(S\log \kappa{}\) from boundedness of lower
order energies.

\step{Step 6: Putting it all together (\cref{sec:putting-it-all-together}).} Together with
certain zeroth order estimates (see \cref{sec:zeroth-order-geometric}), an
induction argument involving the results of Steps 1--5 shows that all our
schematic quantities of order \(\alpha{}\) are controlled by the gauge invariant
initial data norm \(\mathfrak{D}_{\abs{\alpha{}}}\). In particular, the control
in the pointwise norm \(\mathcal{P}_{\alpha{},2-2\epsilon{}}\) (for \(\alpha{} = kS\)) gives the
desired \cref{S-phi-estimate}.

\step{Step 7: Late-time tails (\cref{sec:late-time-tails}).} To prove the late-time tails
result in \cref{tails-theorem}, we appeal to \cite[Main Theorem 4]{luk-oh-tails}.
\Cref{sec:late-time-tails} is dedicated to satisfying the assumptions of
\cite{luk-oh-tails}. The additional ingredient (on top of the estimates obtained
in Step 6) is control of \((r\overline{\partial}_r)\)-derivatives of the scalar
field and the geometry. The observation of \cite{luk-oh-tails} (see also
\cite{MR1374174}) is that a linear combination of the wave operator \(\Box{}_g\)
and powers of the scaling vector field \(S\) is elliptic, which allows one to
obtain \(L^2\)-control of \((r\overline{\partial}_r)\)-derivatives of the scalar
field from \(L^2\)-control of \(S\)-derivatives of the scalar field. To upgrade
this to \(L^\infty\) control, we use the Klainerman's Sobolev type estimate in
\cite[Sec.~5.4]{luk-oh-tails}. These estimates are coupled to weak estimates for
the geometric quantities. Once we obtain control of
\((r\overline{\partial}_r)\)-derivatives of the scalar field in \(L^\infty\), we
establish stronger estimates, including asymptotics, for geometric quantities.
\subsection*{Acknowledgements}
The author thanks their advisor Jonathan Luk for introducing them to the problem
and for many useful discussions and comments on the manuscript. The author also
thanks Maxime Van de Moortel for his interest in this work. This project began
as the author's undergraduate thesis at Stanford University, and much of the
early work was carried out at the SURIM program. This material is based upon
work supported by the National Science Foundation Graduate Research Fellowship
under Grant No. DGE-2039656.
\section{Preliminaries}
\label{sec:orgfb6ad2c}
\subsection{Einstein--Maxwell--scalar field system in spherical symmetry}
\label{sec:orgf69cac0}
\label{sec:emsf}
\subsubsection{Spherically symmetric solutions}
\label{sec:org43eef39}
A \(4\)-dimensional Lorentzian manifold \((\mathcal{M},g)\) is called a
spherically symmetric spacetime if \(\mathcal{M} = \mathcal{Q}\times S^2\) for a
\(2\)-dimensional Lorentzian manifold with boundary \((\mathcal{Q},g_{\mathcal{Q}})\) that has
a global chart \((u,v)\) of ``double null coordinates,'' in which its metric takes
the form
\begin{equation}
g_{\mathcal{Q}} = -\Omega{}^2\dd{}u\dd{}v,
\end{equation}
and
\begin{equation}
g = g_{\mathcal{Q}} + r^2g_{\mathbf{S}^2}.
\end{equation}
Here \(g_{\mathbf{S}^2}\) is the round metric on the unit sphere \(r :
\mathcal{Q}\to [0,\infty)\) is the area-radius function, defined so that if
\(\pi{} : \mathcal{M}\to \mathcal{Q}\) is the quotient map, then
\(\pi{}^{-1}(p)\subset \mathcal{M}\) has area equal to that of a round sphere of radius \(r(p)\).
Observe that a spherically symmetric spacetime admits an action of \(\SO(3)\) by
isometries.

A solution \((\mathcal{M},g,\varphi{},F)\) to \cref{einstein-maxwell} is called spherically
symmetric if \((\mathcal{M},g)\) is a spherically symmetric spacetime,
\(\varphi{}\) is invariant under the \(\SO(3)\) action, \(F\) is invariant under
pullback by the \(\SO(3)\) action, and there exists \(\mathbf{e} : \mathcal{Q}\to
\R\) such that
\begin{equation}
F = \frac{\mathbf{e}}{2(\pi{}^\ast{}r)^2}\pi{}^\ast{}(\Omega{}^2\dd{}u\wedge{}\!\dd{}v).
\end{equation}
The equations imply that \(\mathbf{e}\) is a constant (see \cite[\S{}2]{dafermos03}).
\subsubsection{Equations in spherical symmetry}
\label{sec:org0ce2953}
\label{sec:sph-symmetric-equations}
Let \((u,v)\) be a double null coordinate system on \(\mathcal{Q}\). In
spherical symmetry, the Einstein--Maxwell--scalar field system becomes a system
of wave equations for \((r,\varphi{},\Omega{})\):
\begin{equation}
\begin{cases}
\vspace{1.5ex}\hfill \partial{}_u\partial{}_vr &= -\dfrac{\Omega{}^2}{4r} - \dfrac{\partial{}_ur\partial{}_vr}{r} + \dfrac{\Omega{}^2\mathbf{e}^2}{4r^3}, \\
\vspace{1.5ex}\hfill \partial{}_u\partial{}_v\varphi{} &= - \dfrac{\partial{}_vr\partial{}_u\varphi{}}{r} - \dfrac{\partial{}_ur\partial{}_v\varphi{}}{r}, \\
\hfill \partial{}_u\partial{}_v\log \Omega{} &= -2\partial{}_u\varphi{}\partial{}_v\varphi{} - \dfrac{\Omega{}^2\mathbf{e}^2}{r^4} + \dfrac{\Omega{}^2}{2r^2} + \dfrac{2\partial{}_ur\partial{}_vr}{r^2}.
\end{cases}
\end{equation}
We also have the following Raychaudhuri equations:
\begin{equation}\label{raychaudhuri-equations}
\begin{cases}
\vspace{2ex}\partial{}_v\Bigl(\dfrac{\partial{}_vr}{\Omega{}^2}\Bigr) = -\dfrac{r(\partial{}_v\varphi{})^2}{\Omega{}^2}, \\
\partial{}_u\Bigl(\dfrac{\partial{}_ur}{\Omega{}^2}\Bigr) = -\dfrac{r(\partial{}_u\varphi{})^2}{\Omega{}^2}.
\end{cases}
\end{equation}

We will work with a different set of equations written in terms of the Hawking
mass. Introduce the following notation for null derivatives of \(r\):
\begin{equation}
\lambda{} \coloneqq{} \partial_vr\qquad \nu{} \coloneqq{} \partial_ur.
\end{equation}
Define the Hawking mass \(m : \mathcal{Q}\to \R\) by
\begin{equation}\label{hawking-mass}
m \coloneqq{} \frac{r}{2}(1 - g_{\mathcal{Q}}(\Grad r,\Grad r)) = \frac{r}{2}\Bigl(1 + \frac{4\lambda{}\nu{}}{\Omega{}^2}\Bigr),
\end{equation}
as well as the renormalized Hawking mass
\begin{equation}\label{renormalized-hawking-mass}
\varpi{} \coloneqq{} m + \frac{\mathbf{e}^2}{2r}.
\end{equation}
The wave operator \(\Box{}\coloneqq{}-4\Box_g\) takes the form
\begin{equation}
\Box{} = \frac{1-\mu{}}{\lambda{}(-\nu{})}\Bigl(\partial{}_u\partial{}_v + \frac{\lambda{}}{r}\partial{}_u + \frac{\nu{}}{r}\partial{}_v\Bigr).
\end{equation}
Define
\begin{equation}\label{mu-definition}
\mu{} \coloneqq{} \frac{2m}{r} = \frac{2\varpi{}}{r} - \frac{\mathbf{e}^2}{r^2}
\end{equation}
and set
\begin{equation}
\kappa{} \coloneqq{} \frac{\lambda{}}{1-\mu{}}\qquad \gamma{} \coloneqq{} \frac{\nu{}}{1-\mu{}}.
\end{equation}
When \(1-\mu{}\neq{}0\), a spherically symmetric solution satisfies
\begin{equation}\label{sph-sym-equations-1}
\begin{cases}
\hfill \partial{}_u\partial{}_vr &= \partial{}_u\lambda{} = \partial{}_v\nu{} = \dfrac{2(\varpi{}-\mathbf{e}^2/r)}{r^2} \dfrac{\lambda{}\nu{}}{1-\mu{}} \\
\hfill \Box{}\varphi{} &= 0, \\
\hfill \partial{}_u\varpi{} &= -\dfrac{r^2}{2(-\gamma{})}(\partial{}_u\varphi{})^2, \\
\hfill \partial{}_v\varpi{} &= \dfrac{r^2}{2\kappa{}}(\partial{}_v\varphi{})^2, \\
\end{cases}
\end{equation}
as well as
\begin{equation}\label{sph-sym-equations-2}
\begin{cases}
\vspace{1ex}\hfill\partial{}_u\kappa{} &= -\dfrac{r}{(-\nu{})}(\partial{}_u\varphi{})^2\kappa{} \\
\partial{}_v(-\gamma{}) &= \dfrac{r}{\lambda{}}(\partial{}_v\varphi{})^2(-\gamma{}).
\end{cases}
\end{equation}
\subsection{Admissible initial data}
\label{sec:org06d6784}
\label{sec:admissible-data} See \cite[\S{}2]{luk-oh-scc1} for the definition of
(asymptotically flat) future admissible data that we use.
\subsection{Mass inflation criterion}
\label{sec:org2c2e414}
\label{sec:mass-inflation-and-inextendibility} We now explain how to obtain the
criterion \cref{luk-oh-sr-mass-inflation} for mass inflation from the work of
Luk--Oh--Shlapentokh-Rothman \cite{luk2022scattering}.
\begin{proof}[Proof of \cref{luk-oh-sr-mass-inflation} from \cite{luk2022scattering}]
Fix a solution to \cref{einstein-maxwell} arising from admissible data and let \(I_{p,k}\) be the statement that
\begin{equation}
\int_{\mathcal{H}\cap \set{v\ge 1}} v^p(\partial{}_v^k\varphi{})^2 < \infty.
\end{equation}
Our goal is to show the statement in \cref{luk-oh-sr-mass-inflation}, namely that
\cref{L2-lower-bound} together with \(I_{4,1}\) and \(I_{8,k}\) for some \(k\ge 2\)
implies mass inflation. Inspecting the proof of \cite[Thm.~7.8]{luk2022scattering}
(in particular the role of \cite[Cor.~7.2]{luk2022scattering} in establishing
\cite[Thm.~7.14]{luk2022scattering}) and noting \cite[Rem.~4.3]{luk2022scattering}
shows that mass inflation holds whenever \cref{L2-lower-bound} does if one of the
following statements holds:
\begin{enumerate}
\item \(I_{6,1}\) is false, but \(I_{4,1}\) is true, and \(I_{6,k}\) is true for
some \(k\ge 2\),
\item or \(I_{6,1}\) is true and \(I_{8,k}\) is true for some \(k\ge 2\).
\end{enumerate}
Hence the statement in \cref{luk-oh-sr-mass-inflation} holds, since \(I_{4,1}\) and
\(I_{8,k}\) for some \(k\ge 2\) imply one of the above statements (regardless of
the truth of \(I_{6,1}\)).
\label{mass-inflation-from-main-theorem}
\end{proof}
\begin{remark}
In fact \(I_{6,1}\) is true, in view of the late-time tails result in
\cref{tails-theorem}, but we provide this proof to emphasize that the non-sharp
decay established by \cref{main-theorem-rough} suffices to establish mass inflation.
\end{remark}
\subsection{Coordinate systems}
\label{sec:org3b22535}
\label{sec:coordinate-systems}
We will use four coordinate systems on \(\mathcal{Q}\):
\begin{description}
\item[{Double null gauge normalized on a curve of constant \(r\)}] The double null coordinates
\((u,v)\) defined by the normalizations
\(\frac{(-\partial{}_ur)}{1-\mu{}}|_{\mathcal{I}} =
  \frac{\partial{}_vr}{1-\mu{}}|_{\set{r{=r_{\mathcal{H}}}}} = 1\) (where
\(r_{\mathcal{H}} \coloneqq{} \sup_{\mathcal{H}} r\)), as well as \(u \equiv
  1\) on \(C^{\mathrm{out}}\) and \(v\equiv 1\) on \(C^{\text{in}}\). We write
\((\partial{}_u,\partial{}_v)\) for the coordinate derivatives.
\item[{Double null gauge normalized at the horizon (Eddington--Finkelstein-type)}] The double null coordinates \((u,\bar{v})\), where \(u\) is defined as in the
previous gauge and \(v\) is defined by
\(\frac{\partial_vr}{1-\mu{}}|_{\mathcal{H}} = 1\) and \(v\equiv 1\) on
\(C^{\text{in}}\). We write \((\partial_u,\partial_{\bar{v}})\) for the
coordinate derivatives. In terms of the \((u,v)\) gauge, we have
\(\partial_{\bar{v}}|_{(u,v)} = \kappa{}^{-1}|_{\mathcal{H}}(v)\partial_v|_{(u,v)}\).
\item[{\(\mathcal{I}\)-gauge}] The \((u,r)\)-coordinates, with coordinate
derivatives \((\overline{\partial}_u,\overline{\partial}_r)\). In terms of the \((u,v)\) gauge, we have \(\overline{\partial}_u = \partial_u + \frac{(-\nu{})}{\lambda{}}\partial_v\) and \(\overline{\partial}_r = \frac{1}{\lambda{}}\partial_v\).
\item[{\(\mathcal{H}\)-gauge}] The \((r,v)\)-coordinates, with coordinate
derivatives \((\underline{\partial{}}_r,\underline{\partial{}}_v)\). In terms of the \((u,v)\)
gauge, we have \(\underline{\partial{}}_r = \frac{1}{\nu{}}\partial_u\) and \(\underline{\partial{}}_v =
  \partial_v + \frac{\lambda{}}{(-\nu{})}\partial_u\).
\end{description}
\begin{remark}
The Eddington--Finkelstein-type gauge normalized so that
\(\frac{\partial_vr}{1-\mu{}}|_{\mathcal{H}} = 1\) is more common in the literature (see
for instance \cite{dafermos-rodnianski-price,luk-oh-scc2}) than the gauge
normalized on a curve of constant \(r\). We remark that a double null gauge
normalized on a curve of constant \(r\) is used in the recent work \cite{angelopoulos2024nonlinear}.
\end{remark}
The wave operator in these coordinates takes the following forms:
\begin{lemma}
In a double null gauge, we have
\begin{equation}\label{box-uv-gauge}
\Box{} = \frac{1}{\kappa{}(-\nu{})}\Bigl(\partial{}_u\partial{}_v + \frac{\lambda{}}{r}\partial{}_u + \frac{\nu{}}{r}\partial{}_v\Bigr).
\end{equation}
In the \(\mathcal{I}\)-gauge, we have
\begin{equation}\label{box-I-gauge}
\Box{} = -(1-\mu{})\overline{\partial}_r^2 + \frac{1}{(-\gamma{})}\overline{\partial}_r\overline{\partial}_u - \Bigl(\frac{2}{r} - \frac{2\varpi{}}{r^2}\Bigr)\overline{\partial}_r + \frac{1}{r}\frac{1}{(-\gamma{})}\overline{\partial}_u.
\end{equation}
In the \(\mathcal{H}\)-gauge, we have
\begin{equation}\label{box-H-gauge}
\Box{} = -(1-\mu{})\underline{\partial{}}_r^2 - \frac{1}{\kappa{}}\underline{\partial{}}_r\underline{\partial{}}_v - \Bigl(\frac{2}{r} - \frac{2\varpi{}}{r^2}\Bigr)\underline{\partial{}}_r - \frac{1}{r}\frac{1}{\kappa{}}\underline{\partial{}}_v.
\end{equation}
\label{wave-in-gauge}
\end{lemma}
We introduce a ``time'' parameter \(\tau{}\) defined by
\begin{equation}\label{tau-def}
\tau{}(u,v) \coloneqq{} \begin{cases}
u_{R_0}(v) & r(u,v)\le R_0 \\
u & r(u,v)\ge R_0.
\end{cases}
\end{equation}
Here \(u_{R_0}(v)\) is the \(u\)-value such that \(r(u=u_{R_0}(v),v) = R_0\), and \(R_0\) is a quantity determined from characteristic initial data by
\begin{equation}
R_0\coloneqq{} r(1,1) = r|_{C^{\text{in}}\cap C^{\mathrm{out}}} = \sup_{C^{\text{in}}}r.
\end{equation}
Since the coordinates \(u,v\) are normalized so that \(u,v\ge 1\) in the future of
\(C^{\text{in}}\cup C^{\mathrm{out}}\), we have \(\tau{}\ge 1\). The weight
\(\tau{}\) will be used to quantify decay of energy along a suitable foliation
(\cref{sec:norms:energy}), of pointwise norms of the scalar field
(\cref{sec:norms:pointwise}), and of pointwise norms of the scalar field
(\cref{sec:norms:pointwise}).
\subsection{Commutator vector fields}
\label{sec:org2f17c61}
\label{sec:commutator-vector-fields} Let \(\Rc > 0\). The value of \(\Rc\) will be
fixed to be large in \cref{sec:energy-estimates}. Let \(\chi = \chi{}_{r\lesssim \Rc} :
[0,\infty)\to [0,1]\) be a smooth non-increasing function that equals \(1\) in
\(\set{r\le \Rc}\) and \(0\) in \(\set{r\ge 2\Rc}\) such that \(\abs{\chi^{(n)}(r)}\le
C_nr^{-n}\).

Define the commutator vector fields
\begin{equation}
U \coloneqq{} \frac{1}{(-\partial{}_ur)}\partial{}_u\qquad V\coloneqq{} \chi{}_{r\lesssim \Rc}(r)\underline{\partial{}}_v + (1-\chi{}_{r\lesssim \Rc}(r))\overline{\partial}_r \qquad S\coloneqq{} \chi{}_{r\lesssim \Rc}(r)v\underline{\partial{}}_v + (1-\chi{}_{r\lesssim \Rc}(r))(u\overline{\partial}_u + r\overline{\partial}_r).
\end{equation}
It is immediate that
\begin{equation}\label{Gamma-r-initial-calculation}
Ur = -1\qquad Vr = (1-\chi{}(r))\qquad Sr = r(1-\chi{}(r)).
\end{equation}
\subsection{Notation}
\label{sec:org620ecf8}
\subsubsection{Notation for functions}
\label{sec:org32215e1}
From this point on, fix a spherically symmetric solution
\((\mathcal{M},g,\varphi{},F)\). We will write \(\psi{}\) for a general
sufficiently regular function and write \(\varphi{}\) for the scalar field in the
solution.
\subsubsection{Constants}
\label{sec:orgf151797}
\label{sec:notation:constants} We write \(C\) for a large constant which can change
from line to line. Constants depend only on
the following quantities that can be determined from initial data: the initial
mass \(\varpi_i\), the minimum radius \(r_{\mathrm{min}}\), and the redshift
constant \(c_{\mathcal{H}}\). In view of the subextremalilty of the horizon
(i.e. \(\abs{\mathbf{e}} < \varpi_f\) for \(\varpi_f = \lim_{v\to \infty}\varpi{}|_{\mathcal{H}}(v)\) the
final mass) and the monotonicity properties of the
mass (i.e. \(\varpi{}_f \le \varpi{}_i\)), it follows that \(\abs{\mathbf{e}}\le
C\). We will not track the dependence of constants on \(R_0\), but we
will ensure that \(R_0\) depends only on \(\varpi_i\) (see
\cref{zeroth-order-geometric-bounds}).

Here are the global constants that we use:
\begin{center}
\begin{tabular}{lll}
  \(\varpi_i\) & initial mass of the spacetime (see \cref{sec:reduction-to-characteristic}) & \\[0pt]
  \(c_{\mathcal{H}}\) & redshift constant (see \cref{sec:reduction-to-characteristic}) & \\[0pt]
  \(\eta{}_0\) & global small constant (see \cref{sec:bulk-energy}) & \\[0pt]
  \(r_{\text{min}}\) & minimum radius achieved on the initial data hypersurface & \\[0pt]
  \(R_0\gg\varpi_i\) & radius of the vertex of \(\mathcal{R}_{\text{char}}\) (see \cref{sec:coordinate-systems}), fixed in \cref{zeroth-order-geometric-bounds} & \\[0pt]
  \(\Rc\gg R_0\) & \makecell[l]{radius at which commutator vector fields \(V\)
                   and \(S\) change from being adapted to the \\ horizon to being adapted to null infinity, fixed in \cref{sec:energy-estimates}}\\[0pt]
\end{tabular}
\end{center}
The main role of \(\eta_0\) is in the degeneration of the Morawetz estimate as the
weight in the bulk term approaches \(r^{-1}\) and in the degeneration of the
\(r^p\) estimates as \(p\) approaches \(2\). In \cref{sec:putting-it-all-together},
the constant \(\eta_0\) will be fixed based on the multi-index \(\alpha{}\) that
counts how many times we have commuted. The constant \(\Rc\) will be chosen
large based on the schematic geometric quantities \(\mathfrak{B}_0^{\circ }\)
and \(\mathfrak{g}_0\) (see \cref{sec:schematic-geom-quantities}) as well as the
multi-index \(\alpha{}\).

The notation \(A\lesssim B\) means that \(A\le CB\) for some positive constant \(C
> 0\). We do not require \(A\) or \(B\) to have a sign.
\subsubsection{Multi-index notation for derivatives}
\label{sec:org3813b65}
\label{sec:multi-index-notation}
Let \(\alpha{} = (\alpha{}_1,\alpha{}_2,\alpha{}_3) \in \Z^3\) be a triple of integers. We write \(\abs{\alpha{}} =
\alpha_1 + \alpha_2 + \alpha_3\). We introduce a total order on \(\Z^3\) such
that such that \(\alpha{} < \alpha{}'\) if either \(\abs{\alpha{}} <
\abs{\alpha{}'}\) or \(\alpha_i < \alpha_i'\) for the largest \(i\in
\set{1,2,3}\) such that \(\alpha_i\neq{}\alpha_i'\). Here are some examples of
this ordering:
\begin{equation}
(1,0,0) < (0,1,0) < (0,0,1)\qquad (0,2,0) < (1,0,1) \qquad (0,1,5) < (7,0,0)\quad (0,0,1) < (-1,0,2).
\end{equation}

We will use the commutator vector fields \(U\), \(V\), and \(S\), which were
defined in \cref{sec:commutator-vector-fields}. For this reason, we will often
label the components of a multi-index \(\alpha{}\) as \(\alpha{} =
(\alpha_U,\alpha_V,\alpha_S)\). When performing index arithmetic, we will
identify \(U\) with \((1,0,0)\), \(V\) with \((0,1,0)\), and \(S\) with
\((0,0,1)\). For example, \(\alpha{} + U\) has components \((\alpha_U +
1,\alpha_V,\alpha_S)\).
\begin{remark}
We will only write down multi-indices with non-negative entries (i.e. \(\alpha{}\in \Z_{\ge
0}^3\)), but our formalism is such that a multi-index with negative entries may
arise as the difference of two non-negative multi-indices.
\end{remark}
\subsubsection{Schematic notation}
\label{sec:org34bc048}
\label{sec:schematic-notation}
We use the symbol \(=_{\mathrm{s}}\) to indicate that an equation is to be understood
schematically,\sidenote{A similar schematic notation is used in \cite{dafermos-2017-inter-dynam}.} in the sense that we adopt the following conventions:
\begin{description}
\item[{Numerical constants}] Constants on the left side of a schematic equation are
exact. On the other hand, each term on the right side of a schematic equation
can be multiplied by an implicit constant (which may be zero). For example,
the equation \(\partial{}_v\varphi{} = r^{-1}\partial{}_v(r\varphi{}) -
  r^{-1}(\partial{}_vr)\varphi{}\) can be schematically written
\(\partial_v\varphi{} =_{\mathrm{s}} r^{-1}\partial_v(r\varphi{}) +
  r^{-1}(\partial_vr)\varphi{} + r^{-3}\). Another example is that \(1 - \mu{} =
  1 - 2\varpi{}/r + \mathbf{e}^2/r^2\) can be written \(1 - \mu{} =_{\mathrm{s}} 1 +
  r^{-1}\varpi{} + r^{-2}\), since \(\mathbf{e}^2\) is a constant (bounded by \(\varpi{}_i\)).
\item[{Particular expressions on the left}] When an expression such as \(\Gamma^\alpha{}\)
appears on the left side of a schematic identity, it stands for a particular
\(L\in \Gamma^\alpha{}\); on the right side, \(\Gamma^\alpha{}\) stands for a
sum over all \(L\in \Gamma^\alpha{}\).
\item[{Distributivity over square brackets}] Schematic expressions involving square
brackets should be fully expanded before being parsed. For example, the
schematic expression \(\lambda{}[(1-r) + r]\) involves \(\lambda{}(1-r)\) and
\(\lambda{}r\), and it is not schematically equal to \(\lambda{}\) (but it is
schematically equal to \(\lambda{} + \lambda{}r\)). We will use round brackets
to preserve information about the grouping of terms. For example, \(\lambda{}(5 -
  r + r) =_{\mathrm{s}} \lambda{}\).
\item[{Curly braces notation}] Curly braces are an alternative to the square bracket
notation. For example, we may write \(\lambda{}\set{\kappa{},(-\gamma{})}\) in
place of \(\lambda{}[\kappa{} + (-\gamma{})]\).
\item[{Derivatives (\(\Gamma^\alpha{}\) notation)}] If \(\alpha{}\in \Z_{\ge 0}^3\), then \(\Gamma^\alpha{}\) denotes a
product of \(\alpha_1\) many \(U\)'s, \(\alpha_2\) many \(V\)'s, and
\(\alpha_3\) many \(S\)'s, in any order. If \(\alpha_i < 0\) for some \(i \in
  \set{1,2,3}\), then we set \(\Gamma^\alpha{} = 0\). If \(n\in \Z_{\ge 0}\),
then we schematically write \(\Gamma^n\) for any \(\Gamma^\alpha{}\), where
\(\alpha{}\in \Z_{\ge 0}^3\) with \(\abs{\alpha{}} = n\). We write
\(\Gamma^{\le \alpha{}}\) to stand in for \(\Gamma^\beta{}\), where
\(\beta{}\) is a multi-index with \(\beta{}\le \alpha{}\), and we similarly
define \(\Gamma^{<\alpha{}}\). When such an expression appears on the right of
\(=_{\mathrm{s}}\), it stands for a sum over all possible differential operators it could
represent. For example, if \(\alpha{} = (1,1,0)\), then \(\Gamma^{\le
  \alpha{}}\) on the right of \(=_{\mathrm{s}}\) stands for the schematic expression
\(\set{1,U,V,UV,VU}\).
\item[{Derivatives (\(D^\alpha{}\) notation)}] When considering only \(U\) and \(V\), we
will use the symbol \(D\) in place of \(\Gamma{}\), and the same notational
conventions apply. For example, \(D^2S\) stands schematically for any of the
terms \(U^2S\), \(UVS\), \(VUS\), or \(V^2S\).
\end{description}
It will be convenient to abuse notation and write \(\Gamma^\alpha{}\) not only for a
particular instance of a product of vector fields, but also for the set of such
products of vector fields. For example, we may write \(VUV\in
\Gamma{}^{(1,2,0)}\) or consider the quantity \(\sum_{L\in \Gamma{}^\alpha{}}\abs{Lf}\) for some
function \(f\).

Observe that we can differentiate schematic equations.
\subsubsection{\texorpdfstring{\(\mathcal{O}\)}{O}-notation}
\label{sec:orgcef90a6}
\label{sec:O-notation}
We use expressions \(\mathcal{O}(\cdot )\) as an even coarser form of schematic
notation. We write \(\mathcal{O}(a_1,\ldots,a_k)\) to denote an expression of
the form \(f(r,a_1,\ldots,a_k)\) for a smooth function \(f :
[r_{\mathrm{min}},\infty)\times \R^k\to \R\) such that
\begin{equation}\label{O-notation-definiiton}
\abs{\partial_1^n
\partial_{>1}^mf(t,\set{b_i})}\le C(f,n,m,\max_i \abs{b_i})t^{-n}.
\end{equation}
Any further parameters on which the expression depends should be notated with a
subscript. For example, the cutoff function \(\chi{}_{r\lesssim \Rc}\) with
support in \(\set{r\le 2\Rc}\) (defined in \cref{sec:commutator-vector-fields})
satisfies \(\chi_{r\lesssim \Rc}(r) = \mathcal{O}(1)\) and \(r^k\chi_{r\lesssim
\Rc}(r) = \mathcal{O}_{k,\Rc}(1)\). We now discuss some features of this notation.
\begin{description}
\item[{Order of growth}] The quantitative dependence on parameters can be quite
poor. For example, \(\exp \exp \kappa{} = \mathcal{O}(\kappa{})\).
\item[{Negative powers of \(r\)}] This notation does not keep track of negative
powers of \(r\), since \(r^{-n} = \mathcal{O}_n(1)\). For example, the
equation \(\mu{} = 2\varpi{}/r - \mathbf{e}^2/r^2\) implies \(1 - \mu{} = \mathcal{O}(\varpi{})\)
and \(\mu{} = r^{-1}\mathcal{O}(\varpi{})\).
\item[{Specifying the support in \(r\)}] We write \(\mathbf{1}_{r\ge
  R}\mathcal{O}(\set{a_i})\) when the corresponding smooth function given by
\(f(\cdot ,\set{b_i})\) is supported in \(\set{r\ge R}\) for all arguments
\(b_i\). Similarly, we write \(\mathcal{O}(\mathbf{1}_{r\le
  R}a_1,a_2,\ldots,a_n)\) when \(f(r,a_1,\ldots,a_n) = f(r,0,a_2,\ldots,a_n)\)
for \(r > R\). Similarly we define \(\mathbf{1}_{r\le R}g\) in place of
\(\mathbf{1}_{r\le R}\mathcal{O}(1)g\). Under this convention, we have
\(\mathbf{1}_{r\ge R}\mathcal{O}(\set{a_i}) = \mathcal{O}(\set{\mathbf{1}_{r\le
  R}a_i})\).
\item[{Algebra of \(\mathcal{O}\)-notation}] We have \(\mathcal{O}(\set{a_i}) +
  \mathcal{O}(\set{b_j}),\,\mathcal{O}(\set{a_i})\mathcal{O}(\set{b_j}) =_{\mathrm{s}}
  \mathcal{O}(\set{a_i},\set{b_i})\).
\item[{Differentiating \(\mathcal{O}\)-notation}] Let \(\Gamma{}\) be one of \(U\),
\(V\), or \(S\). If the \(a_i\) are smooth, then we have \(\Gamma{}\mathcal{O}(\set{a_i})
  = \mathcal{O}(\set{a_i},\set{\Gamma{}a_i},r^{-1}\Gamma{}r)\) by the chain rule and
\cref{O-notation-definiiton}. Since \(\Gamma{}r = r\mathcal{O}(1)\) (see
\eqref{Gamma-r-initial-calculation}), we have
\(\Gamma{}^\alpha{}\mathcal{O}(\set{a_i}) = \mathcal{O}(\set{\Gamma{}^{\le
  \alpha{}}a_i})\).
\end{description}
\subsection{Norms and schematic geometric quantities}
\label{sec:orgcee1d42}
\label{sec:norms}
\subsubsection{Initial data norms}
\label{sec:orge4a2282}
For an integer \(n\ge 0\), define the following gauge invariant norm for
characteristic initial data:
\begin{equation}
\mathfrak{D}_n[\psi{}] \coloneqq{}\sup_{C^{\text{in}}} \abs{U^{\le n+1}\psi{}} + \sup_{C^{\mathrm{out}}} \abs{(r^2\overline{\partial}_r)^{\le 1}(r\overline{\partial}_r)^{\le n}(r\psi{})}.
\end{equation}
Write \(\mathfrak{D}_n \coloneqq{} \mathfrak{D}_n[\varphi{}]\). Define
\begin{equation}
\mathcal{D}_\alpha{} \coloneqq{} \mathfrak{D}_0[\Gamma{}^\alpha{}\psi{}].
\end{equation}
\subsubsection{Energy norms}
\label{sec:org3548c59}
\label{sec:norms:energy} Let \(p\ge 0\). For a spacetime region \(\mathcal{R}\), we
define the \(r^p\)-weighted bulk energy quantity
\begin{equation}
E_{p,\mathrm{bulk}}[\psi{},\mathcal{R}] \coloneqq{}\iint_{\mathcal{R}\cap \set{r\ge R_0}} \frac{r^{p-1}}{\lambda{}}(\partial{}_v(r\psi{}))^2(-\nu{})\dd{}u\dd{}v.
\end{equation}
Let \(\mathcal{R}_{\text{char}}\) be the future development of the
characteristic data. We define a norm \(\mathcal{E}_p[\psi{}]\) that has weights
in \(\tau{}\) (see \cref{tau-def}):
\begin{equation}\label{Ep-norm-definition}
\begin{split}
\mathcal{E}_p[\psi{}]^2 &\coloneqq{} \sup_{u,v\in [1,\infty)}\tau{}^p(u,v)\int_u^\infty \frac{r^2}{(-\nu{})}(\partial{}_u\psi{})^2(u',v)\dd{}u' + \sup_{u,v\in [1,\infty)}\tau{}^p(u,v)\int_v^\infty \frac{r^2}{\kappa{}}(\partial{}_v\psi{})^2(u,v')\dd{}v' \\
&\qquad + \mathbf{1}_{p>0}E_{p,\mathrm{bulk}}[\psi{},\mathcal{R}_{\text{char}}] + \mathbf{1}_{p\ge 1+\eta{}_0}\sup_{u\in [1,\infty]}\tau{}^{p-1-\eta{}_0}(u,v_{R_0}(u))\int_{v_{R_0}(u)}^\infty \frac{r^{1+\eta{}_0}}{\lambda{}}(\partial{}_v(r\psi{}))^2\dd{}v'. \\
\end{split}
\end{equation}
Note that \(\mathcal{E}_0[\psi{}]\) does not include a bulk
term. We allow ourselves to write \(\mathcal{E}[\psi{}]\) in place of \(\mathcal{E}_0[\psi{}]\):
\begin{equation}
\mathcal{E}[\psi{}]^2 \coloneqq{} \mathcal{E}_0[\psi{}] = \sup_{v\in [1,\infty)}\int_1^\infty \frac{r^2}{(-\nu{})}(\partial{}_u\psi{})^2(u,v)\dd{}u + \sup_{u\in [1,\infty)}\int_1^\infty \frac{r^2}{\kappa{}}(\partial{}_v\psi{})^2(u,v)\dd{}v. \\
\end{equation}
For \(\Gamma^\alpha{}\varphi{}\), we define the following norms:
\begin{equation}
\mathcal{E}_{\alpha{}}\coloneqq{}\sum_{\beta{}\le \alpha{}}\mathcal{E}[\Gamma{}^{\beta{}}\varphi{}], \qquad\mathcal{E}_{\alpha{},p}\coloneqq{}\sum_{\beta{}\le \alpha{}}\mathcal{E}_p[\Gamma{}^{\beta{}}\varphi{}].
\end{equation}
\subsubsection{Pointwise norm}
\label{sec:org81844e8}
\label{sec:norms:pointwise}
For \(p\ge 0\), define a pointwise norm \(\mathcal{P}_p[\psi{}]\):
\begin{align}
\mathcal{P}_p[\psi{}] &\coloneqq{}\norm{r\psi{}}_{L^\infty} + \norm{rU\psi{}}_{L^\infty} + \norm{r^2\partial{}_v\psi{}}_{L^\infty} + \norm{r^{1/2}\tau{}^{p/2}\psi{}}_{L^\infty} + \mathbf{1}_{p\ge 1+\eta{}_0}\norm{r\tau{}^{p/2-1/2-\eta{}_0/2}\psi{}}_{L^\infty},
\end{align}
where we write \(L^\infty = L^\infty(\mathcal{R}_{\text{char}})\).
\begin{remark}
By interpolation (splitting into regions \(r\le \tau{}\) and \(r\ge \tau{}\)),
\(\mathcal{P}_p[\psi{}]\) for \(p > 0\) controls the quantity \(\norm{r^{\rho{}}\tau^{p/2 +
1/2-\rho{}-\eta_0/2}\psi{}}_{L^\infty}\) for \(\rho{}\in{}[1/2,1]\). In particular,
\(\mathcal{P}_{2-s}[\psi{}]^2\) controls \(\norm{r^{2-s-\eta{}_0}\tau{}\psi{}^2}_{L^\infty}\) (take \(\rho{} = 1 - s/2 - \eta_0/2\)).
\end{remark}
For \(\Gamma^\alpha{}\varphi{}\), we define the following norm:
\begin{equation}
\mathcal{P}_{\alpha{},p} \coloneqq{} \sum_{\beta{}\le \alpha{}}\mathcal{P}_p[\Gamma{}^{\beta{}}\varphi{}].
\end{equation}
\subsubsection{Schematic geometric quantities}
\label{sec:org6b89832}
\label{sec:schematic-geom-quantities} We first define a quantity \(G\) in terms of
which we can compare the double null coordinates \(u\) and \(v\) with the radius
\(r\) and the time \(\tau{}\) (see \cref{v-u-r-compare}).
\begin{lemma}[Comparison of coordinates]
For each characteristic data set and each \(\epsilon{}\in (0,1)\), there is \(C_\epsilon{} > 0\)
depending on \(\epsilon{}\), \(\mathcal{E}_{0,1 + \eta_0}\), \(\eta_0\) and \(\Rc\) such that
\begin{equation}
C^{-1}_\epsilon{}\le G_\epsilon{}\le C_\epsilon{},
\end{equation}
where \(G_\epsilon{}\) is the following quantity computed from the future development
\(\mathcal{R}_{\mathrm{char}}\) of the characteristic data:
\begin{equation}\label{G-def}
G_\epsilon{} \coloneqq{} \sup_{\mathcal{R}_{\mathrm{char}}} [r/v + \mathbf{1}_{r\ge R_0}u/v + (\mathbf{1}_{r\le \epsilon{}^{-1}\Rc} + \mathbf{1}_{r\le (1-\epsilon{})v})v/u].
\end{equation}
\label{G-epsilon}
\end{lemma}
The proof will be given in \cref{sec:zeroth-order-geometric} (see \cref{G-estimate}).

With the notation introduced in \cref{sec:schematic-notation,sec:O-notation} in
mind, we define the following schematic geometric quantities:
\begin{align}
\mathfrak{b}_0 &\coloneqq{}\mathcal{O}_{\varpi{}_i,c_{\mathcal{H}},r_{\text{min}},R_0,\eta{}_0}(1,\kappa{},\kappa{}^{-1},r^{-1}\varpi{},\mathbf{1}_{r\ge R_0}\set{(-\gamma{}),(-\gamma{})^{-1},(1-\mu{})^{-1}}), \label{b0-def} \\
\mathfrak{B}_0^{\circ } &\coloneqq{}\mathcal{O}(\mathfrak{b}_0,\set{r,\mathbf{1}_{r\ge R_0}u,v}\set{(1-\kappa{}),(1-\kappa{}^{-1}),\mathbf{1}_{r\ge R_0}\set{(1-(-\gamma{})),(1-(-\gamma{})^{-1})}}), \label{B0-circ-def}\\
\mathfrak{B}_0 &\coloneqq{}\mathcal{O}_{G_{\eta{}_0},G_{\eta{}_0}^{-1},\Rc}(\mathfrak{B}_0^{\circ },\mathbf{1}_{\Rc\le r\le 2\Rc}(v-u-r)), \label{B0-def}\\
\mathfrak{g}_0 &\coloneqq{}\mathfrak{b}_0\set{1,\varpi{}}, \label{g0-def} \\
\mathfrak{G}_0 &\coloneqq{} \mathfrak{g}_0. \label{G0-def}
\end{align}
For \(\abs{\alpha{}}\ge 1\), define
\begin{align}
\mathfrak{b}_\alpha{} &\coloneqq{}\mathcal{O}_\alpha{}(\Gamma{}^\beta{}\mathfrak{B}_{\beta{}'}|_{\beta{}+\beta{}'<\alpha{}},r^{-1}\mathfrak{G}_{<\alpha{}},r^{-1}\Gamma{}^\alpha{}\varpi{},\Gamma{}^\alpha{}\kappa{},\mathbf{1}_{\alpha{}_{U}>0}r\Gamma{}^\alpha{}\kappa{},\Gamma{}^\alpha{}(-\gamma{})), \label{b-alpha-def}\\
\mathfrak{B}_\alpha{} &\coloneqq{}\mathcal{O}_\alpha{}(\Gamma{}^\beta{}\mathfrak{b}_{\beta{}'}|_{\beta{}+\beta{}'\le \alpha{}},\set{r,\mathbf{1}_{r\ge R_0}u,v}\Gamma{}^\alpha{}\kappa{},r\Gamma{}^{\alpha{}+U}\kappa{},\set{r,\mathbf{1}_{r\ge R_0}ru,v}\mathbf{1}_{r\ge R_0}\Gamma{}^\alpha{}(-\gamma{})), \label{B-alpha-def}\\
\mathfrak{g}_\alpha{} &\coloneqq{}\mathfrak{b}_\alpha{}\set{1,\Gamma{}^\beta{}\mathfrak{G}_{\beta{}'}|_{\beta{}+\beta{}'<\alpha{}},\Gamma{}^\alpha{}\varpi{},r\Gamma{}^\alpha{}(-\gamma{})}, \label{g-alpha-def}\\
\mathfrak{G}_\alpha{} &\coloneqq{}\mathfrak{B}_\alpha{}\set{1,\Gamma{}^\beta{}\mathfrak{g}_{\beta{}'}|_{\beta{}+\beta{}'\le \alpha{}},\set{r,\mathbf{1}_{r\ge R_0}ru,\mathbf{1}_{r\le 2\Rc}v}\mathbf{1}_{\alpha{}_U>0}\Gamma{}^\alpha{}\kappa{}}. \label{G-alpha-def}
\end{align}
Define
\begin{equation}\label{C-alpha-def}
\mathcal{C}_\alpha{} \coloneqq{}\mathcal{O}(\mathfrak{B}_{\alpha{}},r^{-1}\mathfrak{G}_{\alpha{}}).
\end{equation}
The quantities \(\mathcal{C}_\alpha{}\) appear in the computation of the commutators
\([U,S]\) and \([V,S]\) (see \cref{D-comm}). Observe that \(\mathcal{C}_{<\alpha{}}
=_{\mathrm{s}} \mathfrak{b}_\alpha{}\). Finally, define
\begin{equation}
\mathcal{G}_{\alpha,s}\coloneqq{}\mathcal{O}_s(\mathfrak{b}_\alpha{},r^{-s}\mathfrak{g}_\alpha{}).
\end{equation}
The quantities \(\mathcal{G}_{\alpha{},s}\) appear in pointwise bounds for \(\Box{}\Gamma{}^\alpha{}\varphi{}\).
\begin{remark}[Boundedness and growth of schematic geometric quantities]
We will show in our proof that the quantities \(\mathfrak{b}_\alpha{}\) and
\(\mathfrak{B}_\alpha{}\) are bounded, but we only show that the
\(\mathfrak{g}_\alpha{}\) and \(\mathfrak{G}_\alpha{}\) terms grow at a slow
rate (namely \(r^{C_\alpha{}\eta_0}\) for some sequence of constants
\(C_\alpha{}\) that is increasing in \(\alpha{}\)). Observe that the quantities
\(\mathfrak{B}_\alpha{}\) are stronger than the quantities
\(\mathfrak{b}_\alpha{}\), and \(\mathfrak{G}_\alpha{}\) is stronger than
\(\mathfrak{g}_\alpha{}\).
\end{remark}
\begin{remark}[Relation of schematic geometric quantities to the energy]
The weak geometric quantities \(\mathfrak{b}_\alpha{}\) and \(\mathfrak{g}_\alpha{}\) must be
controlled before the order-\(\alpha{}\) energy \(\mathcal{E}_\alpha{}\) is,
because these quantities appear in the commutator
\(\Box{}\Gamma^\alpha{}\varphi{}\) that arises as an error term in the energy
estimate. Once \(\mathcal{E}_\alpha{}\) is controlled, we control the strong
geometric quantities \(\mathfrak{B}_\alpha{}\) and \(\mathfrak{G}_\alpha{}\).
\end{remark}
\begin{remark}[Algebra of schematic geometric quantities]
We have arranged that \(\mathfrak{b}_\alpha{} =_{\mathrm{s}} \mathfrak{B}_\alpha{} =_{\mathrm{s}} \mathfrak{G}_\alpha{}\)
and \(\mathfrak{b}_\alpha{} \equiv_s \mathfrak{g}_\alpha{} =_{\mathrm{s}}
\mathfrak{G}_\alpha{}\). Moreover, the quantities become stronger as
\(\alpha{}\) increases, in the sense that, for example, \(b_\beta{} =_{\mathrm{s}}
\mathfrak{b}_\alpha{}\) whenever \(\beta{}\le \alpha{}\). Finally, the
quantities are compatible with differentiation by the vector fields
\(\Gamma{}\), in the sense that \(\Gamma^\alpha{}\mathfrak{b}_{\alpha{}'}=_{\mathrm{s}}
\mathfrak{b}_{\alpha{} + \alpha{}'}\) (and similarly for the other quantities).
\end{remark}
\subsection{Vector field commutator calculations}
\label{sec:org3ba3a95}
In this section we collect useful calculations related to our vector field
commutators. We relegate the proofs to \cref{sec:vf-commutator-proofs}.
\subsubsection{Comparing vector fields in different coordinate systems}
\label{sec:org4d24805}
\begin{lemma}[Coordinate change associated to \(U\)]
\begin{equation}\label{U-du-coordinate-change}
\mathbf{1}_{r\ge R_0}U =_{\mathrm{s}} \mathfrak{b}_0\partial{}_u
\end{equation}
\end{lemma}
\begin{lemma}[Coordinate change associated to \(V\)]
We have
\begin{align}
V &=_{\mathrm{s}} \mathfrak{b}_0[\partial{}_v + \mathbf{1}_{r\le 2\Rc}U], \label{V-dv-coordinate-change-1} \\
V - \overline{\partial}_r &=_{\mathrm{s}} \mathbf{1}_{r\le 2\Rc}\mathfrak{b}_0[\overline{\partial}_r + U], \label{V-dr-coordinate-change} \\
\partial{}_v &=_{\mathrm{s}} \mathfrak{b}_0[V + \mathbf{1}_{r\le 2\Rc}U]\label{dv-in-terms-of-V-and-U}.
\end{align}
\label{V-dv-coordinate-change}
\end{lemma}
\begin{lemma}[Coordinate change associated to \(S\)]
We have
\begin{align}
S &=_{\mathrm{s}} \mathfrak{b}_0[\set{\mathbf{1}_{r\ge \Rc}u,\mathbf{1}_{r\le 2\Rc}v}U + \set{r,\mathbf{1}_{r\ge \Rc}u,v}V] \label{S-in-terms-of-U-V}\\
\mathbf{1}_{r\ge R_0}(S - r\overline{\partial}_r) &=_{\mathrm{s}} \set{1,\mathbf{1}_{r\ge \Rc}u,\mathbf{1}_{r\le 2\Rc}\mathfrak{b}_0v}\overline{\partial}_u. \label{S-rdr-comparison}
\end{align}
\end{lemma}
\subsubsection{Rearranging commutator vector fields}
\label{sec:org2f10b5a}
\begin{lemma}[Commuting with \(\overline{\partial}_r\)]
We have
\begin{align}
\mathbf{1}_{r\ge R_0}[\overline{\partial}_r,U] &=_{\mathrm{s}} \mathbf{1}_{r\ge R_0}r^{-2}\mathfrak{g}_0[\overline{\partial}_r + U],\label{dr-U-commutation} \\
\mathbf{1}_{r\ge R_0}[\overline{\partial}_r,V] &=_{\mathrm{s}} \mathbf{1}_{R_0\le r\le 2\Rc}\mathfrak{b}_V[\overline{\partial}_r + U], \label{dr-comm-1-V}\\
\mathbf{1}_{r\ge R_0}[\overline{\partial}_r,S] &=_{\mathrm{s}} \mathbf{1}_{r\ge \Rc}\overline{\partial}_r + \mathbf{1}_{R_0\le r\le 2\Rc}\mathfrak{B}_V[\overline{\partial}_r + U]\label{dr-S-comm}.
\end{align}
\label{dr-comm-1}
\end{lemma}
\begin{lemma}[Commuting with one commutator vector field]
We have
\begin{align}
[U,V] &=_{\mathrm{s}} \mathbf{1}_{r\ge \Rc}r^{-2}\mathfrak{g}_0D, \label{U-V-comm}\\
[U,S] &=_{\mathrm{s}} \mathbf{1}_{r\ge \Rc}U + \mathbf{1}_{r\ge \Rc}r^{-1}\mathfrak{G}_UD, \label{U-S-comm-1}\\
[V,S] &=_{\mathrm{s}} V + \mathbf{1}_{\Rc\le r\le 2\Rc}\mathfrak{B}_VD.
\end{align}
\label{D-comm-1}
\end{lemma}
\begin{lemma}[Commuting with a product of commutator vector fields]
Let \(\alpha{}\ge 0\) and let \(L\in \Gamma^\alpha{}\). We have
\begin{align}
[U,L] &=_{\mathrm{s}} \mathcal{C}_{<\alpha{}}[U\Gamma{}^{\le \alpha{}-S} + r^{-2}\mathfrak{G}_{<\alpha{}}D\Gamma{}^{\le \alpha{}-V} + r^{-1}\mathfrak{G}_{<\alpha{}}D\Gamma{}^{\le \alpha{}-S}] , \label{U-comm-weak} \\
[V,L] &=_{\mathrm{s}} V\Gamma{}^{\le \alpha{}-S} + \mathcal{C}_{<\alpha{}}r^{-2}D\Gamma{}^{\le \alpha{}-U}, \label{V-comm-weak} \\
[D,L] &=_{\mathrm{s}} \mathcal{C}_{<\alpha{}}D\Gamma{}^{\le \alpha{}-U}. \label{D-comm-weak}
\end{align}
Moreover, we have
\begin{equation}\label{V2-comm}
[V^2,L] =_{\mathrm{s}} V^2\Gamma{}^{\le \alpha{}-S} + \mathcal{C}_{<\alpha{}+V}r^{-2}D\Gamma{}^{\le \alpha{}-U+V}.
\end{equation}
\label{D-comm}
\end{lemma}
\begin{lemma}[Bringing \(U\) to the front]
Fix \(\alpha{}\ge 1\) such that \(\alpha_U > 0\). Let \(L\in \Gamma^\alpha{}\). Then there is \(L'\in
\Gamma^{\alpha{}-U}\) such that
\begin{equation}
\begin{split}
L - UL' &=_{\mathrm{s}}  \mathcal{O}(\mathfrak{B}_{\alpha{}-U-S+V},r^{-1}\mathfrak{G}_{\alpha{}-S})[U\Gamma{}^{\le \alpha{}-U-V} + r^{-2}\mathfrak{G}_{\alpha{}-U-V}V\Gamma{}^{\le \alpha{}-U-V} + r^{-1}\mathfrak{G}_{\alpha{}-S}V\Gamma{}^{\le \alpha{}-U-S}] \\
&=_{\mathrm{s}} \mathcal{C}_{<\alpha{}}[U\Gamma{}^{<\alpha{}} + r^{-2}\mathfrak{G}_{<\alpha{}}V\Gamma{}^{\le \alpha{}-U-V} + r^{-1}\mathfrak{G}_{<\alpha{}}V\Gamma{}^{\le \alpha{}-U-S}] \\
\end{split}
\end{equation}
\label{U-rearrangement-formula}
\end{lemma}
\begin{lemma}[Rearrangement formula]
Let \(\alpha{}\ge 0\), and let \(L,L'\in \Gamma^\alpha{}\). We have
\begin{equation}
L-L' =_{\mathrm{s}} \mathcal{C}_{\le \alpha{}-S}D\Gamma{}^{\le \alpha{}-U-V}.
\end{equation}
\label{rearrangement-formula}
\end{lemma}
\subsubsection{Commuting with the wave operator}
\label{sec:orgda1dc48}
\label{sec:commute-with-box}
\begin{lemma}[Wave equations for commutator vector fields]
We have
\begin{equation}\label{UV-box}
UV,VU =_{\mathrm{s}} \mathfrak{b}_0[\Box{} + r^{-1}D + \mathbf{1}_{r\le 2\Rc}U^2]
\end{equation}
and
\begin{equation}\label{vU-box}
\partial{}_vU =_{\mathrm{s}} \mathfrak{b}_0[\Box{} + r^{-1}D].
\end{equation}
\label{UV-box-lemma}
\end{lemma}
\begin{lemma}[Commutators of vector field with \(\Box\)]
The following commutation formulas hold for \(\Rc\) large enough depending on
\(\mathfrak{B}_0^{\circ }\):
\begin{align}
[\Box{},U] + f_U U^2 &=_{\mathrm{s}}\mathfrak{b}_U[\Box{} + r^{-2}V + r^{-2}U], \label{U-comm-formula} \\
[\Box{},V] + f_VV^2 &=_{\mathrm{s}} \mathfrak{b}_V[\Box{} + r^{-2}\mathfrak{g}_VV + r^{-2}U + r^{-2}U^2], \label{V-comm-formula} \\
[\Box{},S] &=_{\mathrm{s}} \mathfrak{b}_S[\Box{} + r^{-1}\mathfrak{g}_SV^2 + r^{-2}\mathfrak{g}_SV + \mathbf{1}_{r\le 2\Rc}r^{-2}UU^{\le 1}]. \label{S-comm-formula}
\end{align}
where
\begin{equation}
0\le f_U = \frac{2(\varpi{}-\mathbf{e}^2/r)}{r^2}=_{\mathrm{s}} r^{-2}\mathfrak{g}_0\qquad f_V =_{\mathrm{s}} \mathbf{1}_{r\ge \Rc}r^{-2}\mathfrak{B}_0^{\circ }\mathfrak{g}_0
\end{equation}
\label{wave-comm-formula-UVS}
\end{lemma}
\begin{lemma}[Main formula for commutation with \(\Box\)]
Let \(L\in \Gamma^\alpha{}\). For \(\Rc\) large enough depending on \(\mathcal{G}_0^{\circ }\), we
have
\begin{equation}\label{wave-comm-formula-equation}
\begin{split}
[\Box{},L] + \alpha{}_Uf_UUL + \alpha{}_Vf_VVL =_{\mathrm{s}} \mathfrak{b}_\alpha{}[\Gamma{}^{\le \alpha{}-1}\Box{} + r^{-1}\mathfrak{g}_\alpha{}V^2\Gamma{}^{\le \alpha{}-S} + r^{-2}\mathfrak{g}_\alpha{}V\Gamma{}^{\le \alpha{}-1} + r^{-2}\mathfrak{g}_\alpha{}U\Gamma{}^{<\alpha{}}].
\end{split}
\end{equation}
for \(f_U\) and \(f_V\) as in \cref{wave-comm-formula-UVS}.
\label{wave-comm-formula}
\end{lemma}
In the following pointwise estimate \cref{wave-pointwise-2}, we make crucial use of the good sign \(f_U
\ge 0\) (see \cref{wave-comm-formula-UVS}).
\begin{corollary}
Let \(\abs{\alpha{}}\ge 0\) and let \(L\in \Gamma^\alpha{}\). We have
\begin{equation}\label{wave-pointwise-1}
\begin{split}
r^{2-s}\abs{\Box{}L\varphi{}}&\le C(\mathfrak{B}_0^{\circ },\mathfrak{g}_0,\alpha{})[\alpha{}_V\abs{\partial{}_vL\varphi{}} + \alpha{}_U\abs{UL\varphi{}}] \\
&\qquad + C(\mathcal{G}_{\alpha,s})[\mathbf{1}_{r\ge \Rc}\abs{\partial{}_v(r\Gamma{}^{<\alpha{}}\varphi{})} + \abs{\partial{}_v\Gamma{}^{<\alpha{}}\varphi{}} + \abs{U\Gamma{}^{<\alpha{}}\varphi{}}].
\end{split}
\end{equation}
Moreover, we have
\begin{equation}\label{wave-pointwise-2}
r^{2-s}UL\varphi{}\Box{}L\varphi{} \le C(\mathcal{G}_{\alpha,s})(\abs{UL\varphi{}} + \abs{\partial{}_vL\varphi{}})(\abs{U\Gamma{}^{<\alpha{}}\varphi{}} + \abs{\partial{}_v\Gamma{}^{<\alpha{}}\varphi{}} +  \mathbf{1}_{r\ge \Rc}\abs{\partial{}_v(r\Gamma{}^{<\alpha{}}\varphi{})}),
\end{equation}
and
\begin{equation}\label{wave-pointwise-3}
\begin{split}
r^{2-s}\abs{\partial{}_vL\varphi{}}\abs{\Box{}L\varphi{}}&\le \mathbf{1}_{r\ge \Rc}C(\mathcal{B}_0^{\circ },\mathfrak{g}_0,\alpha{})\abs{\partial{}_vL\varphi{}}(\abs{UL\varphi{}} + \abs{\partial{}_vL\varphi{}}) + \side{RHS}{wave-pointwise-2}.
\end{split}
\end{equation}
\label{wave-pointwise-inequalities}
\end{corollary}
\section{Reduction to a characteristic problem}
\label{sec:org6ba4041}
\label{sec:reduction-to-characteristic}
For the rest of the paper, we study a characteristic problem.
\begin{theorem}[Decay for characteristic problem]
Consider characteristic data on \(C^{\mathrm{in}}\cup C^{\mathrm{out}}\) with
future \(\mathcal{R}_{\mathrm{char}}\). Assume the following:
\begin{enumerate}
\item The initial data norm \(\mathfrak{D}_k[\varphi{}]\) defined on \(C^{\mathrm{in}}\cup
   C^{\mathrm{out}}\) is finite.
\item The data vanishes on a neighbourhood of the outgoing null hypersurface
\(C^{\text{out}}\) (in particular on a portion of \(C^{\text{in}}\) near the
vertex).
\item There exists a double null gauge \((u,v)\) in which \(C^{\mathrm{in}} =
   [1,\infty)_u\times \set{1}_v\) and \(C^{\mathrm{out}} = \set{1}_u \times
   [1,\infty)_v\) (so \(\mathcal{R}_{\mathrm{char}}\) is given by
\([1,\infty)\times [1,\infty)\)) and is normalized such that
\(\frac{\partial_vr}{1-\mu{}}|_{\set{r=r_{\mathcal{H}}}} = 1\), where \(r_{\mathcal{H}} = \sup_{\mathcal{H}}r\),
and \(\lim_{v\to \infty}\frac{(-\partial_ur)}{1-\mu{}}(u,v) = 1\) for each
\(u\).
\item We have \(\partial_ur <0\) and \(\partial_vr \ge 0\) in \(\mathcal{R}_{\mathrm{char}}\) (note this is gauge independent).
\item We have \(\varpi{}\le \varpi{}_i\) in \(\mathcal{R}_{\mathrm{char}}\).
\item There is \(c_{\mathcal{H}} > 0\) such that \(\varpi{} - \mathbf{e}^2/r\ge
   c_{\mathcal{H}}\) in \(\mathcal{R}_{\mathrm{char}}\).
\item There is \(r_{\mathrm{min}}>0\) such that \(r\ge r_{\mathrm{min}}\) in \(\mathcal{R}_{\mathrm{char}}\).
\item The value \(R_0 =
   r|_{C^{\mathrm{in}}\cap C^{\mathrm{out}}}\), is sufficiently large based on
\(\varpi_i\) (as specified in \cref{zeroth-order-geometric-bounds}).
\end{enumerate}
Then for every \(\epsilon{} >0\) and integer \(k\ge 0\) there is \(C =
C(\epsilon{},k,\varpi_i,c_{\mathcal{H}},r_{\mathrm{min}},\mathfrak{D}_{k}[\varphi{}]) > 0\) such that
\begin{equation}
\abs{(v\underline{\partial{}}_v)^k\varphi{}}|_{\mathcal{H}} \le C\mathfrak{D}_k[\varphi{}]v^{-1+\epsilon{}}.
\end{equation}
For \(0\le k\le 4\), the same estimate holds with \(v\underline{\partial{}}_v\) replaced by
\(\bar{v}\partial_{\bar{v}}\).
\label{main-theorem}
\end{theorem}
\begin{figure}[ht]
    \centering
    \vspace{-4ex}
    \def\svgwidth{0.75\columnwidth}
    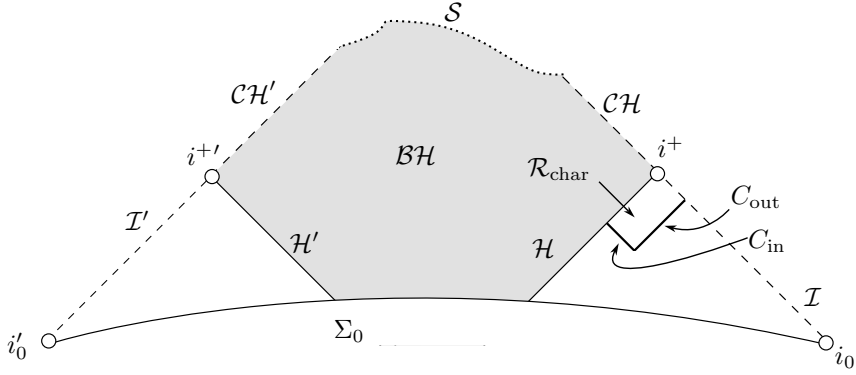

    \caption{The characteristic rectangle \(\mathcal{R}_\text{char}\) depicted on the Penrose diagram of a general solution to \cref{einstein-maxwell}. Note that the achronal singular set \(\mathcal{S}\) may be empty, as is the case in Reissner--Nordström. The statement that the diagram looks as depicted is due to \cite{Dafermos_2014,Kommemi2013TheGS}.}
    \label{fig:penrose}
\end{figure}
\begin{remark}
We have \(r_{\mathcal{H}} =: \sup_{\mathcal{H}}r\le C(\varpi{}_i)\), since by
\cite[Lem.~A.2]{luk-oh-scc1}, we have \(r_{\mathcal{H}} = \varpi_{\mathcal{H}} +
\sqrt{\varpi_{\mathcal{H}} - \mathbf{e}^2}\le 2\varpi_i\), where
\(\varpi_{\mathcal{H}} \coloneqq{} \sup_{\mathcal{H}} \varpi{}\).
\label{rH-size}
\end{remark}
We note that it is enough to consider only one of the ends of the spacetime; the argument is the
same for the other end.
We now show that the future of compactly supported future-admissible Cauchy data
contains a characteristic rectangle satisfying the assumptions of
\cref{main-theorem}, thus reducing the proof of \cref{main-theorem-rough} to
establishing \cref{main-theorem}.
\begin{proof}[Constructing a characteristic rectangle as in \cref{main-theorem} in the setting of \cref{main-theorem-rough}]
By the global well-posedness of \cref{einstein-maxwell} (see \cite{Kommemi2013TheGS}
or \cite[Thm.~4.1]{luk-oh-scc1}), we can first choose a characteristic rectangle in
the exterior region---where assumption (4) is satisfied (see
\cite{Kommemi2013TheGS}, or \cite[Lem.~A.1]{luk-oh-scc1})---and then construct the
gauge in assumption (2) once \(R_0\) is chosen as in assumption (8). Since the
Cauchy data is compactly supported, we can ensure that the data vanishes on a
neighbourhood of \(C^{\text{out}}\), which in particular implies assumptions (1)
and (2) (together with the global well-posedness of the system). By assumption
(4) and \cref{sph-sym-equations-1,sph-sym-equations-2}, we can take \(\varpi_i\) in
assumption (4) to be the supremum of \(\varpi{}\) on the Cauchy data. The future
admissibility condition on the data implies that the event horizon is eventually
subextremal by \cite{Kommemi2013TheGS} (see also \cite[Lem.~A.2]{luk-oh-scc1}), so
if we write \(\varpi{}_{\mathcal{H}}(v) \coloneqq{} \lim_{u\to \infty}
\varpi{}(u,v)\) and \(r_{\mathcal{H}}(v) = \lim_{u\to \infty}r(u,v)\), then
there is \(c_{\mathcal{H}} > 0\) such that \(\lim_{v\to
\infty}(\varpi{}_{\mathcal{H}}-\mathbf{e}^2/r_{\mathcal{H}})(v) =
2c_{\mathcal{H}} > 0\). In particular, there is \(v_\ast{} \ge 1\) such that
\((\varpi_{\mathcal{H}} - \mathbf{e}^2/r_{\mathcal{H}})(v_\ast{})\ge
c_{\mathcal{H}} > 0\). The monotonicity properties of (3) together with the
equations \cref{sph-sym-equations-1,sph-sym-equations-2} imply that \((\varpi{} -
\mathbf{e}^2/r)(u,v)\ge (\varpi_{\mathcal{H}} -
\mathbf{e}^2/r_{\mathcal{H}})(v_\ast{})\) for \(u\ge 1\) and \(v\ge v_\ast{}\).
After using the translation symmetry in the gauge, we can assume \(v_\ast{} =
1\). This settles assumption (6). Assumption (7) holds since
\(\mathcal{R}_{\text{char}}\) is in the exterior region. Finally, since \(r\to
\infty\) along \(C^{\text{out}}\), we can translate \(C^{\text{in}}\) towards
the future to satisfy assumption (8).
\end{proof}
\begin{remark}[Compact support of data]
The assumption that the data is compactly supported in \cref{main-theorem-rough} is
only used to reduce to a characteristic problem (in particular to ensure that we
can choose a characteristic rectangle for which \(\mathfrak{D}_k[\varphi{}]\) is
finite).
\end{remark}
From now on, we work in the region \(\mathcal{R}_{\text{char}}\) and pursue the
proof of \cref{main-theorem}.
\section{Estimates for undifferentiated geometric quantities}
\label{sec:orga2fd1de}
\label{sec:zeroth-order-geometric}
From now on, fix the value of \(R_0\) large enough that the estimates of
\cref{zeroth-order-geometric-bounds} hold and such that \(R_0\ge C(\varpi_i)\ge
r_{\mathcal{H}}\) (see \cref{rH-size}).
\begin{proposition}[Estimates for undifferentiated geometric quantities]\label{zeroth-order-geometric-bounds}
Let \(\eta > 0\). Let \(R\ge R_0\ge 1\). If \(R_0\) is sufficiently large (depending only on \(\varpi{}_i\)), then the following estimates hold (where we allow all constants to depend on \(c_{\mathcal{H}}\), \(r_{\mathrm{min}}\), \(\varpi_i\), and \(R\) in addition to any named parameters):
\begin{enumerate}
\item Globally, we have
\begin{align}
\varpi{}-\mathbf{e}^2/r\ge c_{\mathcal{H}} > 0, & \label{mass-redshift} \\
\varpi{}\le \varpi{}_i, & \label{mass-bounded} \\
0\le 1-\mu{}\le C,                               & \label{mu-estimate-weak} \\
\kappa > 0, \label{kappa-sign} \\
\log (-\gamma) \le 0, \label{gamma-sign} \\
\abs{\log \kappa{}}\le C\mathcal{E}_0^2,            & \label{kappa-estimate} \\
\abs{\log \kappa{}}\le C(\mathcal{E}_{0,1+\eta{}},\eta{})\mathcal{E}_{0,1}^2v^{-1}, &\label{kappa-estimate-strong} \\
0\le \lambda{}\le C\exp (C\mathcal{E}_0^2),         &  \label{lambda-bound} \\
0\le(-\nu)\le C. \label{nu-bound}
\end{align}
\item Away from the horizon, we have
\begin{align}
\abs{\log (-\gamma{})} \le Cr^{-1}\tau^{-p}\mathcal{E}_{0,p}^2 & \quad \text{in }\set{r\ge R}, \label{gamma-estimate} \\
\abs{\log (-\nu)} \le Cr^{-1} & \quad \text{in }\set{r\ge R}, \label{nu-estimate} \\
1-\mu{} \ge 1/2                                 & \quad \text{in }\set{r\ge R},  \label{mu-estimate}                      \\
\lambda \ge C\exp (-C\mathcal{E}_0^2)                                 & \quad \text{in }\set{r\ge R},  \label{lambda-lower-bound}                      \\
(-\nu) \ge 1/2                                 & \quad \text{in }\set{r\ge R}. \label{nu-lower-bound}
\end{align}
\item Near the horizon, there are constants \(c_\nu,c'_\nu,C> 0\) (where \(c'_\nu\) depends on \(\mathcal{E}_0\) and is independent of \(R\) and the other constants depend on \(\eta\) and \(\mathcal{E}_{0,1+\eta}\)) such that
\begin{equation}\label{nu-exponential-bound}
C^{-1}\exp (-c_\nu'(u-v))\le (-\nu{})\le C\exp (-c_\nu(u-v))\quad \text{in }\set{r\le R}.
\end{equation}
\item We have
\begin{equation} \label{v-u-r-estimate}
\abs{v-u-r}\le C(\mathcal{E}_{0,1+\eta{}},\eta{})\log r \quad  \text{in }\set{r\ge R}.
\end{equation}
\item Let \(\epsilon{}\in(0,1)\). For \(G_\epsilon{}\) as defined in \cref{G-epsilon}, there is \(C > 0\) depending on \(\mathcal{E}_{0,1+\eta{}}\), \(\eta\), \(\epsilon{}\), and \(\Rc\) such that
\begin{equation}\label{G-estimate}
C^{-1}\le G_\epsilon{}\le C
\end{equation}
\end{enumerate}
\end{proposition}
\begin{proof}
The estimates
\cref{mass-redshift,mass-bounded,mu-estimate-weak,kappa-sign,kappa-estimate,kappa-estimate-strong,gamma-sign,lambda-bound,nu-bound,gamma-estimate,mu-estimate,lambda-lower-bound,nu-lower-bound,nu-exponential-bound}
follow as in the proof of \cite[Prop.~5.12]{luk-oh-scc2}. Although \cite{luk-oh-scc2}
uses the \((u,\tilde{v})\) gauge, the relevant arguments go through in the
\((u,v)\) gauge. We now show the remaining estimates
\cref{nu-estimate,v-u-r-estimate,G-estimate}.

First, \cref{nu-estimate} follows from the \(\partial{}_v\)-transport equation for
\((-\nu{})\) and \((-\nu{})|_{\mathcal{I}} = 1\) (which follows from \((-\nu{})
= (1-\mu{})(-\gamma{})\), the gauge condition \((-\gamma{})|_{\mathcal{I}} =
1\), and \(1 - \mu{} = 1 + O(r^{-1})\)).

The estimate \cref{v-u-r-estimate} follows from \cref{nu-estimate} (which implies
\(\abs{1 +\nu{}}\le Cr^{-1}\)) and an argument as in the proof of
\cite[Prop.~5.12]{luk-oh-scc2}.

Finally, we establish \cref{G-estimate}. For the remainder of the proof, we
allow the symbol \(\lesssim\) to include dependence on
\(C(\mathcal{E}_{0,1+\eta{}},\eta{})\) as in \cref{v-u-r-estimate}, as well as
\(\epsilon{}\) and \(\Rc\). It is enough to show
\begin{align}
r&\lesssim v, \label{Ge-rv}\\
u&\lesssim v \quad  \text{in }\set{r\ge R_0}, \label{Ge-uv}\\
v&\lesssim u \quad \text{in }\set{r\le \epsilon{}^{-1}\Rc}\cup \set{r\le (1-\epsilon{})v}.\label{Ge-vu}
\end{align}
Indeed, \cref{Ge-rv,Ge-uv,Ge-vu} imply the upper bound in \cref{G-estimate}, and the
lower bound \(G_\epsilon{} \ge 1\) is clear. In \(\set{r\ge R_0}\), we have
\(r-R_0\le v/2\) by \cref{lambda-lower-bound}. Thus \(r\le (R_0+1/2)v\), which
establishes \cref{Ge-rv}. In \(\set{r\ge R_0}\), by \cref{v-u-r-estimate}, we have
\(u\lesssim v + r + \log r\lesssim v\), which establishes \cref{Ge-uv}. If \(r\le
\epsilon{}^{-1}\Rc\), then \(v\lesssim u + \epsilon{}^{-1}\Rc + \log
(\epsilon{}^{-1}\Rc)\) by \cref{v-u-r-compare}, so \(v\lesssim u\) in this region.
If \(r\le  (1-\epsilon{})v\), then \(v = u + r + O(\log r) \le
u+(1-\epsilon{})v + O(\log v)\), so for \(v\ge v_\ast{}\) and \(v_\ast{}\) large
enough we have \(v\le 2\epsilon{}^{-1}u\). For \(v\le v_\ast{}\) we have \(v\le
v_\ast{}u\) since \(u\ge 1\). Thus \(v\le \max (2\epsilon{}^{-1},v_\ast{})u\) in this region.
\end{proof}
As an immediate corollary of \cref{zeroth-order-geometric-bounds}, we control the
zeroth order schematic geometric quantities.
\begin{corollary}[Estimates for zeroth order schematic geometric quantities]
We have
\begin{equation}\label{zeroth-geometric-bound}
\begin{split}
C(\mathfrak{b}_0)&\le C(\mathcal{E}_0,\varpi{}_i,c_{\mathcal{H}},r_{\mathrm{min}},R_0,\eta{}_0), \\
C(\mathfrak{B}_0^{\circ })&\le C(\mathfrak{b}_0,\mathcal{E}_{0,1}), \\
C(\mathfrak{B}_0)&\le C(\Rc,\mathfrak{B}_0^{\circ },\mathcal{E}_{0,1+\eta{}_0}), \\
C(\mathfrak{g}_0)&\le C(\mathfrak{b}_0), \\
C(\mathfrak{G}_0)&\le C(\mathfrak{b}_0). \\
\end{split}
\end{equation}
\label{zeroth-order-schematic-control}
\end{corollary}
We now use \cref{v-u-r-compare} to compare the variables \(r\), \(u\), \(v\), and
\(\tau{}\).
\begin{lemma}
The following estimates hold for \(\epsilon{}\in (0,1)\):
\begin{align}
r + \mathbf{1}_{r\ge R_0}u\le G_\epsilon{}v, & \label{v-u-r-compare-1}\\
v\le G_\epsilon{}u,&\quad \text{in }\set{r\le \epsilon{}^{-1}\Rc}\cup \set{r\le (1-\epsilon{})v} \label{v-u-r-compare-2}\\
v\le (G_\epsilon{} + \epsilon{}^{-1})(u+r) \label{v-u-r-compare-3}.
\end{align}
Moreover, we have
\begin{equation}\label{tau-estimates}
\begin{split}
G_\epsilon^{-1}v \le \tau{} \le G_\epsilon{}v&\quad \text{in }\set{r\le \epsilon{}^{-1}\Rc}\cup \set{r\le (1-\epsilon{})v} \\
\tau{} = u&\quad\text{in } \set{r\ge R_0}.
\end{split}
\end{equation}
\label{v-u-r-compare}
\end{lemma}
\begin{proof}
The definition of \(G\) implies \cref{v-u-r-compare-1}, as well as
\cref{v-u-r-compare-2} in the region \(\set{r\ge R_0}\). To extend
\cref{v-u-r-compare-2} to the region \(\set{r\le R_0}\), note that if \(r(u,v)\le
R_0\), then \(u\ge u_{R_0}(v)\), and so \(v\le Gu_{R_0}(v)\le Gu\). Next,
\cref{v-u-r-compare-3} follows from \cref{v-u-r-compare-2}, because
\begin{equation}
v\le \mathbf{1}_{r\le (1-\epsilon{})v}v + \mathbf{1}_{r\ge (1-\epsilon{})v}v \le G_\epsilon{}u + (1-\epsilon{})^{-1}r.
\end{equation}

By \cref{v-u-r-compare-1,v-u-r-compare-2}, we have
\begin{equation}
G_\epsilon^{-1}v \le u\le G_\epsilon{}v\quad \text{in }(\set{r\le \epsilon{}^{-1}\Rc}\cup \set{r\le (1-\epsilon{})v})\cap \set{r\ge R_0}.
\end{equation}
The estimates in \cref{tau-estimates} now follow from the definition of \(\tau{}\) (see
\cref{tau-def}).
\end{proof}
\section{Boundedness of initial data}
\label{sec:org8eaf06d}
\label{sec:boundedness-initial-data} In this section we explain how the initial data
norm \(\mathcal{D}_\alpha{} = \mathcal{D}[\Gamma^\alpha{}\varphi{}]\) (which is not
gauge invariant) is controlled by the gauge invariant initial data norm
\(\mathfrak{D}_{\abs{\alpha{}}}\). See \cref{sec:norms} for definitions of these
quantities.
\begin{proposition}[Boundedness of initial data]
If the data vanishes in a neighbourhood of \(C^{\text{out}}\), then we have
\begin{equation}\label{boundedness-initial-data-trivial}
\mathcal{D}_0 = \mathfrak{D}_0.
\end{equation}
and for \(\abs{\alpha{}}\ge 1\) we have
\begin{equation}\label{boundedness-initial-data-eq}
\mathcal{D}_\alpha{}\le C(\mathfrak{b}_{<\alpha{}},r^{-1}\mathfrak{g}_{<\alpha{}})\mathfrak{D}_{\abs{\alpha{}}}.
\end{equation}
\label{boundedness-initial-data}
\end{proposition}
\begin{remark}
The restriction that the solution vanish in a neighbourhood of the outgoing null
part of the data hypersurface can be removed, but for the simplicity of the
argument we only give the proof in this case. Indeed, we have reduced to this
case in \cref{sec:reduction-to-characteristic} using the compact support of the
Cauchy data.
\end{remark}
\begin{proof}
First, \cref{boundedness-initial-data-trivial} is trivial. Define
\begin{equation}
\mathfrak{D}_n^{\mathrm{in}}[\psi{}] \coloneqq{} \sup_{C^{\mathrm{in}}} \abs{U^{\le n+1}\psi{}} \qquad \mathfrak{D}_n^{\mathrm{out}}[\psi{}] \coloneqq{} \sup_{C^{\mathrm{out}}} \abs{(r^2\overline{\partial}_r)^{\le 1}(r\overline{\partial}_r)^{\le n}(r\psi{})},
\end{equation}
so that
\begin{equation}
\mathfrak{D}_n[\psi{}] = \mathfrak{D}_n^{\mathrm{in}}[\psi{}] + \mathfrak{D}_n^{\mathrm{out}}[\psi{}].
\end{equation}
The vanishing assumption on the data implies by the domain of dependence
property that the solution vanishes in a neighbourhood of \(C^{\text{out}}\),
and so \(\mathfrak{D}_n^{\text{out}}[\Gamma^\alpha{}\varphi{}] = 0\). In
particular, it is enough to estimate
\(\mathfrak{D}_n^{\text{in}}[\Gamma^\alpha{}\varphi{}]\). The desired estimate
\cref{boundedness-initial-data-eq} follows from the \(n = 0\) case of
\cref{in-en-bound-step-3-1}:
\begin{align}
\mathfrak{D}_n^{\mathrm{in}}[\Gamma{}^\alpha{}\varphi{}]&\le C(\mathfrak{b}_{\alpha{}-1+nU},r^{-1}\mathfrak{g}_{\alpha{}-1+nU})\mathfrak{D}_{n+\abs{\alpha{}}}^{\text{in}}[\varphi{}]. \label{in-en-bound-step-3-1}
\end{align}
We now prove \cref{in-en-bound-step-3-1}.

\step{Step 1: Reduction to \cref{Din-bound-V}.} By induction, \cref{in-en-bound-step-3-1}
follows from
\begin{equation}
\mathfrak{D}_n^{\mathrm{in}}[\Gamma{}^\alpha{}\varphi{}]\le C(\mathfrak{b}_{\alpha{}-1+nU},r^{-1}\mathfrak{g}_{\alpha{}-1+nU})\mathfrak{D}_{n+1}^{\text{in}}[\Gamma{}^{\le \alpha{}-1}\varphi{}].
\end{equation}
Since \(v\equiv 1\) on \(C^{\text{in}}\) and
\(Uv = 0\), we have \(S = V\), and so we can assume \(\alpha_S = 0\). Moreover,
\(U\) commutes with \(V\) on \(C^{\text{in}}\subset \set{r\le \Rc}\). Clearly we have
\begin{equation}\label{Din-bound-U}
\mathfrak{D}_n^{\text{in}}[U\Gamma{}^\alpha{}\varphi{}] \le  \mathfrak{D}_{n+1}^{\text{in}}[\Gamma{}^\alpha{}\varphi{}],
\end{equation}
so it is enough to show that
\begin{equation}\label{Din-bound-V}
\mathfrak{D}_n^{\text{in}}[V\Gamma{}^\alpha{}\varphi{}]\le C(\mathfrak{b}_{\alpha{}+nU},r^{-1}\mathfrak{g}_{\alpha{}+nU})\mathfrak{D}_{n}^{\text{in}}[U\Gamma{}^\alpha{}\varphi{}].
\end{equation}

\step{Step 2: Proof of \cref{Din-bound-V}.} Use \cref{UV-box,wave-comm-formula} and an
induction argument to compute
\begin{equation}
UV\Gamma{}^\alpha{}\varphi{} =_{\mathrm{s}} \mathcal{O}(\mathfrak{b}_\alpha{},r^{-1}\mathfrak{g}_\alpha{})[V\Gamma{}^{\le \alpha{}}\varphi{} + U^{\le 2}\Gamma{}^{\le \alpha{}}\varphi{}]
\end{equation}
It follows from an induction argument that
\begin{equation}\label{Din-bound-1}
\abs{U^nUV\Gamma{}^\alpha{}\varphi{}} \le   C(\mathfrak{b}_{\alpha{}+nU},r^{-1}\mathfrak{g}_{\alpha{}+nU})[\abs{V\Gamma{}^{\le \alpha{}}\varphi{}} + \abs{U^{\le n+1}\Gamma{}^{\le \alpha{}}\varphi{}}]\le C(\mathfrak{b}_{\alpha{}+nU},r^{-1}\mathfrak{g}_{\alpha{}+nU})[\abs{V\Gamma{}^{\le \alpha{}}\varphi{}} + \mathfrak{D}_{n}^{\text{in}}[\Gamma{}^\alpha{}\varphi{}]].
\end{equation}
Integrate \cref{Din-bound-1} for \(n = 0\) to \(C^{\text{in}}\cap C^{\text{out}}\)
(where the data vanishes) to get
\begin{equation}
\begin{split}
\abs{V\Gamma{}^\alpha{}\varphi{}}|_{v=1}(r)&\le \abs{V\Gamma{}^\alpha{}\varphi{}}|_{v=1}(r=R_0) + \int_{r}^{R_0} \abs{UV\Gamma{}^\alpha{}\varphi{}}\dd{}r'\\
&\le C(\mathfrak{b}_\alpha{},r^{-1}\mathfrak{g}_\alpha{})\mathfrak{D}_{0}^{\text{in}}[U\Gamma{}^\alpha{}\varphi{}] + C(\mathfrak{b}_\alpha{},r^{-1}\mathfrak{g}_\alpha{})\int_{r}^{R_0} \abs{V\Gamma{}^{\le \alpha{}}\varphi{}}\dd{}r'.
\end{split}
\end{equation}
Grönwall's inequality and the finite \(r\)-range on \(C^{\text{in}}\) gives
\begin{equation}\label{Din-bound-2}
\sup_{C^{\text{in}}}\abs{V\Gamma{}^\alpha{}\varphi{}}\le C(\mathfrak{b}_\alpha{},r^{-1}\mathfrak{g}_\alpha{})\mathfrak{D}_{0}^{\text{in}}[U\Gamma{}^\alpha{}\varphi{}].
\end{equation}
Substitute \cref{Din-bound-2} into \cref{Din-bound-1} and use \cref{Din-bound-U} to
obtain
\begin{equation}\label{Din-bound-3}
\sup_{C^{\text{in}}}\abs{U^nUV\Gamma^\alpha{}\varphi{}} \le C(\mathfrak{b}_{\alpha{}+nU},r^{-1}\mathfrak{g}_{\alpha{}+nU})\mathfrak{D}_{n}^{\text{in}}[U\Gamma{}^\alpha{}\varphi{}].
\end{equation}
Now \cref{Din-bound-2,Din-bound-3} imply \cref{Din-bound-V}.
\end{proof}
\section{Energy boundedness and decay}
\label{sec:org75fffcb}
\label{sec:energy-estimates}
The goal of this section is to show the following proposition.
\begin{proposition}[Estimates for unweighted and weighted energy norms]
We have
\begin{equation}\label{E-alpha-bound-0}
\mathcal{E}_0\le C(\varpi{}_i,c_{\mathcal{H}},r_{\mathrm{min}})\mathcal{D}_0.
\end{equation}
and
\begin{equation}\label{E-0-p-bound-statement}
\mathcal{E}_{0,2-\eta{}_0}\le C(\mathfrak{b}_0,\mathfrak{g}_0,\mathcal{P}_{0,0})\mathcal{D}_0.
\end{equation}
Let \(\abs{\alpha{}}\ge 1\). If \(\Rc \ge C(\mathfrak{B}_0^{\circ },\mathfrak{g}_0,\alpha{})\), then for
\(s\ge \eta{}_0\) sufficiently small (depending on a numerical constant), we have
\begin{equation}\label{E-alpha-bound-higher}
\mathcal{E}_\alpha{} \le C(\mathfrak{b}_{\alpha{}},r^{-s}\mathfrak{g}_{\alpha{}},\mathcal{P}_{<\alpha{},1+\eta{}_0})[\mathcal{D}_\alpha + \mathcal{E}_{<\alpha{},4s}],
\end{equation}
and when \(s\) is moreover small depending on \(\alpha{}\), we have
\begin{equation}\label{E-alpha-p-bound-statement}
\mathcal{E}_{\alpha{},2-\eta{}_0-C_\alpha{}s} \le C(\mathfrak{b}_{\alpha{}},r^{-s}\mathfrak{g}_{\alpha{}},\mathcal{P}_{\alpha{},0},\mathcal{P}_{<\alpha{},1+\eta{}_0})\mathcal{D}_\alpha{}
\end{equation}
for explicit constants \(C_\alpha{}\) depending only on \(\alpha{}\).
\label{E-alpha-bound}
\end{proposition}
\begin{proof}
The zeroth order estimate \cref{E-alpha-bound-0} follows from
\cref{energy-norm-relation}, \cref{energy-on-data}, and \cref{eb:no-bulk} from
\cref{energy-estimate}. The higher order estimate \cref{E-alpha-bound-higher} is
established in \cref{E-alpha-bound-2}. The boundedness statements
\cref{E-0-p-bound-statement,E-alpha-p-bound-statement} for the weighted energy are
proved in \cref{lem:E-alpha-p-bound}.
\end{proof}
\subsection{Hardy inequalities}
\label{sec:org4303dab}
Both of these lemmas and their proofs can be found in \cite[\S{}8.8]{luk-oh-scc1}.
\begin{lemma}
Let \(a\in \R\), \(1\le u_1<u_2\), and \(1\le v_1<v_2\). For any \(C^1\) function \(f\)
on \(\set{(u,v) : v\in [v_1,v_2]}\), we have
\begin{equation}
\begin{split}
& a\int_{v_1}^{v_2} r^a\lambda{}f^2(u,v)\dd{}v + \int_{v_1}^{v_2} \frac{r^a}{\lambda{}}(\partial{}_v(rf))^2(u,v)\dd{}v + r^{1+a}f^2(u,v_1) \\
&= \int_{v_1}^{v_2} \frac{r^{2+a}}{\lambda{}}(\partial{}_vf)^2(u,v)\dd{}v + r^{1+a}f^2(u,v_2).
\end{split}
\end{equation}
For any \(C^1\) function \(f\)
on \(\set{(u,v) : u\in [u_1,u_2]}\), we have
\begin{equation}
\begin{split}
& a\int_{u_1}^{u_2} r^a(-\nu{})f^2(u,v)\dd{}u + \int_{u_1}^{u_2} \frac{r^a}{(-\nu{})}(\partial{}_u(rf))^2(u,v)\dd{}u + r^{1+a}f^2(u_2,v) \\
&= \int_{u_1}^{u_2} \frac{r^{2+a}}{(-\nu{})}(\partial{}_uf)^2(u,v)\dd{}u + r^{1+a}f^2(u_1,v).
\end{split}
\end{equation}
\label{hardy-type-inequality}
\end{lemma}
\begin{lemma}
Let \(a\in \R\), \(1\le u_1<u_2\), and \(1\le v_1<v_2\). For any \(C^1\) function \(f\)
on \(\set{(u,v) : v\in [v_1,v_2]}\), we have
\begin{equation}
\begin{split}
&\frac{(a+1)^2}{4}\int_{v_1}^{v_2} r^a\lambda{}f^2(u,v)\dd{}v + \int_{v_1}^{v_2} \frac{r^a}{\lambda{}}\Bigl(r\partial{}_vf + \frac{a+1}{2}\lambda{}f\Bigr)^2(u,v)\dd{}v + \frac{a+1}{2}r^{1+a}f^2(u,v_1) \\
&= \int_{v_1}^{v_2} \frac{r^{2+a}}{\lambda{}}(\partial{}_vf)^2(u,v)\dd{}v + \frac{a+1}{2}r^{1+a}f^2(u,v_2).
\end{split}
\end{equation}
For any \(C^1\) function \(f\)
on \(\set{(u,v) : u\in [u_1,u_2]}\), we have
\begin{equation}
\begin{split}
&\frac{(a+1)^2}{4}\int_{u_1}^{u_2} r^a(-\nu{})f^2(u,v)\dd{}u + \int_{u_1}^{u_2} \frac{r^a}{(-\nu{})}\Bigl(r\partial{}_uf + \frac{a+1}{2}(-\nu{})f\Bigr)^2(u,v)\dd{}u + \frac{a+1}{2}r^{1+a}f^2(u_2,v) \\
&= \int_{u_1}^{u_2} \frac{r^{2+a}}{(-\nu{})}(\partial{}_uf)^2(u,v)\dd{}u + \frac{a+1}{2}r^{1+a}f^2(u_1,v).
\end{split}
\end{equation}
\label{hardy-inequality}
\end{lemma}
\subsection{Energy quantities}
\label{sec:orgc72defb}
In order to state the energy boundedness estimate, we introduce notation for
energy quantities along null curves and over spacetime regions.
\subsubsection{Energy along piecewise null curves}
\label{sec:org6d981b3}
\label{sec:energy-along-null-curve}
For \(\Sigma{}\) a null curve, we define
\begin{equation}
E[\psi{},\Sigma{}] \coloneqq{} \begin{cases}
\vspace{1ex}\displaystyle\int_\Sigma{} \dfrac{r^2}{\kappa{}}(\partial_v\psi{})^2\dd{}v & \text{if }u\text{ is constant on }\Sigma, \\
\displaystyle\int_\Sigma{} \dfrac{r^2}{(-\nu{})}(\partial{}_u\psi{})\dd{}u & \text{if }v\text{ is constant on }\Sigma. \\
\end{cases}
\end{equation}
We extend the definition to piecewise null curves \(\Sigma{}\) in the natural way: if
\(\Sigma{} = \bigcup_{i}\Sigma_i\) for null curves \(\Sigma_i\) such that
\(\Sigma_i\cap \Sigma_j\) consists of at most one point when \(i\neq{}j\), then
we set
\begin{equation}
E[\psi{},\Sigma{}] \coloneqq{} \sum_{i}E[\psi{},\Sigma_i].
\end{equation}
\begin{remark}
The energy norm \(\mathcal{E}[\psi{}]\) defined in \cref{sec:norms:energy} is related to
\(E[\psi{},\Sigma{}]\) as follows:
\begin{equation}
\mathcal{E}[\psi{}]^2 = \sup_\Sigma{} E[\psi{},\Sigma{}],
\end{equation}
where the supremum is taken over all null curves \(\Sigma{}\subset \set{u,v\ge 1}\).
\label{energy-norm-relation}
\end{remark}
\begin{lemma}[Energy on initial data]
Let \(p\in C^{\mathrm{in}}\cup C^{\mathrm{out}}\). Then
\begin{equation}
E[\psi{},C^{\mathrm{in}}\cup C^{\mathrm{out}}] + r\psi{}^2(p)\le C(r_{\mathrm{min}},\varpi{}_i,R_0)\mathfrak{D}_0[\psi{}]^2.
\end{equation}
\label{energy-on-data}
\end{lemma}
\begin{proof}
The proof is immediate by the definitions, \cref{mu-estimate-weak}, and a change of
variables.
\end{proof}
\subsubsection{\texorpdfstring{\(r^p\)}{rp}-weighted energy}
\label{sec:org73afe5b}
\label{sec:rp-weighted-energy}
For \(\tau{}_0\ge 1\), define the level set \(\Sigma_{\tau{}_0}\coloneqq{}\set{p : \tau{}(p) = \tau{}_0}\), where
\(\tau{}\) was defined in \cref{tau-def}. Then \(\Sigma_\tau{}\) is piecewise null, with null pieces
\begin{equation}
\Sigma{}_\tau^{\text{in}} \coloneqq{} \Sigma_\tau{}\cap \set{r\le R_0} = \set{(u,v_{R_0}(u)) : u\ge \tau{}}\qquad \Sigma{}_\tau^{\text{out}} \coloneqq{} \Sigma{}_\tau{}\cap
\set{r\le R_0} = \set{(\tau{},v) : v\ge v_{R_0}(\tau{})}.
\end{equation}
\begin{figure}
    \centering
    \vspace{-4ex}
    \def\svgwidth{0.5\columnwidth}
    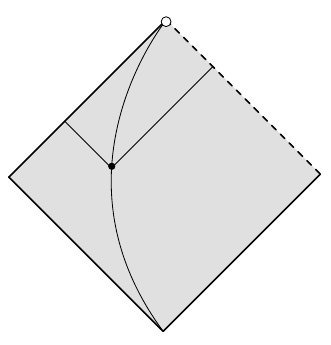

    \caption{The foliation \(\Sigma_\tau\).}
    \label{fig:rp}
\end{figure}
We introduce the \(r^p\)-weighted energy flux associated to the foliation
\(\Sigma_\tau{}\):
\begin{equation}
E_p[\psi{}](\tau{}) \coloneqq{} \int _{\Sigma{}_\tau^{\text{in}}}\frac{r^2}{(-\nu{})}(\partial{}_u\psi{})^2\dd{}u +  \int _{\Sigma{}_\tau^{\mathrm{out}}}\frac{r^p}{\lambda{}}(\partial{}_v(r\psi{}))^2\dd{}v + \abs{\psi{}}^2(\tau{},v_{R_0}(\tau{})).
\end{equation}
We have the following preliminary estimate comparing \(E_p[\psi{}](\tau{})\) and the
energy quantities \(E[\psi{},\Sigma{}_\tau{}]\) (defined in \cref{sec:energy-estimate-proof}).
\begin{lemma}
For any \(p\ge 0\) and \(1\ge \tau{}\), we have
\begin{equation}
E[\psi{},\Sigma{}_\tau{}]\le C(\mathfrak{b}_0)E_p[\psi{}](\tau{}).
\end{equation}
\label{rp-unweighted-energy-bound}
\end{lemma}
\begin{proof}
It is enough to consider the case \(p = 0\),
which is an immediate consequence of a Hardy type inequality
(\cref{hardy-type-inequality} with \(a = 0\) and \(v_1 = v_{R_0}(\tau{})\) in the
limit \(v_2\to \infty\)).
\end{proof}
\begin{lemma}
For \(p < 3\), we have
\begin{equation}
E_p[\psi{}](1)\le C(p,R_0)\mathfrak{D}_0[\psi{}]^2.
\end{equation}
\label{rp:initial-energy-bound}
\end{lemma}
\begin{proof}
This is immediate from the definitions and a change of variables:
\begin{equation}
E_p[\psi{}](1) = \int_{r_{\text{min}}}^{R_0} r^2(U\psi{})^2\dd{}r + \int_{R_0}^\infty r^p(\overline{\partial}_r(r\psi{}))^2\dd{}r + \abs{\psi{}}^2(1,1)\le R_0^2\mathfrak{D}_0[\psi{}]^2 + \mathfrak{D}_0[\psi{}]^2 +  \mathfrak{D}_0[\psi{}]^2\int_{R_0}^\infty r^{p-4}\dd{}r.
\end{equation}
The condition \(p < 3\) is used to ensure that the final integral is finite.
\end{proof}
\subsubsection{Bulk energy over spacetime region}
\label{sec:org89429b4}
\label{sec:bulk-energy} Let \(\mathcal{R}\) be a spacetime region. Define a bulk
energy quantity:
\begin{equation}
\begin{split}
E_{\mathrm{bulk}}[\psi{},\mathcal{R}] &\coloneqq{}\iint_{\mathcal{R}}  \frac{1}{r^{1+\eta{}_0}}\Bigl[\frac{1}{(-\nu{})^2}(\partial{}_u\psi{})^2 + \frac{1}{\kappa{}^2}(\partial{}_v\psi{})^2 + \frac{\psi{}^2}{r^2}\Bigr]r^{2}\kappa{}(-\nu{})\dd{}u\dd{}v \\
&\qquad  + \iint_{\mathcal{R}\cap \set{r\le R_0}} e^{-c_{\nu{}}'(u-v)} (\partial{}_v\psi{})^2\dd{}u\dd{}v,
\end{split}
\end{equation}
where \(\eta_0\) is our global small constant (see \cref{sec:notation:constants}) and
\(c_\nu'\) is as in \cref{nu-exponential-bound}.
\subsection{Vector field multiplier identities}
\label{sec:orgaaf7514}
\label{sec:vector-field-multiplier-identities}
In this section we record the vector field multiplier identities used to derive
an energy estimate in \cref{sec:energy-estimates}. We introduce the notation
\begin{equation}\label{Du-Dv-notation}
D_u = U = \frac{1}{(-\nu{})}\partial{}_u\qquad D_v = \frac{1}{\kappa{}}\partial{}_v.
\end{equation}
The vector field multipliers we use are
\begin{align}
T &=  (1-\mu{})D_u + D_v, & \text{(Kodama vector field)} \\
X &= f(r)((1-\mu{})D_u - D_v), &\text{(Morawetz vector field)} \\
Y &= \chi{}_{\mathcal{H}}(r)D_u, & \text{(Redshift vector field)} \\
Z &=  g(u,v)\partial{}_v. &\text{(Irregular vector field)}
\end{align}
\begin{lemma}[Vector field multiplier identity]
Let \(W = W^uD_u + W^vD_v = (-\nu{})^{-1}W^u\partial_u + \kappa^{-1}W^v\partial_v\). Then
\begin{equation}
\partial{}_u(W^v(D_v\psi{})^2r^2\kappa{}) + \partial{}_v(W^u(D_u\psi{})^2r^2(-\nu{})) = [2W\psi{}\Box{}\psi{} + K^W[\psi{}]]r^2\kappa{}(-\nu{}),
\end{equation}
where
\begin{equation}
K^W[\psi{}] \coloneqq{} (-\nu{})D_v((-\nu{})^{-1}W^u)(D_u\psi{})^2 + \kappa{}D_u(\kappa{}^{-1}W^v)(D_v\psi{})^2 + \frac{2}{r}(-W^u + (1-\mu{})W^v)D_u\psi{}D_v\psi{}.
\end{equation}
\label{vf-multiplier}
\end{lemma}
\begin{proof}
Note that \(W = \widetilde{W}^u\partial_u + \widetilde{W}^v\partial_v\) for
\((\widetilde{W}^u,\widetilde{W}^v) = ((-\nu{})^{-1}W^u,\kappa{}^{-1}W^v)\). Use
the wave equation in double null coordinates to compute
\begin{equation}
\begin{split}
r^2\kappa{}(-\nu{})\partial{}_u\psi{}\Box{}\psi{} &= r^2\partial{}_u\psi{}\Bigl(\partial{}_u\partial{}_v\psi{} + \frac{\partial{}_vr}{r}\partial{}_u\psi{} + \frac{\partial{}_ur}{r}\partial{}_v\psi{}\Bigr) =  r^2\partial{}_u\psi{}\partial{}_v\partial{}_u\psi{} + r\partial{}_vr(\partial{}_u\psi{})^2 + r\partial{}_ur\partial{}_u\psi{}\partial{}_v\psi{} \\
&= \frac{1}{2}\partial{}_v(r^2(\partial{}_u\psi{})^2) + r\partial{}_ur\partial{}_u\psi{}\partial{}_v\psi{}. \\
\end{split}
\end{equation}
Multiply by \(\widetilde{W}^u\) and commute \(W^{u}\) past \(\partial_{v}\) to arrive at
\begin{equation}\label{vf-multiplier-1}
r^2\kappa{}(-\nu{})\widetilde{W}^u\partial{}_u\psi{}\Box{}\psi{} = \frac{1}{2}\partial{}_v(\widetilde{W}^ur^2(\partial{}_u\psi{})^2) - \frac{1}{2}\partial{}_v\widetilde{W}^ur^2(\partial{}_u\psi{})^2 + r\widetilde{W}^u\partial{}_ur\partial{}_u\psi{}\partial{}_v\psi{}.
\end{equation}
Add \cref{vf-multiplier-1} to the result of exchanging \(u\) and \(v\) in
\cref{vf-multiplier-1} to obtain
\begin{equation}\label{vf-multiplier-2}
\begin{split}
r^2\kappa{}(-\nu{})W\psi{}\Box{}\psi{} &= \frac{1}{2}\partial{}_v(\widetilde{W}^ur^2(\partial{}_u\psi{})^2) + \frac{1}{2}\partial{}_u(\widetilde{W}^vr^2(\partial{}_v\psi{})^2) - \frac{1}{2}\partial{}_u\widetilde{W}^vr^2(\partial{}_v\psi{})^2 - \frac{1}{2}\partial{}_v\widetilde{W}^ur^2(\partial{}_u\psi{})^2 \\
&\qquad + (r\widetilde{W}^v\partial{}_vr + r\widetilde{W}^u\partial{}_ur)\partial{}_u\psi{}\partial{}_v\psi{}.
\end{split}
\end{equation}
Express \cref{vf-multiplier-2} in terms of \(D_u\) and \(D_v\), then rearrange to
obtain the desired result.
\end{proof}
\begin{lemma}[Kodama vector field identity]
Let \(T = (1-\mu{})D_u + D_v\). We have
\begin{equation}
\partial{}_u((D_v\psi{})^2r^2\kappa{}) + \partial{}_v((1-\mu{})(D_u\psi{})^2r^2(-\nu{})) = 2T\psi{}\Box{}\psi{} r^2\kappa{}(-\nu{})+ r[-(D_u\psi{})^2(D_v\varphi{})^2 + (D_u\varphi{})^2(D_v\psi{})^2]r^2\kappa{}(-\nu{}),
\end{equation}
\label{T-identity}
\end{lemma}
\begin{proof}
Use \cref{vf-multiplier} with \((T^u,T^v) = (1-\mu{},1)\). In particular we use the
transport equations
\begin{equation}\label{vf-multiplier-transport}
D_v(-\gamma{})^{-1} = -(-\nu{})^{-1}r(D_v\varphi{})^2 \qquad D_u\kappa{}^{-1} = \kappa{}^{-1}r(D_u\varphi{})^2
\end{equation}
to compute
\begin{equation}
K^T[\psi{}] = (-\nu{})D_v((-\gamma{})^{-1})(D_u\psi{})^2 + \kappa{}D_u(\kappa{}^{-1})(D_v\psi{})^2 = r[-(D_u\psi{})^2(D_v\varphi{})^2 + (D_u\varphi{})^2(D_v\psi{})^2].
\end{equation}
\end{proof}
\begin{lemma}[Morawetz vector field]
Let \(f : (0,\infty)_r \to \R\) be a \(C^2\) function. Let \(X = (1-\mu{})f(r)D_u -
f(r)D_v\). Then
\begin{equation}\label{X-estimate-equation}
\begin{split}
&-f'(r)[(1-\mu{})^2(D_u\psi{})^2 + (D_v\psi{})^2]r^2\kappa{}(-\nu{}) + 4 \frac{f(r)}{r}(1-\mu{})D_u\psi{}D_v\psi{}r^2\kappa{}(-\nu{})\\
&\le\partial{}_u(f(r)(D_v\psi{})^2r^2\kappa{}) - \partial{}_v(f(r)(1-\mu{})(D_u\psi{})^2r^2(-\nu{})) + 2X\psi{}\Box{}\psi{}r^2\kappa{}(-\nu{}).
\end{split}
\end{equation}
Moreover,
\begin{equation}\label{X-estimate-equation-2}
\begin{split}
&-f'(r)[(1-\mu{})^2(D_u\psi{})^2 + (D_v\psi{})^2]r^2\kappa{}(-\nu{})  + 4\frac{f(r)}{r}\mu{}D_{u}\psi{}D_{v}\psi{}r^2\kappa{}(-\nu{}) \\
&\qquad +  \Bigl[\frac{2}{1-\mu{}}\frac{f''(r)}{r} + (f'(r)r - f(r)) \frac{4(\varpi{}-\mathbf{e}^2/r)}{r^4}\Bigr]\psi{}^2r^2\kappa{}(-\nu{})\\
&\le\partial{}_u(f(r)(D_v\psi{})^2r^2\kappa{} - 2f'(r)r\lambda{}\psi{}^2 + f(r)r\psi{}^2) - \partial{}_v(f(r)(1-\mu{})(D_u\psi{})^2r^2(-\nu{}) + 2f'(r)r\nu{}\psi{}^2 - f(r) r\psi{}^2) \\
&\qquad + \Bigl(2X\psi{} + 4f(r)\frac{\psi{}}{r}\Bigr)\Box{}\psi{}r^2\kappa{}(-\nu{}).
\end{split}
\end{equation}
\label{X-estimate}
\end{lemma}
\begin{proof}
Use \cref{vf-multiplier} with \((X^u,X^v) = ((1-\mu{})f(r),-f(r))\) to obtain
\begin{equation}
K^{-X}[\psi{}]r^2\kappa{}(-\nu{}) = \partial_u(f(r)(D_v\psi{})^2r^2\kappa{}) - \partial_v(f(r)(1-\mu{})(D_u\psi{})^2r^2(-\nu{})) + 2X\psi{}\Box{}\psi{}.
\end{equation}
From \cref{vf-multiplier-transport} we can compute
\begin{equation}
-(-\nu{})D_v(f(r)(-\gamma{})^{-1}) = f(r)r(D_v\varphi{})^2 - f'(r)(1-\mu{})^2 \qquad \kappa{}D_u(f(r)\kappa{}^{-1}) = f(r)r(D_u\varphi{})^2 - f'(r),
\end{equation}
and then obtain
\begin{equation}
K^{-X}[\psi{}] = [f(r)r(D_v\varphi{})^2 - f'(r)(1-\mu{})^2](D_u\psi{})^2 + [f(r)r(D_u\varphi{})^2 - f'(r)](D_v\psi{})^2 + \frac{4f(r)}{r}(1-\mu{})D_u\psi{}D_v\psi{}.
\end{equation}
Neglecting the terms with a positive sign, we arrive at
\begin{equation}
K^{-X}[\psi{}]\ge -f'(r)[(1-\mu{})^2(D_u\psi{})^2 + (D_v\psi{})^2] + 4 \frac{f(r)}{r}(1-\mu{})D_u\psi{}D_v\psi{}.
\end{equation}

To obtain \cref{X-estimate-equation}, use the wave equation for \(r\) to rewrite
the last term on the left of \cref{X-estimate-equation}:
\begin{equation}
\begin{split}
4 \frac{f}{r}(1-\mu{})D_u\psi{}D_v\psi{}r^2\kappa{}(-\nu{}) &= 4f r\partial{}_{u}\psi{}\partial{}_{v}\psi{} - 4 \frac{f}{r}\mu{}D_u\psi{}D_v\psi{}r^2\kappa{}(-\nu{})\\
&= \partial{}_{u}(-2f'r\lambda{}\psi{}^2 + fr\psi{}^2) + \partial{}_{v}(-2f'r\nu{}\psi{}^2 + f r\psi{}^2) \\
&\qquad + \Bigl[-\frac{4f}{r}\psi{}\Box{}\psi{} + \frac{2}{1-\mu{}}\frac{f''}{r}\psi{}^2 \\
&\qquad + (f'r - f) \frac{4(\varpi{}-\mathbf{e}^2/r)}{r^4}\psi{}^2 + 4 \frac{f}{r}\mu{}D_{u}\psi{}D_{v}\psi{}\Bigr]r^2\kappa{}(-\nu{}).
\end{split}
\end{equation}
\end{proof}
\begin{lemma}[Morawetz vector field with \(f(r) = -1\)]
Let \(\widetilde{X} = -(1-\mu{})D_u + D_v\). We have
\begin{equation}\label{X-tilde-identity-equation}
\begin{split}
& \partial{}_u((D_v\psi{})^2r^2\kappa{} + r\psi{}^2) - \partial{}_v((1-\mu{})(D_u\psi{})^2r^2(-\nu{}) + r\psi{}^2) + \frac{4(\varpi{}-\mathbf{e}^2/r)}{r^4}\psi{}^2r^2\kappa{}(-\nu{}) \\
&\le \Bigl(2\widetilde{X}\psi{} - 4\frac{\psi{}}{r}\Bigr)\Box{}\psi{}r^2\kappa{}(-\nu{})+ \frac{4(2\varpi{}-\mathbf{e}^2/r)}{r^2}D_{u}\psi{}D_{v}\psi{}r^2\kappa{}(-\nu{}).
\end{split}
\end{equation}
\label{X-tilde-identity}
\end{lemma}
\begin{proof}
Set \(f(r) = -1\) in \cref{X-estimate-equation-2} and rearrange.
\end{proof}
\begin{lemma}[Redshift vector field identity]
Let \(\chi_{\mathcal{H}}(r)\) be a \(C^1\) function, and write \(Y =
\chi_{\mathcal{H}}(r)D_{u}\). We have
\begin{equation}
\begin{split}
&\partial{}_v(\chi{}_{\mathcal{H}}(r)(D_u\psi{})^2r^2(-\nu{})) + \frac{2(\varpi{}-\mathbf{e}^2/r)}{r^2}\chi{}_{\mathcal{H}}(r)(D_u\psi{})^2 r^2\kappa{}(-\nu{})\\
&= \Bigl[2Y\psi{}\Box{}\psi{} + (1-\mu{})\chi{}'_{\mathcal{H}}(r)(D_u\psi{})^2 - \frac{2\chi{}_{\mathcal{H}}(r)}{r}D_u\psi{}D_v\psi{}\Bigr]r^2\kappa{}(-\nu{}),
\end{split}
\end{equation}
\label{redshift-identity}
\end{lemma}
\begin{proof}
Set \((Y^u,Y^v) = (\chi{}_{\mathcal{H}}(r),0)\). Compute
\begin{equation}\label{redshift-1}
K^Y[\psi{}] = (-\nu{})D_v((-\nu{})^{-1}\chi{}_{\mathcal{H}}(r))(D_u\psi{})^2 - \frac{2\chi{}_{\mathcal{H}}(r)}{r}D_u\psi{}D_v\psi{}.
\end{equation}
Use the equation for \(\partial{}_v(-\nu{})\) to compute
\begin{equation}\label{redshift-2}
\begin{split}
(-\nu{})D_v((-\nu{})^{-1}\chi{}_{\mathcal{H}}(r)) &=  -\frac{2(\varpi{}-\mathbf{e}^2/r)}{r^2}\chi{}_{\mathcal{H}}(r) + (1-\mu{})\chi{}_{\mathcal{H}}'(r).
\end{split}
\end{equation}
Substitute \cref{redshift-2} into \cref{redshift-1} to get
\end{proof}
\begin{lemma}[Irregular vector field identity]
Let \(g = (u,v)\) be a \(C^1\) function. Then
\begin{equation}\label{irregular-vf-identity}
(-\partial{}_ug)(\partial{}_v\psi{})^2r^2 = -\partial_u(g(\partial{}_v\psi{})^2r^2) + \Bigl[2\kappa{}g\partial{}_v\psi{}\Box{}\psi{} + \frac{2}{r}\lambda{}gDu_\psi{}D_v\psi{}\Bigr]r^2\kappa{}(-\nu{}).
\end{equation}
\end{lemma}
\begin{proof}
Apply \cref{vf-multiplier} with \((W^u,W^v) = (0,\kappa{}g)\).
\end{proof}
\subsection{Energy boundedness estimate}
\label{sec:org98e495d}
\label{sec:energy-estimate-proof} We say \(\mathcal{R}\) is an \emph{admissible spacetime
region} if the past and future boundaries of \(\mathcal{R}\) are each connected
and piecewise null.
\begin{remark}
We apply the energy estimate to characteristic rectangles and regions of
the form \(\bigcup_{\tau{}\in [\tau_1,\tau_2]}\Sigma_\tau{}\), which are both clearly admissible.
\end{remark}
Recall the notations \(D_u\) and \(D_v\) introduced in \cref{Du-Dv-notation}. See
\cref{sec:vector-field-multiplier-identities} for the vector field multipliers and
associated identities used in the proof of \cref{energy-estimate}.
\begin{lemma}[Energy estimate]
Let \(\mathcal{R}\) be an admissible spacetime region with past boundary \(\Sigma_1\)
and future boundary \(\Sigma_2\). We have
\begin{equation}\label{eb:no-bulk}
E[\varphi{},\Sigma{}_2] \le C(\varpi{}_i,c_{\mathcal{H}},r_{\mathrm{min}})E[\varphi{},\Sigma{}_1].
\end{equation}
Let \(p_{\mathrm{fut}}\) be the future endpoint of \(\Sigma_1\). Then
\begin{equation}\label{energy:en-bound}
E[\psi{},\Sigma{}_2] + E_{\mathrm{bulk}}[\psi{},\mathcal{R}]\le C(\mathfrak{b}_0)[E[\psi{},\Sigma{}_1] + r\abs{\psi{}}^2(p_{\mathrm{fut}})] + \mathrm{Err}[\psi{},\mathcal{R}]
\end{equation}
for an error term that admits a decomposition
\begin{equation}
\begin{split}
\mathrm{Err}[\psi{},\mathcal{R}] &\le \mathrm{Err}_{U}[\psi{},\mathcal{R}]  + C(\mathfrak{b}_0)(\mathrm{Err}_{v}[\psi{},\mathcal{R}] + \mathbf{1}_{\psi{}\neq{}\varphi{}}\mathrm{Err}_{\mathrm{quartic}}[\psi{},\mathcal{R}] + \mathrm{Err}_{\mathrm{zo}}[\psi{},\mathcal{R}]),
\end{split}
\end{equation}
where
\begin{equation}
\mathrm{Err}_{U}[\psi{},\mathcal{R}] = \iint_{\mathcal{R}} w_{U}D_u\psi{}\Box{}\psi{}
\end{equation}
for a non-negative weight \(w_U\) satisfying \(0<C(\varpi{}_i,r_{\mathrm{min}})\le w_U\le C(\mathfrak{b}_0)\), and
\begin{align}
\mathrm{Err}_{v}[\psi{},\mathcal{R}] &\coloneqq{}\iint_{\mathcal{R}} \abs{D_v\psi{}}\abs{\Box{}\psi{}}, \\
\mathrm{Err}_{\mathrm{quartic}}[\psi{},\mathcal{R}] &\coloneqq{}\iint_{\mathcal{R}} r\abs{D_u\varphi{}}^2\abs{D_v\psi{}}^2, \\
\mathrm{Err}_{\mathrm{zo}}[\psi{},\mathcal{R}] &\coloneqq{}\iint_{\mathcal{R}} \frac{\abs{\psi{}}}{r}\abs{\Box{}\psi{}}.
\end{align}
\label{energy-estimate}
\end{lemma}
\begin{remark}[Structure of error terms]
Observe that the error terms composing \(\text{Err}[\psi{},\mathcal{R}]\) are
non-negative, except for \(\mathrm{Err}_{U}[\psi{},\mathcal{R}]\), which does
not have a sign. The positivity of the weight \(w_U\) is crucial for closing the
energy estimates. This is because commutation produces a \(D_u\psi{}\) term with
a negative coefficient (see \cref{err-wave-bound}). In the absence of this good
sign, we would not be able to control this term, since there is no smallness to
exploit.
\end{remark}
\begin{proof}
We follow the proof of \cite[Lem.~8.35]{luk-oh-scc2}. We do not mention uses of
\cref{zeroth-order-geometric-bounds} to estimate zeroth order geometric quantities
after Step 1. To simplify the notation, we omit the volume form
\(r^2\kappa{}(-\nu{})\dd{}u\dd{}v\) when integrating over spacetime regions.

\step{Step 1: Energy estimate with a weaker bulk term and proof of \cref{eb:no-bulk}.}
Define the energy quantity
\begin{equation}\label{energy:weak}
E_{\mathrm{weak}}[\psi{},\mathcal{R}] \coloneqq{}\iint_{\mathcal{R}} \frac{1}{r^4}[(D_u\psi{})^2 + (D_v\psi{})^2].
\end{equation}
Compared to the bulk term \(E_{\mathrm{bulk}}\) that we wish to control,
\(E_{\text{weak}}\) lacks a zeroth order term, does not capture improved
integrated local energy decay near the horizon, and has weaker \(r\)-weights. We
will show that
\begin{equation}\label{energy:weak-bulk-bound}
\begin{split}
E[\psi{},\Sigma{}_2] + E_{\text{weak}}[\psi{},\mathcal{R}]&\lesssim E[\psi{},\Sigma{}_1] + \text{Err}^{ }[\psi{},\mathcal{R}],
\end{split}
\end{equation}
where the constant implied by \(\lesssim\) is positive and depends only on and \(\varpi_i\),
\(c_{\mathcal{H}}\), and \(r_{\text{min}}\). Since \(\text{Err}^{ }[\varphi{},\mathcal{R}]
= 0\), we obtain \cref{eb:no-bulk} as an immediate corollary of \cref{energy:weak-bulk-bound}.

\step{Step 1a: Outline of proof.} The proof is a standard application of the vector
field multiplier method, using the Kodama vector field \(T\), the Morawetz
vector field \(X\), and the redshift vector field \(Y\). We will give a careful
proof, in order to emphasize the structure of the right side, in particular the
parameters on which the constants depend and the positivity of the weight
\(w_U\).

We will show the following three estimates:
\begin{align}
\int_{\Sigma_2^v} (D_v\psi{})^2 + \int_{\Sigma_2^u} (1-\mu{})(D_u\psi{})^2 &\lesssim \iint_{\mathcal{R}} T\psi{}\Box{}\psi{} + E[\psi{},\Sigma{}_1] + \mathbf{1}_{\psi{}\neq{}\varphi{}}\text{Err}_{\text{quartic}}[\psi{},\mathcal{R}], \label{energy:T-estimate} \\
\iint_{\mathcal{R}} \frac{1}{r^{3+\eta{}_0}}\Bigl[(1-\mu{})^2(D_u\psi{})^2 + (D_v\psi{})^2\Bigr]&\lesssim \int_{\mathcal{R}} X\psi{}\Box{}\psi{} + \side{RHS}{energy:T-estimate}, \label{energy:X-estimate}\\
\int_{\Sigma{}_2^u\cap \set{r\le R}} (D_u\psi{})^2 + \iint_{\mathcal{R}\cap \set{r\le R}} \frac{1}{r^4}(D_u\psi{})^2 &\lesssim \iint_{\mathcal{R}} Y\psi{}\Box{}\psi{} + E[\psi{},\Sigma{}_1] + \side{RHS}{energy:X-estimate}\label{energy:Y-estimate},
\end{align}
where
\begin{equation}\label{energy:vf-coefficients}
T = (1-\mu{})U + \frac{1}{\kappa{}}\partial{}_v, \qquad X = r^{-3}(1-\mu{})U - r^{-3}\frac{1}{\kappa{}}\partial{}_v \qquad Y = \chi{}_{\mathcal{H}}(r)U,
\end{equation}
and the implicit constants depend only on \(c_{\mathcal{H}}\), \(\varpi{}_i\), and \(r_{\text{min}}\).
Observe that the left sides of \cref{energy:T-estimate,energy:X-estimate} bound the left side of
\cref{energy:weak} from above in a region \(\set{r\ge R}\) for \(R\ge
C(\varpi_i,r_{\text{min}})\) (in view of \eqref{mu-estimate}), and \cref{energy:Y-estimate} gives the control
in the remaining region \(\set{r\le R}\). Thus
\cref{energy:T-estimate,energy:X-estimate,energy:Y-estimate} together give
\begin{equation}\label{energy:weak-1d}
E[\psi{},\Sigma{}_2] + E_{\text{weak}}[\psi{},\mathcal{R}]\le  \iint_{\mathcal{R}} W\psi{}\Box{}\psi{} + CE[\psi{},\Sigma{}_1] + \mathbf{1}_{\psi{}\neq{}\varphi{}}C\text{Err}_{\text{quartic}}[\psi{},\mathcal{R}],
\end{equation}
for a vector field
\begin{equation}
W = C_TT + C_XX + C_{Y}Y,
\end{equation}
where the constants \(C,C_T,C_X,C_Y\) are positive and depend only on \(\varpi_i\),
\(c_{\mathcal{H}}\), and \(r_{\text{min}}\). In view of
\cref{energy:vf-coefficients}, the coefficients of \(W\) are smooth and of size
\(C(\mathfrak{b}_0)\), and the \(U\)-coefficient is bounded below by a positive constant. In particular, we
have
\begin{equation}\label{weight-U-inequality}
W\psi{} = w_{U}U\psi{} + w_{v}\partial{}_{v}\psi{}, \quad 1\lesssim w_{U},\quad \abs{w_{v}}\le C(\mathfrak{b}_{0}),
\end{equation}
and so
\begin{equation}
\iint_{\mathcal{R}} W\psi{}\Box{}\psi{}\le \mathrm{Err}^{ }[\psi{},\mathcal{R}],
\end{equation}
which completes the proof of \cref{energy:weak} when substituted into
\cref{energy:weak-1d}.

\step{Step 1b: Almost conservation law: proof of \cref{energy:T-estimate}.}
Integrate the \(T\)-identity in \cref{T-identity} over \(\mathcal{R}\) to get
\begin{equation}
\begin{split}
\int_{\Sigma_2^v} (D_v\psi{})^2 + \int_{\Sigma_2^u} (1-\mu{})(D_u\psi{})^2 &=  2\iint_{\mathcal{R}} T\psi{}\Box{}\psi{} + \int_{\Sigma_1^v} (D_v\psi{})^2 + \int_{\Sigma_1^u} (1-\mu{})(D_u\psi{})^2 \\
&\qquad + \iint_{\mathcal{R}} r[-(D_u\psi{})^2(D_v\varphi{})^2 + (D_u\varphi{})^2(D_v\psi{})^2].
\end{split}
\end{equation}
Since \(1-\mu{}\le C(\varpi_i,r_{\text{min}})\), the boundary terms on \(\Sigma_1\) are
controlled by \(E[\psi{},\Sigma_1]\). The final term vanishes if \(\psi{} =
\varphi{}\) and is equal to \(\text{Err}_{\text{quartic}}[\psi{},\mathcal{R}]\)
otherwise.

\step{Step 1c: Weak integrated local energy decay estimate: proof of
\cref{energy:X-estimate}.} Let \(f(r) = r^{-3}\), so that
\begin{equation}
-f'(r) - 2f(r)/r = r^{-4}.
\end{equation}
By Young's inequality we have
\begin{equation}
r^{-4}[(1-\mu{})^2(D_u\psi{})^2 + (D_v\psi{})^2]\le -f'(r)[(1-\mu{})^2(D_u\psi{})^2 + (D_v\psi{})^2] + 4 \frac{f(r)}{r}(1-\mu{})D_u\psi{}D_v\psi{} = \side{LHS}{X-estimate-equation}.
\end{equation}
The estimate follows from integrating \cref{X-estimate-equation} with \(f(r) =
r^{-3}\); of the four boundary terms generated, we can neglect two on account of
their sign, and control the other two using Step 1a.

\step{Step 1d: Redshift estimate: proof of \cref{energy:Y-estimate}.} Let \(\chi_{\mathcal{H}}(r)\) be a positive cutoff
function that is \(1\) in \(\set{r\le R}\) and \(0\) in \(\set{r\ge 2R}\).
Recall that \(\varpi{}-\mathbf{e}^2/r\ge c_{\mathcal{H}}\) and apply Young's
inequality in the form
\begin{equation}
\abs[\Big]{\frac{2\chi_{\mathcal{H}}(r)}{r}D_u\psi{}D_v\psi{}}\le c_{\mathcal{H}} \frac{\chi_{\mathcal{H}}(r)}{r^2}(D_u\psi{})^2 + c_{\mathcal{H}}^{-1}\mathbf{1}_{r\le 2R}(D_v\psi{})^2.
\end{equation}
after integrating the identity in \cref{redshift-identity} and noting the support properties of
\(\chi_{\mathcal{H}}\) and \(\chi_{\mathcal{H}}'\) to get
\begin{equation}
\begin{split}
\int_{\Sigma{}_2^u} \chi{}_{\mathcal{H}}(r)(D_u\psi{})^2 + c_{\mathcal{H}}r_{\text{min}}^2\iint_{\mathcal{R}}\chi{}_{\mathcal{H}}(r) \frac{1}{r^4}(D_u\psi{})^2 &\lesssim  \iint_{\mathcal{R}} Y\psi{}\Box{}\psi{} + \iint_{\mathcal{R}\cap \set{R\le r\le 2R}}(1-\mu{})(D_u\psi{})^2 \\
&\qquad + c_{\mathcal{H}}^{-1}\iint_{\mathcal{R}\cap \set{r\le 2R}} (D_v\psi{})^2 + \int_{\Sigma{}_1^u} \chi{}_{\mathcal{H}}(r)(D_u\psi{})^2. \\
\end{split}
\end{equation}
When \(R\ge C(\varpi{}_i,r_{\text{min}})\), we have \(1-\mu{}\ge 1/2\) by \eqref{mu-estimate},
and so the second term can be controlled by Step 1b (up to a multiple of
\(R^4\)), as can the third term. The final term is controlled by data.

\step{Step 2: Control of zeroth order term away from the horizon.} The goal of this step is to show
that for \(R_1\ge R_0\) and \(\epsilon\in (0,1/2)\) we have
\begin{equation}
\begin{split}
\iint_{\mathcal{R}\cap \set{r\ge 2R_1}} \frac{\psi{}^2}{r^4}&\le \epsilon{}^{1-\eta{}_0}C(\mathfrak{b}_{0},R_1)\iint_{\mathcal{R}\cap \set{\epsilon{}^{-1}R_1}} \frac{1}{r^{1+\eta{}_0}}[(D_{u}\psi{})^2 + (D_{v}\psi{})^2] \\
&\qquad + C(\mathfrak{b}_0,R_1,\epsilon{})(E[\psi{},\Sigma{}_1] + \text{Err}^{ }[\psi{},\mathcal{R}] + r\psi{}^2(p_{\mathrm{fut}})).
\end{split}
\end{equation}
We multiply \cref{X-tilde-identity-equation} by a non-decreasing cutoff function \(\chi_{R_1}(r)\)
that is \(0\) in \(\set{r\le R_1}\) and \(1\) in \(\set{r\ge 2R_1}\) and
integrate by parts. We illustrate how to handle the first term on the left side of
\cref{X-tilde-identity-equation}:
\begin{equation}
\begin{split}
&\iint_{\mathcal{R}} \chi{}_{R_1}(r)\partial{}_u((D_v\psi{})^2r^2\kappa{} + r\psi{}^2)\dd{}u\dd{}v \\
&= \iint_{\mathcal{R}} \chi{}_{R_1}'(r)(-\nu{})((D_v\psi{})^2r^2\kappa{} + r\psi{}^2)\dd{}u\dd{}v + \int _{\Sigma{}_1^{\text{out}}}\chi{}_{R_1}(r)((D_v\psi{})^2r^2\kappa{} + r\psi{}^2)\dd{}v  \\
&\qquad - \int _{\Sigma{}_2^{\text{out}}}\chi{}_{R_1}(r)((D_v\psi{})^2r^2\kappa{} + r\psi{}^2)\dd{}v \\
&\ge - \int _{\Sigma{}_2^{\text{out}}}((D_v\psi{})^2r^2\kappa{} + r\psi{}^2)\dd{}v\\
&\ge -C(\mathfrak{b}_0)(E[\psi{},\Sigma{}_1] + \text{Err}^{ }[\psi{},\mathcal{R}] + r\psi{}^2(p_{\mathrm{fut}})).
\end{split}
\end{equation}
The first terms on the second line have a good sign, and the term on the third line is
estimated by Step 1 and multiple uses of Hardy's inequality in both the \(u\)-
and \(v\)-directions (it is here that we use the assumption that the past/future
boundary of \(\mathcal{R}\) is connected). One similarly obtains the analogous
estimate for the second term on the left of \cref{X-tilde-identity-equation}. In
view of the above estimates and \cref{mass-redshift}, we obtain
\begin{equation}\label{energy-2-0}
\iint_{\mathcal{R}\cap \set{r\ge 2R_1}} \frac{\psi{}^2}{r^4}\le C(\mathfrak{b}_0)(E[\psi{},\Sigma{}_1] + \text{Err}^{ }[\psi{},\mathcal{R}] + r\psi{}^2(p_{\mathrm{fut}})) + \iint_{\mathcal{R}\cap \set{r\ge R_1}} \side{RHS}{X-tilde-identity-equation}.
\end{equation}

It is clear that
\begin{equation}\label{energy-2-1}
\begin{split}
\iint_{\mathcal{R}\cap \set{r\ge R_1}} \side{RHS}{X-tilde-identity-equation} &\le  C\iint_{\mathcal{R}} \widetilde{X}\psi{}\Box{}\psi{} + C(\mathfrak{b}_{0})\iint_{\set{r\ge R_1}} \frac{1}{r^2}[(D_{u}\psi{})^2 + (D_{v}\psi{})^2] \\
&\qquad + C(E_{\text{weak}}[\psi{},\mathcal{R}] + r\psi{}^2(p_{\mathrm{fut}})).
\end{split}
\end{equation}
Although the first term on the right of \eqref{energy-2-1} has a bad sign in front
of the \(U\psi{}\Box{}\psi{}\) term (recall that \(\text{Err}_{U}\) has a
positive weight), by Step 1 we can add a multiple of the non-negative quantity
\(E_{\text{weak}}[\psi{},\mathcal{R}]\) (where the constant depends on \(R_1\))
to both sides of \eqref{energy-2-1} and (in conjunction with \cref{energy-2-0}) obtain
\begin{equation}\label{energy-2-2}
\begin{split}
\iint_{\mathcal{R}\cap \set{r\ge 2R_1}} \frac{\psi{}^2}{r^4}&\le C(\mathfrak{b}_{0},R_1)\iint_{\set{r\ge R_1}} \frac{1}{r^2}[(D_{u}\psi{})^2 + (D_{v}\psi{})^2] \\
&\qquad C(\mathfrak{b}_0,R_1)(E[\psi{},\Sigma{}_1] + \text{Err}^{ }[\psi{},\mathcal{R}] + r\psi{}^2(p_{\mathrm{fut}})).
\end{split}
\end{equation}
To complete the proof, note that the first term on the right of \cref{energy-2-2}
is bounded by \(C(R_1,\epsilon{})E_{\text{weak}}[\psi{},\mathcal{R}]\) in the
region \(\set{R_1\le r\le \epsilon^{-1}R_1}\), and in the remaining region we
have
\begin{equation}
\iint_{\set{r\ge \epsilon{}^{-1}R_1}} \frac{1}{r^2}[(D_{u}\psi{})^2 + (D_{v}\psi{})^2]\le  \epsilon{}^{1-\eta{}_0} \iint_{\mathcal{R}\cap \set{r\ge \epsilon{}^{-1}R_1}} \frac{1}{r^{1+\eta{}_0}}[(D_{u}\psi{})^2 + (D_{v}\psi{})^2].
\end{equation}

\step{Step 3: Improved integrated local energy decay away from the horizon.} In this
step we show that for \(R_1 = \max (2,2R_0,\eta{}_0^{-1})\), we have
\begin{equation}\label{energy-Step-3}
\iint_{\mathcal{R}\cap \set{r\ge R_1}} \frac{1}{r^{1+\eta{}_0}}\Bigl[(D_{u}\psi{})^2 + (D_{v}\psi{})^2 + \frac{\psi{}^2}{r^2}\Bigr]\le  C(\mathfrak{b}_{0})(E[\psi{},\Sigma{}_1] + r\psi{}^2(p_{\mathrm{fut}}) + \text{Err}[\psi{},\mathcal{R}]).
\end{equation}
We will track the dependence of constants on \(\eta_0\). Begin by writing
\begin{equation}\label{energy-3-1}
\side{LHS}{X-estimate-equation-2} = [\text{(I)} + \text{(II)} + \text{(III)}]r^2\kappa{}(-\nu{})
\end{equation}
for
\begin{equation}
\begin{split}
\text{(I)} &= -f'(r)[(1-\mu{})^2(D_u\psi{})^2 + (D_v\psi{})^2], \\
\text{(II)} &=  4 \frac{f(r)}{r}\mu{}D_{u}\psi{}D_{v}\psi{}, \\
\text{(III)} &=  \Bigl[\frac{2}{1-\mu{}}f''(r)r + (f'(r)r - f(r)) \frac{4(\varpi{}-\mathbf{e}^2/r)}{r^4}\Bigr]\psi{}^2.
\end{split}
\end{equation}
Now \cref{mu-estimate} and the choice of \(R_1\) imply that for \(f(r) =
r^{-\eta_0}\) and \(r\ge R_1\), we have
\begin{equation}
\frac{\eta{}_0}{r^{1+\eta{}_0}}[(D_{u}\psi{})^2 + (D_{v}\psi{})^2]\lesssim \text{(I)},\quad \abs{\text{(II)}}\lesssim \frac{1}{r^{2+\eta{}_0}}D_{u}\psi{}D_{v}\psi{}\lesssim r^{-1}\text{(I)},\quad \frac{\eta{}_0}{r^{3+\eta{}_0}}\lesssim \text{(III)}.
\end{equation}
We can therefore multiply \eqref{energy-3-1} by a cutoff \(\chi(r)\) that is \(0\) in
\(\set{r\le R_1/2}\) and \(1\) in \(\set{r\ge R_1}\), and integrate over \(\mathcal{R}\) to get
\begin{equation}\label{energy-3-2}
\begin{split}
\iint_{\mathcal{R}\cap \set{r\ge R_1}} \frac{1}{r^{1+\eta{}_0}}\Bigl[(D_{u}\psi{})^2 + (D_{v}\psi{})^2 + \frac{\psi{}^2}{r^2}\Bigr]&\lesssim \eta{}_0^{-1}\iint_{\mathcal{R}\cap \set{r\ge R_1/2}} \chi{}_{R_1}(r)\side{RHS}{X-estimate-equation-2}\dd{}u\dd{}v.
\end{split}
\end{equation}
Now we integrate by parts on the right side. Using Hardy's inequality in both
the \(u\)- and \(v\)-directions and Step 1, one can estimate the boundary terms
by \(E[\psi{},\Sigma{}_1] + \text{Err}^{ }[\psi{},\mathcal{R}] +
r\psi^2(p_{\mathrm{fut}})\). The bulk term involving \(\Box{}\psi{}\) is bounded
by \(\mathrm{Err}_U[\psi{},\mathcal{R}]\) and
\(\mathrm{Err}_{v}[\psi{},\mathcal{R}]\) and
\(\text{Err}_{\text{zo}}[\psi{},\mathcal{R}\cap \set{R_1/2}]\). In view of Steps
1 and 2 (applied with \(R_1/2\) in place of \(R_1\)), the support properties of
\(\chi_{R_1}\), and \cref{zeroth-order-geometric-bounds}, the bulk terms arising
from integration by parts are bounded by
\begin{equation}\label{energy-3-3}
\begin{split}
&C(R_1)\eta{}_0^{-1}\iint_{\mathcal{R}\cap \set{R_1/2\le r\le R_1}} \frac{1}{r^4}\Bigl[(D_{u}\psi{})^2 + (D_{v}\psi{})^2 + \psi{}^2\Bigr] \\
&\le  \epsilon{}^{1-\eta{}_0}C(\mathfrak{b}_{0},R_1,\eta{}_0)\iint_{\mathcal{R}\cap \set{r\ge \epsilon{}^{-1}R_1/2}} \frac{1}{r^{1+\eta{}_0}}[(D_{u}\psi{})^2 + (D_{v}\psi{})^2] \\
&\qquad +C(\mathfrak{b}_0,R_1,\epsilon{},\eta{}_0)(E[\psi{},\Sigma{}_1] + \text{Err}^{ }[\psi{},\mathcal{R}] + r\psi{}^2(p_{\mathrm{fut}})).
\end{split}
\end{equation}
To conclude the proof, choose \(\epsilon{}\) small enough to absorb the first term on the
right of \eqref{energy-3-3} to the left of \eqref{energy-3-2}.

\step{Step 4: Control of zeroth order term near the horizon.} Here we upgrade the
control of the zeroth order term in \cref{energy-Step-3} from Step 3 to
\begin{equation}\label{energy-step-4}
 \iint_{\mathcal{R}} \frac{\psi{}^2}{r^{3+\eta{}_0}}\le \side{RHS}{energy-Step-3}.
\end{equation}
Let \(R > 0\). Use Hardy's inequality in the \(u\)-direction to obtain
\begin{equation}\label{energy-4-1}
\begin{split}
\iint_{\mathcal{R}\cap \set{r\le R}}\frac{\psi{}^2}{r^{3+\eta{}_0}}r^2\kappa{}(-\nu{})\dd{}u\dd{}v &\lesssim \iint_{\mathcal{R}\cap \set{r\le R}} (-\nu{})\chi^2 \psi{}^2\dd{}u\dd{}v \\
&\le  \iint_{\mathcal{R}\cap \set{r\le R}} \frac{r^2}{(-\nu{})}(\partial{}_u\psi{})^2 \dd{}u\dd{}v +\int _{\Sigma{}_1^{\text{out}}\cap \set{r\le R}} r\psi{}^2 \dd{}v \\
&\qquad + \int_{v_1(R)}^{v_2(R)} r\psi{}^2(u_R(v),v)\dd{}v.
\end{split}
\end{equation}
Integrate \cref{energy-4-1} over \(R\in [R_1,R_1 + 1]\), note the monotonicity in
\(R\) of the first term on the left and first two terms on the right, and change
variables in the integration of the third term on the right to get
\begin{equation}\label{energy-4-2}
\begin{split}
&\iint_{\mathcal{R}\cap \set{r\le R_1}}\frac{\psi{}^2}{r^{3+\eta{}_0}}r^2\kappa{}(-\nu{})\dd{}u\dd{}v \\
&\le \iint_{\mathcal{R}\cap \set{r\le R_\ast{}}}\frac{\psi{}^2}{r^{3+\eta{}_0}}r^2\kappa{}(-\nu{})\dd{}u\dd{}v \\
&\le \iint_{\mathcal{R}\cap \set{r\le R_1+1}} \frac{r^2}{(-\nu{})}(\partial{}_u\psi{})^2 \dd{}u\dd{}v +\int _{\Sigma{}_1^{\text{out}}\cap \set{r\le R_1+1}} r\psi{}^2 \dd{}v + \iint_{\mathcal{R}\cap \set{R_1\le r\le R_1+1}} (-\nu{})r\psi{}^2(u_R(v),v)\dd{}u\dd{}v.
\end{split}
\end{equation}
We are done by a Hardy inequality argument as in Step 2 that controls the term
on \(\Sigma_1^{\text{out}}\), together with the results of Steps 2--3 and
\cref{kappa-sign,kappa-estimate}.

\step{Step 5: Improved integrated local energy decay near the horizon.} We now control
the final term in the definition of \(E_{\mathrm{bulk}}\) by showing that
\begin{equation}
\iint_{\mathcal{R}\cap \set{r\le R_0}}e^{-c_\nu'(u-v)}(\partial{}_{v}\psi{})^2\dd{}u\dd{}v\lesssim E[\psi{},\Sigma{}_{1}] + E_{\text{weak}}[\psi{},\mathcal{R}] + \text{Err}[\psi{},\mathcal{R}].
\end{equation}
Let \(g(u,v) = \chi{}(r(u,v))h(u,v)\) for a cutoff function \(\chi{}\) satisfying
\(\chi{}(r) = 1\) in \(\set{r\le R_0}\) and \(\chi{}(r) = 0\) in \(\set{r\ge
2R_0}\) and \(h(u,v) = e^{-c_{\nu}'(u-v)}\), where \(c_{\nu}' > 0\) is as in
\cref{nu-exponential-bound}. Observe that \(h\ge 0\) and \(-\partial_{u}h\ge 0\),
and \(h\lesssim 1\) in \(\set{r\le 2R_0}\) by \cref{nu-exponential-bound}.
Integrating \cref{irregular-vf-identity} gives
\begin{equation}
\begin{split}
\iint_{\mathcal{R}\cap \set{r\le R_0}}(-\partial{}_{u}h)(\partial{}_{v}\psi{})^2r^2\dd{}u\dd{}v &\le  \int_{\Sigma{}_{1}^{out}}g(D_{v}\psi{})r^2\kappa{}\dd{}v + \iint_{\mathcal{R}\cap \set{r\le 2R_0}} 2hD_{v}\psi{}\Box{}\psi{}  \\
&\qquad + \iint_{\mathcal{R}\cap \set{r\le 2R_0}} \frac{2}{r}(1-\mu{})h\abs{D_{u}\psi{}}\abs{D_{v}\psi{}} \\
& \lesssim E[\psi{},\Sigma{}_{1}] + \text{Err}_{v}[\psi{},\mathcal{R}] + E_{\text{weak}}[\psi{},\mathcal{R}].
\end{split}
\end{equation}
We are done after computing
\begin{equation}\label{near-horizon-duh-computation}
-\partial_{u}h = -(-\nu{})\chi_{R}'e^{-c_{\nu}'(u-v)} + c_{\mathcal{H}}\chi_{R}e^{-c_{\nu}'(u-v)}
\end{equation}
and estimating the bulk term arising from the first term on the right of
\cref{near-horizon-duh-computation} by \(E_{\text{weak}}[\psi{},\mathcal{R}]\).
\end{proof}
\subsection{Boundedness of the unweighted energy \texorpdfstring{\(\mathcal{E}_{\alpha{}}\)}{}}
\label{sec:orge59094a}
The goal of this section is to prove \cref{E-alpha-bound-2}.
\begin{lemma}
Let \(\alpha{}\ge 0\) and let \(L\in \Gamma^\alpha{}\). For \(0 < s < 1/4\), we have
\begin{equation}
\begin{split}
\iint_{\mathcal{R}} r^{1+2s}\abs{\Box{}L\varphi{}}^2&\le C(\mathfrak{B}_0^{\circ },\mathfrak{g}_0,\alpha{})E_{\mathrm{bulk}}[L\varphi{},\mathcal{R}] + C(\mathcal{G}_{\alpha,s})[E_{\mathrm{bulk}}[\Gamma{}^{<\alpha{}}\varphi{},\mathcal{R}] + E_{4s,\mathrm{bulk}}[\Gamma{}^{<\alpha{}}\varphi{},\mathcal{R}]].
\end{split}
\end{equation}
\label{energy:wave-pointwise}
\end{lemma}
\begin{proof}
This follows from the pointwise estimate \cref{wave-pointwise-1} and the
definitions of the energy quantities involved.
\end{proof}
\begin{lemma}[Estimate for zeroth order error term]
Let \(\epsilon{} > 0\) and let \(L\in \Gamma^\alpha{}\). For \(s\in{}[\eta_0,1/4]\), we have
\begin{equation}
\begin{split}
\mathrm{Err}_{\mathrm{zo}}[L\varphi{},\mathcal{R}]&\le \epsilon{}E_{\mathrm{bulk}}[L\varphi{},\mathcal{R}] + C(\mathcal{G}_{\alpha,s},\epsilon{})[E_{\mathrm{bulk}}[\Gamma{}^{<\alpha{}}\varphi{},\mathcal{R}] + E_{4s,\mathrm{bulk}}[\Gamma{}^{<\alpha{}}\varphi{},\mathcal{R}]]
\end{split}
\end{equation}
\label{err-zo-bound}
\end{lemma}
\begin{proof}
We omit the volume form \(r^2\kappa{}(-\nu{})\dd{}u\dd{}v\). Introduce a small \(\delta{} > 0\) and
large \(R > 0\) to be chosen in the course of the proof. Use an \(r\)-weighted
Young's inequality to decompose
\begin{equation}\label{energy:err-zero-bound}
\begin{split}
\mathrm{Err}_{\mathrm{zo}}[L\varphi{},\mathcal{R}] &= \iint_{\mathcal{R}\cap \set{r\le R}} \frac{\abs{L\varphi{}}}{r}\abs{\Box{}L\varphi{}}+ \iint_{\mathcal{R}\cap \set{r\ge R}} \frac{\abs{L\varphi{}}}{r}\abs{\Box{}L\varphi{}} \\
&\le \iint_{\mathcal{R}\cap \set{r\le R}} \frac{\abs{L\varphi{}}}{r}\abs{\Box{}L\varphi{}} + \delta{}\iint_{\mathcal{R}\cap \set{r\ge R}} r^{1+2s}\abs{\Box{}L\varphi{}}^2 + \delta{}^{-1}\iint_{\mathcal{R}\cap \set{r\ge R}} r^{-s} \frac{1}{r^{1+s}}\frac{\abs{L\varphi{}}^2}{r^2}  \\
&\coloneqq{}\text{(I)} + \text{(II)} + \text{(III)}.
\end{split}
\end{equation}
To estimate term \(\text{(I)}\), we use the structure of \(\Box{}L\varphi{}\). In
particular, \(\Box{}L\varphi{}\) contains a top order term only when \(L\)
contains a \(U\) or a \(V\), in which case we can bring this vector field to the
front of \(L\) and treat the corresponding term as the derivative of a lower
order term. That is, we write \(\alpha_X\abs{L\varphi{}} \le
C(\mathcal{C}_{<\alpha{}},\alpha{})\abs{D\Gamma^{<\alpha{}}\varphi{}}\) for
\(X\in \set{U,V}\). By \cref{rearrangement-formula,V-dv-coordinate-change}, we have
\begin{equation}\label{zo-I-0}
\begin{split}
\alpha{}_U\abs{L\varphi{}}\abs{UL\varphi{}} + \alpha{}_V\abs{L\varphi{}}\abs{\partial{}_vL\varphi{}}&\le C(\mathcal{C}_{<\alpha{}})(\abs{\partial{}_vL\varphi{}} + \abs{UL\varphi{}})(\abs{\partial{}_v\Gamma{}^{<\alpha{}}\varphi{}} + \abs{U\Gamma{}^{<\alpha{}}\varphi{}}).
\end{split}
\end{equation}
By \cref{wave-pointwise-1,zo-I-0} and Young's inequality, we have
\begin{equation}\label{zo-I-1}
\begin{split}
\frac{\abs{L\varphi{}}}{r}\abs{\Box{}L\varphi{}}&\le  r^{-3+s}C(\mathfrak{B}_0^{\circ },\mathfrak{g}_0,\alpha{})\Bigl[\alpha{}_V \abs{L\varphi{}}\abs{\partial{}_vL\varphi{}} + \alpha{}_U \abs{L\varphi{}}\abs{UL\varphi{}}\Bigr] \\
&\qquad + r^{-3+s}C(\mathcal{G}_{\alpha,s})\abs{L\varphi{}}[\mathbf{1}_{r\ge \Rc}\abs{\partial{}_v(r\Gamma{}^{<\alpha{}}\varphi{})} + \abs{\partial{}_v\Gamma{}^{<\alpha{}}\varphi{}} + \abs{U\Gamma{}^{<\alpha{}}\varphi{}}] \\
&\le \delta{}r^{-3+s}[\abs{\partial{}_vL\varphi{}}^2 + \abs{UL\varphi{}}^2] \\
&\qquad + r^{-3+s}C(\mathcal{G}_{\alpha,s},\delta{})(\abs{\partial{}_v\Gamma{}^{<\alpha{}}\varphi{}}^2 + \abs{U\Gamma{}^{<\alpha{}}\varphi{}}^2 + \mathbf{1}_{r\ge \Rc}\abs{\partial{}_v(r\Gamma{}^{<\alpha{}}\varphi{})}^2).
\end{split}
\end{equation}
Since \(s < 1\), integrating \cref{zo-I-1} with the volume form \(r^2\kappa{}(-\nu{})\dd{}u\dd{}v\)
gives
\begin{equation}\label{zo-I-2}
\text{(I)}\le \delta{}E_{\text{bulk}}[L\varphi{},\mathcal{R}] + C(\mathcal{G}_{\alpha,s},R,\delta{})(E_{\text{bulk}}[\Gamma{}^{<\alpha{}}\varphi{},\mathcal{R}] + E_{s,\text{bulk}}[\Gamma{}^{<\alpha{}}\varphi{},\mathcal{R}]).
\end{equation}
Next, \cref{energy:wave-pointwise} implies that
\begin{equation}\label{zo-II}
\text{(II)}\le \delta{}C(\mathfrak{B}_0^{\circ },\mathfrak{g}_0,\alpha{})E_{\mathrm{bulk}}[L\varphi{},\mathcal{R}] + C(\mathcal{G}_{\alpha,s})[E_{\mathrm{bulk}}[\Gamma{}^{<\alpha{}}\varphi{},\mathcal{R}] + E_{4s,\mathrm{bulk}}[\Gamma{}^{<\alpha{}}\varphi{},\mathcal{R}]].
\end{equation}
Finally, use \(s\ge \eta_0\) to estimate
\begin{equation}\label{zo-III}
\text{(III)} \le \delta{}^{-1}R^{-\eta{}_0}E_{\text{bulk}}[L\varphi{},\mathcal{R}].
\end{equation}
Substitute \cref{zo-I-2,zo-II,zo-III} into \cref{energy:err-zero-bound} to obtain
\begin{equation}
\begin{split}
\mathrm{Err}_{\mathrm{zo}}[L\varphi{},\mathcal{R}] &\le (\delta{} + \delta{}^{-1}R^{-\eta{}_0})C(\mathfrak{B}_0^{\circ },\mathfrak{g}_0,\alpha{})E_{\mathrm{bulk}}[L\varphi{},\mathcal{R}] \\
&\qquad + C(\mathcal{G}_{\alpha,s},R,\delta{})[E_{\mathrm{bulk}}[\Gamma{}^{<\alpha{}}\varphi{},\mathcal{R}] + E_{4s,\mathrm{bulk}}[\Gamma{}^{<\alpha{}}\varphi{},\mathcal{R}]].
\end{split}
\end{equation}
To conclude, choose \(\delta{}\) small based on \(\epsilon{}\) and
\(C(\mathfrak{B}_0^{\circ },\mathfrak{g}_0,\alpha{})\), then choose \(R\) large depending
on \(\eta_0\), \(\delta{}\), and \(\epsilon{}\).
\end{proof}
\begin{lemma}[Estimate for \(U\)- and \(v\)-error terms]
Let \(\epsilon{} > 0\). We have
\begin{equation}
\begin{split}
\mathrm{Err}_{U}[L\varphi{},\mathcal{R}]&\le (\epsilon{} + \Rc^{-1+2s})C(\mathcal{B}_0^{\circ },\mathfrak{g}_0,\alpha{})E_{\mathrm{bulk}}[L\varphi{},\mathcal{R}] \\
&\qquad + C(\mathcal{G}_{\alpha,s},\epsilon{})[E_{\mathrm{bulk}}[\Gamma{}^{<\alpha{}}\varphi{},\mathcal{R}] + E_{3s,\mathrm{bulk}}[\Gamma{}^{<\alpha{}}\varphi{},\mathcal{R}]].
\end{split}
\end{equation}
The same estimate holds with \(\mathrm{Err}_{v}[L\varphi{},\mathcal{R}]\) on the left side.
\label{err-wave-bound}
\end{lemma}
\begin{proof}
In what follows, we will freely use the bounds
\((-\nu{}),\kappa{},\kappa{}^{-1},\mathbf{1}_{r\ge
R}\lambda{}^{-1},r_{\mathrm{min}}\le C(\mathfrak{b}_0)\). Use the positivity of
\(w_U\), \(\kappa{}\), and \((-\nu{})\), the bound \(\abs{w_U}\le
C(\mathfrak{b}_0)\), the estimates for \(UL\varphi{}\Box{}L\varphi{}\) in
\cref{wave-pointwise-2}, and an \(r\)-weighted Young's inequality to obtain
\begin{equation}\label{err-wave-bound-prep}
\begin{split}
w_UUL\varphi{}\Box{}L\varphi{}r^{2}\kappa{}(-\nu{}) &\le r^{s}C(\mathcal{G}_{\alpha,s})(\abs{UL\varphi{}} + \abs{\partial{}_vL\varphi{}})(\abs{U\Gamma{}^{<\alpha{}}\varphi{}} + \abs{\partial{}_v\Gamma{}^{<\alpha{}}\varphi{}} +  \mathbf{1}_{r\ge \Rc}\abs{\partial{}_v(r\Gamma{}^{<\alpha{}}\varphi{})})\kappa{}(-\nu{}) \\
&\le\epsilon{}r^{1-s}(\abs{UL\varphi{}}^2 + \abs{\partial{}_vL\varphi{}}^2)\kappa{}(-\nu{}) \\
&\qquad + r^{-1+3s}C(\mathcal{G}_{\alpha,s},\epsilon{})(\abs{U\Gamma{}^{<\alpha{}}\varphi{}}^2 + \abs{\partial{}_v\Gamma{}^{<\alpha{}}\varphi{}})\kappa{}(-\nu{}) \\
&\qquad + C(\mathcal{G}_{\alpha,s},\epsilon{})r^{-1+3s}\mathbf{1}_{r\ge \Rc}\abs{\partial{}_v(r\Gamma{}^{<\alpha{}}\varphi{})}.
\end{split}
\end{equation}
Integrate (noting that \(s\ge \eta_0\)) to obtain
\begin{equation}\label{err-wave-bound-prep-2}
\text{Err}_U[L\varphi{},\mathcal{R}]\le \epsilon{}E_{\text{bulk}}[L\varphi{},\mathcal{R}] + C(\mathcal{G}_{\alpha,s},\epsilon{})[E_{\text{bulk}}[\Gamma{}^{<\alpha{}}\varphi{},\mathcal{R}] + E_{3s,\text{bulk}}[\Gamma{}^{<\alpha{}}\varphi{},\mathcal{R}]].
\end{equation}
Using the estimate for \(\partial{}_vL\varphi{}\Box{}L\varphi{}\) in
\cref{wave-pointwise-3}, one
obtains
\begin{equation}
\begin{split}
\abs{\partial{}_vL\varphi{}}\abs{\Box{}L\varphi{}}r^2\kappa{}(-\nu{})&\le \mathbf{1}_{r\ge \Rc}C(\mathcal{B}_0^{\circ },\mathfrak{g}_0,\alpha{})r^{s}\abs{\partial{}_vL\varphi{}}(\abs{UL\varphi{}} + \abs{\partial{}_vL\varphi{}}) + \side{RHS}{err-wave-bound-prep} \\
&\le C(\mathcal{B}_0^{\circ },\mathfrak{g}_0,\alpha{})\Rc^{-1+2s}r^{1-s}(\abs{\partial{}_vL\varphi{}}^2 + \abs{UL\varphi{}}^2) + \side{RHS}{err-wave-bound-prep} \\
\end{split}
\end{equation}
Integrate to obtain
\begin{equation}\label{err-wave-bound-prep-3}
\text{Err}_{v}[L\varphi{},\mathcal{R}]\le \Rc^{-1+2s}C(\mathcal{B}_0^{\circ },\mathfrak{g}_0,\alpha{})E_{\text{bulk}}[L\varphi{},\mathcal{R}] + \side{RHS}{err-wave-bound-prep-2} \\
\end{equation}
Combine \cref{err-wave-bound-prep-2,err-wave-bound-prep-3} to complete the proof.
\end{proof}
For the next two lemmas we introduce the following notation:
\begin{itemize}
\item \(C^{\mathrm{out}}(u)\) for the constant-\(u\) curve \(\set{(u',v') \in  \mathcal{R}: u' = u}\),
\item \(v_1(u)\) and \(v_2(u)\) for the past and future endpoints of
\(C^{\text{out}}(u)\) i(the assumption that \(\mathcal{R}\) is an admissible
spacetime region ensures that \(C^{\text{out}}(u)\) is connected),
\item \(u_1\) and \(u_2\) for the smallest and largest \(u\)-values attained in \(\mathcal{R}\).
\end{itemize}
\begin{lemma}[Estimate for quartic error term]
Let \(\abs{\alpha{}}\ge 1\) and \(L\in \Gamma^\alpha{}\). There is \(c_{\nu} > 0\) such that for \(p\ge 1\) and
\(\epsilon{} > 0\), we have
\begin{equation}
\begin{split}
\mathrm{Err}_{\mathrm{quartic}}[L\varphi{},\mathcal{R}]&\le \epsilon{}\mathcal{P}_{<\alpha{},p}^2 E_{\mathrm{bulk}}[L\varphi{},\mathcal{R}] \\
&\qquad + C(\mathfrak{B}_0,\mathcal{P}_{<\alpha{},p},\epsilon{})\Bigl[E_{\mathrm{bulk}}[\Gamma{}^{<\alpha{}}\varphi{},\mathcal{R}] + \int_{u_1}^{u_2} (u^{-p} + e^{-c_{\nu}u})E[L\varphi{},C^{\mathrm{out}}(u)]\dd{}u\Bigr].
\end{split}
\end{equation}
\label{err-quartic-bound}
\end{lemma}
\begin{proof}
In Step 1 (see \cref{err-quartic-step-1}), we prove the estimate for \(\alpha{} = U\), and
Step 2 (see \cref{err-quartic-step-2}) implies the estimate for \(\alpha{} > U\).

\step{Step 1: Estimate for \(\mathrm{Err}_{\mathrm{quartic}}[U\varphi{},\mathcal{R}]\).} We start by showing that
\begin{equation}\label{err-quartic-step-1}
\mathrm{Err}_{\mathrm{quartic}}[U\varphi{},\mathcal{R}]\le C(\mathfrak{b}_0)\mathcal{P}_{0,0}^2E_{\mathrm{bulk}}[\varphi{},\mathcal{R}].
\end{equation}
By \cref{vU-box}, we have
\begin{equation}
(\partial{}_vU\varphi{})^2\le r^{-2}C(\mathfrak{b}_0)[(\partial{}_v\varphi{})^2 + (U\varphi{})^2]\le C(\mathfrak{b}_0)\mathcal{P}_{0,0}^2.
\end{equation}
Thus
\begin{equation}
\begin{split}
\mathrm{Err}_{\mathrm{quartic}}[U\varphi{},\mathcal{R}] &\le  \iint_{\mathcal{R}} \frac{r}{\kappa{}}(U\varphi{})^2(\partial{}_vU\varphi{})^2(-\nu{})\dd{}u\dd{}v \le C(\mathfrak{b}_0)\mathcal{P}_{0,0}^2\iint_{\mathcal{R}} \frac{r^{-1}}{\kappa{}}(U\varphi{})^2(-\nu{})\dd{}u\dd{}v\\
&\le C(\mathfrak{b}_0)\mathcal{P}_{0,0}^2E_{\mathrm{bulk}}[\varphi{},\mathcal{R}],
\end{split}
\end{equation}
which proves \cref{err-quartic-step-1}.

\step{Step 2: Estimate for \(\mathrm{Err}_{\mathrm{quartic}}[\psi{},\mathcal{R}]\) in terms
of norms of \(U\varphi{}\).} Let \(\epsilon{} > 0\) and \(p\ge 1\). We claim that
\begin{equation}\label{err-quartic-step-2}
\mathrm{Err}_{\mathrm{quartic}}[\psi{},\mathcal{R}]\le \epsilon{}\mathcal{P}_{U,p}^2E_{\mathrm{bulk}}[\psi{},\mathcal{R}] + \mathcal{P}_{U,p}^2C(\mathfrak{B}_0,\epsilon{})\int_{u_1}^{u_2} (u^{-p} + e^{-c_{\nu}u})E[\psi{},C^{\mathrm{out}}_{v_1,v_2}(u)]\dd{}u.
\end{equation}
Since \(r_{\mathrm{min}}^{-1}\le C\), we immediately obtain
\begin{equation}\label{err-quartic-step-2-prep}
\begin{split}
\mathrm{Err}_{\mathrm{quartic}}[\psi{},\mathcal{R}]\le C\mathcal{P}_{U,p}^2\iint_{\mathcal{R}} \frac{r^2}{\kappa{}}\tau{}^{-p}(\partial{}_v\psi{})^2(-\nu{})\dd{}u\dd{}v.
\end{split}
\end{equation}
Split the integral in \cref{err-quartic-step-2-prep} as
\begin{equation}\label{err-quartic-step-2-prep-2}
\iint_{\mathcal{R}} =
\iint_{\mathcal{R}\cap \set{r\le R_0}\cap \set{\tau{}\ge \epsilon{}^{-1}}} + \iint_{\mathcal{R}\cap \set{r\le R_0}\cap \set{\tau{}\le \epsilon{}^{-1}}} +  \iint_{\mathcal{R}\cap \set{r\ge R_0}} \coloneqq{} \text{(I)} + \text{(II)} + \text{(III)}.
\end{equation}
To handle term \(\text{(I)}\), let \(c_{\nu}>0\) be as in \cref{nu-exponential-bound}
and control \((-\nu{})\) with \cref{nu-exponential-bound}, use the smallness of
\(\tau{}^{-1}\) and the near-horizon term in \(E_{\text{bulk}}\) (see \cref{sec:bulk-energy}):
\begin{equation}\label{err-quartic-step-2-1}
\text{(I)} \le C\epsilon{}\iint_{\mathcal{R}\cap \set{r\ge R_0}} e^{-c_{\nu}(u-v)} \frac{r^2}{\kappa{}}(\partial{}_v\psi{})^2\dd{}u\dd{}v\le C\epsilon{}E_{\mathrm{bulk}}[\psi{},\mathcal{R}].
\end{equation}
Now we treat term \(\text{(II)}\). Recall the notation introduced before the statement of this lemma. By
\cref{v-u-r-compare}, in the region \(\set{r\le R_0}\cap \set{\tau{}\le
\epsilon^{-1}}\) we have \(v\le v_\ast{}\) for some \(v_\ast{}\le
C(\mathfrak{B}_0,\epsilon{})\). Since \(\tau{}\ge 1\), the bound on \((-\nu{})\) in
\(\set{r\le R_0}\) from \cref{nu-exponential-bound} implies
\begin{equation}\label{err-quartic-step-2-2}
\begin{split}
\text{(II)} &\le   C\int_{u_1}^{u_2} \int_{v_1(u)}^{\min (v_2(u),v_{R_0}(u),v_\ast{})} e^{-c_{\nu}(u-v)}\frac{r^2}{\kappa{}}(\partial{}_v\psi{})^2(u,v)\dd{}v\dd{}u\\
&\le  Ce^{c_{\nu}v_\ast{}}\int_{u_1}^{u_2} e^{-c_{\nu}u}\int_{v_1}^{v_2} \frac{r^2}{\kappa{}}(\partial{}_v\psi{})^2(u,v)\dd{}v\dd{}u = C(\mathcal{B}_0,\epsilon{})\int_{u_1}^{u_2} e^{-c_{\nu}u} E[\psi{},C^{\mathrm{out}}(u)]\dd{}u.
\end{split}
\end{equation}
Finally, we handle term \(\text{(III)}\). We have \((-\nu{})\le C\) by \cref{nu-bound},
and \(\tau{} = u\) in \(\set{r\ge R_0}\) by definition, so
\begin{equation}\label{err-quartic-step-2-3}
\begin{split}
\text{(III)} &= \int_{u_1}^{u_2} \int_{\max(v_1(u),v_{R_0}(u))}^{v_2(u)} \frac{r^2}{\kappa{}}\tau{}^{-p}(\partial{}_v\psi{})^2(-\nu{})(u,v)\dd{}v\dd{}u \\
&\le C\int_{u_1}^{u_2} u^{-p}\int_{\max(v_1(u),v_{R_0}(u))}^{v_2(u)} \frac{r^2}{\kappa{}}(\partial{}_v\psi{})^2(u,v)\dd{}v\dd{}u\le C\int_{u_1}^{u_2} u^{-p}E[\psi{},C^{\mathrm{out}}(u)]\dd{}u.
\end{split}
\end{equation}
We conclude \cref{err-quartic-step-2} by substituting
\cref{err-quartic-step-2-1,err-quartic-step-2-2,err-quartic-step-2-3} into
\cref{err-quartic-step-2-prep-2,err-quartic-step-2-prep}.
\end{proof}
\begin{lemma}
Let \(L\in \Gamma^\alpha{}\). For \(\Rc\) sufficiently large depending on
\(C(\mathcal{B}_0^{\circ },\mathfrak{g}_0,\alpha{})\) and \(s\ge \eta_0\) sufficiently
small, we have
\begin{equation}
E[L\varphi{},\Sigma{}_2] + E_{\mathrm{bulk}}[L\varphi{},\mathcal{R}]\le C(\mathcal{G}_{\alpha,s},\mathcal{P}_{<\alpha{},1+\eta{}_0})[E[\Gamma{}^{\le \alpha{}}\varphi{},\Sigma{}_1] + r\abs{\Gamma{}^{\le \alpha{}}\varphi{}}^2(p_{\mathrm{fut}}) + E_{4s,\mathrm{bulk}}[\Gamma{}^{<\alpha{}}\varphi{},\mathcal{R}]].
\end{equation}
\label{energy-control-1}
\end{lemma}
\begin{proof}
We induct on \(\alpha{}\). The base case \(\alpha{}= 0\) follows from \cref{energy:en-bound},
since \(\mathrm{Err}[\varphi{},\mathcal{R}]\) vanishes. It is therefore enough
to show that
\begin{equation}\label{energy-control-1-prep}
E[L\varphi{},\Sigma{}_2] + E_{\mathrm{bulk}}[L\varphi{},\mathcal{R}]\le C(\mathcal{G}_{\alpha,s},\mathcal{P}_{<\alpha{},p})[E[L\varphi{},\Sigma{}_1] + r\abs{L\varphi{}}^2(u_1,v_2) + E_{\mathrm{bulk}}[\Gamma{}^{<\alpha{}}\varphi{},\mathcal{R}] +  E_{4s,\mathrm{bulk}}[\Gamma{}^{<\alpha{}}\varphi{},\mathcal{R}]].
\end{equation}

Let \(\epsilon_1,\epsilon_2,\epsilon_3 > 0\). Combining
\cref{err-zo-bound,err-wave-bound,err-quartic-bound} (and using \(\abs{\alpha{}}\ge
1\) so that \(\mathfrak{B}_0^{\circ }\le C(\mathcal{G}_{\alpha{},s})\)) gives
\begin{equation}
\begin{split}
\text{Err}[L\varphi{},\mathcal{R}]&\le(\epsilon{}_1 + \epsilon{}_2 + \epsilon{}_3\mathcal{P}_{<\alpha{},1+\eta{}_0}^2 + \Rc^{-1+2s})C(\mathcal{B}_0^{\circ },\mathfrak{g}_0,\alpha{})E_{\mathrm{bulk}}[L\varphi{},\mathcal{R}] \\
&\qquad + C(\mathcal{G}_{\alpha,s},\mathcal{P}_{<\alpha{},1+\eta{}_0},\epsilon{}_1,\epsilon{}_2,\epsilon{}_3)\Bigl[E_{\mathrm{bulk}}[\Gamma{}^{<\alpha{}}\varphi{},\mathcal{R}] + E_{4s,\mathrm{bulk}}[\Gamma{}^{<\alpha{}}\varphi{},\mathcal{R}] \\
&\qquad + \int_{u_1}^{u_2} (u^{-1-\eta{}_0} + e^{-c_{\nu}u})E[L\varphi{},C^{\mathrm{out}}(u)]\dd{}u\Bigr].
\end{split}
\end{equation}
One can choose \(\epsilon_1\), \(\epsilon_2\), and \(\epsilon_3\) small based on \(C(\mathcal{B}_0^{\circ },\mathfrak{g}_0
,\mathcal{P}_{<\alpha{},1 + \eta_0},\alpha{})\) and \(\Rc\) large based on
\(C(\mathcal{B}_0^{\circ },\mathfrak{g}_0,\alpha{})\), to obtain
\begin{equation}\label{energy:en-bound-1}
\begin{split}
\mathrm{Err}[L\varphi{},\mathcal{R}]&\le \frac{1}{2}E_{\mathrm{bulk}}[L\varphi{},\mathcal{R}] + C(\mathcal{G}_{\alpha,s},\mathcal{P}_{<\alpha{},1+\eta{}_0})\Bigl[E_{\mathrm{bulk}}[\Gamma{}^{<\alpha{}}\varphi{},\mathcal{R}] +E_{4s,\mathrm{bulk}}[\Gamma{}^{<\alpha{}}\varphi{},\mathcal{R}] \\
&\qquad + \int_{u_1}^{u_2} (u^{-p} + e^{-c_{\nu}u})E[L\varphi{},C^{\mathrm{out}}(u)\Bigr]. \\
\end{split}
\end{equation}
Now return to \cref{energy:en-bound} and absorb the
\(E_{\mathrm{bulk}}[L\varphi{},\mathcal{R}]\) term on the right of
\cref{energy:en-bound-1} to the left, and then apply the same analysis to the
subregion \(\mathcal{R}(u_\ast{})\coloneqq{}\mathcal{R}\cap \set{u\le u_\ast{}}\) for \(u_\ast{}\in [u_1,u_2]\) to obtain
\begin{equation}\label{energy:en-bound-gronwall}
\begin{split}
E[L\varphi{},\Sigma{}_2(u_\ast{})] + E_{\mathrm{bulk}}[L\varphi{},\mathcal{R}(u_\ast{})]&\le C(\mathcal{G}_{\alpha,s},\mathcal{P}_{<\alpha{},1+\eta{}_0})\Bigl[E[L\varphi{},\Sigma{}_1] + r\abs{L\varphi{}}^2(u_1,v_2) + E_{\mathrm{bulk}}[\Gamma{}^{<\alpha{}}\varphi{},\mathcal{R}] \\
&\qquad + E_{4s,\mathrm{bulk}}[\Gamma{}^{<\alpha{}}\varphi{},\mathcal{R}] + \int_{u_1}^{u_\ast{}} (u^{-1-\eta{}_0} + e^{-c_{\nu}u})E[L\varphi{},C^{\mathrm{out}}(u)]\dd{}u\Bigr],
\end{split}
\end{equation}
where we have written \(\Sigma_2(u_\ast{})\) for the future boundary of
\(\mathcal{R}(u_\ast{})\). Since \(E[L\varphi{},C^{\mathrm{out}}(u_\ast{})]\le
E[L\varphi{},\Sigma_2(u_\ast{})]\), we can use Grönwall's
inequality in \cref{energy:en-bound-gronwall} to obtain
\begin{equation}\label{energy:en-bound-2}
\begin{split}
\max _{u\in [u_1,u_\ast{}]}E[L\varphi{},C^{\mathrm{out}}(u)]&\le C(\mathcal{G}_{\alpha,s},\mathcal{P}_{<\alpha{},p})[E[L\varphi{},\Sigma{}_1] + r\abs{L\varphi{}}^2(u_1,v_2) \\
&\qquad +  E_{\mathrm{bulk}}[\Gamma{}^{<\alpha{}}\varphi{},\mathcal{R}] + E_{4s,\mathrm{bulk}}[\Gamma{}^{<\alpha{}}\varphi{},\mathcal{R}]].
\end{split}
\end{equation}
Combine \cref{energy:en-bound-2,energy:en-bound-gronwall} and take \(u_\ast{} = u_2\)
to conclude \cref{energy-control-1-prep}.
\end{proof}
From now on, we fix \(\Rc\) large enough that \cref{energy-control-1} holds.
\begin{lemma}[Energy boundedness estimate on characteristic rectangles]
Let \(\mathcal{R}\) be a characteristic rectangle with past \(\Sigma_1\). For any
null curve \(\Sigma{}\subset \mathcal{R}\), when \(s\ge \eta_0\) is sufficiently small we have
\begin{equation}
E[L\varphi{},\Sigma{}] + E_{\mathrm{bulk}}[\psi{},\mathcal{R}] \le C(\mathcal{G}_{\alpha,s},\mathcal{P}_{<\alpha{},1+\eta{}_0})[E[\Gamma{}^{\le \alpha{}}\varphi{},\Sigma{}_1] + r\abs{\Gamma{}^{\le \alpha{}}\varphi{}}^2(p_{\mathrm{fut}}) + E_{4s,\mathrm{bulk}}[\Gamma{}^{<\alpha{}}\varphi{},\mathcal{R}]].
\end{equation}
\label{energy-boundedness-curves-in-future}
\end{lemma}
\begin{proof}
A null curve \(\Sigma{}\subset \mathcal{R}\) can be extended to a piecewise null curve
\(\Sigma_2\) so that \(\Sigma_1\) and \(\Sigma{}_2\) are the past and future
boundaries of a rectangular region \(\mathcal{R}'\subset \mathcal{R}\). Apply
\cref{energy-control-1} for the region \(\mathcal{R}'\), and take a supremum over
null curves \(\Sigma{}\subset \mathcal{R}\) (and hence \(\mathcal{R}'\subset \mathcal{R}\)).
\end{proof}
\begin{corollary}[Estimate for unweighted energy norm]
When \(s\ge \eta_0\) is sufficiently small, we have
\begin{equation}
\mathcal{E}_\alpha^2 + E_{\mathrm{bulk}}[L\varphi{},\mathcal{R}_{\mathrm{char}}]\le C(\mathcal{G}_{\alpha,s},\mathcal{P}_{<\alpha{},1+\eta{}_0})[\mathcal{D}_\alpha^2 + \mathcal{E}_{<\alpha{},4s}^2].
\end{equation}
\label{E-alpha-bound-2}
\end{corollary}
\begin{proof}
Apply \cref{energy-boundedness-curves-in-future} to a characteristic rectangle
\(\mathcal{R}\subset \mathcal{R}_{\text{char}}\) whose past \(\Sigma_1\) is
contained in the past of \(\mathcal{R}_{\text{char}}\). Since \(\Sigma_1\) and
\(p_{\mathrm{fut}}\) are contained in the data surface \(C^{\text{in}}\cup
C^{\text{out}}\), the energy \(E[\psi{},\Sigma_1]\) and the value
\(r\psi{}^2(p_{\mathrm{fut}})\) are controlled by data (see \cref{energy-on-data}). Take
a supremum over \(\mathcal{R}\subset \mathcal{R}_{\text{char}}\).
\end{proof}
\subsection{\texorpdfstring{\(r^p\)}{rp}-weighted energy estimate}
\label{sec:org05e270c}
We now adapt the work of Dafermos--Rodnianski \cite{rp-method} to derive an
\(r^p\)-weighted energy estimate. Recall the weighted energy quantity \(E_p\)
defined in \cref{sec:rp-weighted-energy}.
\begin{lemma}[Preliminary estimate for the quantity \(E_p\)]
Suppose that \(\lim_{v\to \infty}r\psi^2(u,v) = 0\) for each \(u\ge 1\). Then for any \(p\ge
0\) and \(R\ge R_0 + 1\) and \(\tau{}\ge 1\), we have
\begin{equation}\label{rp-Ep-bound}
E_p[\psi{}](\tau{})\le \int _{\Sigma{}_\tau^{\mathrm{out}}\cap \set{r\ge R}}\frac{r^p}{\lambda{}}(\partial{}_v(r\psi{}))^2\dd{}v + C(\mathfrak{b}_0,R,p)E[\psi{},\Sigma{}_\tau{}],
\end{equation}
and for \(1\le \tau_1<\tau_2\) we have
\begin{equation}\label{rp-Ep-bound-integrated}
\int_{\tau{}_1}^{\tau{}_2} E_p[\psi{}](\tau{})\le \int_{\tau{}_1}^{\tau{}_2} \Bigl(\int _{\Sigma{}_\tau^{\mathrm{out}}\cap \set{r\ge R}}\frac{r^p}{\lambda{}}(\partial{}_v(r\psi{}))^2\dd{}v\Bigr)\dd{}\tau{} + C(\mathfrak{b}_0,R,p)E_{\mathrm{bulk}}[\psi{},\mathcal{R}(\tau{}_1,\tau{}_2)],
\end{equation}
\label{rp-Ep-bounds}
\end{lemma}
\begin{proof}
Introduce the following energy quantity that is localized to the region \(\set{r\le R}\):
\begin{equation}
\begin{split}
E_{\le R}^{\circ }[\psi{}](\tau{})&\coloneqq{}\int _{\Sigma{}_\tau^{\text{in}}}\frac{r^2}{(-\nu{})}(\partial{}_u\psi{})^2\dd{}u +  \int _{\Sigma{}_\tau^{\mathrm{out}}\cap \set{r\le R}}\frac{r^2}{\kappa{}}(\partial{}_v\psi{})^2\dd{}v + \int _{\Sigma{}_\tau^{\mathrm{out}}\cap \set{R-1\le r\le R}} \lambda{}r\psi{}^2\dd{}v \\
&= E[\psi{},\Sigma{}_\tau{}\cap \set{r\ge R}] + \int _{\Sigma{}_\tau^{\mathrm{out}}\cap \set{R-1\le r\le R}} \lambda{}r\psi{}^2\dd{}v.
\end{split}
\end{equation}
Observe that the desired \cref{rp-Ep-bound,rp-Ep-bound-integrated} are a consequence of the
following three estimates:
\begin{align}
E_p[\psi{}](\tau{})&\le \int _{\Sigma{}_\tau^{\mathrm{out}}\cap \set{r\ge R}}\frac{r^p}{\lambda{}}(\partial{}_v(r\psi{}))^2\dd{}v + C(\mathfrak{b}_0,R,p)E_{\le R}^{\circ }(\tau{}) \label{rp-Ep-bound-1} \\
E_{\le R}^{\circ }(\tau{}) &\le C(\mathfrak{b}_0)[E[\psi{},\Sigma{}_\tau\cap \set{r\ge R}] + \limsup_{v\to \infty}r\psi^2|_{u=\tau{}}(v))] \label{rp-Ep-bound-2} \\
\int_{\tau{}_1}^{\tau{}_2} E[\psi{},\Sigma{}_\tau\cap \set{r\le R}]\dd{}\tau{} &\le C(\mathfrak{b}_0,R)E_{\mathrm{bulk}}[\psi{},\mathcal{R}(\tau{}_1,\tau{}_2)] \label{rp-Ep-bound-3}.
\end{align}
We now prove \cref{rp-Ep-bound-1,rp-Ep-bound-2,rp-Ep-bound-3} one by one.

\step{Step 1: Proof of \cref{rp-Ep-bound-1}.} By the definitions of \(E_p\) and \(E_{\le
 R}^{\circ }\), it is enough to show that
\begin{align}
\int _{\Sigma{}_\tau^{\mathrm{out}}\cap \set{r\le R}}\frac{r^p}{\lambda{}}(\partial{}_v(r\psi{}))^2\dd{}v + \abs{\psi{}}^2(\tau{},v_{R_0}(\tau{}))&\le C(\mathfrak{b}_0,R,p)E_{\le R}^{\circ }[\psi{}](\tau{}). \label{rp-in-bound-dv}
\end{align}
To prove \cref{rp-in-bound-dv}, we start with the following pointwise estimate on
\(\Sigma_\tau^{\mathrm{out}}\cap \set{r\le R}\subset \set{R_0\le r\le R}\):
\begin{equation}\label{rp-in-bound-dv-prep}
\frac{r^p}{\lambda{}}(\partial_v(r\psi{}))^2\le R^p[\lambda^2\psi^2 + \frac{r^2}{\lambda{}}(\partial{}_v\psi{})^2]\le C(\mathfrak{b}_0,R,p)\bigl[\lambda{}r\psi{}^2 + \frac{r^2}{\kappa{}}(\partial{}_v\psi{})^2\bigr].
\end{equation}
Integrate \cref{rp-in-bound-dv-prep} over \(v\in [v_{R_0}(\tau{}),v_R(\tau{})]\) and add
\(R_0\abs{\psi{}}^2(\tau{},v_{R_0}(\tau{}))\) to both sides to get
\begin{equation}\label{rp-in-bound-dv-prep-2}
\side{LHS}{rp-in-bound-dv}\le C(\mathfrak{b}_0,R,p)\Bigl[\int _{\Sigma{}_\tau^{\mathrm{out}}\cap \set{r\le R}}\lambda{}r\psi{}^2\dd{}v + R_0\abs{\psi{}}^2(\tau{},v_{R_0}(\tau{})) + \int _{\Sigma{}_\tau^{\mathrm{out}}\cap \set{r\le R}}\frac{r^2}{\kappa{}}(\partial{}_v\psi{})^2\dd{}v\Bigr].
\end{equation}
For any \(R'\in [R-1,R]\), we can control the first two terms on the right side of
\cref{rp-in-bound-dv-prep-2} using Hardy's inequality (\cref{hardy-inequality} with
\(a=0\)) and arrive at
\begin{equation}
\side{LHS}{rp-in-bound-dv}\le C(\mathfrak{b}_0,R,p)\Bigl[R'\abs{\psi{}}^2(\tau{},v_{R'}(\tau{})) + \int _{\Sigma{}_\tau^{\mathrm{out}}\cap \set{r\le R}}\frac{r^2}{\kappa{}}(\partial{}_v\psi{})^2\dd{}v\Bigr].
\end{equation}
Averaging over \(R'\in [R-1,R]\) yields \cref{rp-in-bound-dv}.

\step{Step 2: Proof of \cref{rp-Ep-bound-2}.} It is enough to prove \cref{rp-Ep-bound-2}
with the left side replaced by \(\int _{\Sigma_\tau^{\mathrm{out}}\cap \set{R-1\le
r\le R}} \lambda{}r\psi^2\dd{}v\); this estimate is a consequence of Hardy's
inequality (\cref{hardy-inequality} with \(a=0\) in the limit \(v_2\to \infty\)).

\step{Step 3: Proof of \cref{rp-Ep-bound-3}.} Recall from \cref{tau-def} that \(\tau{}\) is a
function of \(v\) in \(\set{r\le R_0}\) (with \(r(\tau{},v) = R_0\)) and a
function of \(u\) in \(\set{r\ge R_0}\) (with \(\tau{} = u\)). Thus
\cref{rp-Ep-bound-3} follows from changing variables with the following
computation,
\begin{equation}
\dv{\tau{}}{v}(u,v) = \frac{\kappa{}}{(-\gamma{})}(u_{R_0}(v),v)\quad \text{in }\set{r\le R_0}, \qquad \dv{\tau{}}{u}(u,v) = 1\quad \text{in }\set{r\ge R_0}.
\end{equation}
\end{proof}
\begin{lemma}
Let \(p\in \R\). We have
\begin{equation}
\begin{split}
&\partial{}_u\Bigl(\frac{r^p}{\lambda{}}(\partial{}_v(r\psi{}))^2\Bigr) + \partial{}_v\Bigl(\frac{r^p}{\lambda{}} \frac{(-\partial{}_u\partial{}_vr)}{r}(r\psi{})^2\Bigr) \\
&\qquad + \Bigl(p(-\nu{}) - \frac{2(\varpi{}-\mathbf{e}^2/r)}{r}(-\gamma{})\Bigr) \frac{r^{p-1}}{\lambda{}}(\partial{}_v(r\psi{}))^2 + (3-p)r^{p-4}r^2(-\partial{}_u\partial{}_vr)(r\psi{})^2 \\
&= 2\frac{r^{p+1}}{\lambda{}}\partial{}_v(r\psi{})\kappa{}(-\nu{})\Box{}\psi{} + r^{p-3}\partial{}_v\Bigl(r^2 \frac{(-\partial{}_u\partial{}_vr)}{\lambda{}}\Bigr)\psi{}^2.
\end{split}
\end{equation}
\label{rp-vf-identity}
\end{lemma}
\begin{proof}
We refer to \cite[Lem.~8.52]{luk-oh-scc2} for details of the computation.
\end{proof}
\begin{lemma}[\(r^p\)-weighted energy estimate]
Suppose that \(\lim_{v\to \infty}r\psi^2(u,v) = 0\) for each \(u\ge 1\). Then for \(1\le \tau_1\le
\tau_2\) and \(p\in (0,3)\), we have
\begin{equation}\label{rp-estimate-equation}
\begin{split}
&E_p[\psi{}](\tau{}_2) + \int_{\tau{}_1}^{\tau{}_2} E_{p-1}[\psi{}](\tau{})\dd{}\tau{}\\
&\le C(\mathfrak{b}_0,\mathcal{P}_{0,0},p)\Bigl[E_p[\psi{}](\tau{}_1) + \iint_{\mathcal{R}(\tau{}_1,\tau{}_2)}r^{p+3}\abs{\Box{}\psi{}}^2\kappa{}(-\nu{})\dd{}u\dd{}v + E_{\mathrm{bulk}}[\psi{},\mathcal{R}(\tau{}_1,\tau{}_2)] + E[\psi{},\Sigma{}_{\tau{}_2}]\Bigr].
\end{split}
\end{equation}
\label{rp-estimate}
\end{lemma}
\begin{proof}
Introduce the notation \(\Psi{} \coloneqq{} r\psi{}\). By \cref{rp-Ep-bounds}, it is enough to show
that for \(R = C(\mathfrak{b}_0,\mathcal{P}_{0,0},p)\), we have
\begin{equation}\label{rp-estimate-prep}
\begin{split}
& \int _{\Gamma_{\tau_2}^{\mathrm{out}}\cap \set{r\ge R}}\frac{r^p}{\lambda{}}(\partial_v\Psi)^2\dd{}v + \int_{\tau{}_1}^{\tau{}_2}\Bigl( \int _{\Gamma_{\tau}^{\mathrm{out}}\cap \set{r\ge R}} \frac{r^{p-1}}{\lambda{}}(\partial{}_v\Psi)^2\dd{}v\Bigr)\dd{}\tau{}\le \side{RHS}{rp-estimate-equation}.
\end{split}
\end{equation}
To show \cref{rp-estimate-prep}, let \(\chi{}_R(r)\) be a smooth non-decreasing function
such that \(\chi_R \equiv 1\) on \(\set{r\ge R}\) and \(\chi_R\equiv 0\) on
\(\set{r\le R/2}\). To begin, let \(R \ge 3R_0\). Define
\begin{equation}
\begin{split}
\text{(I)} &\coloneqq{} \chi{}_R(r)\partial{}_u\Bigl(\frac{r^p}{\lambda{}}(\partial{}_v\Psi{})^2\Bigr) + \chi{}_R(r)\partial{}_v\Bigl(\frac{r^p}{\lambda{}} \frac{(-\partial{}_u\partial{}_vr)}{r}\Psi{}^2\Bigr), \\
\text{(II)} &\coloneqq{} \chi{}_R(r)\Bigl[p(-\nu{}) - \frac{2(\varpi{}-\mathbf{e}^2/r)}{r}(-\gamma{})\Bigr] \frac{r^{p-1}}{\lambda{}}(\partial{}_v\Psi{})^2  + \chi{}_R(r)2(3-p)r^{p-4}(\varpi{}-\mathbf{e}^2/r)\lambda{}(-\gamma{})\Psi{}^2, \\
\text{(III)} &\coloneqq{}\chi{}_R(r)\Bigl[2\frac{r^{p+1}}{\lambda{}}\partial{}_v\Psi{}r\Box{}\psi{} + r^{p-3}\partial{}_v\Bigl(r^2 \frac{(-\partial{}_u\partial{}_vr)}{\lambda{}}\Bigr)\psi{}^2\Bigr],
\end{split}
\end{equation}
so that the result of multiplying the statement of \cref{rp-vf-identity} by \(\chi_R(r)\) reads \(\text{(I)} + \text{(II)} = \text{(III)}\).
By \cref{lambda-lower-bound,mu-estimate,nu-lower-bound}, for \(R\ge
C(\mathfrak{b}_0,p,\varpi{}_i,r_{\mathrm{min}})\) and \(c = c(c_{\mathcal{H}})\), the
square bracketed term in \(\text{(I)}\) is positive, and we have
\begin{equation}\label{rp-lower-bound-II}
\text{(II)} \ge \chi_R(r)cp \frac{r^{p-1}}{\lambda{}}(\partial{}_v\Psi{})^2 + \chi_R(r)c(3-p)r^{p-4}\Psi{}^2.
\end{equation}
Since \(p \in (0,3)\), the right side of \cref{rp-lower-bound-II} is positive. Since
\begin{equation}
\mathbf{1}_{r\ge R_0}\abs[\Big]{\partial{}_v\Bigl(r^2 \frac{(-\partial{}_u\partial{}_vr)}{\lambda{}}\Bigr)} = \mathbf{1}_{r\ge R_0}\abs[\Big]{\partial{}_v(2(\varpi{}-\mathbf{e}^2/r)(-\gamma{}))}\le C(\mathfrak{b}_0,\mathcal{P}_{0,0})r^{-2},
\end{equation}
we have
\begin{equation}\label{rp-upper-bound-III}
\begin{split}
\abs{(\text{III})}&\le  C(\mathfrak{b}_0)\side{RHS}{rp-lower-bound-II}^{1/2}\cdot \bigl(r^{p+3}\abs{\Box{}\psi{}}^2\bigr)^{1/2} + r^{-1}C(\mathfrak{b}_0,\mathcal{P}_{0,0})\side{RHS}{rp-lower-bound-II} \ \\
&\le \frac{1}{2}\text{(II)} + C(\mathfrak{b}_0,\mathcal{P}_{0,0})r^{p+3}\abs{\Box{}\psi{}}^2.
\end{split}
\end{equation}
Integrate by parts, using \(R\ge 3R_0\) to note that \(\Supp \chi{}_R\) is disjoint
from \(\Gamma_\tau^{\text{in}}\):
\begin{equation}
\begin{split}
\iint_{\mathcal{R}(\tau_1,\tau_2)} \text{(I)} &= \int _{\Gamma_{\tau_2}^{\mathrm{out}}} \chi_R(r) \frac{r^p}{\lambda{}}(\partial_v\Psi{})^2\dd{}v - \int _{\Gamma_{\tau_1}^{\mathrm{out}}} \chi_R(r) \frac{r^p}{\lambda{}}(\partial_v\Psi{})^2\dd{}v \\
&\qquad - \iint_{\mathcal{R}(\tau_1,\tau_2)} \chi_R'(r)\Bigl[\frac{r^p}{\lambda{}}(\partial{}_v\Psi{})^2 + \frac{r^p}{\lambda{}} \frac{(-\partial{}_u\partial{}_vr)}{r}\Psi{}^2\Bigr]\dd{}u\dd{}v.
\end{split}
\end{equation}
Integrating the identity \(\text{(I)} + \text{(II)} = \text{(III)}\) and using
\cref{rp-lower-bound-II,rp-upper-bound-III} gives
\begin{equation}\label{rp-1}
\begin{split}
&\int _{\Gamma_{\tau_2}^{\mathrm{out}}} \chi_R(r) \frac{r^p}{\lambda{}}(\partial_v\Psi{})^2\dd{}v + \int_{\tau{}_1}^{\tau{}_2} \int _{\Gamma_{\tau}^{\mathrm{out}}} \chi_R(r) \frac{r^{p-1}}{\lambda{}}(\partial_v\Psi{})^2\dd{}v\dd{}\tau{}  \\
&\lesssim  \int _{\Gamma_{\tau_1}^{\mathrm{out}}} \chi_R(r) \frac{r^p}{\lambda{}}(\partial_v\Psi{})^2\dd{}v + \iint_{\mathcal{R}(\tau_1,\tau_2)} \chi_R'(r)\Bigl[\frac{r^p}{\lambda{}}(\partial{}_v\Psi{})^2 + \frac{r^p}{\lambda{}} \frac{(-\partial{}_u\partial{}_vr)}{r}\Psi{}^2\Bigr]\dd{}u\dd{}v \\
&\qquad + \iint_{\mathcal{R}(\tau{}_1,\tau{}_2)} \chi{}_R(r)r^{p+3}\abs{\Box{}\psi{}}^2\dd{}u\dd{}v \\
&\lesssim \side{RHS}{rp-estimate-equation} + \iint_{\mathcal{R}(\tau{}_1,\tau{}_2)} \chi{}_R'(r) \frac{r^p}{\lambda{}} \frac{(-\partial{}_u\partial{}_vr)}{r}\Psi{}^2\dd{}u\dd{}v,
\end{split}
\end{equation}
where the implied constant is \(C(\mathfrak{b}_0,\mathcal{P}_{0,0},p)\). Since
\(\chi{}'_R\) is supported in a finite-\(r\) region, the last term on the right is
controlled by \(E_{\text{bulk}}[\psi{},\mathcal{R}(\tau{}_1,\tau{}_2)]\). For the remaining
parts of \(E_p[\psi{}]\) on the left of \cref{rp-estimate-equation}, we have
\begin{equation}\label{rp-2}
\int _{\Gamma_{\tau_2}^{\text{in}}} \frac{r^2}{(-\nu{})}(\partial{}_u\varphi{})^2\dd{}u + \int_{\tau{}_1}^{\tau{}_2} \int _{\Gamma_{\tau}^{\text{in}}} \frac{r^2}{(-\nu{})}(\partial{}_u\varphi{})^2\dd{}u\dd{}\tau{} \lesssim_{R_0} E[\psi{},\Sigma{}_{\tau{}_2}] + E_{\text{bulk}}[\psi{},\mathcal{R}(\tau{}_1,\tau{}_2)].
\end{equation}
An argument using Hardy's inequality allows us to estimate
\begin{equation}\label{rp-3}
\int _{\Gamma{}_{\tau{}}^{\text{out}}\cap \set{R_0\le r\le R}}\frac{r^p}{\lambda{}}(\partial{}_v\Psi{})^2\dd{}v + \int _{\Gamma{}_{\tau{}}^{\text{out}}\cap \set{R_0\le r\le R}}\frac{r^{p-1}}{\lambda{}}(\partial{}_v\Psi{})^2\dd{}v \lesssim_{R,\mathfrak{b}_0} E[\psi{},\Sigma{}_\tau{}] + E_{\text{bulk}}[\psi{},\mathcal{R}(\tau{}_1,\tau{}_2)].
\end{equation}
Add \cref{rp-1,rp-2,rp-3} to complete the proof.
\end{proof}
\subsection{Boundedness of the weighted energy \texorpdfstring{\(\mathcal{E}_{\alpha{},p}\)}{}}
\label{sec:org1659546}
\label{sec:E-alpha-p-bound} In this section, we complete the proof of
\cref{E-alpha-bound} by controlling the \(p\)-weighted energy (for \(p\) close to
\(2\)) by initial data, as well as geometric quantities, unweighted pointwise
norms, and lower order weighted pointwise norms. To this end, we first formulate
the pigeonhole argument of \cite{rp-method} by which one obtains pointwise decay
from \(r^p\)-weighted energy estimates.
\begin{lemma}[Pigeonhole argument]
Suppose \(F_p : [1,\infty)\to \R_{\ge 0}\) is a family of non-negative functions
defined for \(p\in \R\) such that
\begin{enumerate}
\item \(F_p(\tau{})\lesssim_{p,p'}F_{p'}(\tau{})\) whenever \(p\le p'\),
\item and there are \(q\ge 0\), \(\epsilon{} > 0\), and \(A\ge 0\) such that for any \(p\in
   (0,2-\epsilon{}]\) and \(1\le \tau_1\le \tau_2\), we have
\end{enumerate}
\begin{equation}\label{rp-F-estimate}
F_p(\tau{}_2) + \int_{\tau{}_1}^{\tau{}_2} F_{p-1}(\tau{})\dd{}\tau{} \lesssim_p F_{p}(\tau{}_1) + A\tau{}_1^{p-q}.
\end{equation}
Then
\begin{align}
F_{2-\epsilon{}}(\tau{}) &\lesssim F_{2-\epsilon{}}(1) + A, \label{rp-F-conclusion-2} \\
F_{1-\epsilon{}}(\tau{}) &\lesssim \tau{}^{-1}(F_{2-\epsilon{}}(1) + A\max (1,\tau{}^{2-\epsilon{}-q}), \label{rp-F-conclusion-1} \\
\int_{\tau{}/2}^\tau{} F_{-\epsilon{}}(\tau{}')\dd{}\tau{}'&\lesssim \tau{}^{-1}((F_{2-\epsilon{}}(1) + A\max (1,\tau{}^{2-\epsilon{}-q})) \label{rp-F-conclusion-0}
\end{align}
for \(\tau{}\ge 1\), with implicit constants depending on \(\epsilon{}\).
\label{rp-pigeonhole-argument}
\end{lemma}
\begin{proof}
We do not track the dependence of implicit constants on \(\epsilon{}\) (that arises from
the use of \cref{rp-F-estimate} for fixed values of \(p\) that depend on \(\epsilon{}\)).

\step{Step 1: Proof of \cref{rp-F-conclusion-2} and proof of
\cref{rp-F-conclusion-1,rp-F-conclusion-0} for \(\tau{}\in [1,2]\).} First,
\cref{rp-F-conclusion-2} follows from \cref{rp-F-estimate} for \(p = 2-\epsilon\) and
\(\tau_1 = 1\) (since the \(F_p\) are non-negative). Next, suppose \(1\le
\tau{}\le 2\). Start with \cref{rp-F-estimate} for \(p = 1-\epsilon{}\) over the
interval \([1,\tau{}]\), and then use the consequence \(F_{1-\epsilon{}}(1)\lesssim
F_{2-\epsilon{}}(1)\) of the monotonicity assumption (1) to get
\begin{equation}
F_{1-\epsilon{}}(\tau{}) + \int_{1}^{\tau{}} F_{-\epsilon{}}(\tau{}')\dd{}\tau{}' \lesssim F_{1-\epsilon{}}(1) + A \lesssim F_{2-\epsilon{}}(1) + A.
\end{equation}
This implies \cref{rp-F-conclusion-1,rp-F-conclusion-0} when \(\tau{}\in [1,2]\). We can
therefore suppose for the rest of the proof that \(\tau{}\ge 2\).

\step{Step 2: Proof of \cref{rp-F-conclusion-1} for \(\tau{}\ge 2\).} The estimate
\cref{rp-F-estimate} for \(p = 2-\epsilon{}\) and the non-negativity of the \(F_p\) give
\begin{equation}\label{rp-F-conclusion-1-step-0}
\int_{\tau{}/2}^{\tau{}} F_{1-\epsilon{}}(\tau{}')\dd{}\tau{}'\lesssim F_{2-\epsilon{}}(\tau{}/2) +  A\tau{}^{2-\epsilon{}-q}.
\end{equation}
There is \(\conj{\tau{}}\in [\tau{}/2,\tau{}]\) at which \(F_{1-\epsilon{}}\) is bounded by its
average value over the interval (``the pigeonhole principle''). Combine this
observation with \cref{rp-F-conclusion-1-step-0} and \cref{rp-F-conclusion-2} to
arrive at
\begin{equation}\label{rp-pigeonhole-1}
F_{1-\epsilon{}}(\conj{\tau{}})\le 2\tau{}^{-1}\int_{\tau{}/2}^\tau{} F_{1-\epsilon{}}(\tau{}')\dd{}\tau{}'\lesssim \tau{}^{-1}(F_{2-\epsilon{}}(1) + A\max (1,\tau{}^{2-\epsilon{}-q})).
\end{equation}
Conclude \cref{rp-F-conclusion-1} by substituting \cref{rp-pigeonhole-1} into
\cref{rp-F-estimate} with \(p = 1-\epsilon{}\) over the interval \([\conj{\tau{}},\tau{}]\).

\step{Step 3: Proof of \cref{rp-F-conclusion-0} for \(\tau{}\ge 2\).} Start with
\cref{rp-F-estimate} for \(p = 1-\epsilon\) over the interval
\([\tau{}/2,\tau{}]\), and then use the just established \cref{rp-F-conclusion-1}
and the non-negativity of \(F_p\).
\end{proof}
\begin{lemma}[Estimate for weighted energy norm]
There are explicit constants \(C_\alpha{}\) depending only on \(\alpha{}\) such that for \(s \ge
\eta_0\) sufficiently small (depending on a numerical constant), we have
\begin{equation}\label{E-alpha-p-bound-eq}
\mathcal{E}_{\alpha{},2-\eta{}_0-C_\alpha{}s} \le C(\mathcal{G}_{\alpha,s},\mathcal{P}_{\alpha{},0},\mathcal{P}_{<\alpha{},1+\eta{}_0},s)\mathcal{D}_\alpha{}.
\end{equation}
\label{lem:E-alpha-p-bound}
\end{lemma}
\begin{proof}
We use the notation \(\mathcal{R}(\tau_1,\tau_2) \coloneqq{} \bigcup_{\tau_1\le \tau{}\le \tau_2}\Sigma_\tau{}\). Write \(A_{\alpha{},s} \coloneqq{}
C(\mathcal{G}_{\alpha,s},\mathcal{P}_{\alpha{},0},\mathcal{P}_{<\alpha{},1+\eta{}_0},s)\).

\step{Step 1: \(r^p\)-weighted energy estimate and energy boundedness for the
foliation \(\Sigma_\tau{}\).} Let \(1\le \tau_1<\tau_2\). In this step we show that for \(p\in
(0,2-4B_\alpha{}s-\eta_0]\), we have
\begin{align}
E[L\varphi{},\Sigma{}_{\tau{}_2}] + E_{\mathrm{bulk}}[L\varphi{},\mathcal{R}(\tau{}_1,\tau{}_2)]&\le C(A_{\alpha,s},s)[E[\Gamma{}^{\le \alpha{}}\varphi{},\Sigma{}_{\tau{}_1}] + E_{4B_\alpha{}s}[\Gamma{}^{<\alpha{}}\varphi{}](\tau{}_1)], \label{rp:sigma-tau-boundedness} \\
E_p[L\varphi{}](\tau{}_2) + \int_{\tau{}_1}^{\tau{}_2} E_{p-1}[L\varphi{}](\tau{})\dd{}\tau{}&\le C(A_{\alpha,s},p,s)[E_p[\Gamma{}^{\le \alpha{}}\varphi{}](\tau{}_1) + E_{p+4B_\alpha{}s}[\Gamma{}^{<\alpha{}}\varphi{}](\tau{}_1)], \label{rp:controlled-error}
\end{align}
where \(B_\alpha{} = \card{}\set{\beta{} : 0\le \beta{}<\alpha{}}\) is the number of multi-indices less than
\(\alpha{}\). These follow from Steps 2abc below (see
\cref{rp:sigma-tau-boundedness-prep,rp:controlled-error-prep,rp:bulk-bound}).

\step{Step 1a: Energy boundedness.} We begin by showing that
\begin{equation}\label{rp:sigma-tau-boundedness-prep}
\begin{split}
E[L\varphi{},\Sigma{}_{\tau{}_2}] + E_{\mathrm{bulk}}[L\varphi{},\mathcal{R}(\tau{}_1,\tau{}_2)]&\le C(A_{\alpha,s})[E[\Gamma{}^{\le \alpha{}}\varphi{},\Sigma{}_{\tau{}_1}] + E_{4s,\mathrm{bulk}}[\Gamma{}^{<\alpha{}}\varphi{},\mathcal{R}(\tau{}_1,\tau{}_2)]].
\end{split}
\end{equation}
In fact, this is a consequence of \cref{energy-control-1} applied to the region
\(\mathcal{R}(\tau{}_1,\tau{}_2)\). The limiting boundary terms on
\(\mathcal{H}\) and \(\mathcal{I}\) can be controlled by another application of
\cref{energy-control-1} to the region \(\mathcal{R}(\tau_1,\tau_2)\). We can
neglect the \(p_{\mathrm{fut}}\) term because \(\Sigma_{\tau{}_1}\) extends to
\(v = \infty\) and we can assume \(\limsup_{v\to \infty}r\abs{\Gamma^{\le
\alpha{}}\varphi{}}^2(\tau_1,v) = 0\) by including \(\mathcal{P}_{\alpha{},0}\)
in the constant on the right (since the limit supremum vanishes if
\(\mathcal{P}_{\alpha{},0}\) is finite).

\step{Step 1b: Controlling the error terms in the \(r^p\)-weighted energy estimate.}
Let \(p\in (0,2-4s-\eta_0]\). We show that
\begin{equation}\label{rp:controlled-error-prep}
E_p[L\varphi{}](\tau{}_2) + \int_{\tau{}_1}^{\tau{}_2} E_{p-1}[L\varphi{}](\tau{})\dd{}\tau{}\le C(A_{\alpha,s},p)[E_p[\Gamma{}^{\le \alpha{}}\varphi{}](\tau{}_1) + E_{p+4s,\mathrm{bulk}}[\Gamma{}^{<\alpha{}}\varphi{},\mathcal{R}(\tau{}_1,\tau{}_2)]].
\end{equation}
Since \(p\le 2-2s-\eta_0\), \cref{wave-pointwise-1} implies
\begin{equation}\label{rp:error-prep}
\begin{split}
&\iint_{\mathcal{R}(\tau{}_1,\tau{}_2)}r^{p+3}\abs{\Box{}L\varphi{}}^2\kappa{}(-\nu{})\dd{}u\dd{}v\\
&\le C(\mathcal{G}_{\alpha,s})\iint_{\mathcal{R}(\tau{}_1,\tau{}_2)}r^{p-1+2s}[\abs{\partial{}_v\Gamma{}^{\le \alpha{}}\varphi{}}^2 + \abs{U\Gamma{}^{\le \alpha{}}\varphi{}}^2 + \mathbf{1}_{r\ge \Rc}\abs{\partial{}_v(r\Gamma{}^{<\alpha{}}\varphi{})}^2]\kappa{}(-\nu{})\dd{}u\dd{}v \\
&\le C(\mathcal{G}_{\alpha,s})[E_{\text{bulk}}[\Gamma{}^{\le \alpha{}}\varphi{},\mathcal{R}(\tau{}_1,\tau{}_2)] + E_{p+2s,\text{bulk}}[\Gamma{}^{<\alpha{}}\varphi{},\mathcal{R}(\tau{}_1,\tau{}_2)]].
\end{split}
\end{equation}

Substitute \cref{rp:sigma-tau-boundedness-prep,rp:error-prep} into the
\(r^p\)-weighted energy estimate \cref{rp-estimate-equation} (noting as in Step 1a
that we can assume \(\limsup_{v\to \infty}r\abs{\Gamma^{\le
\alpha{}}\varphi{}}^2(\tau{},v) = 0\) for all \(\tau{}\ge 1\) by including
\(\mathcal{P}_{\alpha{},0}\) in the constant on the right) to obtain
\begin{equation}
\begin{split}
E_p[L\varphi{}](\tau{}_2) + \int_{\tau{}_1}^{\tau{}_2} E_{p-1}[L\varphi{}](\tau{})\dd{}\tau{}&\le C(A_{\alpha,s},p)\bigl[E_p[\Gamma{}^{\le \alpha{}}\varphi{}](\tau{}_1) + E[\Gamma{}^{\le \alpha{}}\varphi{},\Sigma{}_{\tau{}_1}] \\
&\qquad + E_{\max (p+2s,4s),\mathrm{bulk}}[\Gamma{}^{<\alpha{}}\varphi{},\mathcal{R}(\tau{}_1,\tau{}_2)]\bigr].
\end{split}
\end{equation}
To complete the proof, use \(\max(p + 2s,4s)\le p + 4s\) and control \(E[\Gamma^{\le \alpha{}}\varphi{},\Sigma_{\tau_1}]\) by \(E_p[\Gamma^{\le \alpha{}}\varphi{}](\tau{}_1)\) using
\cref{rp-unweighted-energy-bound}.

\step{Step 1c: Estimate for \(r^p\)-weighted bulk term.} We claim that for \(\tau{}\ge1\)
and \(p\in (0,2-4B_\alpha{}s-\eta{}_0]\), we have
\begin{equation}\label{rp:bulk-bound}
E_{p,\mathrm{bulk}}[\Gamma{}^{\le \alpha{}}\varphi{},\mathcal{R}(\tau{},\infty)] \le C(A_{\alpha,s},p)[E_{p}[\Gamma{}^{\le \alpha{}}\varphi{}](\tau{}) + E_{p+4B_\alpha{}s}[\Gamma{}^{<\alpha{}}\varphi{}](\tau{})],
\end{equation}
where \(B_\alpha{} = \card{}\set{\beta{}:0\le \beta{}<\alpha{}}\) is the number of multi-indices smaller than \(\alpha{}\).
Indeed, Step 1b (see \cref{rp:controlled-error-prep}) implies
\begin{equation}
\begin{split}
E_{p,\mathrm{bulk}}[\Gamma{}^{\le \alpha{}}\varphi{},\mathcal{R}(\tau{}_1,\tau{}_2)]&\le C(\mathfrak{b}_0)\int_{\tau{}_1}^{\tau{}_2} E_{p-1}[\Gamma{}^{\le \alpha{}}\varphi{}](\tau{})\dd{}\tau{}\\
&\le C(A_{\alpha,s})[E_{p}[\Gamma{}^{\le \alpha{}}\varphi{}](\tau{}_1) + E_{p+4s,\mathrm{bulk}}[\Gamma{}^{<\alpha{}}\varphi{},\mathcal{R}(\tau{}_1,\tau{}_2)]].
\end{split}
\end{equation}
An induction argument establishes \cref{rp:bulk-bound}.

\step{Step 2: Interpolation.} Observe that for any \(\tau{}\ge 1\) and \(\epsilon{}\ge 0\) and \(q\ge -\epsilon{}\), we have
\begin{equation}\label{rp:interpolation}
\begin{split}
E_q[\psi{}](\tau{}) &\le \int _{\Sigma_\tau^{\text{in}}} \frac{r^2}{(-\nu{})}(\partial_u\psi{})^2\dd{}u + r\abs{\psi{}}^2(\tau{},v_{R_0}(\tau{})) \\
&\qquad + \tau{}^{q+\epsilon}\int _{\Sigma{}_\tau^{\text{out}}\cap \set{r\le \tau{}}} \frac{r^{-\epsilon}}{\lambda{}}(\partial{}_v(r\psi{}))^2\dd{}v + \tau^{q-1+\epsilon}\int _{\Sigma{}_\tau^{\text{out}}\cap \set{r\ge \tau{}}} \frac{r^{1-\epsilon}}{\lambda{}}(\partial{}_v(r\psi{}))^2\dd{}v \\
&\le \tau{}^{q+\epsilon}E_{-\epsilon}[\psi{}](\tau{}) + \tau{}^{q-1+\epsilon}E_{1-\epsilon}[\psi{}](\tau{}).
\end{split}
\end{equation}
We have bounded the first two terms on the right in the first line by
\(E_{-\epsilon}[\psi{}](\tau{})\).

\step{Step 3: Energy decay through the foliation \(\Sigma{}_\tau{}\).} We will prove inductively that for \(p\in
(0,2-4B_\alpha{}s-\eta{}_0]\) we have
\begin{equation}\label{rp:decay-induction-hypothesis}
E_{p}[\Gamma^{\le \alpha{}}\varphi{}](\tau{})\le \tau^{p-q_\alpha{}}C(A_{\alpha{},s},p)\mathcal{D}_\alpha^2,
\end{equation}
where \(q_\alpha{} = 2 - \eta{}_0 - 4\tilde{B}_\alpha{}s\) for \(\tilde{B}_\alpha{} \coloneqq{} \sum_{\beta{}\le \alpha{}}B_\beta{}\).

We first show that \cref{rp:decay-induction-hypothesis} holds for \(\alpha{} = 0\) with
\(q_0 = 2 - \eta{}_0\). By \cref{rp:controlled-error}, we have
\begin{equation}\label{rp:phi-estimate}
E_p[\varphi{}](\tau{}_2) + \int_{\tau{}_1}^{\tau{}_2} E_{p-1}[\varphi{}](\tau{})\dd{}\tau{}\le C(A_{0,s},p)E_p[\varphi{}](\tau{}_1)
\end{equation}
for \(p\in (0,2-\eta{}_0]\) and \(1\le \tau{}_1\le \tau{}_2\). Now \cref{rp-pigeonhole-argument} implies
that
\begin{equation}\label{rp:phi:rp-cons}
E_{1-\eta{}_0}[\varphi{}](\tau{}) + \int_{\tau{}/2}^\tau{} E_{-\eta{}_0}[\varphi{}](\tau{}')\dd{}\tau{}' \lesssim \tau{}^{-1}E_{2-\eta{}_0}[\varphi{}](1)\le \tau{}^{-1}C(A_{0,s})\mathcal{D}_0^2, \\
\end{equation}
where we used \cref{rp:initial-energy-bound} to control \(E_{2-\eta_0}[\varphi{}](1)\) by
initial data. Interpolate using \cref{rp:interpolation} to obtain
\begin{equation}\label{rp:E-p-phi-bound}
E_p[\varphi{}](\tau{})\le \tau{}^{p+\eta{}_0}E_{-\eta{}_0}[\varphi{}](\tau{}) + \tau{}^{p-1+\eta{}_0}E_{1-\eta{}_0}[\varphi{}](\tau{}).
\end{equation}
A pigeonhole argument applies to \cref{rp:phi:rp-cons} implies that there is \(\conj{\tau{}}\in [\tau{}/2,\tau{}]\) such that
\begin{equation}\label{rp:E-neg-phi-bound}
E_{-\eta_0}[\varphi{}](\conj{\tau{}})\le \frac{2}{\tau{}}\int_{\tau{}/2}^\tau{} E_{-\eta_0}[\varphi{}](\tau{}')\dd{}\tau{}'\le \tau^{-2}C(A_{0,s})\mathcal{D}_0^2.
\end{equation}
Substitute \cref{rp:E-neg-phi-bound,rp:phi:rp-cons} into \cref{rp:E-p-phi-bound} to get
\begin{equation}\label{rp:phi-estimate-2}
E_p[\varphi{}](\conj{\tau{}})\le \tau{}^{p-2+\eta{}_0}C(A_{0,s})\mathcal{D}_0^2.
\end{equation}
Use \cref{rp:phi-estimate-2} and \cref{rp:phi-estimate} over the interval
\([\conj{\tau{}},\tau{}]\) to obtain
\begin{equation}
E_p[\varphi{}](\tau{})\le C(A_{0,s},p)E_p[\varphi{}](\conj{\tau{}})\le \tau{}^{p-2+\eta{}_0}C(A_{0,s},p)\mathcal{D}_0^2.
\end{equation}

Now we show the inductive step. Suppose \cref{rp:decay-induction-hypothesis} holds
for multi-indices \(<\alpha{}\). Then \cref{rp:controlled-error} and the induction
hypothesis give
\begin{equation}\label{rp:controlled-error-cons}
E_p[L\varphi{}](\tau_2) + \int_{\tau_1}^{\tau_2} E_{p-1}[L\varphi{}](\tau{})\dd{}\tau{} \le C(A_{\alpha{},s},p)[E_p[\Gamma^{\le \alpha{}}\varphi{}](\tau_1) + \tau_1^{p + 4B_\alpha{}s-q_{<\alpha{}}}A_{\alpha{},s}\mathcal{D}_\alpha^2]
\end{equation}
for \(1\le \tau{}_1\le \tau{}_2\). Sum over \(L\in \Gamma^\alpha{}\) and use \cref{rp-pigeonhole-argument} to get
\begin{equation}\label{rp:pigeonhole-consequence}
\begin{split}
E_{1-\eta{}_0}[\Gamma^{\le \alpha{}}\varphi{}](\tau{}) + \int_{\tau{}/2}^\tau{} E_{-\eta{}_0}[\Gamma^{\le \alpha{}}\varphi{}](\tau{}')\dd{}\tau{}' & \lesssim \tau^{1-\eta{}_0-q_{<\alpha{}} + 4B_\alpha{}s}C(A_{\alpha{},s})\mathcal{D}_\alpha^2, \\
\end{split}
\end{equation}
We have used \cref{rp:initial-energy-bound} to control \(E_{2-\eta{}_0}[\Gamma{}^{\le \alpha{}}\varphi{}](1)\).
By a pigeonhole argument, there is \(\conj{\tau{}}\in [\tau{}/2,\tau{}]\) such that
\begin{equation}\label{rp:E-epsilon-bound}
E_{-\eta{}_0}[\Gamma^{\le \alpha{}}\varphi{}](\conj{\tau{}})\le \frac{2}{\tau{}} \int_{\tau{}/2}^\tau{} E_{-\eta{}_0}[\Gamma^{\le \alpha{}}\varphi{}](\tau{}')\dd{}\tau{}'\lesssim \tau^{-\eta{}_0-q_{<\alpha{}} + 4B_\alpha{}s}C(A_{\alpha{},s})\mathcal{D}_\alpha^2.
\end{equation}
Interpolate using \cref{rp:interpolation} and then use
\cref{rp:pigeonhole-consequence,rp:E-epsilon-bound} to get
\begin{equation}
\begin{split}
E_p[\Gamma^{\le \alpha{}}\varphi{}](\conj{\tau{}}) & \le \tau^{p+\eta_0}E_{-\eta_0}[\Gamma^{\le \alpha{}}\varphi{}](\conj{\tau{}}) + \tau^{p-1 + \eta_0}E_{1-\eta_0}[\Gamma^{\le \alpha{}}\varphi{}](\conj{\tau{}}) \le \tau^{p-q_{<\alpha{}} + 4B_\alpha{}s}C(A_{\alpha{},s},p)\mathcal{D}_\alpha^2.
\end{split}
\end{equation}
Apply \cref{rp:controlled-error-cons} over the interval \([\conj{\tau{}},\tau{}]\) to obtain
\begin{equation}
\begin{split}
E_p[\Gamma{}^{\le \alpha{}}\varphi{}](\tau{})&\le C(A_{\alpha{},s},p)[E_p[\Gamma^{\le \alpha{}}\varphi{}](\conj{\tau}) + \conj{\tau}^{p + 4B_\alpha{}s-q_{<\alpha{}}}A_{\alpha{},s}\mathcal{D}_\alpha^2] \le \tau{}^{p-(q_{<\alpha{}}-4B_\alpha{}s)}C(A_{\alpha{},s},p)\mathcal{D}_{\alpha{}}^2.
\end{split}
\end{equation}
This completes the inductive proof of \cref{rp:decay-induction-hypothesis}.

\step{Step 4: Completing the proof of \cref{E-alpha-p-bound-eq}.} Write \(p_{\text{max}}
\coloneqq{} 2 - \eta_0 - 4(B_\alpha{} + 1)s - 4\tilde{B}_\alpha{}s\), so that it is enough to show that
\begin{equation}\label{rp:conc-goal}
\mathcal{E}_{p_{\text{max}}}^2\le C(A_{\alpha{},s})\mathcal{D}_\alpha^2,
\end{equation}
 For \(\tau{}\ge 1\), let \(\mathcal{S}_{\tau{}}\) be the collection of
null curves whose past endpoint lies in \(\Sigma_{\tau}\). Then
\begin{equation}\label{rp:conc-1}
\mathcal{E}_{\alpha{},p}^2\le \sup_{\tau{}\ge 1} \sup_{\Sigma{}\in \mathcal{S}_\tau{}} \tau{}^{p}E[\Gamma{}^{\le \alpha{}}\varphi{},\Sigma{}] + \sup_{\tau{}\ge 1}\tau{}^{p-1-\eta{}_0}E_{1+\eta{}_0}[\Gamma{}^{\le \alpha{}}\varphi{},\Sigma{}_\tau{}] + E_{p,\mathrm{bulk}}[\Gamma{}^{\le \alpha{}}\varphi{},\mathcal{R}(1,\infty)].
\end{equation}
Fix a null curve \(\Sigma{}\in \mathcal{S}_\tau{}\), so that \(\Sigma{}\subset \mathcal{R}(\tau{},\infty)\). Apply
the energy boundedness statement of \cref{energy-boundedness-curves-in-future} to
obtain
\begin{equation}
E[\Gamma{}^{\le \alpha{}}\varphi{},\Sigma{}]\le C(A_{\alpha{},s})[E[\Gamma{}^{\le \alpha{}}\varphi{},\Sigma{}_\tau{}] + E_{4s,\mathrm{bulk}}[\Gamma{}^{<\alpha{}}\varphi{},\mathcal{R}(\tau{},\infty)]].
\end{equation}
The \(p_{\text{fut}}\) term vanishes because \(r\to \infty\) along \(\Sigma_\tau{}\) and the
pointwise norm \(\mathcal{P}_{\alpha{},0}\) appears on the right. By
\cref{rp-unweighted-energy-bound,rp:bulk-bound,rp:decay-induction-hypothesis}, we
have
\begin{equation}\label{rp:conc-2}
\begin{split}
E[L\varphi{},\Sigma{}]&\le C(A_{\alpha{},s})[E_{4s}[\Gamma{}^{\le \alpha{}}\varphi{}](\tau{}) + E_{4(B_\alpha{} + 1)s}[\Gamma{}^{<\alpha{}}\varphi{}](\tau{})]\le  C(A_{\alpha{},s})E_{4(B_\alpha{} + 1)s}[\Gamma{}^{\le \alpha{}}\varphi{}](\tau{}) \\
&\le \tau^{-2 + \eta{}_0 + 4(B_\alpha{} + 1)s + 4\tilde{B}_\alpha{}s}C(A_{\alpha{},s})\mathcal{D}_\alpha^2 = \tau^{-p_{\text{max}}}C(A_{\alpha{},s})\mathcal{D}_\alpha^2
\end{split}
\end{equation}
By \cref{rp:decay-induction-hypothesis}, we have
\begin{equation}\label{rp:conc-3}
E_{1+\eta{}_0}[\Gamma^{\le \alpha{}}\varphi{}](\tau{})\le \tau^{1+\eta{}_0-2+\eta{}_0+4\tilde{B}_\alpha{}s}C(A_{\alpha{},s})\mathcal{D}_\alpha^2\le \tau^{1+\eta{}_0-p_{\text{max}}}C(A_{\alpha{},s})\mathcal{D}_\alpha^2.
\end{equation}
Finally, \cref{rp:bulk-bound,rp:decay-induction-hypothesis} gives
\begin{equation}\label{rp:conc-4}
E_{p_{\text{max}},\mathrm{bulk}}[\Gamma{}^{\le \alpha{}}\varphi{},\mathcal{R}(\tau{},\infty)] \le C(A_{\alpha,s})E_{p_{\text{max}}+4B_\alpha{}s}[\Gamma{}^{\le \alpha{}}\varphi{}](\tau{})\le \tau^{-4(B_\alpha{}+1)s}C(A_{\alpha{},s})\mathcal{D}_\alpha^2\le C(A_{\alpha{},s})\mathcal{D}_\alpha^2.
\end{equation}
Combine \cref{rp:conc-1,rp:conc-2,rp:conc-3,rp:conc-4} to obtain the desired
\cref{rp:conc-goal}.
\end{proof}
\section{Pointwise estimates}
\label{sec:org1618878}
\label{sec:pointwise-estimates}
In \cref{sec:pointwise-general}, we estimate weighted \(L^\infty\)-norms similar to
\(P_p[\psi{}]\) in terms of \(\Box{}\psi{}\), the initial data norm \(D[\psi{}]\),
and the energy norms \(\mathcal{E}[\psi{}]\) and \(\mathcal{E}_p[\psi{}]\). In
\cref{sec:pointwise-P-alpha}, we specialize the estimates in
\cref{sec:pointwise-general} to functions of the form \(\Gamma^{\le
\alpha{}}\varphi{}\). This requires an understanding of \(\Box{}\Gamma^{\le
\alpha{}}\varphi{}\), as obtained in \cref{sec:commute-with-box}, and culminates in
a proof of the following result:
\begin{proposition}[Estimate for pointwise norm]
Let \(\alpha{}\ge 0\) and \(p\ge 0\). When \(s > 0\) is sufficiently small depending on
\(\alpha{}\), we have
\begin{equation}
\mathcal{P}_{\alpha{},p}\le C(\mathfrak{b}_\alpha{},r^{-s}\mathfrak{g}_\alpha{})(\mathcal{D}_\alpha{} + \mathcal{E}_{\alpha{},p}).
\end{equation}
\label{P-alpha-bound}
\end{proposition}
\begin{proof}
Combine \cref{r12-tau-phi-bound,UL-phi-bound,dv-L-phi-bound,r-tau-phi-bound}.
\end{proof}
\subsection{Estimates for general functions}
\label{sec:orga867d23}
\label{sec:pointwise-general}
\begin{lemma}
We have
\begin{equation}
\norm{r^{1/2}\psi{}}_{L^\infty}\le \mathfrak{D}_0[\psi{}] + \mathcal{E}[\psi{}].
\end{equation}
\label{r12-phi-bound}
\end{lemma}
\begin{proof}
We use the fundamental theorem of calculus along a constant-\(v\) curve from
\((u,v)\) to \((1,v)\in C^{\text{in}}\), and then apply Cauchy-Schwarz, the
assumption \(\min_{v\in [1,\infty)} r(1,v) \ge 1\), the change of variables
\(\dd{}r = (-\nu{})\dd{}u\), and the monotonicity \(r(u,v)\le r(1,v)\). This gives
\begin{equation}
\begin{split}
\abs{\psi{}}(u,v)&\le \abs{\psi{}}(1,v) + \int_1^u \abs{\partial{}_u\psi{}}(u',v)\dd{}u'\le r^{-1}(1,v)D[\psi{}] \\
&\qquad + \Bigl(\int_1^u r^{-2}(u',v)(-\nu{})\dd{}u'\Bigr)^{1/2}\Bigl(\int_1^u \frac{r^2}{(-\nu{})}(\partial{}_u\psi{})^2(u',v)\dd{}u'\Bigr)^{1/2}\\
&\le r^{-1/2}(1,v)\mathfrak{D}_0[\psi{}] + \Bigl(\int_{r(u,v)}^{r(1,v)}  r^{-2}\dd{}r\Bigr)^{1/2}\mathcal{E}[\psi{}] \le r^{-1/2}(u,v)(\mathfrak{D}_0[\psi{}] + \mathcal{E}[\psi{}]).
\end{split}
\end{equation}
\end{proof}
\begin{lemma}
Let \(s > 1\). Then
\begin{equation}
\norm{r\psi{}}_{L^\infty}\le C(\mathfrak{b}_0,s)(\mathfrak{D}_0[\psi{}] + \mathcal{E}[\psi{}]  + \norm{r^s\partial_v(r\psi{})}_{L^\infty}).
\end{equation}
\label{rphi-bound}
\end{lemma}
\begin{proof}
Fix \(R\ge R_0\). We integrate \(\partial_v(r\psi{})\) in the \(v\)-direction from
\((u,v_{R}(u))\in \set{r=R}\) to \((u,v)\). To control the boundary term at
\(\set{r=R}\), we use \cref{r12-phi-bound}, and we estimate the integral using the
hypothesis, a change of variables, and the fact that \(\mathbf{1}_{r\ge
R_0}\lambda{}^{-1}\le C(\mathfrak{b}_0)\). Writing \(A \coloneqq{}
\norm{r^s\partial_v(r\psi{})}_{L^\infty}\), we get
\begin{equation}
\begin{split}
r\abs{\psi{}}(u,v)&\le R\abs{\psi{}}(u,v_R(u)) + \int_{v_R(u)}^v \abs{\partial{}_v(r\psi{})}(u,v')\dd{}v'\le R^{1/2}(\mathfrak{D}_0[\psi{}] + \mathcal{E}[\psi{}]) + A\int_{v_R(u)}^v r^{-s}(u,v')\dd{}v' \\
&\le R^{1/2}(\mathfrak{D}_0[\psi{}] + \mathcal{E}[\psi{}]) + C(\mathfrak{b}_0)A\int_R^{r(u,v)}  r^{-s}\dd{}r \le C(\mathfrak{b}_0,s)(\mathfrak{D}_0[\psi{}] + \mathcal{E}[\psi{}]  + A).
\end{split}
\end{equation}
\end{proof}
\begin{lemma}
Let \(p \ge 0\). If \(\mathfrak{D}_0[\psi{}]< \infty\), then
\begin{equation}\label{psi-rtau-bound}
\norm{r^{1/2}\tau{}^{p/2}\psi{}}_{L^\infty}\le 3\mathcal{E}_p[\psi{}].
\end{equation}
\label{r12-tau-phi-bound}
\end{lemma}
\begin{proof}
We first consider the case where \(r(u,v)\ge R_0\), and then the case where
\(r(u,v)\le R_0\). We can assume that \(\mathcal{E}[\psi{}] < \infty\), as
otherwise \cref{psi-rtau-bound} is trivial since the right side is infinite.

\step{Step 1: \(r(u,v) \ge R_0\).} Suppose that \(r(u,v)\ge R_0\). Since \(r(u,v')\to \infty\) as
\(v'\to \infty\), \cref{r12-phi-bound} implies that \(\abs{\psi{}}(u,v')\to 0\) as
\(v'\to \infty\). We can therefore use the fundamental theorem of calculus in
\(v\) from \((u,v)\) to \((u,\infty)\), then use Cauchy-Schwarz and the lower
bound \(1-\mu{}\ge 1/2\) in \(\set{r\ge R_0}\) from \cref{mu-estimate} to get
\begin{equation}
\begin{split}
\abs{\psi{}}(u,v)&\le\int_v^\infty \abs{\partial{}_v\psi{}}(u,v')\dd{}v'\le  \Bigl(\int_v^\infty r^{-2}\kappa{}\dd{}v'\Bigr)^{1/2}\Bigl(\int_v^\infty \frac{r^2}{\kappa{}}(\partial{}_v\psi{})^2(u,v')\dd{}v'\Bigr)^{1/2}\\
&\le \Bigl(\int_{r(u,v)}^\infty r^{-2}(1-\mu{})^{-1}\dd{}r\Bigr)^{1/2}\tau{}^{-p/2}\mathcal{E}_p\le 2r^{-1}\tau{}^{-p/2}\mathcal{E}_p[\psi{}].
\end{split}
\end{equation}

\step{Step 2: \(r(u,v)\le R_0\).} If \(r(u,v)\le R_0\), then we integrate on a
constant-\(v\) curve from \((u_{R_0}(v),v)\) to \((u,v)\) and use Cauchy-Schwarz
and a change of variables as before. The boundary term on \(\set{r = R_0}\) is
controlled by Step 1, and the integral decays in \(\tau{}\). This gives
\begin{equation}
\begin{split}
\abs{\psi{}}(u,v)&\le \abs{\psi{}}(u_{R_0}(v),v) + \int_{u_{R_0}(v)}^u \abs{\partial{}_u\psi{}}(u',v)\dd{}u'\\
&\le  2R_0^{-1/2}\tau{}^{-p/2}(u_{R_0}(v),v)\mathcal{E}_p[\psi{}] + \Bigl(\int_{r(u,v)}^{R_0} r^{-2}\dd{}r\Bigr)^{1/2}\Bigl(\int_{u_{R_0}(v)}^u \frac{r^2}{(-\nu{})}(\partial{}_u\psi{})^2(u',v)\dd{}u'\Bigr)^{1/2} \\
&\le 3r^{-1/2}\tau{}^{-p/2}\mathcal{E}_p[\psi{}].
\end{split}
\end{equation}
In the last line we have used \(\tau{}(u,v) = \tau{}(u_{R_0}(v),v)\) for \(r(u,v)\le R_0\).
\end{proof}
\begin{lemma}
Let \(p \ge 1 + \eta_0\). If \(\mathfrak{D}_0[\psi{}]< \infty\), then
\begin{equation}
\norm{r\tau{}^{p/2-1/2-\eta{}_0/2}\psi{}}_{L^\infty}\le C(\varpi{}_i,\eta{}_0)\mathcal{E}_p[\psi{}].
\end{equation}
\label{r-tau-phi-bound}
\end{lemma}
\begin{proof}
In view of the previous lemma and \(\tau{}\ge 1\), it is enough to obtain the result in
\(\set{r\ge R_0}\). As in the previous lemma, we integrate
\(\partial_v(r\psi{})\) in the \(v\)-direction from \((u,v_{R_0}(u))\in
\set{r=R_0}\) to \((u,v)\). This time, we estimate the integral using the
\(p\)-weighted energy. We get
\begin{equation}
\begin{split}
r\abs{\psi{}}(u,v)&\le R_0\abs{\psi{}}(u,v_{R_0}(u)) + \int_{v_{R_0}(u)}^v \abs{\partial{}_v(r\psi{})}(u,v')\dd{}v'\\
&\le R_0^{1/2}\tau{}^{-p/2}C(\mathfrak{b}_0)\mathcal{E}_p[\psi{}] + \Bigl(\int_{v_{R_0}(u)}^v r^{-1-\eta{}_0}\dd{}r\Bigr)^{1/2}\Bigl(\int_{v_{R_0}(u)}^v \frac{r^{1+\eta{}_0}}{\lambda{}}(\partial{}_v(r\psi{}))^2\dd{}v'\Bigr)^{1/2} \\
&\le \tau{}^{-p/2}C(\varpi{}_i)\mathcal{E}_p[\psi{}] + C(\eta{}_0)\tau{}^{-p/2+1/2+\eta{}_0/2}\mathcal{E}_p[\psi{}].
\end{split}
\end{equation}
In passing to the last line, we used \(R_0\le C(\varpi{}_i)\).
\end{proof}
\begin{lemma}
Let \(s\in{}[0,1]\). We have
\begin{equation}
\abs{\partial{}_v(r\psi{})}\le r^{-3/2}C(\mathfrak{b}_0)(\mathfrak{D}_0[\psi{}] + \mathcal{E}[\psi{}]) + \int_1^u r\kappa{}(-\nu{})\abs{\Box{}\psi{}}(u',v)\dd{}u'.
\end{equation}
\label{dvrpsi-bound}
\end{lemma}
\begin{proof}
The wave equation \cref{box-uv-gauge} implies
\begin{equation}\label{dudv-rpsi}
\partial{}_u\partial{}_v(r\psi{}) = r\kappa{}(-\nu{})\Box{}\psi{} + (\partial{}_u\partial{}_vr) \psi{} = r\kappa{}(-\nu{})\Box{}\psi{} - \frac{2(\varpi{}-\mathbf{e}^2/r)}{r^2}\kappa{}(-\nu{}) \psi{}.
\end{equation}
Integrate \cref{dudv-rpsi} in \(u\) to get
\begin{equation}\label{dvrpsi-2}
\begin{split}
\abs{\partial{}_v(r\psi{})}(u,v) &\le  \underbrace{\abs{\partial{}_v(r\psi{})}(1,v)}_{\coloneqq{}\text{(I)}}  + \underbrace{\int_1^u\frac{2(\varpi{}-\mathbf{e}^2/r)}{r^2}\kappa{}(-\nu{}) \abs{\psi{}}(u',v)\dd{}u'}_{\coloneqq{}\text{(II)}} + \int_1^u r\kappa{}(-\nu{})\abs{\Box{}\psi{}}(u',v)\dd{}u'.
\end{split}
\end{equation}
For term \(\text{(I)}\), use the monotonicity \(r(1,v)\ge r(u,v)\) and the fact
 that \(r\) has a global lower bound to estimate \(r^{-2}(1,v)\)
 in terms of \(r^{-3/2}(u,v)\). We also use the fact that \(\lambda{}\le C(\mathfrak{b}_0)\) on \(C^{\mathrm{out}}\). For term \(\text{(II)}\), use
 \(\kappa{},\varpi{}=_{\mathrm{s}} \mathfrak{b}_0\), change variables \((-\nu{})\dd{}u = \dd{}r\), and use \cref{r12-phi-bound}.
 This gives
\begin{equation}\label{dvrpsi-3}
\begin{split}
\text{(I)} + \text{(II)}&\le r^{-2}(1,v)\mathfrak{D}_0[\psi{}] + C(\mathfrak{b}_0)\norm{r^{1/2}\psi{}}_{L^\infty}\int_1^u r^{-1-3/2}(-\nu{})\dd{}u'\\
&\le r^{-3/2}(u,v)C(\mathfrak{b}_0)(\mathfrak{D}_0[\psi{}] + \mathcal{E}[\psi{}]).
\end{split}
\end{equation}
Substitute \cref{dvrpsi-3} into \cref{dvrpsi-2} to get the result.
\end{proof}
\begin{lemma}
We have
\begin{equation}
\frac{1}{2}r^2\abs{U\psi{}}^2(u,v)+  \int_1^v c_{\mathcal{H}}(U\psi{})^2(u,v')\dd{}v'\le C(\mathfrak{b}_0)(\mathfrak{D}_0[\psi{}]^2 + \mathcal{E}[\psi{}]^2) + 2\int_1^v r^2\kappa{}U\psi{}\Box{}\psi{}(u,v')\dd{}v'.
\end{equation}
\label{Upsi-bound}
\end{lemma}
\begin{proof}
It is enough to show the pointwise inequality
\begin{equation}\label{Upsi-step-1}
\frac{1}{2}\partial{}_v(rU\psi{})^2 + c_{\mathcal{H}}(U\psi{})^2 \le 2r^2\kappa{}U\psi{}\Box{}\psi{} + C(\mathfrak{b}_0)\frac{r^2}{\kappa{}}(\partial{}_v\psi{})^2.
\end{equation}
Indeed, \cref{Upsi-bound} follows immediately from \cref{Upsi-step-1} after recalling
the definitions of \([\psi{}]\) and \(\mathcal{E}[\psi{}]\).

We now show \cref{Upsi-step-1}. Apply the redshift inequality \(\varpi{}-\mathbf{e}^2/r\ge
c_{\mathcal{H}} > 0\) (see \cref{mass-redshift}) to the wave equation to compute
\begin{equation}
\partial{}_v(rU\psi{}) = r\kappa{}\Box{}\psi{} + \partial{}_v\psi{} - \frac{2(\varpi{}-\mathbf{e}^2/r)}{r^2}\kappa{}rU\psi{} \le r\kappa{}\Box{}\psi{} + \partial{}_v\psi{} - 2c_{\mathcal{H}}\kappa{}r^{-1}U\psi{}.
\end{equation}
To conclude \cref{Upsi-step-1}, multiply both sides by \(rU\psi{}\), use Young's
inequality for the term \(r\partial{}_v\psi{}U\psi{}\), and apply \(\kappa{}^{-1},c_{\mathcal{H}}^{-1}\le C(\mathfrak{b}_0)\).
\end{proof}
\subsection{Controlling the pointwise norm \texorpdfstring{\(\mathcal{P}_{\alpha{},p}\)}{}}
\label{sec:org12df586}
\label{sec:pointwise-P-alpha}
\begin{lemma}
For \(\alpha{}\ge 0\), let \(B_\alpha{} \coloneqq{} \card{}\set{\beta{} : 0\le \beta{}<\alpha{}}\) be the number of
multi-indices smaller than \(\alpha{}\). For \(s > 0\) small enough depending on
\(\alpha{}\), we have
\begin{equation}
\norm{r\Gamma{}^{\le \alpha{}}\varphi{}}_{L^\infty} + \norm{r^2\partial{}_v\Gamma{}^{\le \alpha{}}\varphi{}}_{L^\infty} + \norm{r^{3/2-B_\alpha{}s}\partial{}_v(r\Gamma{}^{\le \alpha{}}\varphi{})}_{L^\infty}\le C(\mathcal{G}_{\alpha,s})(\mathcal{D}_\alpha{} + \mathcal{E}_\alpha{}).
\end{equation}
\label{dv-L-phi-bound}
\end{lemma}
\begin{proof}
It is enough to establish that
\begin{equation}\label{dv-L-phi-bound-equation}
\norm{r^{3/2-B_{\alpha{}}s}\partial{}_v(r\Gamma{}^{\le \alpha{}}\varphi{})}_{L^\infty}\le C(\mathcal{G}_{\alpha,s})(\mathcal{D}_\alpha{} + \mathcal{E}_\alpha{}).
\end{equation}
Indeed, the estimate for \(\norm{r\Gamma{}^{\le \alpha{}}\varphi{}}_{L^\infty}\) follows from \cref{rphi-bound} and
\cref{dv-L-phi-bound-equation}. Moreover, the estimate for
\(\norm{r^2\partial{}_v\Gamma{}^{\le \alpha{}}\varphi{}}\) follows from
\begin{equation}
r^2\partial_v\Gamma^{\le \alpha{}}\varphi{}\le r\partial{}_v(r\Gamma{}^{\le \alpha{}}\varphi{}) - \lambda{}r\Gamma{}^{\le \alpha{}}\varphi{}
\end{equation}
and \(\lambda{}\le C(\mathfrak{b}_0)\). Note that here we use the smallness of \(s\) with
respect to \(\alpha{}\) to conclude that \(3/2 - B_\alpha{}s > 1\).

We prove \cref{dv-L-phi-bound-equation} by induction on \(\alpha{}\). When \(\alpha{} = 0\),
\cref{dv-L-phi-bound-equation} is an immediate consequence of \cref{dvrpsi-bound}
applied to \(\varphi{}\). Now suppose that \(\abs{\alpha{}}\ge 1\) and
\cref{dv-L-phi-bound-equation} holds for multi-indices \(<\alpha{}\) with the
exponent \(B_\alpha{} = B_{<\alpha{}} + 1\). We want to show that
\cref{dv-L-phi-bound-equation} holds for the multi-index \(\alpha{}\).

Fix a point \((u,v)\) at which we want to estimate \(\abs{\partial_v(rL\varphi{})}\). By
\cref{dvrpsi-bound,r12-phi-bound}, we have
\begin{equation}\label{dvLphi-prep}
\abs{\partial{}_v(rL\varphi{})}\le r^{-3/2}C(\mathfrak{b}_0)(\mathcal{D}_\alpha{} + \mathcal{E}_\alpha{}) + \int_1^u r\kappa{}(-\nu{})\abs{\Box{}\psi{}}(u',v)\dd{}u'.
\end{equation}
The estimate for \(\abs{\Box{}L\varphi{}}\) in \cref{wave-pointwise-inequalities} and the
inequality \(\kappa{}\le C(\mathfrak{b}_0)\) gives
\begin{equation}\label{dvLphi-step-0}
\begin{split}
 \int_1^u r\kappa{}(-\nu{})\abs{\Box{}\psi{}}(u',v)\dd{}u'&\le C(\mathcal{G}_{\alpha,s})\underbrace{\int _1^u r^{-1+s}\abs{U\Gamma{}^{\le \alpha{}}\varphi{}}(-\nu{})(u',v)\dd{}u'}_{\text{:=(I)}} \\
&\qquad + C(\mathcal{G}_{\alpha,s})\underbrace{\int_1^u r^{-1+s}\bigl[\abs{\partial{}_v(r\Gamma{}^{<\alpha{}}\varphi{})} + \abs{\partial{}_v\Gamma{}^{<\alpha{}}\varphi{}} + \abs{\partial{}_vL\varphi{}}\bigr](-\nu{})(u',v)\dd{}u'}_{\text{:=(II)}} \\
\end{split}
\end{equation}
Since \(r(u,v)\le r(u',v)\) for all \(u'\in [1,u]\), we have
\begin{equation}\label{dvLphi-step-1}
\text{(I)}\le \Bigl(\int_{r(u,v)}^{r(1,v)} r^{-4+2s}\dd{}r\Bigr)^{1/2}\Bigl(\int_1^u r^2(U\Gamma{}^{\le \alpha{}}\varphi{})^2(-\nu{})(u',v)\dd{}u'\Bigr)^{1/2}\le r^{-3/2+s}(u,v)\mathcal{E}_\alpha{}.
\end{equation}
We now turn to term \(\text{(II)}\). Use \(\abs{\partial_v\psi{}}\le r^{-1}\abs{\partial_v(r\psi{})} +
r^{-1}\lambda{}\abs{\psi{}}\), the induction hypothesis, and \cref{r12-phi-bound} to get
\begin{equation}\label{dvLphi-step-2}
\begin{split}
\text{(II)}&\le \int_1^u [r^{-1+s}\abs{\partial{}_v(r\Gamma{}^{<\alpha{}}\varphi{})} + r^{-2+s}\abs{\Gamma{}^{\le \alpha{}}\varphi{}} + r^{-2+s}\abs{\partial{}_v(rL\varphi{})}](-\nu{})(u',v)\dd{}u' \\
&\le \int_1^u [r^{-1-3/2+(B_{<\alpha{}} + 1)s}C(\mathcal{G}_{\alpha,s})(\mathfrak{D}_\alpha{} + \mathfrak{E}_\alpha{}) + r^{-2+s}\abs{\partial{}_v(rL\varphi{})}](-\nu{})(u',v)\dd{}u'   \\
&\le C(\mathcal{G}_{\alpha,s})(\mathfrak{D}_\alpha{} + \mathcal{E}_\alpha{})r^{-3/2+(B_{<\alpha{}} + 1)s}(u,v) + \int_1^u  r^{-2+s}\abs{\partial{}_v(rL\varphi{})}(-\nu{})(u',v)\dd{}u'.
\end{split}
\end{equation}
By \cref{dvLphi-prep,dvLphi-step-0,dvLphi-step-1,dvLphi-step-2}, we obtain
\begin{equation}
\begin{split}
\abs{\partial{}_v(rL\varphi{})}(u,v)&\le C(\mathcal{G}_{\alpha,s})(\mathfrak{D}_\alpha{} + \mathcal{E}_\alpha{})r^{-3/2+(B_{<\alpha{}} + 1)s}(u,v)  \\
&\qquad + C(\mathcal{G}_{\alpha,s})\int_1^u r^{-2+2\eta{}_0}\abs{\partial{}_v(rL\varphi{})}(-\nu{})(u',v)\dd{}u'.
\end{split}
\end{equation}
Since \(r^{-3/2+(B_{<\alpha{}}+1)}(\cdot ,v)\) is non-decreasing (in \(u\)), Grönwall's
inequality implies \cref{dv-L-phi-bound-equation} with \(B_\alpha{} = B_{<\alpha{}} + 1\), as
desired.
\end{proof}
\begin{lemma}
Let \(L\in \Gamma^\alpha{}\) for some \(\alpha{}\ge 0\). For \(s > 0\) small enough depending on \(\alpha{}\),
we have
\begin{equation}
\norm{rUL\varphi{}}_{L^\infty}\le C(\mathcal{G}_{\alpha,s})(\mathcal{D}_\alpha{} + \mathcal{E}_\alpha{}).
\end{equation}
\label{UL-phi-bound}
\end{lemma}
\begin{proof}
\step{Step 1: Preliminary observation.} Observe that \cref{dv-L-phi-bound} implies that
for \(s\) small enough depending on \(\alpha{}\), we have
\begin{equation}\label{UL-phi-bound-prep}
r^{2s}\abs{\partial_v(r\Gamma^{<\alpha{}}\varphi{})}^2\le r^{-2}C(\mathfrak{b}_{<\alpha{}},r^{-s}\mathfrak{g}_{<\alpha{}})(\mathcal{D}_{<\alpha{}}^2+\mathcal{E}_{<\alpha{}}^2).
\end{equation}

\step{Step 2.} We prove \cref{UL-phi-bound} by induction on \(\alpha{}\). When \(\alpha{} = 0\),
\cref{Upsi-bound} follows from \cref{UL-phi-bound} (for \(\psi{} = \varphi{}\)). Now
suppose inductively that \(\abs{\alpha{}}\ge 1\) and \cref{UL-phi-bound} holds for
all multi-indices \(< \alpha{}\). Use \cref{Upsi-bound} (for \(\psi{} =
L\varphi{}\)), the estimate for \(UL\varphi{}\Box{}L\varphi{}\) in
\cref{wave-pointwise-2}, an \(r\)-weighted Young's inequality with \(\epsilon{}\),
the fact that \(\kappa{},\kappa^{-1},c_{\mathcal{H}}^{-1},r^{-1}\le
C(\mathfrak{b}_0)\), the induction hypothesis, and \cref{UL-phi-bound-prep} to
obtain
\begin{equation}
\begin{split}
&r^2\abs{UL\varphi{}}^2(u,v) + \int_1^v 2c_{\mathcal{H}}(UL\varphi{})^2(u,v')\dd{}v'\\
&\le C(\mathfrak{b}_0)(\mathcal{D}_\alpha^2 + \mathcal{E}_\alpha^2) + \epsilon{}\int_1^v (UL\varphi{})^2(u,v')\dd{}v' \\
&\qquad + C(\mathcal{G}_{\alpha,s},\epsilon{})\int_1^v \Bigl[\frac{r^2}{\kappa{}}\abs{\partial{}_v\Gamma{}^{\le \alpha{}}\varphi{}}^2 + \mathbf{1}_{r\ge \Rc}r^{2s}\abs{\partial{}_v(r\Gamma{}^{<\alpha{}}\varphi{})}^2 + r^{2s}\abs{U\Gamma{}^{<\alpha{}}\varphi{}}^2\Bigr]\dd{}v' \\
&\le C(\mathfrak{b}_0)(\mathcal{D}_\alpha^2 + \mathcal{E}_\alpha^2) + \epsilon{}\int_1^v (UL\varphi{})^2(u,v')\dd{}v' + C(\mathcal{G}_{\alpha,s},\epsilon{})\int_1^v \Bigl[\frac{r^2}{\kappa{}}\abs{\partial{}_v\Gamma{}^{\le \alpha{}}\varphi{}}^2 + r^{-2+2s}(\mathcal{D}_\alpha^2 + \mathcal{E}_\alpha^2)\Bigr]\dd{}v' \\
&\le \epsilon{}\int_1^v (UL\varphi{})^2(u,v')\dd{}v' +  C(\mathcal{G}_{\alpha,s},\epsilon{})(\mathcal{D}_\alpha^2+ \mathcal{E}_\alpha^2). \\
\end{split}
\end{equation}
To conclude, absorb the first term on the right to the left by taking \(\epsilon{} >0\) small enough
depending on \(c_{\mathcal{H}}\).
\end{proof}
\section{Estimates for differentiated geometric quantities}
\label{sec:org53790d4}
\label{sec:geometric-estimates}
The goal of this section is to establish the following result:
\begin{proposition}[Estimates for higher order geometric quantities]
For \(\abs{\alpha{}}\ge 1\) and \(s > 0\) sufficiently small, we have
\begin{align}
C(\mathfrak{b}_\alpha{})&\le C(\alpha{},\mathfrak{B}_0,r^{-s}\mathfrak{G}_0,\mathcal{E}_{<\alpha{},1},\mathcal{P}_{<\alpha{},2-s+\eta{}_0}), \label{higher-geometric-bound-b} \\
C(\mathfrak{B}_\alpha{})&\le C(\alpha{},\mathfrak{B}_0,r^{-s}\mathfrak{G}_0,\mathcal{E}_{\alpha{},1},\mathcal{P}_{\alpha{},2-s+\eta{}_0}),\label{higher-geometric-bound-B} \\
C(r^{-s}\mathfrak{g}_\alpha{})&\le C(\alpha{},\mathfrak{B}_0,r^{-s}\mathfrak{G}_0,\mathcal{E}_{<\alpha{},1},\mathcal{P}_{<\alpha{},2-s+\eta{}_0}),\label{higher-geometric-bound-g} \\
C(r^{-s}\mathfrak{G}_\alpha{})&\le C(\alpha{},\mathfrak{B}_0,r^{-s}\mathfrak{G}_0,\mathcal{E}_{\alpha{},1},\mathcal{P}_{\alpha{},2-s+\eta{}_0}).\label{higher-geometric-bound-G}
\end{align}
\label{geometric-quantities-bound}
\end{proposition}
\begin{proof}
\Cref{preliminary-geometric-estimate} and
\cref{Gamma-omega-bounds,Gamma-kappa-bounds,Gamma-gamma-bounds} imply that for
\(\abs{\alpha{}}\ge 1\) we have
\begin{align}
C(\mathfrak{b}_\alpha{})&\le C(\alpha{},\mathfrak{b}_0,\mathfrak{B}_{<\alpha{}},r^{-s}\mathfrak{G}_{<\alpha{}},\mathcal{E}_{<\alpha{},1},\mathcal{P}_{<\alpha{},1+\eta{}_0}), \label{higher-geometric-b-prep}\\
C(\mathfrak{B}_\alpha{})&\le C(\alpha{},\mathfrak{b}_\alpha{},\mathfrak{B}_{<\alpha{}},r^{-s}\mathfrak{G}_{<\alpha{}},\mathcal{E}_{\alpha{},1},\mathcal{P}_{\alpha{},1}), \label{higher-geometric-B-prep}\\
C(r^{-s}\mathfrak{g}_\alpha{})&\le C(\alpha{},\mathfrak{b}_{\le \alpha{}},r^{-s}\mathfrak{G}_{<\alpha{}},\mathcal{E}_{<\alpha{},1},\mathcal{P}_{<\alpha{},2-s+\eta{}_0}), \label{higher-geometric-g-prep}\\
C(r^{-s}\mathfrak{G}_\alpha{})&\le C(\alpha{},\mathfrak{B}_{\le \alpha{}},r^{-s}\mathfrak{g}_{\le \alpha{}},r^{-s}\mathfrak{G}_{<\alpha{}},\mathcal{E}_{\alpha{},1},\mathcal{P}_{\alpha{},2-s+\eta{}_0}). \label{higher-geometric-G-prep}
\end{align}
Substitute \cref{higher-geometric-b-prep} into \cref{higher-geometric-B-prep} and
substitute \cref{higher-geometric-B-prep,higher-geometric-g-prep} into
\cref{higher-geometric-G-prep} and induct on \(\alpha{}\) to obtain
\cref{higher-geometric-bound-b,higher-geometric-bound-B,higher-geometric-bound-g,higher-geometric-bound-G}.
\end{proof}
\subsection{Preliminary estimate for schematic geometric quantities}
\label{sec:org275b2c8}
To estimate the geometric quantities \(\mathfrak{b}_\alpha{}\), \(\mathfrak{B}_\alpha{}\),
\(\mathfrak{g}_\alpha{}\), and \(\mathfrak{G}_\alpha{}\), it is enough to estimate the zeroth order geometric
quantities (which has been done in \cref{sec:zeroth-order-geometric}) and estimate
derivatives of \(\varpi{}\), \(\log \kappa{}\), and \(\log (-\gamma{})\).
\begin{lemma}
For \(\abs{\alpha{}}\ge 1\), we have
\begin{equation}
\begin{split}
C(\mathfrak{b}_\alpha{})&\le C(\alpha{},\mathfrak{b}_0,\mathfrak{B}_{<\alpha{}},r^{-1}\mathfrak{G}_{<\alpha{}},r^{-1}\Gamma{}^{\le \alpha{}}\varpi{},\Gamma{}^{\le \alpha{}}\log \kappa{},\mathbf{1}_{\beta{}_U>0}r\Gamma{}^\beta{}\log \kappa{}|_{\beta{}\le \alpha{}},\Gamma{}^{\le \alpha{}}\log (-\gamma{})), \\
C(\mathfrak{B}_\alpha{})&\le C(\alpha{},\mathfrak{B}_0,\mathfrak{b}_\alpha{},r^{-1}\mathfrak{G}_{<\alpha{}},v\Gamma{}^{\le \alpha{}}\log \kappa{},r\Gamma{}^{\le \alpha{}+U}\log \kappa{},r\tau{}\Gamma{}^{\le \alpha{}}\log (-\gamma{})), \\
C(r^{-s}\mathfrak{g}_\alpha{})&\le C(\alpha{},\mathfrak{b}_{\le \alpha{}},r^{-s}\mathfrak{G}_{<\alpha{}},r^{-s}\Gamma{}^{\le \alpha{}}\varpi{},\mathbf{1}_{r\ge R_0}r^{1-s}\Gamma{}^{\le \alpha{}}\log (-\gamma{})), \\
C(r^{-s}\mathfrak{G}_\alpha{})&\le C(\alpha{},\mathfrak{B}_{\le \alpha{}},r^{-s}\mathfrak{g}_{\le \alpha{}},r^{-s}\Gamma{}^{\le \alpha{}}\varpi{},r^{1-s}\tau{}\mathbf{1}_{\beta{}_U>0}\Gamma{}^\beta{}\log \kappa{}|_{\beta{}\le \alpha{}}). \\
\end{split}
\end{equation}
\label{preliminary-geometric-estimate}
\end{lemma}
The proof is provided at the end of the section.
\begin{lemma}
Let \(\abs{\alpha{}}\ge 1\), and let \(L\in \Gamma^\alpha{}\). We have
\begin{align}
Lr &=_{\mathrm{s}} \mathcal{O}(1)r, \label{Lr-est}\\
Lv - \mathbf{1}_{\abs{\alpha{}} = \alpha{}_S}v &=_{\mathrm{s}} (\Gamma{}^{\le \alpha{}-1}\mathfrak{B}_{0})[1 + \mathbf{1}_{r\ge \Rc}\mathbf{1}_{\alpha{}_S > 0}(-v + u + r)], \label{Lv-est}\\
\mathbf{1}_{r\ge R_0}Lu &=_{\mathrm{s}} (\Gamma{}^{\le \alpha{}-1}\mathfrak{B}_{0})[1 + \mathbf{1}_{\alpha{}_S>0}\set{\mathbf{1}_{r\ge R_0}u,\mathbf{1}_{r\le 2\Rc}v}], \label{Lu-est}\\
\mathbf{1}_{r\ge R_0}L(-v + u + r) &=_{\mathrm{s}} \Gamma{}^{\le \alpha{}-1}\mathfrak{B}_{0} \label{LG-est}.
\end{align}
\label{gamma-uv}
\end{lemma}
\begin{proof}
First, \cref{Lr-est} follows from the calculation \cref{Gamma-r-initial-calculation}.
For the remaining estimates, we use an inductive argument.

\step{Step 1: The case \(\abs{\alpha{}} = 1\).} We first establish \cref{Lv-est,Lu-est,LG-est} for
the case \(\abs{\alpha{}} = 1\).

\step{Step 1a: Proof of \cref{Lv-est} when \(\abs{\alpha{}} = 1\).} We need to control \(Uv\), \(Vv\),
and \(Sv - v\). A computation reveals
\begin{equation}\label{Lv-est-step-1}
Uv = 0 \qquad Vv = \chi{} + (1-\chi{})\lambda{}^{-1}\qquad Sv - v =  (1-\chi{})((-v + u + r) - r(1-\lambda{}^{-1}) - \frac{u}{\kappa{}}\bigl(\kappa{} - (-\gamma{})\bigr)).
\end{equation}
It follows that
\begin{equation}\label{Lv-est-step-1-conc}
Dv =_{\mathrm{s}} \mathfrak{b}_0 \qquad Sv - v =_{\mathrm{s}} \mathfrak{B}_0[1 + \mathbf{1}_{r\ge R_0}(-v + u + r)].
\end{equation}

\step{Step 1b: Proof of \cref{Lu-est} when \(\abs{\alpha{}} = 1\).} We need to control \(Uu\), \(Vu\),
and \(Su\). We compute
\begin{equation}\label{Lu-est-step-1}
Uu = \frac{1}{(-\nu{})}\qquad Vu = \chi{} \frac{\kappa{}}{(-\gamma{})}\qquad Su = \chi{}v \frac{\kappa{}}{(-\gamma{})} + (1-\chi{})u
\end{equation}
It follows that
\begin{equation}\label{Lu-est-step-1-conc}
\mathbf{1}_{r\ge R_0}Du =_{\mathrm{s}} \mathfrak{b}_0\qquad \mathbf{1}_{r\ge R_0}Su =_{\mathrm{s}} \mathfrak{b}_0\set{\mathbf{1}_{r\ge R_0}u,\mathbf{1}_{r\le 2\Rc}v}.
\end{equation}

\step{Step 1c: Proof of \cref{LG-est} when \(\abs{\alpha{}} = 1\).} By
\cref{Lv-est-step-1-conc,Lu-est-step-1-conc} and \(Ur,Vr =_{\mathrm{s}} \mathcal{O}(1)\) we
have
\begin{equation}\label{LG-est-step-1a}
\mathbf{1}_{r\ge R_0}D(-v + u - r) =_{\mathrm{s}} \mathfrak{b}_0.
\end{equation}
From \cref{Lv-est-step-1-conc} and the discussion that follows in Step 1a, we
can express
\begin{equation}
Sv = v + (1-\chi{})(-v + u + r) + \mathcal{E} = \chi{}v + (1-\chi{})r + (1-\chi{})u + \mathcal{E}
\end{equation}
for \(\mathcal{E} =_{\mathrm{s}} \mathfrak{B}_0\). Together with \cref{Lu-est-step-1-conc}
and \(Sr = (1-\chi{})r\), we compute
\begin{equation}
\begin{split}
S(-v + u + r) &= \frac{\chi{}v}{(-\gamma{})} (\kappa{} - (-\gamma{}))  - \mathcal{E}.
\end{split}
\end{equation}
It follows as in Step 1a that \(v(\kappa{} - (-\gamma{}))=_{\mathrm{s}} \mathfrak{B}_0\). Since
\(\mathbf{1}_{r\ge R_0}(-\gamma{})^{-1}=_{\mathrm{s}} \mathfrak{b}_0\), we get
\begin{equation}\label{LG-est-step-1b}
\mathbf{1}_{r\ge R_0}S(-v + r + r) =_{\mathrm{s}} \mathfrak{B}_0.
\end{equation}
We conclude from \cref{LG-est-step-1a,LG-est-step-1b} that
\cref{LG-est} holds when \(\abs{\alpha{}} = 1\).

\step{Step 2: Inductive step.} Let \(\abs{\alpha_1},\abs{\alpha_2}\ge 1\), and let \(L_i\in
\Gamma^{\alpha_i}\). Write \(\alpha{} = \alpha_1 + \alpha_2\). Suppose that
\cref{Lv-est,Lu-est,LG-est} hold for multi-indices \(\beta{}\) with \(1\le
\abs{\beta{}} < \abs{\alpha{}}\). We will show that they hold for \(L_1L_2\in
\Gamma^{\alpha{}}\). Recall that by definition \(\Gamma^{\le
\alpha{}}\mathfrak{B}_0 =_{\mathrm{s}} \mathfrak{B}_\alpha{}\). Differentiating \cref{Lv-est}
for the multi-index \(\alpha_2\) by \(L_1\in \Gamma^{\alpha_1}\) and using
\cref{Lv-est} for multi-indices \(\le \abs{\alpha_1}\) gives
\begin{align}
L_1L_2v - \mathbf{1}_{(\alpha{}_2)_S = \abs{\alpha{}_2}}L_1v &=_{\mathrm{s}} (\Gamma{}^{\le \alpha{}_1+\alpha{}_2-1}\mathfrak{B}_0)[1 + \mathbf{1}_{r\ge R_0}\mathbf{1}_{\alpha{}_S > 0}(-v + u + r)].
\end{align}
Use \cref{Lv-est} for \(L_1\in \Gamma^{\alpha_1}\) to handle the second term on the left
and use \(\mathbf{1}_{(\alpha{}_2)_S =
\abs{\alpha{}_2}}\mathbf{1}_{(\alpha{}_1)_S=\abs{\alpha{}_1}}v =
\mathbf{1}_{\alpha{}_S = \abs{\alpha{}}}\) to conclude \cref{Lv-est} for the multi-index \(\alpha{}\). A consequence of \cref{Lv-est} for the multi-index \(\alpha{}\) and the fact that
\(r^{-1} =_{\mathrm{s}} \mathcal{O}(1)\) is that
\begin{equation}\label{Lv-est-consequence}
\begin{split}
Lv &=_{\mathrm{s}} (\Gamma{}^{\le \alpha{}-1}\mathfrak{B}_0)\set{r,\mathbf{1}_{r\ge R_0}u,v}.
\end{split}
\end{equation}
Differentiate \cref{Lv-est} for the multi-index \(\alpha_2\) by \(L_1\in
\Gamma^{\alpha_1}\) and use \cref{Lv-est-consequence,Lu-est} for
multi-indices \(\le \alpha{}_1\) to conclude \cref{Lu-est} for the multi-index \(\alpha{}\).
Finally, it is clear how to prove \cref{LG-est} inductively.
\end{proof}
\begin{proof}[Proof of \cref{preliminary-geometric-estimate}]
Define
\begin{align}
\mathfrak{\tilde{b}}_0 \coloneqq{} \mathfrak{b}_0, \quad \mathfrak{\tilde{B}}_0 \coloneqq{}\mathfrak{B}_0,\quad \mathfrak{\tilde{g}}_0 \coloneqq{}\mathfrak{g}_0,\quad \mathfrak{\tilde{G}}_0 \coloneqq{} \mathfrak{G}_0,
\end{align}
and for \(\abs{\alpha{}}\ge 1\), define the
following schematic quantities that capture the top order terms in
\(\mathfrak{b}_\alpha{}\), \(\mathfrak{B}_\alpha{}\), \(\mathfrak{g}_\alpha{}\), and \(\mathfrak{G}_\alpha{}\):
\begin{align}
\mathfrak{\tilde{b}}_\alpha{} & \coloneqq{} \mathcal{O}_\alpha{}(r^{-1}\Gamma^\alpha{}\varpi{},\Gamma^\alpha{}\log \kappa{},\mathbf{1}_{\alpha_U>0}r\Gamma^\alpha{}\log \kappa{},\Gamma^\alpha{}\log (-\gamma{})), \label{b-tilde-alpha-def}                                            \\
\mathfrak{\tilde{B}}_\alpha{} & \coloneqq{}\mathcal{O}_\alpha{}(\set{r,\mathbf{1}_{r\ge R_0}u,v}\Gamma^\alpha{}\log\kappa{},r\Gamma^{\alpha{} + U}\log \kappa{},\set{r,\mathbf{1}_{r\ge R_0}ru,v}\mathbf{1}_{r\ge R_0}\Gamma^\alpha{}\log (-\gamma{})), \label{B-tilde-alpha-def} \\
\mathfrak{\tilde{g}}_\alpha{} & \coloneqq{}\mathfrak{b}_\alpha{}\set{1,\Gamma^\alpha{}\varpi{},r\Gamma^\alpha{}\log (-\gamma{})}, \label{g-tilde-alpha-def} \\
\mathfrak{\tilde{G}}_\alpha{} & \coloneqq{} \mathfrak{B}_\alpha{}\set{1, \set{r,\mathbf{1}_{r\ge R_0}ru,\mathbf{1}_{r\le 2\Rc}v}\mathbf{1}_{\alpha{}_U>0}\Gamma{}^\alpha{}\log \kappa{}}. \label{G-tilde-alpha-def}
\end{align}

\step{Step 1: Relating the original schematic geometric quantities to the ones with
tildes.} Observe that for \(\abs{\alpha{}}\ge 1\)
\begin{equation}\label{b-alpha-diff-prep-1}
\begin{split}
\mathfrak{b}_\alpha{} &= \mathcal{O}_\alpha{}(\Gamma{}^\beta{}\mathfrak{\tilde{b}}_{\beta{}'}|_{\beta{}+\beta{}'\le \alpha{}},\Gamma{}^\beta{}\mathfrak{\tilde{B}}_{\beta{}'}|_{\beta{}+\beta{}'<\alpha{}},r^{-1}\Gamma{}^\beta{}\mathfrak{\tilde{g}}_{\beta{}'}|_{\beta{}+\beta{}'<\alpha{}},r^{-1}\Gamma{}^\beta{}\mathfrak{G}_{\beta{}'}|_{\beta{}+\beta{}'<\alpha{}}), \\
\mathfrak{B}_\alpha{} &= \mathcal{O}_\alpha{}(\Gamma{}^\beta{}\mathfrak{\tilde{b}}_{\beta{}'}|_{\beta{}+\beta{}'\le \alpha{}},\Gamma{}^\beta{}\mathfrak{\tilde{B}}_{\beta{}'}|_{\beta{}+\beta{}'\le \alpha{}},r^{-1}\Gamma{}^\beta{}\mathfrak{\tilde{g}}_{\beta{}'}|_{\beta{}+\beta{}'<\alpha{}},r^{-1}\Gamma{}^\beta{}\mathfrak{G}_{\beta{}'}|_{\beta{}+\beta{}'<\alpha{}}), \\
\mathfrak{g}_\alpha{} &= \mathfrak{b}_\alpha{}\set{1,\Gamma{}^\beta{}\tilde{\mathfrak{g}}_{\beta{}'}|_{\beta{}+\beta{}'\le \alpha{}},\Gamma{}^\beta{}\mathfrak{\tilde{G}}_{\beta{}+\beta{}'<\alpha{}}},\\
\mathfrak{G}_\alpha{} &= \mathfrak{B}_\alpha{}\set{1,\Gamma{}^\beta{}\tilde{\mathfrak{g}}_{\beta{}'}|_{\beta{}+\beta{}'\le \alpha{}},\Gamma{}^\beta{}\mathfrak{\tilde{G}}_{\beta{}+\beta{}'\le \alpha{}}}.
\end{split}
\end{equation}
This can be seen from the definitions of the original and tilded quantities and the
consequence of the chain rule that \(\Gamma{}^\alpha{}f =_{\mathrm{s}}
\mathcal{O}_\alpha{}(\Gamma{}^{\le \alpha{}}\log f,f,f^{-1})\) for any positive
function \(f\).

\step{Step 2.} We claim that
\begin{align}
\Gamma{}^\beta{}\mathfrak{\tilde{b}}_{\beta{}'}|_{\beta{}+\beta{}'\le \alpha{}} &=_{\mathrm{s}} \mathcal{O}(\mathfrak{\tilde{b}}_{\le \alpha{}}), \label{b-tilde-diff} \\
\Gamma{}^\beta{}\mathfrak{\tilde{B}}_{\beta{}'}|_{\beta{}+\beta{}'\le \alpha{}} &=_{\mathrm{s}} \mathcal{O}(\mathfrak{\tilde{b}}_{\le \alpha{}},\mathfrak{\tilde{B}}_{\le \alpha{}}),  \label{B-tilde-diff} \\
\Gamma{}^\beta{}\mathfrak{\tilde{g}}_{\beta{}'}|_{\beta{}+\beta{}'\le \alpha{}} &=_{\mathrm{s}} \mathcal{O}(\mathfrak{\tilde{b}}_{\le \alpha{}})\set{1,\mathfrak{\tilde{g}}_{\le \alpha{}}},  \label{g-tilde-diff} \\
\Gamma{}^\beta{}\mathfrak{\tilde{G}}_{\beta{}'}|_{\beta{}+\beta{}'\le \alpha{}} &=_{\mathrm{s}} \mathcal{O}(\mathfrak{\tilde{B}}_{\le \alpha{}})\set{1,\mathfrak{\tilde{g}}_{\le \alpha{}},\mathfrak{\tilde{G}}_{\le \alpha{}}}.  \label{G-tilde-diff}
\end{align}

\step{Step 2a: Proof of \cref{b-tilde-diff,g-tilde-diff}.} This follows from \cref{Lr-est} and
differentiating \(\mathfrak{\tilde{b}}_0 =\mathfrak{b}_0\) (see \cref{b0-def}) in the case
\(\beta{}' = 0\) and \(\mathfrak{\tilde{b}}_{\beta{}'}\)
(see \cref{b-tilde-alpha-def}) in the case \(\abs{\beta{}'} > 0\). The proof of
\cref{g-tilde-diff} is similar.

\step{Step 2b: Proof of \cref{B-tilde-diff,G-tilde-diff}.} Differentiate \(\mathfrak{B}_0\) from
\cref{B0-def} use \cref{Lv-est,Lu-est,LG-est} to control the differentiated
weights \(r\), \(u\), and \(v\), and use \cref{b-tilde-diff} to arrive at
\begin{equation}\label{B-tilde-diff-prep-1}
\begin{split}
\Gamma{}^\alpha{}\mathfrak{\tilde{B}}_0 =_{\mathrm{s}} \Gamma{}^\alpha{}\mathfrak{B}_0 =_{\mathrm{s}} \mathcal{O}(\Gamma{}^\alpha{}\mathfrak{b}_0,\Gamma{}^{\le \alpha{}-1}\mathfrak{B}_0,\mathfrak{\tilde{B}}_{\le \alpha{}})=_{\mathrm{s}}\mathcal{O}(\Gamma{}^{\le \alpha{}-1}\mathfrak{B}_0,\mathfrak{\tilde{b}}_{\le \alpha{}},\mathfrak{\tilde{B}}_{\le \alpha{}})
\end{split}
\end{equation}
Differentiate \cref{B-tilde-alpha-def} and use \cref{Lv-est,Lu-est,LG-est} to get
\begin{equation}\label{B-tilde-diff-prep-2}
\Gamma{}^\beta{}\mathcal{\widetilde{G}}_{\beta{}'}|_{\beta{}+\beta{}'\le \alpha{},\abs{\beta{}'}>0} =_{\mathrm{s}} \mathcal{O}_\alpha{}(\Gamma{}^{\le \alpha{}-1}\mathfrak{B}_0,\mathfrak{\tilde{B}}_{\le \alpha{}}).
\end{equation}
Combine \cref{B-tilde-diff-prep-1,B-tilde-diff-prep-2} with \cref{b-tilde-diff} for
\(\beta{}' = 0\) to conclude
\begin{equation}\label{B-tilde-diff-prep-3}
\Gamma{}^\beta{}\mathfrak{\tilde{B}}_{\beta{}'}|_{\beta{}+\beta{}'\le \alpha{}} =_{\mathrm{s}} \mathcal{O}(\Gamma{}^{\le \alpha{}-1}\mathfrak{B}_0,\mathfrak{\tilde{b}}_{\le \alpha{}},\mathfrak{\tilde{B}}_{\le \alpha{}}).
\end{equation}
To prove \cref{B-tilde-diff}, induct on \(\alpha{}\) with trivial base case \(\alpha{} = 0\) and
inductive step given by \cref{B-tilde-diff-prep-3}. The proof of \cref{G-tilde-diff}
is similar.

\step{Step 3: Completing the proof.} The results of Steps 1 and 2 imply that for \(\abs{\alpha{}}\ge
1\), we have
\begin{equation}\label{alpha-diff-step-3}
\begin{split}
\mathfrak{b}_\alpha{} &=_{\mathrm{s}} \mathcal{O}_\alpha{}(\mathfrak{\tilde{b}}_{\le \alpha{}},\mathfrak{\tilde{B}}_{<\alpha{}},r^{-1}\mathfrak{G}_{<\alpha{}}), \\
\mathfrak{B}_\alpha{} &=_{\mathrm{s}} \mathcal{O}_\alpha{}(\mathfrak{\tilde{b}}_{\le \alpha{}},\mathfrak{\tilde{B}}_{\le \alpha{}},r^{-1}\mathfrak{G}_{<\alpha{}}), \\
\mathfrak{g}_\alpha{} &=_{\mathrm{s}} \mathcal{O}_\alpha{}(\mathfrak{\tilde{b}}_{\le \alpha{}})\set{1,\mathfrak{\tilde{g}}_{\le \alpha{}},\mathfrak{\tilde{G}}_{<\alpha{}}}, \\
\mathfrak{G}_\alpha{} &=_{\mathrm{s}} \mathcal{O}_\alpha{}(\mathfrak{\tilde{B}}_{\le \alpha{}})\set{1,\mathfrak{\tilde{g}}_{\le \alpha{}},\mathfrak{\tilde{G}}_{\le \alpha{}}}.
\end{split}
\end{equation}
It is immediate from \cref{b-tilde-alpha-def,g-tilde-alpha-def} that for \(\abs{\alpha{}}\ge 1\) we have
\begin{equation}\label{b-alpha-diff-step-3-1}
\begin{split}
C(\mathfrak{\tilde{b}}_\alpha{})&\le C(r^{-1}\Gamma{}^\alpha{}\varpi{},\Gamma{}^\alpha{}\log \kappa{},\mathbf{1}_{\alpha{}_U>0}r\Gamma{}^\alpha{}\log \kappa{},\Gamma{}^\alpha{}\log (-\gamma{})),\\
C(\mathfrak{\tilde{g}}_\alpha{})&\le C(\mathfrak{b}_\alpha{})(1 + \abs{\Gamma{}^\alpha{}\varpi{}} + \mathbf{1}_{r\ge R_0}\abs{r\Gamma{}^\alpha{}\log (-\gamma{})}),\\
\end{split}
\end{equation}
Since \(\eta_0 < 1/2\) and \(G_{\eta_0}\le C(\mathfrak{B}_0)\), applying
\cref{v-u-r-compare} to \cref{B-tilde-alpha-def,G-tilde-alpha-def} implies that for \(\abs{\alpha{}}\ge 1\) we have
\begin{equation}\label{B-alpha-diff-step-3-2}
\begin{split}
C(\mathfrak{\tilde{B}}_\alpha{}) &\le  C(\mathfrak{B}_0,v\Gamma{}^\alpha{}\log \kappa{},r\Gamma{}^{\alpha{}+U}\log \kappa{},r\tau{}\Gamma{}^\alpha{}\log (-\gamma{})), \\
C(\mathfrak{\tilde{G}}_\alpha{}) &\le  C(\mathfrak{B}_\alpha{})(1 + \abs{r\tau{}\mathbf{1}_{\alpha{}_U>0}\Gamma{}^\alpha{}\log \kappa{}}).
\end{split}
\end{equation}
To complete the proof of \cref{preliminary-geometric-estimate}, combine
\cref{alpha-diff-step-3} with \cref{b-alpha-diff-step-3-1,B-alpha-diff-step-3-2}.
\end{proof}
\subsection{Estimates for \texorpdfstring{\(\Gamma{}^\alpha{}\varpi\)}{derivatives of the renormalized Hawking mass}}
\label{sec:org69995bf}
\begin{proposition}[Estimate for \(\Gamma{}^\alpha{}\varpi{}\)]
For \(\abs{\alpha{}}\ge 1\), we have
\begin{equation}
\max (r^{-s-\eta{}_0},\tau{}^{-s-\eta{}_0})\abs{\Gamma{}^\alpha{}\varpi{}}\le C(\mathfrak{B}_{<\alpha{}})\mathcal{P}_{<\alpha{},2-s}^2.
\end{equation}
\label{Gamma-omega-bounds}
\end{proposition}
\begin{proof}
\step{Step 1: Formula for \(\Gamma^\alpha{}\varpi{}\).} Let \(\abs{\alpha{}}\ge 1\). We begin by showing that
\begin{equation}\label{Gamma-omega-bound-step-1}
\begin{split}
\Gamma{}^\alpha{}\varpi{} &=_{\mathrm{s}} \mathfrak{b}_{\le \alpha{}-1}\bigl[r^2D\Gamma{}^{\le \alpha{}-1}\varphi{}D\Gamma{}^{\le \alpha{}-1}\varphi{} + r^2\set{1,\mathbf{1}_{r\ge \Rc}u,\mathbf{1}_{r\le 2\Rc}v,v/r}D\Gamma{}^{\le \alpha{}-S}\varphi{}D\Gamma{}^{\le \alpha{}-S}\varphi{} \\
&\qquad + r^2\set{1,r,\mathbf{1}_{r\ge \Rc}u,v}V\Gamma{}^{\le \alpha{}-S}\varphi{}V\Gamma{}^{\le \alpha{}-S}\varphi{}\bigr].
\end{split}
\end{equation}

\step{Step 1a: \(\abs{\alpha{}} = 1\).} In this step we compute \(U\varpi{}\), \(V\varpi{}\), and \(S\varpi{}\). Use
the transport equation for \(\varpi{}\) in the \(u\)-direction (see
\cref{sph-sym-equations-1}) to get
\begin{equation}\label{Uomega-bound}
U\varpi{} = \mathfrak{b}_0r^2(U\varphi{})^2.
\end{equation}
The transport equation for \(\varpi{}\) in the \(v\)-direction (see \cref{sph-sym-equations-1}) and the identities in
\cref{V-dv-coordinate-change} relating the \((U,\partial{}_v)\) derivatives and the
\((U,V)\) derivatives give
\begin{equation}\label{Vomega-bound}
V\varpi{} = \mathfrak{b}_0[r^2(V\varphi{})^2 + \mathbf{1}_{r\le 2\Rc}r^2V\varphi{}U\varphi{} + \mathbf{1}_{r\le 2\Rc}r^2(U\varphi{})^2].
\end{equation}
Finally, we express \(S\) in terms of \(U\) and \(V\) using
\cref{S-in-terms-of-U-V}, and then use \cref{Uomega-bound,Vomega-bound} to compute
\begin{equation}\label{Somega-bound}
\begin{split}
S\varpi{} &= \mathfrak{b}_0[\set{\mathbf{1}_{r\ge \Rc}u,\mathbf{1}_{r\le 2\Rc}v}U\varpi{} + \set{r,\mathbf{1}_{r\ge \Rc}u,v}V\varpi{}] \\
&= \mathfrak{b}_0[r^2\set{1,\mathbf{1}_{r\ge \Rc}u,\mathbf{1}_{r\le 2\Rc}v}D\varphi{}D\varphi{} + r^2\set{r,\mathbf{1}_{r\ge \Rc}u,v}(V\varphi{})^2]. \\
\end{split}
\end{equation}
Combining \cref{Uomega-bound,Vomega-bound,Somega-bound} gives
\cref{Gamma-omega-bound-step-1} for \(\abs{\alpha{}}=1\).

\step{Step 1b: \(\abs{\alpha{}}\ge 1\).} The following consequences of \cref{gamma-uv} hold for all
\(\abs{\beta{}}\ge 0\):
\begin{equation}\label{omega-bound-lemma-consequences}
\begin{split}
\Gamma{}^\beta{}\set{r,\mathbf{1}_{r\ge \Rc}u,v} &=_{\mathrm{s}} \mathfrak{B}_{\le \beta{}-1}\set{r,\mathbf{1}_{r\ge \Rc}u,v}, \\
\Gamma{}^\beta{}\set{\mathbf{1}_{r\ge \Rc}u,\mathbf{1}_{r\le 2\Rc}v} &=_{\mathrm{s}} \mathfrak{B}_{\le \beta{}-1}\set{1,\mathbf{1}_{r\ge \Rc}u,\mathbf{1}_{r\le 2\Rc}v}, \\
\Gamma{}^\beta{}(v/r) &=_{\mathrm{s}} \mathfrak{B}_{\le \beta{}-1}\set{1,\mathbf{1}_{r\ge \Rc}u/r,v/r}.
\end{split}
\end{equation}
Indeed, the first two lines follow directly from \cref{gamma-uv}, and the third
line follows from the first. Differentiate \cref{Gamma-omega-bound-step-1} for a
multi-index \(\alpha_2\) of size \(1\) by \(\Gamma^{\alpha_1}\), and then use
\cref{omega-bound-lemma-consequences} and the commutation formulas
\cref{D-comm-weak,V-comm-weak} to conclude \cref{Gamma-omega-bound-step-1} for the
multi-index \(\alpha_1 + \alpha_2\). By induction, we have established
\cref{Gamma-omega-bound-step-1} for all \(\alpha{}\ge 1\). The reason for the
weight \(v/r\) in front of general derivatives of \(\varpi{}\) (which does not
appear when \(\alpha{} = 1\)) is because the weight \(v\) appears in front of
derivatives starting with \(V\), and commuting \(\Gamma{}^\alpha{}\) past \(V\)
(in the third term) can produce derivatives starting with \(U\), but only with
an \(r\)-weight that decays faster than \(r^{-1}\).

\step{Step 2: Completing the proof.} It follows from \cref{v-u-r-compare} that
\begin{equation}\label{omega-coordinate-bounds}
\begin{split}
\set{1,\mathbf{1}_{r\ge \Rc}u,\mathbf{1}_{r\le 2\Rc}v,v/r}&\lesssim \mathbf{1}_{r\ge \Rc}u + \mathbf{1}_{r\le v/2}v\le C(\mathfrak{B}_0)\tau{}, \\
\set{1,r,\mathbf{1}_{r\ge \Rc}u,v}&\le C(\mathfrak{B}_0)v. \\
\end{split}
\end{equation}
Applying \cref{omega-coordinate-bounds} to \cref{Gamma-omega-bound-step-1} (and noting
\(\alpha{} > 0\)) gives
\begin{equation}\label{omega-bound-penultimate}
\begin{split}
\abs{\Gamma{}^\alpha{}\varpi{}} &\le C(\mathfrak{B}_{<\alpha{}})[r^2\abs{D\Gamma{}^{\le \alpha{}-1}\varphi{}}^2 + r^2\tau{}\abs{D\Gamma{}^{\le \alpha{}-S}\varphi{}}^2 + r^2v\abs{V\Gamma{}^{\le \alpha{}-S}\varphi{}}^2].
\end{split}
\end{equation}
To complete the proof, note that the first term of \cref{omega-bound-penultimate}
appears in \(\mathcal{P}_{<\alpha{},0}^2\). Control the \(v\)-weight on the
final term by \(\tau{} + r\) (see \cref{v-u-r-compare}). This produces a term in
the form of the second term and a new term \(r^3\abs{V\Gamma^{\le
\alpha{}-S}\varphi{}}^2\) that is controlled by
\(C(\Rc,\mathfrak{b}_0)\mathcal{P}_{<\alpha{},0}^2\le
C(\mathfrak{B}_0)\mathcal{P}_{<\alpha{},0}^2\). The remaining term
\(r^2\tau{}\abs{D\Gamma^{\le \alpha{}-S}\varphi{}}^2\) is of the form
\(r^2\tau{}\abs{\Gamma{}^{<\alpha{}}\varphi{}}^2\), which is controlled by
\(r^{s+\eta{}_0}\mathcal{P}_{<\alpha{},2-s}^2\), or by \(\tau^{s +
\eta_0}\mathcal{P}_{<\alpha{},2-s}^2\).
\end{proof}
\subsection{Estimates for \texorpdfstring{\(\Gamma{}^\alpha{}\kappa\)}{derivatives of κ}}
\label{sec:orgc8768ff}
The main result of this section is \cref{Gamma-kappa-bounds}, whose proof we now
outline. The gauge condition on \(v\) normalizes \(\log \kappa{} = 0\) on the
curve \(\set{r=r_{\mathcal{H}}}\). In particular, derivatives of \(\kappa{}\)
involving only \(V\) or \(S\) vanish on this curve, since in the region
\(\set{r\le r_{\mathcal{H}}\le \Rc}\), the vector fields \(V\) and \(S\) are
tangent to curves of constant \(r\). We estimate \(\Gamma^\alpha{}\log
\kappa{}\) by computing \(U\Gamma^\alpha{}\log \kappa{}\) with the
\(U\)-transport equation for \(\log \kappa{}\) in \cref{U-kappa-alpha-bound} and
integrating towards \(\set{r=r_{\mathcal{H}}}\) in \cref{Gamma-kappa-bounds}. In
the case where \(\alpha_U > 0\), we can rearrange \(\Gamma^\alpha{}\log \kappa{}
= U\Gamma^{\alpha{}-U}\log \kappa{}\) and use the \(U\)-transport equation
directly, without integrating it.
\begin{lemma}
We have
\begin{equation}\label{S-bounded-by-V}
\abs{S\psi{}}\le C(\mathfrak{B}_0)(\tau{}\abs{U\psi{}} + v\abs{V\psi{}})
\end{equation}
and
\begin{equation}\label{V-bounded-by-S}
\abs{V\psi{}}\le C(\mathfrak{B}_0)(v^{-1}\abs{S\psi{}} + \abs{U\psi{}}).
\end{equation}
\end{lemma}
\begin{proof}
We omit the computation.
\end{proof}
\begin{lemma}[Estimates for \(U\Gamma{}^\alpha{}\log \kappa{}\)]
Let \(L\in \Gamma^\alpha{}\) for \(\abs{\alpha{}}\ge 0\). For sufficiently small \(s > 0\), we have
\begin{equation}\label{kappa-alpha-bound-1}
\abs{UL\log \kappa{}}\le C(\mathcal{C}_{<\alpha{}},r^{-s}\mathfrak{G}_{<\alpha{}})[r\abs{U\Gamma{}^{\le \alpha{}}\varphi{}}^2 + r^{-2}\tau{}^{-1}\mathcal{P}_{\alpha{},1}^2 + r^{-2+s}\abs{\Gamma{}^{\le \alpha{}}\log \kappa{}}],
\end{equation}
and if \(\alpha_V + \alpha_S > 0\), then
\begin{equation}\label{kappa-alpha-bound-2}
\abs{UL\log \kappa{}}\le C(\mathcal{C}_{<\alpha{}},r^{-s}\mathfrak{G}_{<\alpha{}})[rv\abs{U\Gamma{}^{<\alpha{}}\varphi{}}^2 + r^{-3}\mathcal{P}_{<\alpha{},1}^2 + r^{-2+s}\abs{\Gamma{}^{\le \alpha{}}\log \kappa{}}].
\end{equation}
\label{U-kappa-alpha-bound}
\end{lemma}
\begin{proof}
\step{Step 0: Controlling the commutator \([U,L]\psi{}\).} Let \(L\in \Gamma^\alpha{}\). We claim that
\begin{equation}\label{kappa-alpha-step-0}
\abs{[U,L]f}\le C(\mathcal{C}_{<\alpha{}},r^{-s}\mathfrak{G}_{<\alpha{}})[\abs{U\Gamma{}^{<\alpha{}}f} + r^{-2+s}\abs{\Gamma{}^{\le \alpha{}}f}].
\end{equation}
By \cref{U-comm-weak}, we have
\begin{equation}
\abs{[U,L]f}\le C(\mathcal{C}_{<\alpha{}},r^{-s}\mathfrak{G}_{<\alpha{}})[\abs{U\Gamma{}^{<\alpha{}}f} + r^{-2+s}\abs{\Gamma{}^{\le \alpha{}}f} + r^{-1+s}\abs{V\Gamma{}^{\le \alpha{}-S}f}].
\end{equation}
Use \cref{V-bounded-by-S} and \cref{v-u-r-compare} to control the last term by the
first two terms.

\step{Step 1: Preliminary estimate for \(\abs{UL \log \kappa{}}\).} We show by induction on \(\alpha{}\)
that
\begin{equation}\label{kappa-alpha-step-1}
\abs{UL\log \kappa{}}\le r\abs{LU\varphi{}}\abs{U\varphi{}} + C(\mathcal{C}_{<\alpha{}},r^{-s}\mathfrak{G}_{<\alpha{}})[r\abs{U\Gamma{}^{<\alpha{}}\varphi{}}^2 + r^{-3+2s}\abs{\Gamma{}^{<\alpha{}}\varphi{}}^2 + r^{-2+s}\abs{\Gamma{}^{\le \alpha{}}\log \kappa{}}].
\end{equation}
When \(\alpha{} = 0\), \cref{kappa-alpha-step-1} follows from the equation \(U\log \kappa{} =
-r(U\varphi{})^2\). For the inductive step, use \cref{kappa-alpha-step-0}, the
equation for \(U\kappa{}\), the estimate \(\abs{\Gamma^{\le \alpha{}}r}\le
C_\alpha{}r\), and \cref{kappa-alpha-bound-1} for multi-indices \(< \alpha{}\) to
estimate
\begin{equation}
\begin{split}
\abs{UL\log \kappa{}}&\le \abs{LU\log \kappa{}} + \abs{[U,L]\log \kappa{}} \le \abs{L(r(U\varphi{})^2)} + C(\mathcal{C}_{<\alpha{}},r^{-s}\mathfrak{G}_{<\alpha{}})[\abs{U\Gamma{}^{<\alpha{}}\log \kappa{}} + r^{-2+s}\abs{\Gamma{}^{\le \alpha{}}\log \kappa{}}] \\
&\le \side{RHS}{kappa-alpha-step-1}.
\end{split}
\end{equation}

\step{Step 2: Proof of \cref{kappa-alpha-bound-1}.} Use Young's inequality and
\cref{kappa-alpha-step-0} on the first term in \cref{kappa-alpha-step-1}, and control
the term involving \(\abs{\Gamma{}^{\le \alpha{}}\varphi{}}^2\) using the
pointwise norm defined in \cref{sec:norms}.

\step{Step 3: Proof of \cref{kappa-alpha-bound-2}.} In Steps 3ab we show
that if \(\alpha_V + \alpha_S > 0\), we have
\begin{equation}\label{kappa-alpha-step-3}
\abs{LU\varphi{}}\le C(\mathcal{C}_{<\alpha{}},r^{-s}\mathfrak{G}_{<\alpha{}})[v\abs{U\Gamma{}^{<\alpha{}}\varphi{}} + r^{-2+s}v\abs{\Gamma{}^{<\alpha{}}\varphi{}} + \abs{V\Gamma{}^{<\alpha{}}\varphi{}}].
\end{equation}
In Step 3c we deduce \cref{kappa-alpha-bound-2} from \cref{kappa-alpha-step-3}.

\step{Step 3a: Proof of \cref{kappa-alpha-step-3} when \(\alpha_V > 0\).} Use
\cref{rearrangement-formula} to bring \(V\) to the end of \(L\), then use the wave
equation \cref{UV-box} and the commutation formulas \cref{D-comm-weak,U-comm-weak} to
get
\begin{equation}\label{kappa-alpha-bound-step-3a}
LU\varphi{}  =_{\mathrm{s}} \Gamma{}^{\le \alpha{}-V}VU\varphi{} =_{\mathrm{s}} \Gamma{}^{\le \alpha{}-V}(\mathfrak{b}_0[D\varphi{} + U^2\varphi{}]) + \mathcal{C}_{<\alpha{}}D\Gamma{}^{\le \alpha{}-V}\varphi{} =_{\mathrm{s}} \mathcal{C}_{<\alpha{}}[U\Gamma{}^{<\alpha{}}\varphi{} + V\Gamma{}^{\le \alpha{}-V}\varphi{}].
\end{equation}

\step{Step 3b: Proof of \cref{kappa-alpha-step-3} when \(\alpha_S > 0\).} First,
\cref{kappa-alpha-step-0} implies that
\begin{equation}\label{kappa-alpha-V-gamma-alpha-S}
V\Gamma^{\alpha{}-S}U\varphi{} =_{\mathrm{s}} \Gamma{}^{<\alpha{}}U\varphi{} =_{\mathrm{s}} C(\mathcal{C}_{<\alpha{} },r^{-s}\mathfrak{G}_{<\alpha{}})[U\Gamma{}^{<\alpha{}}\varphi{} + r^{-2+s}\Gamma{}^{<\alpha{}}\varphi{}],
\end{equation}
Now use \cref{rearrangement-formula} to bring \(S\) to the front of \(L\),
\cref{S-bounded-by-V} to rewrite \(S\) in terms of \(U\) and \(V\), and
\cref{kappa-alpha-V-gamma-alpha-S} to rewrite \(V\Gamma^{\alpha{}-S}U\varphi{}\),
and hence obtain
\begin{equation}\label{kappa-alpha-bound-step-3b}
\begin{split}
\abs{LU\varphi{}}&\le \abs{S\Gamma{}^{\alpha{}-S}U\varphi{}} + C(\mathcal{C}_{<\alpha{}})\abs{D\Gamma{}^{<\alpha{}}\varphi{}}\le C(\mathcal{C}_{<\alpha{}})[\tau{}\abs{U\Gamma{}^{\alpha{}-S}U\varphi{}} + v\abs{V\Gamma{}^{\alpha{}-S}U\varphi{}} + \abs{D\Gamma{}^{<\alpha{}}\varphi{}}] \\
&\le C(\mathcal{C}_{<\alpha{}},r^{-s}\mathfrak{G}_{<\alpha{}})[v\abs{U\Gamma{}^{<\alpha{}}\varphi{}} + r^{-2+s}v\abs{\Gamma{}^{<\alpha{}}\varphi{}} + \abs{V\Gamma{}^{<\alpha{}}\varphi{}}].
\end{split}
\end{equation}

\step{Step 3c: Deducing \cref{kappa-alpha-bound-2} from \cref{kappa-alpha-step-3}.} First,
\cref{kappa-alpha-step-3} and Young's inequality imply
\begin{equation}\label{kappa-alpha-step-3-cons}
r\abs{LU\varphi{}}\abs{U\varphi{}}\le C(\mathcal{C}_{<\alpha{}},r^{-s}\mathfrak{G}_{<\alpha{}})[rv\abs{U\Gamma{}^{<\alpha{}}\varphi{}}^2 + r^{-3+2s}v\abs{\Gamma{}^{<\alpha{}}\varphi{}}^2 + r\abs{V\Gamma{}^{<\alpha{}}\varphi{}}^2].
\end{equation}
Combine \cref{kappa-alpha-step-3-cons,kappa-alpha-step-1} to get
\begin{equation}\label{kappa-alpha-3c-1}
\abs{UL\log \kappa{}}\le C(\mathcal{C}_{<\alpha{}},r^{-s}\mathfrak{G}_{<\alpha{}})[rv\abs{U\Gamma{}^{<\alpha{}}\varphi{}}^2 + r^{-3+2s}v\abs{\Gamma{}^{<\alpha{}}\varphi{}}^2 + r\abs{V\Gamma{}^{<\alpha{}}\varphi{}}^2 + r^{-2+s}\abs{\Gamma{}^{\le \alpha{}}\log \kappa{}}].
\end{equation}
We have the pointwise estimates
\begin{equation}\label{kappa-alpha-3c-2}
r^{-3+2s}v\abs{\Gamma{}^{<\alpha{}}\varphi{}}^2\le \mathbf{1}_{r\le v/2}r^{-4+2s}\tau{}^{-1}v\mathcal{P}_{<\alpha{},1}^2 + \mathbf{1}_{r\ge v/2}r^{-5+2s}v\mathcal{P}_{<\alpha{},0}^2\le C(\mathcal{B}_0)r^{-4+2s}\mathcal{P}_{<\alpha{},1}^2\le C(\mathcal{C}_{<\alpha{}})r^{-3}\mathcal{P}_{<\alpha{},1}^2,
\end{equation}
where we used \cref{v-u-r-compare} and \(s<1/2\), and
\begin{equation}\label{kappa-alpha-3c-3}
r\abs{V\Gamma{}^{<\alpha{}}\varphi{}}^2\le C(\mathfrak{B}_0,\Rc)r^{-3}\mathcal{P}_{<\alpha{},0}^2\le C(\mathcal{C}_{<\alpha{}})r^{-3}\mathcal{P}_{<\alpha{},1}^2.
\end{equation}
Substitute \cref{kappa-alpha-3c-2,kappa-alpha-3c-3} into \cref{kappa-alpha-3c-1} to
obtain \cref{kappa-alpha-bound-2}.
\end{proof}
Before we integrate \cref{U-kappa-alpha-bound} towards \(\set{r=r_{\mathcal{H}}}\), we introduce an
integration lemma.
\begin{lemma}
We have
\begin{equation}\label{u-energy-estimate-equation}
\int_u^\infty \frac{r}{(-\nu{})}(\partial{}_u\psi{})^2(u',v)\dd{}u'\le v^{-1}C(\mathfrak{B}_0)\mathcal{E}_{1}[\psi{}]^2.
\end{equation}
\label{u-energy-estimate}
\end{lemma}
\begin{proof}
Fix \((u,v)\). Define \(u_\ast{}\in [u,\infty)\) so that \(r(u_\ast{},v) = v/2\) if such a
\(u_\ast{}\) exists, and otherwise set \(u_\ast{} = u\). Split the integration
range as \([u,\infty) = [u,u_\ast{}]\cup [u,\infty)\). The \(r\)-weight in the
energy \(\mathcal{E}[\psi{}]^2\) is one power stronger than in the integral on
the left of \cref{u-energy-estimate-equation}, so by the monotonicity of \(r\) we
gain one power of decay in \(r(u_\ast{},v)\) when integrating over
\([u,u_\ast{}]\). The energy \(\mathcal{E}_p[\psi{}]^2\) controls the integral
over \([u_\ast{},\infty)\) with a weight \(\tau{}(u_\ast{},v)\), which is
comparable to \(v\) by \cref{v-u-r-compare}. Thus:
\begin{equation}
\begin{split}
\int_u^\infty \frac{r}{(-\nu{})}(\partial{}_u\psi{})^2(u',v)\dd{}u' &= r^{-1}(u_\ast{},v)\int_u^{u_\ast{}} \frac{r^2}{(-\nu{})}(\partial{}_u\psi{})^2(u',v)\dd{}u' + r_{\mathrm{min}}^{-1}\int_{u_\ast{}}^{\infty}  \frac{r^2}{(-\nu{})}(\partial{}_u\psi{})^2(u',v)\dd{}u'\\
&\le 2v^{-1}\mathcal{E}[\psi{}]^2 + r_{\mathrm{min}}^{-1}\tau{}^{-1}(u_\ast{},v)E_1[\psi{}]^2\le v^{-1}C(\mathfrak{B}_0)\mathcal{E}_1[\psi{}]^2.
\end{split}
\end{equation}
\end{proof}
\begin{proposition}[Estimates for \(\Gamma{}^\alpha{}\log \kappa{}\)]
Let \(\abs{\alpha{}}\ge 1\), and let \(s > 0\) be sufficiently small. We have decay estimates
for order-\(\alpha{}\) derivatives of \(\log \kappa{}\) when norms of \(\Gamma^{\le \alpha{}}\varphi{}\) are
bounded:
\begin{align}
v\abs{\Gamma{}^\alpha{}\log \kappa{}}&\le C(\mathfrak{B}_{<\alpha{}},r^{-s}\mathfrak{G}_{<\alpha{}})(\mathcal{E}_{\alpha{},1}^2 + \mathcal{P}_{\alpha{},1}^2), \label{kappa-alpha-bound-after-energy} \\
r^{1-s-\eta{}_0}\tau{}\abs{\Gamma{}^\alpha{}\log \kappa{}}&\le C(\mathfrak{B}_{<\alpha{}},r^{-s}\mathfrak{G}_{<\alpha{}})(\mathcal{E}_{\alpha{},1}^2 + \mathcal{P}_{\alpha{},2-s}^2)\quad \text{if }\alpha{}_U > 0, \label{kappa-alpha-bound-when-U}\\
r\abs{\Gamma{}^{\alpha{}+U}\log \kappa{}}&\le C(\mathfrak{B}_{<\alpha{}},r^{-s}\mathfrak{G}_{<\alpha{}})(\mathcal{E}_{\alpha{},1}^2 + \mathcal{P}_{\alpha{},1}^2), \label{U-kappa-alpha-bound-top-order}
\end{align}
Moreover, order-\(\alpha{}\) derivatives of \(\log \kappa{}\) are bounded if lower order
derivatives of \(\varphi{}\) are:
\begin{align}
\abs{\Gamma^\alpha{}\log \kappa{}}&\le C(\mathfrak{B}_{<\alpha{}},r^{-s}\mathfrak{G}_{<\alpha{}})(\mathcal{E}_{<\alpha{},1}^2 + \mathcal{P}_{<\alpha{},1}^2). \label{kappa-alpha-bound-before-energy}                   \\
r\abs{\Gamma^\alpha{}\log \kappa{}}&\le C(\mathfrak{B}_{<\alpha{}},r^{-s}\mathfrak{G}_{<\alpha{}})(\mathcal{E}_{<\alpha{},1}^2 + \mathcal{P}_{<\alpha{},1}^2)\quad \text{if }\alpha{}_U > 0. \label{kappa-alpha-U-bound-before-energy}
\end{align}
\label{Gamma-kappa-bounds}
\end{proposition}
\begin{proof}
In Step 1, we establish a preliminary estimate. In Step 2, we use Step 1 and
\cref{U-kappa-alpha-bound} to show \cref{kappa-alpha-bound-after-energy}. In Step 3,
we prove \cref{kappa-alpha-bound-before-energy}. In Step 4 we use the results of
the previous steps to show
\cref{kappa-alpha-bound-when-U,U-kappa-alpha-bound-top-order,kappa-alpha-U-bound-before-energy}.

\step{Step 1: Estimate for \(\Gamma^{\alpha{} + U}\log \kappa{}\) in terms of \(\Gamma^\alpha{}\log \kappa{}\).} Use
\cref{U-rearrangement-formula} to move the \(U\) in \(\Gamma^{\alpha{} + U}\) to
the front, and then use \cref{kappa-alpha-bound-1} to obtain
\begin{equation}\label{U-kappa-alpha-bound-top-order-prep}
\begin{split}
\abs{\Gamma{}^{\alpha{}+U}\log \kappa{}}&\le C(\mathcal{C}_{<\alpha{}},r^{-s}\mathfrak{G}_{<\alpha{}})[\abs{U\Gamma{}^{\le \alpha{}}\log \kappa{}} + r^{-1+s}\abs{V\Gamma{}^{\le \alpha{}-V}\log \kappa{}}] \\
&\le C(\mathcal{C}_{<\alpha{}},r^{-s}\mathfrak{G}_{<\alpha{}})[r\abs{U\Gamma{}^{\le \alpha{}}\varphi{}}^2 + r^{-2}\tau{}^{-1}\mathcal{P}_{\alpha{},1}^2 + r^{-1+s}\abs{\Gamma{}^{\le \alpha{}}\log \kappa{}}].
\end{split}
\end{equation}

\step{Step 2: Proof of \cref{kappa-alpha-bound-after-energy}.} We will write, for
example, \((\ref{kappa-alpha-bound-after-energy})_{\le \alpha{}}\) to mean the
collection of statements \((\ref{kappa-alpha-bound-after-energy})_\beta{}\) for
multi-indices \(\beta{}\) such that \(\beta{}\le \alpha{}\) and
\(\abs{\beta{}}\ge 1\). In Step 2a, we show that Step 1 reduces us to estimating
\(U\log \kappa{}\) and \(\Gamma^\alpha{}\log \kappa{}\) for \(\alpha_U = 0\).
The first case (Step 2b) is easy, and in the second case (Step 2c), we can
integrate \cref{U-kappa-alpha-bound} to \(\set{r = R_0}\), since \(\Gamma^\alpha{}\log \kappa{}\)
vanishes there when \(\alpha_U = 0\).

\step{Step 2a: Proof that \((\ref{kappa-alpha-bound-after-energy})_{\le \alpha{}}\) implies
\((\ref{kappa-alpha-bound-after-energy})_{\le \alpha{} + nU}\) for all \(n\ge 0\).} By
induction on \(n\), it is enough to obtain the conclusion
\((\ref{kappa-alpha-bound-after-energy})_{\le \alpha{} + U}\). Since
\(\beta{}\le \alpha{} + U\) if and only if \(\beta{} = \alpha{}\) or \(\beta{} =
\beta{}' + U\) for some \(\beta{}'\le \alpha{}\), we reduce to showing that
\((\ref{kappa-alpha-bound-after-energy})_{\le \alpha{}}\) implies
\((\ref{kappa-alpha-bound-after-energy})_{\alpha{} + U}\). This follows from
substituting \((\ref{kappa-alpha-bound-after-energy})_{\le \alpha{}}\) into
\cref{U-kappa-alpha-bound-top-order-prep} and using
\begin{equation}\label{kappa-bound-step2a-prep}
r\abs{U\Gamma^{\le \alpha{}}\varphi{}}^2\le \mathbf{1}_{r\le v/2}\tau{}^{-1}\mathcal{P}_{\alpha{}+U,1}^2 + \mathbf{1}_{r\ge v/2}r^{-1}\mathcal{P}_{\alpha{},0}^2\le v^{-1}C(\mathfrak{B}_0)\mathcal{P}_{\alpha{}+U,1}^2.
\end{equation}
and \cref{v-u-r-compare} and \(s < 1\).

\step{Step 2b: Proof of \((\ref{kappa-alpha-bound-after-energy})_{\alpha{}}\) for \(\alpha{} = U\).}
The transport equation for \(\log \kappa{}\) is \(U\log \kappa{} =
-r(U\varphi{})^2\). To complete the proof, control the right side by
\(v^{-1}\mathcal{P}_{U,1}^2\) as in \cref{kappa-bound-step2a-prep}.

\step{Step 2c: Proof of \((\ref{kappa-alpha-bound-after-energy})_{\alpha{}}\) if \(\alpha_U = 0\).}
Fix \((u,v)\). Write \(r_0 = r(u,v)\). Integrate \cref{kappa-alpha-bound-1} to
\(\set{r=r_{\mathcal{H}}}\), noting that \(L\log
\kappa{}|_{\set{r=r_{\mathcal{H}}}} = 0\) because \(\alpha{}_U = 0\):
\begin{equation}\label{kappa-alpha-integrate-1}
\begin{split}
\abs{L\log \kappa{}}(r=r_0,v)&\le C(\mathcal{C}_{<\alpha{}},r^{-s}\mathfrak{G}_{<\alpha{}})\Bigl[\underbrace{\int_{\min (r_0,r_{\mathcal{H}})}^{\max (r_0,r_{\mathcal{H}})} r\abs{U\Gamma{}^{\le \alpha{}}\varphi{}}^2\dd{}r}_{\coloneqq{}\text{(I)}} \\
&\qquad + \underbrace{\mathcal{P}_{\alpha{},1}^2\int_{r_{\mathcal{H}}}^{r_0} r^{-2}\tau{}^{-1}\dd{}r}_{\coloneqq{}\text{(II)}} + \int_{\min (r_0,r_{\mathcal{H}})}^{\max (r_0,r_{\mathcal{H}})} r^{-2+s}\abs{\Gamma{}^{\le \alpha{}}\log \kappa{}}\dd{}r\Bigr]. \\
\end{split}
\end{equation}
Terms \(\text{(I)}\) and \(\text{(II)}\) decay like \(v^{-1}\). Indeed, for term
\((\text{I})\), use \cref{u-energy-estimate}, and for term \(\text{(II)}\), split
the integral into the regions \(r\le v/2\) and \(r\ge v/2\) and use
\cref{v-u-r-compare}. Thus
\begin{equation}
\abs{L\log \kappa{}}(r=r_0,v)\le v^{-1}C(\mathcal{C}_{<\alpha{}},r^{-s}\mathfrak{G}_{<\alpha{}})(\mathcal{E}_{\alpha{},1}^2 + \mathcal{P}_{\alpha{},1}^2) + C(\mathcal{C}_{<\alpha{}},r^{-s}\mathfrak{G}_{<\alpha{}})\int_{\min (r_0,r_{\mathcal{H}})}^{\max (r_0,r_{\mathcal{H}})} r^{-2+s}\abs{\Gamma{}^{\le \alpha{}}\log \kappa{}}\dd{}r.
\end{equation}
Sum over \(L\in \Gamma^{\le \alpha{}}\) and apply Grönwall's inequality (noting that \(s < 1\)) and recall the
definition of \(\mathcal{C}_\alpha{}\) (see \cref{C-alpha-def}) to complete the proof.

\step{Step 2d: Completing the proof by induction.} Let \(\abs{\alpha{}}\ge 1\), and suppose we have
shown \((\ref{kappa-alpha-bound-after-energy})_{<\alpha{}}\). Write \(\alpha{} =
\beta{} + nU\) for \(\abs{\beta{}}\ge 1\) and \(n\ge 0\) maximal. Then either
\(\beta{} = U\) or \(\beta_U = 0\), so by Step 2b or Step 2c we obtain
\((\ref{kappa-alpha-bound-after-energy})_{\beta{}}\). Together with the
induction hypothesis, this implies \((\ref{kappa-alpha-bound-after-energy})_{\le
\beta{}}\), so Step 2a now implies \((\ref{kappa-alpha-bound-after-energy})_{\le
\beta{}+nU} = (\ref{kappa-alpha-bound-after-energy})_{\le \alpha{}}\), which
completes the induction.

\step{Step 3: Proof of \cref{kappa-alpha-bound-before-energy}.} Follow Step 2, but instead
integrate \cref{kappa-alpha-bound-2} (which contains only lower order norms on the
right side, as opposed to \cref{kappa-alpha-bound-1}) as in Step 2c. We omit the
details.

\step{Step 4: Proof of \cref{kappa-alpha-bound-when-U,U-kappa-alpha-bound-top-order,kappa-alpha-U-bound-before-energy}
from \cref{kappa-alpha-bound-after-energy,kappa-alpha-bound-before-energy}.} Let
\(\abs{\alpha{}}\ge 1\) satisfy \(\alpha_U > 0\). Start
with \cref{U-kappa-alpha-bound-top-order-prep} to obtain
\begin{equation}
\abs{\Gamma{}^{\alpha{}}\log \kappa{}}\le C(\mathcal{G}_{<\alpha{}})[r\abs{U\Gamma{}^{\le \alpha{}-U}\varphi{}}^2 + r^{-2}\tau{}^{-1}\mathcal{P}_{<\alpha{},1} + r^{-1+s}\abs{\Gamma{}^{<\alpha{}}\log \kappa{}}]. \\
\end{equation}
To show \cref{kappa-alpha-bound-when-U} estimate
\begin{equation}
\abs{U\Gamma^{\le \alpha{}-U}\varphi{}}^2\le \abs{\Gamma{}^{\le \alpha{}}\varphi{}}^2\le r^{-2 + s + \eta{}_0}\tau^{-1}\mathcal{P}_{\alpha{},2-s}^2
\end{equation}
and use \cref{kappa-alpha-bound-after-energy} to control \(\abs{\Gamma^{\le \alpha{}}\log
\kappa{}}\). The same argument establishes
\cref{U-kappa-alpha-bound-top-order} and \cref{kappa-alpha-U-bound-before-energy}, but we
estimate \(\abs{U\Gamma^{\le \alpha{}}\varphi{}}\) (resp. \(U\Gamma^{\le \alpha{}-U}\)) by \(r^{-1}\mathcal{P}_{\alpha{},0}\)
instead.
\end{proof}
\subsection{Estimates for \texorpdfstring{\(\Gamma{}^\alpha{}(-\gamma{})\)}{derivatives of γ}}
\label{sec:orgc1fa4d2}
The main result of this section is \cref{Gamma-gamma-bounds}. The gauge condition
on \(u\) normalizes \(\log (-\gamma{}) = 0\) on \(\mathcal{I}\). In fact,
\(\Gamma^\alpha{}\log (-\gamma{}) = 0\) on \(\mathcal{I}\) for all \(\alpha{}\),
which we establish as part of the proof of \cref{Gamma-gamma-bounds}. To complete
estimates in which we integrate the transport equation for \(\Gamma^\alpha{}\log
(-\gamma{})\) (see \cref{gamma-bound-step-1c}) to \(\mathcal{I}\), we first
formulate a variant of Grönwall's inequality.
\begin{lemma}[Grönwall inequality on half-infinite interval]
Let \(I = [t_0,\infty)\). Suppose \(f,g,h : I\to \R\) are continuous functions such that \(f\)
is bounded, \(g\) is non-increasing, \(h\) is non-negative and integrable on
\(I\), and
\begin{equation}
f(t)\le g(t) + \int_t^\infty h(s)f(s)\dd{}s \quad\text{for } t\in I.
\end{equation}
Then
\begin{equation}
f(t)\le g(t)\exp \Bigl(\int_t^\infty h(s)\dd{}s\Bigr)\quad\text{for } t\in I.
\end{equation}
\label{decaying-gronwall}
\end{lemma}
\begin{proof}
Let \(\epsilon{} > 0\). Since \(h\in L^1\) and \(f\in L^\infty\), we have \(hf\in L^1\), so there is
\(T > t_0\) such that \(\abs{\int_T^\infty hf}\le \epsilon{}\). The assumption now implies
\begin{equation}
f(t)\le (g(t) + \epsilon{}) + \int_t^T h(s)f(s)\dd{}s
\end{equation}
for \(t\in [t_0,T]\). The usual Grönwall inequality (after a change of variables
\(t\mapsto T-t\)) and the non-negativity of \(h\) give
\begin{equation}
f(t)\le (g(t) + \epsilon{})\exp \Bigl(\int_t^T h(s)\dd{}s\Bigr)\le (g(t) + \epsilon{})\exp \Bigl(\int_t^\infty h(s)\dd{}s\Bigr).
\end{equation}
We are done since \(\epsilon{} > 0\) is arbitrary.
\end{proof}
\begin{lemma}[Estimate for \(\overline{\partial}_rL\log (-\gamma{})\)]
Let \(\abs{\alpha{}}\ge 0\) and let \(L\in \Gamma{}^\alpha{}\). We have
\begin{equation}
\mathbf{1}_{r\ge R_0}\abs{\overline{\partial}_rL\log (-\gamma{})}\le C(\mathcal{G}_{\alpha{},s})[r\abs{\overline{\partial}_r\Gamma{}^{\le \alpha{}}\varphi{}}^2 + r^{-3+2s}\abs{\Gamma{}^{\le \alpha{}}\varphi{}}^2 + r^{-2+s}\abs{\Gamma{}^{\le \alpha{}}\log (-\gamma{})}].
\end{equation}
\label{gamma-bound-step-1c}
\end{lemma}
\begin{proof}
\step{Step 1: Computing the commutator \([\overline{\partial}_r,L]\)}. We have
\begin{equation}\label{gamma-bound-commutator}
\begin{split}
\mathbf{1}_{r\ge R_0}[\overline{\partial}_r,L] =_{\mathrm{s}} \mathcal{O}(\mathfrak{b}_\alpha{},r^{-s}\mathfrak{g}_{<\alpha{}})[\overline{\partial}_r\Gamma{}^{<\alpha{}} + r^{-2+s}\Gamma{}^{\le \alpha{}}]
\end{split}
\end{equation}
We omit the proof, but the argument is an induction with base case
\cref{dr-comm-1}. A similar argument is done in detail in the proof of \cref{V-comm}.

\step{Step 2: Proof of the desired \cref{gamma-bound-step-1c}.} It is enough to show
\begin{equation}\label{gamma-bound-step-1b}
\mathbf{1}_{r\ge R_0}\overline{\partial}_rL\log (-\gamma{}) =_{\mathrm{s}} \mathcal{G}_{\alpha{},s}[r\set{\overline{\partial}_r\Gamma{}^{\le \alpha{}}\varphi{},r^{-2+s}\Gamma{}^{\le \alpha{}}\varphi{}}^2 + r^{-2+s}\Gamma{}^{\le \alpha{}}\log (-\gamma{})].
\end{equation}
When \(\abs{\alpha{}} = 0\), \cref{gamma-bound-step-1b} follows from the equation \(\overline{\partial}_r\log (-\gamma{}) =
r(\overline{\partial}_r\varphi{})^2\). From this equation and \cref{gamma-bound-commutator}, we
get
\begin{equation}\label{gamma-bound-step-1b-prep}
\begin{split}
\overline{\partial}_rL\log (-\gamma{}) &= L\overline{\partial}_r\log (-\gamma{}) + [\overline{\partial}_r,L]\log (-\gamma{}) \\
&=_{\mathrm{s}} \mathcal{G}_{\alpha{},s}[r\Gamma{}^{\le \alpha{}}\overline{\partial}_r\varphi{}\Gamma{}^{\le \alpha{}}\overline{\partial}_r\varphi{} + \overline{\partial}_r\Gamma{}^{<\alpha{}}\log (-\gamma{}) + r^{-2+s}\Gamma{}^{\le \alpha{}}\log (-\gamma{})]. \\
\end{split}
\end{equation}
To complete the inductive proof of \cref{gamma-bound-step-1b} for \(L\in \Gamma^\alpha{}\), use
\cref{gamma-bound-commutator} for the first term on the right of
\cref{gamma-bound-step-1b-prep} and for the second term use \cref{gamma-bound-step-1b}
for multi-indices \(< \alpha{}\).
\end{proof}
\begin{lemma}[Estimate for \(\partial{}_u^n\log (-\gamma{})\)]
Let \(n\ge 0\). We have
\begin{equation}
\mathbf{1}_{r\ge R_0}\abs{\partial{}_u^n\log (-\gamma{})}\le r^{-2}C(\mathfrak{b}_{U^{<n}})\mathcal{P}_{U^{<n},0}^2.
\end{equation}
\label{gamma-bound-step-3}
\end{lemma}
\begin{proof}
\step{Step 1: Computing \(\partial{}_v\partial{}_u^n\varphi{}\).} We claim that for \(n\ge 1\), we have
\begin{equation}\label{gamma-bound-step-3a}
\partial{}_v\partial{}_u^n\varphi{} =_{\mathrm{s}} r^{-1}\mathcal{O}(\partial{}_u^{\le n}\lambda{},\partial{}_u^{<n}(-\nu{}))[\partial{}_u\partial{}_u^{<n}\varphi{} + \partial{}_v\varphi{}].
\end{equation}
When \(n = 1\), this follows from the wave equation,
\begin{equation}\label{gamma-bound-step-3-wave}
\partial{}_v\partial{}_u\varphi{} = \kappa{}(-\nu{})\Box{}\varphi{} + r^{-1}(-\lambda{}\partial{}_u\varphi{} + (-\nu{})\partial{}_v\varphi{}) = r^{-1}(-\lambda{}\partial{}_u\varphi{} + (-\nu{})\partial{}_v\varphi{}),
\end{equation}
and the inductive step follows since \(\partial{}_ur^{-1} = r^{-1}\cdot r^{-1}(-\nu{})\).

\step{Step 2: Estimating \(\partial{}_v\partial{}_u^n\log (-\gamma{})\).} We show that for \(n\ge 1\), we have
\begin{equation}\label{gamma-bound-step-3b}
\mathbf{1}_{r\ge R_0}\abs{\partial{}_v\partial{}_u^n\log (-\gamma{})}\le r^{-3}C(\mathfrak{b}_{U^{<n}})\mathcal{P}_{U^{<n},0}^2.
\end{equation}
We first show that
\begin{equation}\label{gamma-bound-step-3-step-2}
\mathbf{1}_{r\ge R_0}\partial{}_v\partial{}_u^n\log (-\gamma{}) =_{\mathrm{s}} \mathfrak{b}_{U^{<n}}[r^{-1}\partial{}_u\partial{}_u^{<n}\varphi{}\partial{}_u\partial{}_u^{<n}\varphi{} + \partial{}_u\partial{}_u^{<n}\varphi{}\partial{}_v\varphi{} + (\partial{}_v\varphi{})^2].
\end{equation}
For the case \(n = 1\), we compute
\begin{equation}
\begin{split}
\partial{}_u\partial{}_v\log (-\gamma{}) &= \partial{}_u(\lambda{}^{-1}r(\partial{}_v\varphi{})^2) \\
&= (-\nu{})\lambda{}^{-1}(\partial{}_v\varphi{})^2 - 2r^{-1}\lambda{}^{-1}\kappa{}(-\nu{})(\varpi{}-\mathbf{e}^2/r)(\partial{}_v\varphi{})^2 + \partial{}_u\varphi{}\partial{}_v\varphi{} + (-\nu{})\lambda{}^{-1}(\partial{}_v\varphi{})^2 \\
&=_{\mathrm{s}} \mathcal{O}(\lambda{}^{-1},(-\nu{}),\kappa{},r^{-1}\varpi{})[\partial{}_u\varphi{}\partial{}_v\varphi{} + (\partial{}_v\varphi{})^2],
\end{split}
\end{equation}
so that
\begin{equation}
\mathbf{1}_{r\ge R_0}\partial{}_u\partial{}_v\log (-\gamma{}) =_{\mathrm{s}} \mathfrak{b}_0[\partial{}_u\varphi{}\partial{}_v\varphi{} + (\partial{}_v\varphi{})^2].
\end{equation}
For the inductive step, we have
\begin{equation}
\begin{split}
\mathbf{1}_{r\ge R_0}\partial{}_v\partial{}_u\log (-\gamma{}) &=_{\mathrm{s}} \partial{}_u^{\le 1}\mathfrak{b}_{U^{<n}}[r^{-1}\partial{}_u\partial{}_u^{<n+1}\varphi{}\partial{}_u\partial{}_u^{<n+1}\varphi{} + \partial{}_u\partial{}_u^{<n+1}\varphi{}\partial{}_v\varphi{}\\
&\qquad + (\partial{}_v\varphi{})^2 + \partial{}_u\partial{}_u^{<n+1}\varphi{}\partial{}_u\partial{}_v\varphi{} + \partial{}_u\partial{}_v\varphi{}\partial{}_v\varphi{}].
\end{split}
\end{equation}
The first three terms are of the desired form. For the last two terms, use the
wave equation \cref{gamma-bound-step-3-wave}. Since \(\partial{}_u = (-\nu{})U\), the statement
\cref{gamma-bound-step-3-step-2} follows.

The statements
\begin{equation}\label{gamma-bound-step-3b-2}
\partial{}_u^n =_{\mathrm{s}} \mathcal{O}(U^{<n}(-\nu{}))UU^{<n}\qquad \lambda{},(-\nu{}),\partial{}_u\lambda{} =_{\mathrm{s}}\mathcal{O}(r^{-1}\varpi{},\kappa{},(-\gamma{}))
\end{equation}
follow from an induction argument starting from \(\partial_u = (-\nu{})U\) and from the
equation for \(\partial{}_u\lambda{}\), respectively. Apply
\cref{gamma-bound-step-3b-2} to \cref{gamma-bound-step-3-step-2} to get
\begin{equation}
\mathbf{1}_{r\ge R_0}\partial{}_v\partial{}_u^n\log (-\gamma{}) =_{\mathrm{s}} \mathfrak{b}_{U^{<n}}[r^{-1}UU^{<n}\varphi{}UU^{<n}\varphi{} + UU^{<n}\varphi{}\partial{}_v\varphi{} + (\partial{}_v\varphi{})^2].
\end{equation}
The desired \cref{gamma-bound-step-3b} now follows from the definition of the
pointwise norm.

\step{Step 3: Showing \(\lim_{v\to \infty}\partial{}_u^n\log (-\gamma{})(u,v)=0\).} Let \(n \ge 1\). We show
that a boundedness statement \(C(\mathfrak{b}_{U^{<n}})\mathcal{P}_{U^{<n},0}^2
< \infty\) implies the qualitative result \(\lim_{v\to \infty}\partial_u^{\le
n}\log (-\gamma{})(u,v) = 0\). The hypothesis together with Step 2 implies that
\(\mathbf{1}_{r\ge R_0}\abs{\partial_v\partial_u^{\le n}\log (-\gamma{})}\) is
integrable in \(v\) towards infinity (uniformly on compact subsets of \(u\)),
and so \(\partial{}_u^{\le n}\log (-\gamma{})(u,v)\) converges uniformly as
\(v\to \infty\) on compact subsets of \(u\). A basic fact from real analysis is
that if \(f : \R^2_{u,v}\to \R\) is such that \(\lim_{v\to \infty}f(u,v)\)
exists and \((\partial{}_uf)(u,v)\) converges uniformly as \(v\to \infty\), then
\(\partial{}_u\lim_{v\to \infty}f(u,v)\) exists and is equal to \(\lim_{v\to
\infty}\partial{}_uf(u,v)\). Since \(\lim_{v\to \infty}\log (-\gamma{})(u,v) =
0\), an induction argument completes the proof.

\step{Step 4: Completing the proof of \cref{gamma-bound-step-3}.} Let \(r(u,v)\ge R_0\).
Integrate \cref{gamma-bound-step-3b} in \(v\) on \([v,\infty)\). The boundary term
at infinity vanishes by Step 3c, so we conclude \cref{gamma-bound-step-3}.
\end{proof}
\begin{lemma}
For \(\abs{\alpha{}}\ge 1\), we have
\begin{equation}
\mathbf{1}_{r\ge R_0}L =_{\mathrm{s}} \mathbf{1}_{r\ge R_0}\mathcal{O}(\mathfrak{b}_{<\alpha{}},r^{-1}\mathfrak{G}_{<\alpha{}})[\partial{}_u^{\le \abs{\alpha{}}} + \overline{\partial}_r\Gamma{}^{\le \alpha{}-V} + S\Gamma{}^{\le \alpha{}-S}]
\end{equation}
\label{gamma-bound-step-2}
\end{lemma}
\begin{proof}
We will in fact show that for a a multi-index \(\alpha{}\) and a non-negative integer
\(k\), we can rewrite \(L\in \Gamma^\alpha{}\) as follows:
\begin{equation}\label{gamma-bound-step-2-desired}
\mathbf{1}_{r\ge R_0}L = \mathbf{1}_{r\ge R_0}\mathcal{O}(\mathfrak{b}_{\le \alpha{}-U},r^{-1}\mathfrak{G}_{\le \alpha{}-U})[\partial{}_u^{\le k}\Gamma{}^{\le \alpha{}-kU} + \overline{\partial}_r\Gamma{}^{\le \alpha{}-V} + S\Gamma{}^{\le \alpha{}-S}].
\end{equation}
Schematically write \(\text{good}_{\alpha{},k}\) for terms appearing on the right of
\cref{gamma-bound-step-2-desired} for the pair \((\alpha{},k)\). The strategy is to induct
on the pair \((\alpha{},k)\).

\step{Step 1: Preliminary observations and base case.} Observe that
\cref{gamma-bound-step-2-desired} for \((\alpha{},0)\) is trivial. Moreover, since
\(\Gamma^{\le \alpha{}-kU} = \set{0}\) when \(k\ge \abs{\alpha{}}\), we know
\cref{gamma-bound-step-2-desired} for \((\alpha{},\abs{\alpha{}})\) implies
\cref{gamma-bound-step-2-desired} for \((\alpha{},k)\) whenever \(k\ge \abs{\alpha{}}\).

We now prove \cref{gamma-bound-step-2-desired} for pairs \((\alpha{},1)\). Let \(L\in
\Gamma^\alpha{}\). Since \(L\) begins with either \(U\), \(V\), or \(S\), we
have
\begin{equation}
\begin{split}
L &=_{\mathrm{s}} U\Gamma{}^{\le \alpha{}-U} + V\Gamma{}^{\le \alpha{}-V} + S\Gamma{}^{\le \alpha{}-S}. \\
\end{split}
\end{equation}
After rewriting the first derivative that makes up \(L\) using
\cref{U-du-coordinate-change,V-dr-coordinate-change,S-in-terms-of-U-V}, we get
\begin{equation}
\mathbf{1}_{r\ge R_0}L =_{\mathrm{s}} \mathbf{1}_{r\ge R_0}\mathfrak{b}_0[\partial{}_u\Gamma{}^{\le \alpha{}-U} + \overline{\partial}_r\Gamma{}^{\le \alpha{}-V} + S\Gamma{}^{\le \alpha{}-S}] = \text{good}_{\alpha{},1}.
\end{equation}
as desired. We have used \(\abs{\alpha{}}\ge 1\) to conclude \(\mathfrak{b}_0=_{\mathrm{s}}\mathfrak{b}_{\le
\alpha{}-U}\).

\step{Step 2: Inductive step.} Now suppose \(\abs{\alpha{}}\ge 2\) and \(1\le k
<\abs{\alpha{}}\) and \cref{gamma-bound-step-2-desired} holds for pairs \((\beta{},j)\)
such that \(\beta{} < \alpha{}\) or such that \(\beta{} = \alpha{}\) and \(0\le
j \le k\). We will show that \cref{gamma-bound-step-2-desired} holds for the pair
\((\alpha{},k + 1)\). First, \cref{gamma-bound-step-2-desired} for the pair
\((\alpha{},k)\) gives
\begin{equation}
\begin{split}
\mathbf{1}_{r\ge R_0}L &=_{\mathrm{s}} \mathbf{1}_{r\ge R_0}\mathcal{O}(\mathfrak{b}_{\le \alpha{}-U},r^{-1}\mathfrak{G}_{\le \alpha{}-U})[\partial{}_u^{\le k}\Gamma{}^{\le \alpha{}-kU} + \overline{\partial}_r\Gamma{}^{\le \alpha{}-V} + S\Gamma{}^{\le \alpha{}-S}] \\
&=_{\mathrm{s}} \mathbf{1}_{r\ge R_0}\mathcal{O}(\mathfrak{b}_{\le \alpha{}-U},r^{-1}\mathfrak{G}_{\le \alpha{}-U})\partial{}_u^{\le k}\Gamma{}^{\le \alpha{}-kU}+ \text{good}_{\alpha{},k+1}.
\end{split}
\end{equation}
Now use \cref{gamma-bound-step-2-desired} for the pair \((\alpha{}-kU,1)\) to get
\begin{equation}
\begin{split}
\mathbf{1}_{r\ge R_0}L &=_{\mathrm{s}} \mathbf{1}_{r\ge R_0}\mathcal{O}(\mathfrak{b}_{\le \alpha{}-U},r^{-1}\mathfrak{G}_{\le \alpha{}-U})\\
&\qquad \partial{}_u^{\le k}[\mathcal{O}(\mathfrak{b}_{\le \alpha{}-(k+1)U},r^{-1}\mathfrak{G}_{\le \alpha{}-(k+1)U})[\partial{}_u^{\le 1}\Gamma{}^{\le \alpha{}-(k+1)U} + \overline{\partial}_r\Gamma{}^{\le \alpha{}-kU-V} + S\Gamma{}^{\le \alpha{}-kU-S}]] \\
&\qquad + \text{good}_{\alpha{},k+1} \\
&=_{\mathrm{s}} \mathbf{1}_{r\ge R_0}\mathcal{O}(\mathfrak{b}_{\le \alpha{}-U},r^{-1}\mathfrak{G}_{\le \alpha{}-U})(\partial{}_u^{\le k}\mathcal{O}(\mathfrak{b}_{\le \alpha{}-(k+1)U},r^{-1}\mathfrak{G}_{\le \alpha{}-(k+1)U}))[\partial{}_u^{\le k+1}\Gamma{}^{\le \alpha{}-(k+1)U} \\
&\qquad + \partial{}_u^{\le k}\overline{\partial}_r\Gamma{}^{\le \alpha{}-kU-V} + \partial{}_u^{\le k}S\Gamma{}^{\le \alpha{}-kU-S}] + \text{good}_{\alpha{},k+1}.
\end{split}
\end{equation}
Apply \(\mathbf{1}_{r\ge R_0}\partial{}_u^k =_{\mathrm{s}} \mathfrak{b}_{U^{<k}}U^{\le k}\) (see \cref{gamma-bound-step-3b-2}) and
\(\mathfrak{b}_{U^{<k}} =_{\mathrm{s}} \mathfrak{b}_{\le \alpha{}-U}\) (which holds since
\(k<\abs{\alpha{}}\)) to get
\begin{equation}\label{gamma-bound-step-2-a}
\begin{split}
\mathbf{1}_{r\ge R_0}L &=_{\mathrm{s}} \mathbf{1}_{r\ge R_0}\mathcal{O}(\mathfrak{b}_{\le \alpha{}-U},r^{-1}\mathfrak{G}_{\le \alpha{}-U})[U^{\le k}\overline{\partial}_r\Gamma{}^{\le \alpha{}-kU-V} + U^{\le k}S\Gamma{}^{\le \alpha{}-kU-S}] + \text{good}_{\alpha{},k+1} \\
&=_{\mathrm{s}} \mathbf{1}_{r\ge R_0}\mathcal{O}(\mathfrak{b}_{\le \alpha{}-U},r^{-1}\mathfrak{G}_{\le \alpha{}-U})[ [\overline{\partial}_r,U^{\le k}]\Gamma{}^{\le \alpha{}-kU-V} + [S,U^{\le k}]\Gamma{}^{\le \alpha{}-kU-S}] + \text{good}_{\alpha{},k+1}. \\
\end{split}
\end{equation}

We now compute the commutator terms associated to \(\overline{\partial}_r\) and to
\(S\). By \cref{gamma-bound-commutator}, we have
\begin{equation}\label{gamma-bound-step-2-dr}
\begin{split}
\mathbf{1}_{r\ge R_0}[\overline{\partial}_r,U^{\le k}]\Gamma{}^{\le \alpha{}-kU-V} &=_{\mathrm{s}} \mathcal{O}(\mathfrak{b}_{kU},r^{-1}\mathfrak{g}_{<kU})[\overline{\partial}_r\Gamma{}^{\le \alpha{}-V} + \Gamma{}^{<\alpha{}}] \\
&=_{\mathrm{s}} \mathcal{O}(\mathfrak{b}_{\le \alpha{}-U},r^{-1}\mathfrak{G}_{\le \alpha{}-U})[\overline{\partial}_r\Gamma{}^{\le \alpha{}-V} + \Gamma{}^{<\alpha{}}] =_{\mathrm{s}} \text{good}_{\alpha{},k+1}.
\end{split}
\end{equation}
We can pass to the second line because \(kU\le \alpha{}-U\) (since
\(\abs{\alpha{}}\ge k + 1\) and \((k + 1)U\) is the smallest multi-index of
order \(k+1\)). In the final identity we used \cref{gamma-bound-step-2-desired} for
pairs \((<\alpha{},k+1)\).

Next, one can easily prove by induction on \(k\), with base case \(k = 1\) given
by \cref{U-S-comm-1}, that
\begin{equation}
[S,U^k] =_{\mathrm{s}} \mathcal{O}(\mathfrak{b}_{\le kU},r^{-1}\mathfrak{G}_{\le kU})\Gamma{}^{\le (k-1)U+V}.
\end{equation}
It follows that
\begin{equation}\label{gamma-bound-step-2-S}
[S,U^k]\Gamma{}^{\le \alpha{}-kU-S} =_{\mathrm{s}} \mathcal{O}(\mathfrak{b}_{\le kU},r^{-1}\mathfrak{G}_{\le kU})\Gamma{}^{\le \alpha{}-U-S+V} =_{\mathrm{s}} \mathcal{O}(\mathfrak{b}_{\le \alpha{}-U},r^{-1}\mathfrak{G}_{\le \alpha{}-U})\Gamma{}^{<\alpha{}} =_{\mathrm{s}} \text{good}_{\alpha{},k+1}.
\end{equation}
In the second last identity we used \(\alpha{}\ge kU\) (since \(\abs{\alpha{}}\ge k + 1\)). In the
final identity we used \cref{gamma-bound-step-2-desired} for pairs \((<\alpha{},k+1)\).

Substitute \cref{gamma-bound-step-2-dr,gamma-bound-step-2-S} into
\cref{gamma-bound-step-2-a} to establish \cref{gamma-bound-step-2-desired} for
\((\alpha{},k + 1)\) and hence complete the induction.
\end{proof}
\begin{proposition}[Estimates for \(\Gamma{}^\alpha{}\log (-\gamma{})\)]
Let \(\abs{\alpha{}}\ge 0\). For \(s >0\) sufficiently small, we have
\begin{align}
r\tau{}\abs{\Gamma{}^{\le \alpha{}}\log (-\gamma{})}&\le C(\mathfrak{b}_\alpha{},r^{-s}\mathfrak{g}_\alpha{})[\mathcal{E}_{\alpha{},1}^2 + \mathcal{P}_{\alpha{},1}^2]&\quad \text{in }\set{r\ge R_0}, \label{gamma-bound-strong} \\
r^{1-s-\eta{}_0}\abs{\Gamma{}^{\le \alpha{}}\log (-\gamma{})}&\le C(\mathfrak{b}_{<\alpha{}},r^{-1}\mathfrak{G}_{<\alpha{}})[\mathcal{E}_{<\alpha{},1}^2 + \mathcal{P}_{<\alpha{},2-s}^2]&\quad \text{in }\set{r\ge R_0}. \label{gamma-bound-weak}
\end{align}
\label{Gamma-gamma-bounds}
\end{proposition}
\begin{proof}
We first explain the logic of the proof. Consider the following auxiliary
qualitative decay statement for a multi-index \(\alpha{}\) with \(\abs{\alpha{}}\ge 0\):
\begin{equation}\label{gamma-qualitative-decay}
\lim_{v\to \infty} L \log (-\gamma{})(u,v) = 0\text{ for }L\in \Gamma^{\le \alpha{}} \text{ and }u\in [1,\infty).
\end{equation}
To prove the quantitative decay estimates \cref{gamma-bound-weak,gamma-bound-strong}, we
establish the following implications:
\begin{equation*}
\cref{gamma-qualitative-decay}_{<\alpha{}} \implies  \cref{gamma-bound-strong}_{<\alpha{}}  \implies \cref{gamma-bound-weak}_\alpha{} \implies \cref{gamma-qualitative-decay}_{\alpha{}}.
\end{equation*}
That is, \cref{gamma-qualitative-decay} for multi-indices \(< \alpha{}\) implies
\cref{gamma-bound-strong} for the multi-index \(\alpha{}\), and so on. Since
\cref{gamma-qualitative-decay} holds when \(\alpha{} = 0\) by the normalization at
\(\mathcal{I}\) of the \(u\)-coordinate, an induction argument establishes
\cref{gamma-bound-weak,gamma-bound-strong} for all multi-indices \(\alpha{}\). The final
implication is clear, so we prove the first two implications.

\step{Step 1: Proof that \(\cref{gamma-qualitative-decay}_{\alpha{}}\) implies
\(\cref{gamma-bound-strong}_{\alpha{}}\).} Let \(\abs{\alpha{}}\ge 0\) and let \(L\in \Gamma^\alpha{}\). By
\cref{gamma-bound-step-1c} and the definition of the pointwise norm, we have
\begin{equation}\label{gamma-bound-step-4-dr-L}
\mathbf{1}_{r\ge R_0}\abs{\overline{\partial}_rL\log (-\gamma{})}\le C(\mathcal{G}_{\alpha{},s})[r\abs{\overline{\partial}_r\Gamma{}^{\le \alpha{}}\varphi{}}^2 + r^{-2}\tau{}^{-1}\mathcal{P}_{\alpha{},1}^2 + r^{-2+s}\abs{\Gamma{}^{\le \alpha{}}\log (-\gamma{})}].
\end{equation}
Fix \((u,v)\) such
that \(r(u,v)\ge R_0\) and integrate \cref{gamma-bound-step-4-dr-L} to get
\begin{equation}
\begin{split}
\abs{L\log (-\gamma{})}(u,v)&\le \limsup_{v'\to \infty}\abs{L\log (-\gamma{})}(u,v') + \int_v^\infty \abs{\overline{\partial}_rL\log (-\gamma{})}(u,v')\dd{}v' \\
&\le C(\mathcal{G}_{\alpha{},s})\Bigl[\int_v^\infty r\abs{\overline{\partial}_r\Gamma{}^{\le \alpha{}}\varphi{}}^2(u,v')\dd{}v' + \mathcal{P}_{\alpha{},1}^2\int_v^\infty r^{-2}\tau{}^{-1}(u,v')\dd{}v' \\
&\qquad + \int_v^\infty r^{-2+s}\abs{\Gamma{}^{\le \alpha{}}\log (-\gamma{})}(u,v')\dd{}v'\Bigr] \\
&\le r^{-1}\tau{}^{-1}(u,v)C(\mathcal{G}_{\alpha{},s})(\mathcal{E}_{\alpha{},1}^2 + \mathcal{P}_{\alpha{},1}^2) + C(\mathcal{G}_{\alpha{},s})\int_v^\infty r^{-2+s}\abs{\Gamma{}^{\le \alpha{}}\log (-\gamma{})}(u,v')\dd{}v'.
\end{split}
\end{equation}
Sum over \(L\in \Gamma^\alpha{}\), take \(s < 1\), and use the Grönwall inequality in
\cref{decaying-gronwall} (which applies because of the assumed
\(\cref{gamma-qualitative-decay}_{\alpha{}}\)) to conclude
\begin{equation}
r\tau{}\abs{L\log (-\gamma{})}\le C(\mathcal{G}_{\alpha{},s})(\mathcal{E}_{\alpha{},1}^2 + \mathcal{P}_{\alpha{},1}^2).
\end{equation}

\step{Step 2: Proof that \(\cref{gamma-bound-strong}_{<\alpha{}}\) implies
\(\cref{gamma-bound-weak}_\alpha{}\).} To ease the notation, all estimates in
this section will be done in the region \(\set{r\ge R_0}\); that is, all
estimates should be understood to have an implicit \(\mathbf{1}_{r\ge R_0}\) at
the front of each term. It is enough to prove that
\begin{equation}\label{gamma-bound-step-4-L}
\abs{L\log (-\gamma{})}\le r^{-1+s+\eta{}_0}C(\mathfrak{b}_{<\alpha{}},r^{-1}\mathfrak{G}_{<\alpha{}})[\mathcal{P}_{<\alpha{},2-s}^2 + r\tau{}\abs{\Gamma{}^{<\alpha{}}\log (-\gamma{})}],
\end{equation}
since the final term can be handled with \(\cref{gamma-bound-strong}_{<\alpha{}}\).
First, \cref{gamma-bound-step-2} implies
\begin{equation}\label{gamma-bound-step-4a}
\abs{L\log (-\gamma{})}\le C(\mathfrak{b}_{<\alpha{}},r^{-1}\mathfrak{G}_{<\alpha{}})[\abs{\partial{}_u^{\le \abs{\alpha{}}}\log (-\gamma{})} + \abs{\overline{\partial}_r\Gamma{}^{\le \alpha{}-V}\log (-\gamma{})} + \abs{S\Gamma{}^{\le \alpha{}-S}\log (-\gamma{})}].
\end{equation}
Observe that the following three estimates together with \cref{gamma-bound-step-4a}
imply \cref{gamma-bound-step-4-L}:
\begin{align}
\abs{\partial{}_u^{\le \abs{\alpha{}}}\log (-\gamma{})}&\le r^{-2}C(\mathfrak{b}_{<\alpha{}})\mathcal{P}_{<\alpha{},0}^2 \label{gamma-bound-4-U} \\
\abs{\partial{}_r\Gamma{}^{\le \alpha{}-V}}&\le C(\mathfrak{b}_{<\alpha{}})[r^{-3}\mathcal{P}_{<\alpha{},0}^2 + r^{-2}\tau{}^{-1}\mathcal{P}_{<\alpha{},1}^2 +
r^{-2}\abs{\Gamma{}^{<\alpha{}}\log (-\gamma{})}]  \label{gamma-bound-4-V} \\
\abs{S \Gamma{}^{\le \alpha{}-S}\log (-\gamma{})}&\le C(\mathfrak{b}_{<\alpha{}})[r^{-1+s+\eta{}_0}\mathcal{P}_{<\alpha{},2-s}^2 + \tau{}\abs{\Gamma{}^{<\alpha{}}\log (-\gamma{})}]. \label{gamma-bound-4-S}
\end{align}
Now \cref{gamma-bound-4-U} and \cref{gamma-bound-4-V} follow from
\cref{gamma-bound-step-3} and \cref{gamma-bound-step-1c}, respectively. We now
establish \cref{gamma-bound-4-S}. Rewrite \(S\) in \(\set{r\ge R_0}\) using
\cref{v-u-r-compare,V-dr-coordinate-change,S-in-terms-of-U-V}, and then use
\cref{gamma-bound-step-1c}:
\begin{equation}\label{gamma-bound-step-4c}
\begin{split}
\abs{S \Gamma{}^{\le \alpha{}-S}\log (-\gamma{})}&\le C(\mathfrak{b}_0)[v\abs{\overline{\partial}_r\Gamma{}^{\le \alpha{}-S}\log (-\gamma{})} + \tau{}\abs{U\Gamma{}^{\le \alpha{}-S}\log (-\gamma{})}] \\
&\le C(\mathfrak{b}_{<\alpha{}})[rv\abs{\overline{\partial}_r\Gamma{}^{\le \alpha{}-S}\varphi{}}^2 + r^{-3+2s}v\abs{\Gamma{}^{\le \alpha{}-S}\varphi{}}^2 + \tau{}\abs{\Gamma{}^{<\alpha{}}\log (-\gamma{})}] \\
&\le C(\mathfrak{b}_{<\alpha{}})[r^{-1+s+\eta{}_0}\mathcal{P}_{<\alpha{},2-s}^2 + \tau{}\abs{\Gamma{}^{<\alpha{}}\log (-\gamma{})}].
\end{split}
\end{equation}
In the last line, we used \(\abs{\overline{\partial}_r\Gamma^{\le \alpha{}-S}\varphi{}}^2\le
\mathbf{1}_{r\le v/2}r^{-2+s+\eta{}_0}\tau^{-1}v\mathcal{P}_{<\alpha{},2-s}^2 + \mathbf{1}_{r\ge v/2}r^{-4}\mathcal{P}_{<\alpha{},0}^2\).
\end{proof}
\section{Putting it all together: proof of \cref{main-theorem}}
\label{sec:org374b5f4}
\label{sec:putting-it-all-together} Our main theorem (\cref{main-theorem}, see also the
rough versions \cref{main-theorem-rough} and \cref{main-theorem-double-null}) follows
immediately from \cref{S-decay,S-vdv-comparison}.
\subsection{Control of \texorpdfstring{\(S^k\varphi\)}{S-derivatives of the scalar field} along the horizon}
\label{sec:orgf46599b}
\begin{proposition}
For \(N\ge 0\) and \(\epsilon{} > 0\), we have
\begin{equation}\label{conc:goal-2}
\abs{S^N\varphi{}}|_{\mathcal{H}}\le  C(\epsilon{},N,\mathfrak{D}_{N},\varpi{}_i,c_{\mathcal{H}},r_{\mathrm{min}})\mathfrak{D}_Nv^{-1+\epsilon{}}.
\end{equation}
\label{S-decay}
\end{proposition}
\begin{proof}
In view of \cref{v-u-r-compare}, it is enough to show that
\begin{equation}\label{conc:goal}
\mathcal{P}_{NS,2-C_{NS}'\eta{}_0}\le C(N,\eta{}_0,\mathfrak{D}_N,\varpi{}_i,c_{\mathcal{H}},r_{\text{min}})\mathfrak{D}_{N}
\end{equation}
when \(\eta_0\) is sufficiently small based on \(N\) and then fix \(\eta_0\) small
based on \(\epsilon{}\) and \(C_{NS}'=C(N)\).

We now establish \cref{conc:goal}. The results of
\cref{boundedness-initial-data,E-alpha-bound,P-alpha-bound,geometric-quantities-bound}
imply that we can fix \(\Rc\ge C(\mathfrak{B}_0^{\circ },\mathfrak{g}_0,N)\)
large enough that for \(\abs{\alpha{}}\le N\) and \(s\) sufficiently small
depending on \(N\), we have
\begin{align}
\mathcal{D}_\alpha{} &\le C(\mathfrak{b}_{<\alpha{}},r^{-1}\mathfrak{g}_{<\alpha{}})\mathfrak{D}_{\abs{\alpha{}}}, \\
C(\mathfrak{b}_\alpha{})&\le C(N,\mathfrak{B}_0,\mathfrak{G}_0,\mathcal{E}_{<\alpha{},1},\mathcal{P}_{<\alpha{},2-s+\eta{}_0}), \\
C(r^{-s}\mathfrak{g}_\alpha{})&\le C(N,\mathfrak{B}_0,\mathfrak{G}_0,\mathcal{E}_{<\alpha{},1},\mathcal{P}_{<\alpha{},2-s+\eta{}_0}), \\
\mathcal{E}_\alpha{} &\le C(\mathfrak{b}_\alpha{},r^{-s}\mathfrak{g}_\alpha{},\mathcal{P}_{<\alpha{},1+\eta{}_0})[\mathcal{D}_\alpha + \mathcal{E}_{<\alpha{},4s}], \\
\mathcal{P}_{\alpha{},p}&\le C(\mathfrak{b}_\alpha{},r^{-s}\mathfrak{g}_\alpha{})(\mathcal{D}_\alpha{} + \mathcal{E}_{\alpha{},p}) \\
\mathcal{E}_{\alpha{},2-\eta{}_0-C_\alpha{}s} &\le C(\mathfrak{b}_\alpha{},r^{-s}\mathfrak{g}_\alpha{},\mathcal{P}_{\alpha{},0},\mathcal{P}_{<\alpha{},1+\eta{}_0})\mathcal{D}_\alpha{}
\end{align}
for explicit constants \(C_\alpha{}\). In particular, the above equations imply
\begin{equation}
\mathcal{E}_{\alpha{},2-\eta{}_0-C_\alpha{}(s+\eta{}_0)} \le C(N,\mathfrak{D}_N,\mathfrak{B}_0,\mathfrak{G}_0,\mathcal{E}_{<\alpha{},2-s})\mathfrak{D}_{N}
\end{equation}
for \(\abs{\alpha{}}\le N\). By an induction argument (starting with \(s = \eta_0\) when \(\alpha{} =
0\)), there are constants \(C_\alpha'\) depending only on \(\alpha{}\) such that if \(\eta_0\)
is small enough depending on \(N\), we have
\begin{equation}\label{conc:energy-bound}
\mathcal{E}_{\alpha{},2-C_\alpha'\eta{}_0}\le C(N,\mathfrak{D}_N,\mathfrak{B}_0,\mathfrak{G}_0,\mathcal{E}_{0,2-\eta{}_0})\mathfrak{D}_{N}
\end{equation}
Substituting \cref{conc:energy-bound} into the above equations (and performing
another induction argument) gives
\begin{equation}\label{conc:goal-1}
\mathcal{P}_{NS,2-C_{NS}'\eta{}_0}\le C(N,\mathfrak{D}_N,\mathfrak{B}_0,\mathfrak{G}_0,\mathcal{E}_{0,2-\eta{}_0})\mathfrak{D}_{N}.
\end{equation}
The zeroth order estimates, due to
\cref{zeroth-order-schematic-control,boundedness-initial-data,E-alpha-bound,P-alpha-bound,geometric-quantities-bound},
are
\begin{align}
\mathcal{D}_0 &= \mathfrak{D}_0, \\
\mathcal{E}_0 &\le  C(\varpi{}_i,c_{\mathcal{H}},r_{\text{min}})\mathcal{D}_0, \\
\mathfrak{b}_0 &\le  C(\mathcal{E}_0,\varpi{}_i,c_{\mathcal{H}},r_{\text{min}},R_0,\eta{}_0), \\
\mathcal{P}_{0,0}&\le C(\mathfrak{b}_0)[\mathcal{D}_0 + \mathcal{E}_{0}], \\
\mathcal{E}_{0,2-\eta{}_0} &\le C(\mathfrak{b}_0,\mathcal{P}_{0,0})\mathcal{D}_0, \\
\mathcal{P}_{0,2-2\eta{}_0} &\le C(\mathfrak{b}_0)[\mathcal{D}_0 + \mathcal{E}_{0,2-2\eta{}_0}], \\
\mathcal{B}_0^{\circ }&\le C(\mathfrak{b}_0,\mathcal{E}_{0,1}), \\
\mathfrak{B}_0&\le C(\Rc,\mathfrak{B}_0^{\circ },\mathcal{E}_{0,1+\eta{}_0}), \\
\mathfrak{g}_0&\le C(\mathfrak{b}_0), \\
\mathfrak{G}_0&\le C(\mathfrak{b}_0).
\end{align}
Noting the value of \(\Rc\) that we have fixed and the value of \(R_0 =
R_0(\varpi_i)\) that we have fixed (in \cref{zeroth-order-geometric-bounds}),
chaining the zeroth order estimates with \cref{conc:goal-1} (for \(\alpha{} = NS\))
gives \cref{conc:goal}. Another induction argument in fact shows that all the
schematic quantities are controlled by
\(C(\epsilon{},N,\mathfrak{D}_{N},\varpi{}_i,c_{\mathcal{H}},r_{\mathrm{min}})\).
\end{proof}
\subsection{Control of \texorpdfstring{\((\bar{v}\partial_{\bar{v}})^k\varphi{}\)}{v-derivatives of the scalar field in the Eddington--Finkelstein-type gauge} along the horizon}
\label{sec:org52fa9b2}
\label{sec:Stovdv-translation} The goal of this section is to establish the following result comparing \(S^k\)
(in the constant-\(r\) normalized gauge)and \((\bar{v}\partial_{\bar{v}})^k\) (in the Eddington--Finkelstein-type gauge
normalized on the horizon).
\begin{proposition}
For \(0\le k\le 4\), we have
\begin{equation}
\abs{S^k\varphi{} - (\bar{v}\partial{}_{\bar{v}})^k\varphi{}}\le C(\epsilon{},k,\varpi{}_i,r_{\mathrm{min}},c_{\mathcal{H}},\mathfrak{D}_k)v^{-1+\epsilon{}}.
\end{equation}
\label{S-vdv-comparison}
\end{proposition}
In this section, we allow all constants to depend on \(\varpi_i\),
\(c_{\mathcal{H}}\), and \(r_{\text{min}}\). We will also freely apply the
results obtained in (the proof of) \cref{S-decay}, namely that our schematic energy
and geometric quantities of order \(\abs{\alpha{}} = k\) are controlled by the
data quantity \(\mathfrak{D}_k\). We first introduce an integration lemma.
\begin{lemma}
For \(\alpha{}\ge 0\), \(c>0\), and \(v\ge 1\), we have
\begin{equation}
\int_v^\infty e^{-c(v'-v)}v'^{-\alpha{}}\dd{}v'\le  c^{-1}v^{-\alpha{}}.
\end{equation}
\label{S-vdv-integration-lemma}
\end{lemma}
\begin{proof}
We have
\begin{equation}
\int_v^\infty e^{-c(v'-v)}v'^{-\alpha{}}\dd{}v'\le v^{-\alpha{}}\int_v^\infty e^{-c(v'-v)}\dd{}v' = c^{-1}v^{-\alpha{}}.
\end{equation}
\end{proof}
To estimate \(\lambda{}\) along the horizon, we use an ODE argument similar to the one
used in \cite[Prop.~B.1]{luk-oh-scc1} (albeit in a different gauge).
\begin{lemma}[Decay of \(\lambda\) along the horizon]
We have
\begin{equation}
\lambda{}|_{\mathcal{H}}\le C(\epsilon{},\varpi{}_i,r_{\mathrm{min}},c_{\mathcal{H}},\mathfrak{D}_1)v^{-4+\epsilon{}}.
\end{equation}
\label{lambda-horizon-bound}
\end{lemma}
\begin{proof}
Since \(\lambda{} = \kappa{}(1-\mu{})\le C(\mathfrak{b}_0)(1-\mu{})\), it is enough to obtain the claimed
estimate for \((1-\mu{}){|_{\mathcal{H}}}\) in place of
\(\lambda{}|_{\mathcal{H}}\).

\step{Step 1: Estimate for \((1-\mu{})|_{\mathcal{H}}\).} The goal of this step is to use
the transport equation for \(1-\mu{}\) along the horizon to obtain
\begin{equation}\label{mu-horizon-bound-1}
(1-\mu{})|_{\mathcal{H}}(v) \le  \int_v^\infty e^{-c_\ast{}(v'-v)}r(\partial{}_v\varphi{})^2(v')\dd{}v'
\end{equation}
for a constant \(c_\ast{} > 0\). The equation for \((1-\mu{})\) (see
\cref{mu-definition,sph-sym-equations-1}) gives
\begin{equation}
\partial{}_v(1-\mu{}) - \frac{2(\varpi{}-\mathbf{e}^2/r)}{r^2}\kappa{}(1-\mu{}) = -\frac{r}{\kappa{}}(\partial{}_v\varphi{})^2.
\end{equation}
The method of integrating factors now gives
\begin{equation}\label{mu-horizon-equation}
\begin{split}
(1-\mu{})|_{\mathcal{H}}(v) &= \lim_{v'\to \infty}\exp \Bigl(\int_v^{v'}  - \frac{2(\varpi{}-\mathbf{e}^2/r)}{r^2}\kappa{}\dd{}v''\Bigr)(1-\mu{})|_{\mathcal{H}}(v') \\
&\qquad  + \int_v^\infty \exp \Bigl(-\int_v^{v'} \frac{2(\varpi{}-\mathbf{e}^2/r)}{r^2}\kappa{}\dd{}v''\Bigr)r(\partial{}_v\varphi{})^2(v')\dd{}v'.
\end{split}
\end{equation}
By \cref{mass-redshift,kappa-estimate}, we have
\begin{equation}\label{mu-horizon-2}
\exp \Bigl(-\int_v^{v'} \frac{2(\varpi{}-\mathbf{e}^2/r)}{r^2}\kappa{}\dd{}v''\Bigr)\le e^{c_\ast{}(v'-v)},
\end{equation}
where \(c_\ast{}\) is a constant depending on \(\varpi{}_i\), \(c_{\mathcal{H}}\), and
\(\mathcal{E}_0\). Since \(\lim_{v\to \infty}(1-\mu{})|_{\mathcal{H}}(v)=0\)
(see \cite[Prop.~B.1]{luk-oh-scc1}), it follows from \cref{mu-horizon-2} that the first term on the right
side of \cref{mu-horizon-equation} vanishes. Noting this observation and
\cref{mu-horizon-2}, return to \cref{mu-horizon-equation} to conclude the proof of \cref{mu-horizon-bound-1}.

\step{Step 2: Preliminary decay for \((1-\mu{})|_{\mathcal{H}}\).} In this step we show
that
\begin{equation}\label{mu-preliminary-decay}
(1-\mu{})|_{\mathcal{H}}(v)\le C(\epsilon{},\mathfrak{D}_0)v^{-2+\epsilon{}}.
\end{equation}
Control the exponential in \cref{mu-horizon-bound-1} by \(1\) and use the decay of
the energy to obtain
\begin{equation}\label{mu-horizon-preliminary-decay}
(1-\mu{})|_{\mathcal{H}}(v) \le  \int_v^\infty r(\partial{}_v\varphi{})^2(v')\dd{}v'\le C(\mathfrak{b}_0) \int_v^\infty \frac{r^2}{\kappa{}}(\partial{}_v\varphi{})^2(v')\dd{}v'\le \tau{}^{-p}(\infty,v)C(\mathfrak{b}_0)\mathcal{E}_{0,p}\le C(p,\mathfrak{D}_0)v^{-p}
\end{equation}
for \(p < 2\).

\step{Step 3: Strong decay for \((1-\mu{})|_{\mathcal{H}}\).} We now show that
\begin{equation}
(1-\mu{})|_{\mathcal{H}}(v)\le C(\epsilon{},\mathfrak{D}_1)v^{-4+\epsilon{}}.
\end{equation}
Use \(S = v\partial{}_v + v\lambda{}U\) and \cref{mu-preliminary-decay} to obtain
\begin{equation}\label{mu-horizon-3}
\abs{\partial{}_v\varphi{}}|_{\mathcal{H}}\le v^{-1}\abs{S\varphi{}}|_{\mathcal{H}} + \kappa{}(1-\mu{})|_{\mathcal{H}}\abs{U\varphi{}}|_{\mathcal{H}}\le \mathcal{P}_{S,2-\epsilon{}}v^{-2+\epsilon{}} + C(\epsilon{},\mathfrak{b}_0)\mathcal{P}_{U,2-\epsilon{}}v^{-3+\epsilon{}}\le C(\epsilon{},\mathfrak{D}_1)v^{-2+\epsilon{}}.
\end{equation}
Substitute \cref{mu-horizon-3} into \cref{mu-horizon-bound-1} and use
\cref{S-vdv-integration-lemma} to conclude the proof.
\end{proof}
To estimate \(S^n\lambda{}\) along the horizon, we use the relation \(\lambda{} = \kappa{}(1-\mu{})\) and
estimate separately \(S^n\kappa{}\) and \(S^n\mu{}\).
\begin{lemma}[Decay for \(S^n\lambda\) along the horizon]
Let \(n\ge 1\). We have
\begin{equation}
\abs{S^n\lambda{}}|_{\mathcal{H}}\le C(\epsilon{},\varpi{}_i,r_{\mathrm{min}},c_{\mathcal{H}},\mathfrak{D}_n)v^{-3+\epsilon{}}.
\end{equation}
\label{S-lambda-horizon-decay}
\end{lemma}
\begin{proof}
Suppose for the sake of induction that the lemma has been proven for \(0\le m <
n\). This holds for \(n = 1\) by \cref{lambda-horizon-bound}.

\step{Step 1: Decay for \((S^n\kappa{})|_{\mathcal{H}}\).} We first show that for \(n\ge 1\) we have
\begin{equation}\label{S-kappa-horizon-decay}
\abs{S^n\kappa{}}|_{\mathcal{H}}\lesssim C(\epsilon{},\mathfrak{D}_n)v^{-5+\epsilon{}}.
\end{equation}
Write \(r_{\mathcal{H}}(v)
\coloneqq{} r(u=\infty,v)\) (so that \(r_{\mathcal{H}} =
r_{\mathcal{H}}(\infty)\)). The strategy is to integrate the equation for
\(US^n\log \kappa{}\) to \(\set{r=r_{\mathcal{H}}}\) and use the decay of
\(\abs{r_{\mathcal{H}} - r_{\mathcal{H}}(v)}\).

\step{Step 1a: Decay of \(\abs{r_{\mathcal{H}} - r_{\mathcal{H}}(v)}\).} By
\cref{lambda-horizon-bound}, we have
\begin{equation}\label{r-horizon-decay}
\abs{r_{\mathcal{H}} - r_{\mathcal{H}}(v)} = \int_v^\infty \partial{}_vr\dd{}v\lesssim C(\mathfrak{D}_1)\int_v^\infty v^{-4+\epsilon{}}\dd{}v\lesssim C(\mathfrak{D}_1)v^{-3+\epsilon{}}.
\end{equation}

\step{Step 1b: Decay of \(S^n\kappa{}\).} By a computation using the chain rule, \(S^n\log \kappa{} -
\kappa{}^{-1}S^n\kappa{}\) is a polynomial in \(\kappa{}^{-1}S^i\kappa{}\) for
\(1\le i<n\) with no constant term, so it is enough to obtain decay for
\(S^n\log\kappa{}\). Since \(U\) and \(S\) commute near the horizon, we have
\begin{equation}\label{S-kappa-U-equation}
US^n\log \kappa{} = S^nU\log \kappa{} = S^n(r(U\varphi{})^2) =_{\mathrm{s}} rUS^{\le n}\varphi{}US^{<n}\varphi{},
\end{equation}
for \(n\ge 0\) and so
\begin{equation}
\sup_{r\in [r_{\mathcal{H}}(v),r_{\mathcal{H}}]}\abs{US^n\log \kappa{}}(r,v)\le C(\mathfrak{b}_0)v^{-2+\epsilon{}}\mathcal{P}_{nS,2-\epsilon{}}^2\le C(\mathfrak{D}_n)v^{-2+\epsilon{}}.
\end{equation}
Since \(S^n\log \kappa{}|_{\set{r=r_{\mathcal{H}}}} = 0\), we can integrate
\cref{S-kappa-U-equation} to get
\begin{equation}\label{S-log-kappa-decay}
S^n\log \kappa{}|_{\mathcal{H}}(v)\le \abs{r_{\mathcal{H}} - r_{\mathcal{H}}(v)}\sup_{r\in [r_{\mathcal{H}}(v),r_{\mathcal{H}}]}\abs{US^n\log \kappa{}}(r,v) \lesssim C(\mathfrak{D}_n)v^{-5+\epsilon{}}
\end{equation}
for \(n\ge 0\).

\step{Step 2: Decay for \(S^n\mu{}\).} In this step we show that for \(n\ge 1\) we have
\begin{equation}
\abs{S^n\mu{}}|_{\mathcal{H}}\le C(\epsilon{},\mathfrak{D}_n)v^{-3+\epsilon{}},
\end{equation}
By \cref{mu-definition} and \(Sr = 0\) near the
horizon, we have \(S\mu{} = 2r^{-1}S\varpi{}\) near the horizon, so it is enough to prove
the estimate for \(S\varpi{}\). Using the
transport equation for \(\varpi{}\), we compute
\begin{equation}
S\varpi{} = \frac{1}{2}r^2v\Bigl[\frac{1}{\kappa{}}(\partial_v\varphi{})^2 + \lambda{}(1-\mu{})(U\varphi{})^2\Bigr] =_{\mathrm{s}} r^2v\mathfrak{b}_0[\set{v^{-1}S\varphi{},\lambda{}U\varphi{}}^2].
\end{equation}
For \(n\ge 1\), it follows that
\begin{equation}
S^n\varpi{} =_{\mathrm{s}} r^2v\mathfrak{b}_{(n-1)S}\set{v^{-1}S^{\le n}\varphi{},S^{<n}\lambda{}US^{<n}\varphi{}}^2.
\end{equation}
By the inductive hypothesis and the control on the geometric and pointwise
quantities, we have
\begin{equation}
\abs{S^n\varpi{}}|_{\mathcal{H}}\lesssim C(\epsilon{},\mathfrak{D}_n)v^{-3+\epsilon{}},
\end{equation}
and so

\step{Step 3: Closing the induction.} We are done by the relation \(\lambda{} = \kappa{}(1-\mu{})\), the
results of Steps 1 and 2, and the product rule.
\end{proof}
\begin{proof}[Proof of \cref{S-vdv-comparison}]
The claim follows from \cref{S-vdv,S-vdv-2,S-vdv-3,S-vdv-4} below
together with \cref{S-rho-bound,S-lambda-horizon-decay,lambda-horizon-bound,S-decay}.

\step{Step 1: Schematic expressions for \(S^k - (\bar{v}\partial_{\bar{v}})^k\) for \(1\le
k\le 4\).} Define a function of \(v\) alone by
\begin{equation}
\rho{} \coloneqq{} 1 - \frac{\bar{v}}{v}\frac{1}{\kappa{}|_{\mathcal{H}}}.
\end{equation}
We show that
\begin{align}
S - \bar{v}\partial_{\bar{v}} &=  \rho{}S + (1-\rho{})v\lambda{}U, \label{S-vdv}\\
S^2 - (\bar{v}\partial{}_{\bar{v}})^2 &=_{\mathrm{s}} (S^{\le 1}\rho{})^{\le 1}(S^{\le 1}(v\lambda{}))^{\le 1}(v\lambda{}U(v\lambda{}))^{\le 1}\set{U,S}^{1\le 2}, \label{S-vdv-2} \\
S^3 - (\bar{v}\partial{}_{\bar{v}})^3 &=_{\mathrm{s}} (S^{\le 2}\rho{})^{\le 3}(S^{\le 2}(v\lambda{}))^{\le 2}(S^{\le 1}(v\lambda{})S^{\le 1}U(v\lambda{}))^{\le 1}(v\lambda{}(U^{\le 2}(v\lambda{}))^{\le 2})^{\le 2}\set{U,S}^{1\le 3}, \label{S-vdv-3} \\
S^4 - (\bar{v}\partial{}_{\bar{v}})^4 &=_{\mathrm{s}} (S^{\le 3}\rho{})^{\le 4}(S^{\le 3}(v\lambda{}))^{\le 3}(S^{\le 2}(v\lambda{})(S^{\le 2}U(v\lambda{}))^{\le 2})^{\le 2} \nonumber\\
&\qquad (v\lambda{}(S^{\le a}U^{\le b}(v\lambda{}))^{\le 3})^{\le 3}|_{a\le 2,b\le 3,a+b\le 4}\set{U,S}^{1\le 4}. \label{S-vdv-4}
\end{align}
Indeed, we have
\begin{equation}
S = v\partial_v + v\lambda{}U,\quad\text{and}\quad  \bar{v}\partial{}_{\bar{v}} = \frac{\bar{v}}{\kappa{}|_{\mathcal{H}}}\partial{}_v.
\end{equation}
It follows that
\begin{equation}\label{S-vdv-1-prep}
S - \bar{v}\partial_{\bar{v}} = \Bigl(v - \frac{\bar{v}}{\kappa{}|_{\mathcal{H}}}\Bigr)\partial_v + v\lambda{}U = \Bigl(1 - \frac{\bar{v}}{v} \frac{1}{\kappa{}|_{\mathcal{H}}}\Bigr)(S - v\lambda{}U) + v\lambda{}U = \rho{}S + (1-\rho{})v\lambda{}U,
\end{equation}
which is \cref{S-vdv}. We now explain how to compute \cref{S-vdv-2}. From
\cref{S-vdv-1-prep}, one obtains
\begin{equation}
S^2 - S(\bar{v}\partial{}_{\bar{v}}) =_{\mathrm{s}} \sum_{\substack{a,b\le 1 \\ a + b \ge 1}}(S^{\le 1}\rho{})^a(S^{\le 1}(v\lambda{}))^b\set{U,S}^{1\le 2}
\end{equation}
and
\begin{equation}\label{S-vdv-1-prep-2}
\begin{split}
S(\bar{v}\partial{}_{\bar{v}}) - (\bar{v}\partial{}_{\bar{v}})^2 &=_{\mathrm{s}} (\rho{}S + (1-\rho{})v\lambda{}U)((1-\rho{})S + (1-\rho{})v\lambda{}U) \\
&=_{\mathrm{s}} \sum_{\substack{a\le 2,b,c\le 1 \\ a + b + c\ge 1}}(S^{\le 1}\rho{})^{a}(S^{\le 1}(v\lambda{}))^{b}(v\lambda{}U(v\lambda{}))^{c}\set{U,S}^{1\le 2},
\end{split}
\end{equation}
where in \cref{S-vdv-1-prep-2} we used the already computed expression for \(S -
\bar{v}\partial_{\bar{v}}\). Adding \cref{S-vdv-1-prep,S-vdv-1-prep-2} gives
\cref{S-vdv-2}. One computes \cref{S-vdv-3,S-vdv-4} similarly.

\step{Step 2: Estimate for \(S^k(\kappa{}|_{\mathcal{H}})\) when \(1\le k\le 3\).} We show that
\begin{equation}\label{S-kappa-on-the-horizon}
S^k\kappa{}|_{\mathcal{H}}\le C(\epsilon{},\mathfrak{D}_3)v^{-2+\epsilon{}}\quad \text{for }1\le k\le 3.
\end{equation}
This follows from
\cref{skappa-h-1,SU-lambda,SU-kappa,skappa-h-1,skappa-h-2,skappa-h-3}, which are
proven below.

\step{Step 2a: Schematic expressions for \(S^k(\kappa{}|_{\mathcal{H}})\) and \(1\le k\le 3\).} Since \(\kappa{}|_{\mathcal{H}}\) is a function of \(v\) alone, we have
\(U(\kappa{}|_{\mathcal{H}}) = 0\). It follows that
\begin{equation}\label{skappa-h-1}
S^k(\kappa{}|_{\mathcal{H}}) = (v\partial{}_v)^k\kappa{}|_{\mathcal{H}} = ((v\partial{}_v)^k\kappa{})|_{\mathcal{H}} = ((S-v\lambda{}U)^k\kappa{})|_{\mathcal{H}}.
\end{equation}
We compute
\begin{equation}\label{skappa-h-2}
(S - v\lambda{}U)^2 = S^2 + S(v\lambda{})U + v\lambda{}SU + (v\lambda{})^2U^2 + v\lambda{}U(v\lambda{})U =_{\mathrm{s}} (S^{\le 1}(v\lambda{}))^{\le 2}(v\lambda{}U(v\lambda{}))^{\le 1}\set{U,S}^{1\le 2}
\end{equation}
and
\begin{equation}\label{skappa-h-3}
(S - v\lambda{}U)^3 =_{\mathrm{s}} (S^{\le 2}(v\lambda{}))^{\le 3}(S^{\le 1}(v\lambda{})(S^{\le 1}U^{\le 2}(v\lambda{}))^{\le 2})^{\le 2}\set{U,S}^{1\le 3}.
\end{equation}

\step{Step 2b: Estimate for derivatives of \(\lambda{}\).} The transport equation for \(\lambda{}\) is
\begin{equation}
U\lambda{} = \frac{2(\varpi{}-\mathbf{e}^2/r)}{r^2}\kappa{}=_{\mathrm{s}}\mathfrak{b}_0.
\end{equation}
Using \cref{lambda-horizon-bound,S-lambda-horizon-decay} for the \(b = 0\) case, it
follows that
\begin{equation}\label{SU-lambda}
\abs{S^aU^b\lambda{}}\le C(\mathfrak{D}_{a+b}).
\end{equation}

\step{Step 2c: Estimate for derivatives of \(\kappa{}\) along the horizon.}
It is easy to compute from the transport equation for \(\kappa{}\) that for \(b\ge 1\)
one has
\begin{equation}
S^aU^b\log \kappa{} =_{\mathrm{s}} r^{\le 1}S^{\le a}U^{\le b}\varphi{}S^{\le a}U^{\le b}\varphi{}.
\end{equation}
Together with \cref{S-kappa-horizon-decay} for the case \(b = 0\) (and a chain rule
argument), we obtain (for \(a,b\ge 0\) and \(a + b\ge 1\))
\begin{equation}\label{SU-kappa}
\abs{S^aU^b\kappa{}}|_{\mathcal{H}}\le C(\epsilon{},\mathfrak{D}_{a+b})v^{-2+\epsilon{}}.
\end{equation}

\step{Step 3: Estimate for \(S^k\rho{}\).} We claim that
\begin{equation}\label{S-rho-bound}
\abs{S^k\rho{}}\le v^{-1}C(k,\mathfrak{D}_{k})\quad \text{for }1\le k\le 3.
\end{equation}
We first consider the case \(k = 0\). Observe that
\begin{equation}
\rho{} = v^{-1}(v-\bar{v}) + \frac{\bar{v}}{v}(1-1/\kappa{}|_{\mathcal{H}}).
\end{equation}
Since \(\abs{1-1/\kappa{}{|_{\mathcal{H}}}}\lesssim \log \kappa{}|_{\mathcal{H}}\le C(\mathfrak{D}_0)v^{-4}\) by Step 1b of
\cref{S-lambda-horizon-decay}, it is enough to show that \(\bar{v}(v) = v +
C(\mathfrak{D}_0)\). From \(\dv{\bar{v}}{v}(v) = \kappa{}|_{\mathcal{H}}(v)\)
and \(\bar{v}(v=1) = 1\), one obtains
\begin{equation}
\begin{split}
\abs{\bar{v}(v) - v}  &\le  \abs[\Big]{1 - v + \int_1^v \kappa{}|_{\mathcal{H}}(v')\dd{}v'} \le \int_1^v \abs{\kappa{}|_{\mathcal{H}}(v') - 1}\dd{}v'\le C(\mathfrak{b}_0)\int_1^v \log \kappa{}\dd{}v'\\
&\le C(\mathfrak{D}_0)\int_1^v v^{-4}\dd{}v'\le C(\mathfrak{D}_0).
\end{split}
\end{equation}
This establishes the cases \(k = 0\). We now show that for \(k\ge 1\) we have
\begin{equation}\label{S-k-rho-horizon}
S^k\rho{} =_{\mathrm{s}} \sum_{\substack{a+b\ge 1 \\ a\le 1,b\le k}}\rho{}^{a}(S^{\le k}\log \kappa{}|_{\mathcal{H}})^{b}
\end{equation}
Indeed, in the case \(k = 1\), one computes explicitly
\begin{equation}
S\rho{} = -\rho{} - \rho{}S\log \kappa{}|_{\mathcal{H}} + S\log \kappa{}|_{\mathcal{H}},
\end{equation}
and the general case follows by an easy induction. Now
\cref{S-kappa-on-the-horizon,S-k-rho-horizon} and the case \(k = 0\) together
establish \cref{S-rho-bound}.
\end{proof}
\section{Late-time tails}
\label{sec:org785e069}
\label{sec:late-time-tails} The goal of this section is to prove
\cref{tails-theorem-precise} (the precise version of \cref{tails-theorem}), as an
application of \cite[Main Theorem 4]{luk-oh-tails}. The estimates we have already
obtained do not suffice to satisfy the assumptions of \cite{luk-oh-tails}. In
particular, we require additional control of
\((r\overline{\partial}_r)\)-derivatives of the scalar field and of geometric
quantities, together with precise asymptotics for geometric quantities in the
region \(\set{u\lesssim r}\). We prove such estimates in
\cref{sec:rdr-derivatives-of-scalar-field-and-geometry}. The estimates for
\((r\overline{\partial}_r)\)-derivatives of the scalar field are coupled to the
corresponding estimates for the geometry. To obtain such estimates from control
of \(S\)-derivatives of the scalar field and the geometry (which has been done
in the preceding sections), we introduce the spacetime elliptic estimates (see
\cref{sec:spacetime-elliptic-estimates}) and Klainerman's Sobolev-type inequality
(see \cref{sec:klainerman-sobolev}) used in \cite{luk-oh-tails}.

In this section we will freely use the results of
\cref{sec:putting-it-all-together}, namely that our schematic quantities of order
\(\alpha{}\) (see \cref{sec:schematic-geom-quantities}) are controlled by data
\(\mathfrak{D}_{\abs{\alpha{}}}\). We also allow our constants to depend on
\(\varpi{}_i\), \(c_{\mathcal{H}}\), and \(r_{\text{min}}\).
\subsection{Spacetime elliptic estimate}
\label{sec:org2d3157c}
\label{sec:spacetime-elliptic-estimates}
We first formulate an abstract second order weighted elliptic estimate in two dimensions.
\begin{lemma}[Abstract weighted elliptic estimate]
Let \(U \coloneqq{} (0,\infty)_{x_1}\times (0,\infty)_{x_2}\subset \R^2_{x_1,x_2}\). Fix a point \(\conj{x}\in
U\). For \(\theta{}\in (0,1)\), write \(\mathcal{R}_\theta{} =
\mathcal{R}_\theta{}(\conj{x}) \coloneqq{} \set{x\in U : \abs{x_i - \conj{x}_i}
< \theta{}\conj{x}_i\text{ for } i = 1,2}\). Fix \(\theta{},\theta{}'\in (0,1)\)
such that \(\theta{} < \theta{}'\). Let \(L\) be a differential operator on
\(\mathcal{R}_{\theta{}'}\) of the form
\begin{equation}\label{Lpsi-equation}
L\psi{} = \sum_{i,j=1}^2 w_{ij}\partial{}_{ij}^2\psi{} + \sum_{i=1}^2 w_i\partial{}_i\psi{}
\end{equation}
for smooth coefficients \(w_{ij}, w_i\) satisfying \(w_{ij} = w_{ji}\). Suppose
that there are \(\delta{}\in (0,1)\) and \(A\ge 1\) such that the following
conditions hold on \(\mathcal{R}_{\theta{}'}\):
\begin{align}
w_{ii} & > 0 \label{positivity-condition}\\
\abs{w_{12}}^2&\le (1-\delta{})w_{11}w_{22}, \label{ellipticity-condition}\\
\abs{\partial{}_iw_{ij}}^2 + \abs{x_i^{-1}w_{ij}}^2 &\le A^2w_{jj} \label{weight-condition-1}\\
\abs{w_i}^2 &\le A^2w_{ii} \label{weight-condition-2} \\
w_{jj}\abs{\partial{}_jw_{ii}}^2&\le A^2w_{ii}^2 \label{weight-condition-3}.
\end{align}
Then
\begin{equation}
\sum_{i}\norm{w_{ii}^{1/2}\partial{}_i\psi{}}_{L^2(\mathcal{R}_\theta{})} + \sum_{i,j}\norm{w_{ii}^{1/2}w_{jj}^{1/2}\partial{}_{ij}^2\psi{}}_{L^2(\mathcal{R}_\theta{})}\lesssim_{\theta{},\theta{}',\delta{}} A(\norm{\psi{}}_{L^2(\mathcal{R}_{\theta{}'})} + \norm{L\psi{}}_{L^2(\mathcal{R}_{\theta{}'})}).
\end{equation}
\label{abstract-elliptic-estimate}
\end{lemma}
\begin{remark}
Observe that conditions
\cref{weight-condition-1,weight-condition-2,weight-condition-3} ``scale correctly,''
in the sense that the indices \(i\) and \(j\) appear the same number of times on
each side when counted with multiplicity, where a derivative counts with
multiplicity \(-1\).
\end{remark}
\begin{proof}
The result follows from \cref{Lpsi-first-order,Lpsi-second-order} below. For ease
of notation, we will assume \(\theta{}' = 4\theta{}\). Let \(\chi{}\) be a
smooth function satisfying
\begin{equation}
0\le \chi{}\le 1, \quad \chi{}\equiv 1\text{ on }\mathcal{R}_\theta{},\quad \Supp \chi{}\subset \mathcal{R}_{2\theta{}},\qquad \abs{\partial{}_i\chi{}}\le C\theta{}^{-1}\conj{x}_i^{-1}\le C\theta{}^{-1}x_i^{-1}\text{ on }\mathcal{R}_{2\theta{}}.
\end{equation}

\step{Step 1: Estimate for first order derivatives.} We first show that
\begin{equation}\label{Lpsi-first-order}
\sum_{i=1}^2\norm{w_{ii}^{1/2}\partial{}_i\psi{}}_{L^2(\mathcal{R}_\theta{})}\le CA\theta{}^{-1}\delta{}^{-2}(\norm{\psi{}}_{L^2(\mathcal{R}_{2\theta{}})} + \norm{L\psi{}}_{L^2(\mathcal{R}_{2\theta{}})}).
\end{equation}
Multiply \cref{Lpsi-equation} by \(\chi^2\psi{}\) and rewrite and rearrange (``integrate by parts'') to get
\begin{equation}\label{abstract-elliptic-1}
\sum_{i,j=1}^2 \chi{}^2w_{ij}\partial{}_i\psi{}\partial{}_j\psi{} =  -\chi{}^2\psi{}L\psi{} + \sum_{i} \chi{}^2w_i\psi{}\partial{}_i\psi{} +  \sum_{i,j=1}^2 [\partial{}_i(\chi{}^2w_{ij}\psi{}\partial{}_j\psi{}) - 2\chi{}\partial{}_i\chi{}w_{ij}\psi{}\partial{}_j\psi{} - \chi{}^2\partial{}_iw_{ij}\psi{}\partial{}_j\psi{}.
\end{equation}
The ellipticity condition \cref{ellipticity-condition} and Young's inequality imply
\begin{equation}\label{abstract-elliptic-1-1}
\delta{}\sum_{i}\chi{}^2w_{ii}(\partial{}_i\psi{})^2 \le \side{LHS}{abstract-elliptic-1}.
\end{equation}
Let \(\epsilon{} > 0\). Use Young's inequality and then the conditions
\cref{weight-condition-1,weight-condition-2} to obtain
\begin{equation}\label{abstract-elliptic-2}
\begin{split}
\side{RHS}{abstract-elliptic-1} &\le \epsilon{}\sum_{i}\chi{}^2w_{ii}(\partial{}_i\psi{})^2 + C\epsilon{}^{-1}(\abs{\psi{}}^2 + \abs{L\psi{}}^2) + \sum_{i,j=1}^2 [\partial{}_i(\chi{}^2w_{ij}\psi{}\partial{}_j\psi{}) \\
&\qquad+ C\epsilon{}^{-1}\psi{}^2\sum_{i,j}[\chi{}^2w_{jj}^{-1}w_j^2 + w_{jj}^{-1}\abs{\partial{}_i\chi{}w_{ij}}^2 + w_{jj}^{-1}\abs{\partial{}_iw_{ij}}] \\
&\le \epsilon{}\sum_{i=1}^2\chi{}^2w_{ii}(\partial{}_i\psi{})^2 + \sum_{i,j} \partial{}_i(\chi{}^2w_{ij}\psi{}\partial{}_j\psi{}) + CA\theta{}^{-1}\epsilon{}^{-1}(\abs{\psi{}}^2 + \abs{L\psi{}}^2)
\end{split}
\end{equation}
Take \(\epsilon{} < \delta{}/4\) to absorb the first term on the right of
\cref{abstract-elliptic-2} to the left side of \cref{abstract-elliptic-1-1}. Integrate
over \(\mathcal{R}_{2\theta{}}\), noting that the total derivative term has vanishing
integral, to conclude \cref{Lpsi-first-order}.

\step{Step 2: Estimate for second order derivatives.} We now show that
\begin{equation}\label{Lpsi-second-order}
\sum_{i,j}\norm{w_{ij}\partial_{ij}^2\psi{}}_{L^2(\mathcal{R}_\theta{})} + \sum_{i\neq{}j}\norm{w_{ii}^{1/2}w_{jj}^{1/2}\partial_{ij}^2\psi{}}_{L^2(\mathcal{R}_\theta{})}\le CA\theta^{-1}\delta^{-1}(\norm{\psi{}}_{L^2(\mathcal{R}_{4\theta{}})} + \norm{L\psi{}}_{L^2(\mathcal{R}_{4\theta{}})})
\end{equation}
Begin by computing
\begin{equation}\label{abstract-elliptic-3}
\begin{split}
2\sum_{i\neq{}j}w_{ij}^2(\partial_{ij}^2\psi{})^2 + \sum_{i,j}w_{ii}w_{jj}\partial_{ii}^2\psi{}\partial_{jj}^2\psi{} &= (L\psi{})^2 - \sum_{i\neq{}j}\underbrace{2w_{ii}w_{ij}\partial_{ii}^2\psi{}\partial_{ij}^2\psi{}}_\text{(I)} \\
&\qquad  - \underbrace{2\sum_{i,j,k}w_{ij}w_k\partial{}_{ij}^2\psi{}\partial{}_k\psi{}}_{\text{(II)}} - \underbrace{\sum_{i,j}w_iw_j\partial{}_i\psi{}\partial{}_j\psi{}}_{\text{(III)}}.
\end{split}
\end{equation}
Treat term \(\text{(I)}\) using \cref{ellipticity-condition} and Young's
inequality. Terms \(\text{(II)}\) and \(\text{(III)}\) can be handled with
Young's inequality:
\begin{equation}\label{abstract-elliptic-4}
\begin{split}
\abs{\text{(I)}}&\le 2w_{ii}w_{ii}^{1/2}w_{jj}^{1/2}\partial{}_{ii}^2\psi{}\partial{}_{ij}^2\psi{}\le (1-\delta{})w_{ii}^2(\partial{}_{ii}\psi{})^2 + (1-\delta{})w_{ii}w_{jj}(\partial{}_{ij}\psi{})^2 \\
\text{(II)} + \text{(III)}&\le \epsilon{}\sum_{i,j}w_{ij}(\partial{}_{ij}\psi{})^2 + C\epsilon{}^{-1}\sum_{i}w_i(\partial{}_i\psi{})^2.
\end{split}
\end{equation}
Substitute \cref{abstract-elliptic-4} into \cref{abstract-elliptic-3}, take \(\epsilon{} =
\delta{}/2\) and absorb the \(\epsilon{}\)-term to the left to get
\begin{equation}\label{abstract-elliptic-8}
\begin{split}
\frac{\delta{}}{2}\sum_{i,j}w_{ij}^2(\partial_{ij}^2\psi{})^2 + \sum_{i\neq{}j}w_{ii}w_{jj}\partial_{ii}^2\psi{}\partial_{jj}^2\psi{} &\le  (L\psi{})^2 + C\delta{}^{-1}\sum_{i}w_i^2(\partial{}_i\psi{})^2
\end{split}
\end{equation}
We now multiply by \(\chi^2\) and integrate the second term on the left by parts
twice:
\begin{equation}\label{abstract-elliptic-5}
\chi{}^2w_{ii}w_{jj}\partial{}_{ii}^2\psi{}\partial{}_{jj}^2\psi{} = \chi{}^2w_{ii}w_{jj}(\partial{}_{ij}\psi{})^2 -\partial{}_i(\chi{}^2w_{ii}w_{jj}\partial{}_{ji}^2\partial{}_j\psi{}) + \partial{}_j(\chi{}^2w_{ii}w_{jj}\partial{}_{ii}^2\psi{}\partial{}_{j}\psi{}) + \mathcal{E}_{ij},
\end{equation}
where
\begin{equation}\label{abstract-elliptic-6}
\begin{split}
\mathcal{E}_{ij} &\coloneqq{}  2\chi{}\partial{}_i\chi{}w_{ii}w_{jj}\partial{}_{ij}^2\psi{}\partial{}_j\psi{} + \chi{}^2\partial{}_iw_{ii}w_{jj}\partial{}_{ij}^2\psi{}\partial{}_j\psi{} + \chi{}^2w_{ii}\partial{}_iw_{jj}\partial{}_{ij}^2\psi{}\partial{}_j\psi{}  \\
&\qquad - 2\chi{}\partial{}_j\chi{}w_{ii}w_{jj}\partial{}_{ii}^2\psi{}\partial{}_{j}\psi{} - \chi{}^2\partial{}_jw_{ii}w_{jj}\partial{}_{ii}^2\psi{}\partial{}_j\psi{} - \chi{}^2w_{ii}\partial{}_jw_{jj}\partial{}_{ii}^2\psi{}\partial{}_j\psi{}.
\end{split}
\end{equation}
Estimate the terms in \cref{abstract-elliptic-6} using Young's inequality
\cref{weight-condition-1,weight-condition-3}:
\begin{equation}\label{abstract-elliptic-7}
\begin{split}
\abs{\mathcal{E}_{ij}}\le C\epsilon{}\chi^2w_{ii}w_{jj}(\partial_{ij}\psi{})^2 + C\epsilon{}\chi^2w_{ii}^2(\partial{}_{ii}^2\psi{})^2 + CA^2\theta{}^{-2}\epsilon^{-1}w_{jj}(\partial_j\psi{})^2.
\end{split}
\end{equation}
Substitute \cref{abstract-elliptic-6,abstract-elliptic-7} into
 \cref{abstract-elliptic-5} and move all terms but
 \(\chi{}^2w_{ii}w_{jj}(\partial{}_{ij}\psi{})^2\) to the right side of
 \cref{abstract-elliptic-8}. Take \(\epsilon{} > 0\) small enough based on
 \(\delta{}\) to absorb the \(\epsilon{}\)-weighted terms to the left, integrate
 over \(\mathcal{R}_{2\theta{}}\), and use \cref{Lpsi-first-order} (with
 \(2\theta{}\) in place of \(\theta{}\)) to conclude \cref{Lpsi-second-order}.
\end{proof}
Now we construct elliptic operators involving \(S\), \(S^2\), and \(\Box{}\) to
which we apply \cref{abstract-elliptic-estimate}.
\begin{lemma}
Fix \(0<\theta{}<\theta{}'<1\). Let \(U_0\) be sufficiently large
depending on \(\mathfrak{D}_1\). Fix a point \((u_0,r_0)\in \set{r\ge \Rc}\cap \set{u\ge
U_0}\). Write \(\mathcal{R}_\theta{} \coloneqq{} \mathcal{R}_\theta{}(u_0,r_0)\). We have
\begin{equation}\label{elliptic-estimate}
\sum_{i+j\le 2}\norm{(u\overline{\partial}_u)^i(r\overline{\partial}_r)^j\psi{}}_{L^2(\mathcal{R}_{\theta{}})}\lesssim_{\theta{},\theta{}'} \norm{S^{\le 2}\psi{}}_{L^2(\mathcal{R}_{\theta{}'})} + r_0\min (u_0,r_0)\norm{\Box{}\psi{}}_{L^2(\mathcal{R}_{\theta{}'})}.
\end{equation}
\label{spacetime-elliptic-estimate}
\end{lemma}
\begin{proof}
\step{Step 1: Elliptic operator in the region \(\set{\Rc\le r\le \epsilon{}u}\).} We show that there
is \(\epsilon{} > 0\) (depending on \(\mathcal{G}_0\) and \(G_{\eta_0}\)) such that
\begin{equation}\label{elliptic-near}
\sum_{i+j\le 2}u^ir^i\norm{\overline{\partial}_u^i\overline{\partial}_r^j\psi{}}_{L^2(\mathcal{R}_{\theta{}}(u,r))}\lesssim_{\theta{},\theta{}'} \norm{S^{\le 2}\psi{}}_{L^2(\mathcal{R}_{\theta{}'}(u,r))} + ur\norm{\Box{}\psi{}}_{L^2(\mathcal{R}_{\theta{}'}(u,r))}\quad \text{in }\set{\Rc\le r\le \epsilon{}u},
\end{equation}
Write \(\chi{} =
\chi{}_{\lesssim \Rc}(r)\). Compute
\begin{equation}\label{S2-minus-S-near}
\begin{split}
S^2 - S &= (u +\chi{}\cdot (v-u))^2\overline{\partial}_u^2 + r^2\overline{\partial}_r^2 + 2r(u + \chi{}\cdot (v-u))\overline{\partial}_u\overline{\partial}_r \\
&\qquad + r(\chi{}'(v-u) + \lambda{}^{-1}\chi{})\partial{}_u - u\chi{}(1 - (-\gamma{})/\kappa{})\overline{\partial}_r \\
\end{split}
\end{equation}
Let \(B > 0\) be a constant to be chosen. Now \(\eqref{S2-minus-S-near} - Br^2\cdot \eqref{box-I-gauge}\) is
\begin{equation}
\begin{split}
S^2 - S - Br^2\Box{} &= (u + \chi{}\cdot (v-u))^2\overline{\partial}_u^2 + (1 + B(1-\mu{}))r^2\overline{\partial}_r + (2r(u + \chi{}\cdot (v-u)) - Br^2(-\gamma{})^{-1})\overline{\partial}_u\overline{\partial}_r \\
&\qquad + r(\chi{}'\cdot (v-u) + \lambda{}^{-1}\chi{} + B(2-2\varpi{}/r))\overline{\partial}_r - (u\chi{}(1-(-\gamma{})/\kappa{}) + Br(-\gamma{})^{-1})\overline{\partial}_u.
\end{split}
\end{equation}
We now verify conditions
\cref{positivity-condition,ellipticity-condition,weight-condition-1,weight-condition-2,weight-condition-3}
of \cref{abstract-elliptic-estimate}. In \(\set{r\le \eta{}_0u}\) we have
\begin{equation}
w_{uu} = u^2 + O(1)ur \quad w_{rr} = (1 + B(1-\mu{}))r^2 \quad w_{ur} = (2ur + (BO(1)) + O(1))r^2),
\end{equation}
where we write \(O(1)\) for a term controlled by \(\mathfrak{B}_0\) and \(G_{\eta_0}\).
Now let \(\epsilon{} < \eta_0\). It follows that in \(\set{r\le \epsilon{}u}\) for \(\epsilon{}\) sufficiently
small depending on \(\mathcal{G}_0\), \(G_{1/2}\), and \(B\), we have
\begin{equation}
\abs{w_{uu}} \ge  \frac{1}{2}u^2 \quad w_{rr} = (1 + B(1-\mu{}))r^2 \quad \abs{w_{ur}} \le 3ur.
\end{equation}
Since \(1-\mu{}\ge 1/2\) in \(\set{r\ge R_0}\), it follows that in \(\set{R_0\le r\le \epsilon{}u}\),
\cref{positivity-condition,ellipticity-condition} are satisfied for \(B = 36\)
(with \(\delta{} > 0\) an explicit numerical constant).
Conditions \cref{weight-condition-1,weight-condition-2,weight-condition-3} follow
from computations involving \(\mathfrak{B}_V\) and \(G_{\eta_0}\).

\step{Step 2: Elliptic operator in the region \(\set{r\ge \epsilon{}u}\).} Let \(\epsilon{}\) be as in
Step 1. In fact, the value of \(\epsilon{}\) does not play a role in the elliptic
estimate for this region. We show that for \(u\) large depending on
\(C(\mathfrak{D}_2)\), we have
\begin{equation}\label{elliptic-far}
\sum_{i+j\le 2}u^ir^i\norm{\overline{\partial}_u^i\overline{\partial}_r^j\psi{}}_{L^2(\mathcal{R}_{\theta{}}(u,r))}\lesssim_{\theta{},\theta{}'} \norm{S^{\le 2}\psi{}}_{L^2(\mathcal{R}_{\theta{}'}(u,r))} + ur\norm{\Box{}\psi{}}_{L^2(\mathcal{R}_{\theta{}'}(u,r))}\quad \text{in }\set{r\ge \epsilon{}u}.
\end{equation}
For \(u\ge \epsilon^{-1} \Rc\), we have \(r\ge \Rc\), so
\begin{equation}\label{S2-minus-S-far}
S^2 - S = u^2\overline{\partial}_u^2 + r^2\overline{\partial}_r^2 + 2ur\overline{\partial}_u\overline{\partial}_r.
\end{equation}
Now \(\eqref{S2-minus-S-far} - 2ur\cdot \eqref{box-I-gauge}\) is
\begin{equation}
S^2 - S - 2ur\Box{} = u^2\overline{\partial}_u^2 + (r^2 + 2ur)\overline{\partial}_r^2 + 2ur\Bigl(1-\frac{1}{(-\gamma{})}\Bigr)\overline{\partial}_u\overline{\partial}_r + 4u(1 - \varpi{}/r)\overline{\partial}_u - u\frac{2}{(-\gamma{})}\overline{\partial}_r.
\end{equation}
To apply \cref{abstract-elliptic-estimate} and obtain \cref{elliptic-far}, we need to check that this operator
satisfies the conditions
\cref{positivity-condition,weight-condition-1,weight-condition-2,weight-condition-3}.
Clearly \cref{positivity-condition} is satisfied. Since \(\abs{1-1/(-\gamma{})}\le
r^{-1}u^{-1}C(\mathfrak{B}_0)\), \cref{ellipticity-condition} is satisfied for \(u\)
large enough depending on \(\mathfrak{B}_0\). Since
\(\abs{\conj{\partial}(-\gamma{})}\le r^{-1}u^{-1}C(\mathfrak{B}_V)\),
\cref{weight-condition-1} is satisfied for \(u\) large enough. Since \(u\lesssim
r\), \cref{weight-condition-2,weight-condition-3} are satisfied.
\end{proof}
\subsection{Klainerman's Sobolev-type inequality}
\label{sec:orge564f18}
\label{sec:klainerman-sobolev}
In this section we upgrade the \(L^2\) estimates in
\cref{sec:spacetime-elliptic-estimates} to \(L^2\)--\(L^\infty\) Sobolev inequalities.
\begin{lemma}
Let \(c > 1\). Let \(Q \coloneqq{} Q_{\ell{}_1,\ell{}_2}\subset \R^2\) be a rectangle of side lengths \(\ell_1\) and \(\ell_2\). We
have
\begin{equation}
\norm{f}_{L^\infty(Q)}\lesssim_c \ell{}_1^{-1/2}\ell{}_2^{-1/2}\sum_{i+j\le 2}\norm{(\ell{}_1\partial{}_{x^1})^i(\ell{}_2\partial{}_{x^2})^jf}_{L^2(cQ)}
\end{equation}
\label{scaled-sobolev-embedding}
\end{lemma}
\begin{proof}
Apply the usual Sobolev embedding on a unit rectangle to the rescaled function
\(f(\ell_1x^1,\ell_2x^2)\).
\end{proof}
\begin{lemma}[Weighted \(L^2\)--\(L^{\infty}\) estimate]
Fix \(0<\theta{}<\theta{}'<1\). Fix a point \((u_0,r_0)\in \set{r\ge 5\Rc}\cap \set{u\ge U_0}\) for
\(U_0\) as in \cref{spacetime-elliptic-estimate}. Write \(\mathcal{R}_\theta{}
\coloneqq{} \mathcal{R}_\theta{}(u_0,r_0)\) as in
\cref{abstract-elliptic-estimate}. We have
\begin{equation}
u_0^{1/2}r_0^{1/2}\norm{\psi{}}_{L^\infty(\mathcal{R}_\theta{})}\lesssim_{\theta{},\theta{}'} \norm{S^{\le 2}\psi{}}_{L^2(\mathcal{R}_{\theta{}'})} + \norm{r^2\Box{}\psi{}}_{L^2(\mathcal{R}_{\theta{}'})}.
\end{equation}
\label{klainerman-sobolev}
\end{lemma}
\begin{proof}
Combine \cref{spacetime-elliptic-estimate} with \cref{scaled-sobolev-embedding},
noting that \(\mathcal{R}_\theta{}\) is a rectangle with side lengths proportional
to \(u_0\) and \(r_0\) and that the coordinates \(u\) and \(r\) are comparable
to \(u_0\) and \(r_0\) in \(\mathcal{R}_\theta{}(u_0,r_0)\).
\end{proof}
\subsection{Estimates for \((r\overline{\partial}_r)\)-derivatives of the scalar field and of geometric quantities}
\label{sec:org9c55fc9}
\label{sec:rdr-derivatives-of-scalar-field-and-geometry}
\subsubsection{Estimates for \((r\overline{\partial}_r)\)-derivatives of the scalar field}
\label{sec:orgbbbcd79}
\begin{lemma}[Commutation formula for \(r\overline{\partial}_r\) and \(S\)]
For \(n\ge 0\) and \(m\ge 0\) and \(n + m\ge 1\), we have
\begin{equation}\label{rdr-S-n-comm-equation}
\begin{split}
&[\Box{},(r\overline{\partial}_r)^nS^m] \\
&=_{\mathrm{s}} \sum_{\substack{a+b\le m-1 \\ i+j+k\le n}}\set{1,1 + (r\overline{\partial}_r)^{i}\log (-\gamma{})}^{\le 1}\set{1,(r\overline{\partial}_r)^{j}S^{1+a}\log (-\gamma{})}^{\le 1}(r\overline{\partial}_r)^{k}S^b\Box{} \\
&\,+ \frac{1}{r^2}\mathbf{1}_{m\ge 1}\sum_{\substack{a+b+c\le m \\ i + j + k\le n \\ c\le m-1}}\set{O_\infty(1),r^{-1}(r\overline{\partial}_r)^iS^a\varpi{}}((r\overline{\partial}_r)^jS^{b}\log (-\gamma{}))^{\le 1}(r\overline{\partial}_r)(r\overline{\partial}_r)^{\le 1+k}S^c \\
&\, + \frac{1}{r^2}\mathbf{1}_{n\ge 1}\sum_{\substack{a+b+c\le m \\ i+j+k+\ell{}\le n+1 \\ \ell{}\ge 1}}(1 + (r\overline{\partial}_r)^{i}\log (-\gamma{}))^{\le 1}\set{O_\infty(1),r^{-1}(r\overline{\partial}_r)^{j}S^a\varpi{}}((r\overline{\partial}_r)^{k}S^{b}\log (-\gamma{}))^{\le 1}(r\overline{\partial}_r)^{\ell{}}S^c, \\
\end{split}
\end{equation}
where we write \(O_\infty(1)\) for a constant-coefficient polynomial in \(r^{-1}\).
\label{rdr-S-n-comm}
\end{lemma}
\begin{proof}
The lemma follows from \cref{rdr-n-comm,S-n-comm}, which we prove below.

\step{Step 1: Commutation formula for \(r\overline{\partial}_r\).} A computation reveals that
\begin{equation}\label{r-dr-commutation-formula}
\begin{split}
[\Box{},r\overline{\partial}_r] &= (1 + r\overline{\partial}_r\log (-\gamma{}))\Box{} + \frac{1}{r^2}\Bigl[-1 + \frac{4\varpi{}}{r} - \frac{3\mathbf{e}^2}{r^2}\Bigr](r\overline{\partial}_r)^2 \\
&\qquad + \frac{1}{r^2}\Bigl[-1 + \frac{3\mathbf{e}^2}{r^2}- \frac{2}{r}r\overline{\partial}_r\varpi{} + \Bigl(2 - \frac{2\varpi{}}{r}\Bigr)r\overline{\partial}_r\log (-\gamma{})\Bigr]r\overline{\partial}_r \\
&=_{\mathrm{s}} (1 + r\overline{\partial}_r\log (-\gamma{}))\Box{} + \frac{1}{r^2}\bigl[O_\infty(1) + r^{-1}\varpi{}\bigr](r\overline{\partial}_r)^2 \\
&\qquad + \frac{1}{r^2}\bigl[O_{\infty}(1) + r^{-1}r\overline{\partial}_r\varpi{} + \set{1,r^{-1}\varpi{}}r\overline{\partial}_r\log (-\gamma{})\bigr]r\overline{\partial}_r.
\end{split}
\end{equation}
It follows by induction that for \(n\ge 1\), we have
\begin{equation}\label{rdr-n-comm}
\begin{split}
[\Box{},(r\overline{\partial}_r)^n] &=_{\mathrm{s}} \sum_{i+j\le n-1}(1 + (r\overline{\partial}_r)^{1+i}\log (-\gamma{}))(r\overline{\partial}_r)^{j}\Box{} \\
&\qquad + \sum_{\substack{i+j+k\le n+1 \\k\ge 1}}\frac{1}{r^2}\set{O_\infty(1),r^{-1}(r\overline{\partial}_r)^{i}\varpi{}}((r\overline{\partial}_r)^{j}\log (-\gamma{}))^{\le 1}(r\overline{\partial}_r)^{k} \\
\end{split}
\end{equation}
\step{Step 2: Commutation formula for \(S\).} In the region \(\set{r\ge 5\Rc}\) of interest, where \(S = r\overline{\partial}_r +
u\overline{\partial}_u\), the computation \cref{S-far-commutation} gives
\begin{equation}
\begin{split}
[\Box{},S] &=_{\mathrm{s}} (2 + S\log (-\gamma{}))\Box{} + \frac{1}{r^2}\bigl[O_\infty(1) + \set{O_\infty(1),r^{-1}\varpi{}}S\log (-\gamma{}) + r^{-1}S^{\le 1}\varpi{} \bigr](r\overline{\partial}_r)(r\overline{\partial}_r)^{\le 1}. \\
\end{split}
\end{equation}
It follows by induction that for \(n\ge 1\), we have
\begin{equation}\label{S-n-comm}
\begin{split}
[\Box{},S^n] &=_{\mathrm{s}} \sum_{a+b\le n-1}\set{1,S^{1+a}\log (-\gamma{})}S^b\Box{} \\
&\qquad + \sum_{a+b+c\le n-1}\frac{1}{r^2}\bigl[O_\infty(1) + \set{O_\infty(1),r^{-1}S^a\varpi{}}S^{1+b}\log (-\gamma{}) + r^{-1}S^{\le 1+a}\varpi{} \bigr](r\overline{\partial}_r)(r\overline{\partial}_r)^{\le 1}S^c. \\
\end{split}
\end{equation}
\end{proof}
\begin{lemma}
Fix \(0<\theta{}<\theta{}'<1/2\). Fix a point \((u_0,r_0)\in \set{r\ge 5\Rc}\cap \set{u\ge U_0}\) for
\(U_0\) as in \cref{spacetime-elliptic-estimate}. Write \(\mathcal{R}_\theta{}
\coloneqq{} \mathcal{R}_\theta{}(u_0,r_0)\) as in
\cref{abstract-elliptic-estimate}. For \(n,m\ge 0\), we have
\begin{equation}
\norm{(r\overline{\partial}_r)^nS^m\varphi{}}_{L^\infty(\mathcal{R}_\theta{})}\lesssim C(n,m,\theta{},\theta{}',\mathfrak{D}_{n+m+2})\norm{S^{\le n+m+3}\varphi{}}_{L^\infty(\mathcal{R}_{\theta{}'})}.
\end{equation}
\label{elliptic:Linf-scalar-field}
\end{lemma}
\begin{corollary}
Let \(n,m\ge 0\) and let \(\epsilon{} > 0\). In the region \(\set{r\ge 5\Rc}\), we have
\begin{equation}
(r^{1/2}\tau{}^{1-\epsilon{}} + r\tau{}^{1/2-\epsilon{}})\abs{(r\overline{\partial}_r)^nS^m\varphi{}}\le C(\epsilon{},n,m,\mathfrak{D}_{n+m+3}).
\end{equation}
\label{elliptic:pointwise-scalar-field}
\end{corollary}
\begin{proof}
In the region \(\set{u\ge U_0}\), this is an immediate consequence of
\cref{elliptic:Linf-scalar-field}, the definition of the pointwise norm (see
\cref{sec:norms:pointwise}), and the results of \cref{sec:putting-it-all-together}. In
the region \(\set{u\le U_0}\), we can write \(r\overline{\partial}_r = S -
u\overline{\partial}_u\) to reduce to estimating expressions of the form
\((u\overline{\partial}_u)^{n'}S^{m'}\varphi{}\) for \(n' + m' = n + m\). From the calculation
\((u\overline{\partial}_u)^n =_{\mathrm{s}} u^{\le n}\overline{\partial}_u^{1\le n} =_{\mathrm{s}} u^{\le
n}C(\mathfrak{b}_{nV})D^{1\le n}\) and the boundedness of \(u\), we have
\begin{equation}
\abs{(u\overline{\partial}_u)^{n'}S^{m'}\varphi{}}\lesssim_{U_0} C(\mathfrak{D}_n)\abs{D^{\le n'}S^{\le m'}\varphi{}},
\end{equation}
which we can estimate with the pointwise norm. Recalling that \(U_0\) depends
only on \(\mathfrak{D_1}\) completes the proof.
\end{proof}
\begin{proof}[Proof of \cref{elliptic:Linf-scalar-field}]
Let \(c > 1\) be a parameter that will be chosen in the proof. We allow all
estimates to depend on \(n\), \(m\), \(\theta{}\), \(\theta{}'\), and \(c\). The
desired estimate follows from \cref{elliptic:Linf-goal}, after choosing \(c > 1\)
so that \(c^{n + m + 3}\theta{} \le \theta{}'\). The restriction \(\theta{}' <
1/2\) ensures that \(r > 2\Rc\) in the region \(\mathcal{R}_{\theta{}'}\) (and
so \(S = r\overline{\partial}_r + u\overline{\partial}_u\)). In this proof, we
write \(\mathcal{R}(\theta{})\) in place of \(\mathcal{R}_\theta{}\). The right
sides of our equations are to be understood schematically, as in
\cref{sec:schematic-notation}.

\step{Step 1: \(L^2\) estimate for the commutator.} We claim that for \(n\ge 0\) and \(m\ge
0\) and \(n+m\ge 1\), we have
\begin{equation}\label{elliptic:rS-box}
\begin{split}
&\norm{r^2\Box{}(r\overline{\partial}_r)^nS^m\varphi{}}_{L^2(\mathcal{R}(\theta{}))} \\
&\lesssim C(\mathfrak{D}_m,\norm{(r\overline{\partial}_r)^{\le n-1}S^{\le m}\varphi{}}_{L^\infty(\mathcal{R}(\theta{}))},\norm{(r\overline{\partial}_r)^{\le n}S^{\le m-1}\varphi{}}_{L^\infty(\mathcal{R}(\theta{}))})\\
&\qquad [\norm{(r\overline{\partial}_r)^{\le n+1}S^{\le m}\varphi{}}_{L^2(\mathcal{R}(\theta{}))} + \mathbf{1}_{m\ge 1}\norm{(r\overline{\partial}_r)^{\le n+2}S^{\le m-1}\varphi{}}_{L^2(\mathcal{R}(\theta{}))}].
\end{split}
\end{equation}
The idea is to use the transport equation for \((-\gamma{})\) and \(\varpi{}\) to express the
coefficients in \cref{rdr-S-n-comm-equation} in terms of the scalar field, and then
put the top order terms in \(L^2\) and the lower order terms in \(L^\infty\).

\step{Step 1a: Estimate for derivatives of \((-\gamma{})\).} When \(n = 0\), we can estimate
\begin{equation}\label{elliptic:gamma-bound-0}
\abs{S^m\log (-\gamma{})}\lesssim C(\mathfrak{D}_m)
\end{equation}
by the results of \cref{sec:putting-it-all-together}. When \(n\ge 1\),
differentiating the transport equation \((r\overline{\partial}_r)\log (-\gamma{}) =
((r\overline{\partial}_r)\varphi{})^2\) gives
\begin{equation}\label{elliptic:gamma-bound}
\abs{(r\overline{\partial}_r)^nS^m\log (-\gamma{})} \lesssim  C(\abs{(r\overline{\partial}_r)^{\le n-1}S^{\le m}\varphi{}},\abs{(r\overline{\partial}_r)^{\le n}S^{\le m-1}\varphi{}})\abs{(r\overline{\partial}_r)^{\le n}S^{\le m}\varphi{}}
\end{equation}

\step{Step 1b: Estimate for derivatives of \(\varpi{}\).} When \(n = 0\), we can estimate
\begin{equation}\label{elliptic:mass-bound-0}
\abs{S^m\varpi{}}\lesssim C(\mathfrak{D}_m)
\end{equation}
by the results of \cref{sec:putting-it-all-together}. When \(n\ge 1\), we claim
\begin{equation}\label{elliptic:mass-bound}
r^{-1}\abs{(r\partial{}_r)^nS^m\varpi{}}\le C(\mathfrak{D}_m,\abs{(r\overline{\partial}_r)^{\le n-1}S^{\le m}\varphi{}},(r\overline{\partial}_r)^{\le n}S^{\le m-1}\varphi{})\abs{(r\overline{\partial}_r)^{\le n}S^{\le m}\varphi{}}.
\end{equation}
The transport equation for \(\varpi{}\) is
\begin{equation}
r^{-1}(r\overline{\partial}_r)\varpi{} = \frac{1}{2}(1-\mu{})((r\overline{\partial}_r)\varphi{})^2 =_{\mathrm{s}} \mathcal{O}(1,r^{-1}\varpi{})((r\overline{\partial}_r)\varphi{})^2.
\end{equation}
Differentiate to obtain
\begin{equation}
\begin{split}
r^{-1}(r\overline{\partial}_r)^nS^m\varpi{} &=_{\mathrm{s}} \mathcal{O}(1,r^{-1}S^{\le m}\varpi{},r^{-1}(r\overline{\partial}_r)^{1\le n-1}S^{\le m}\varpi{})((r\overline{\partial}_r)^{\le n}S^{\le m}\varphi{})\\
&\qquad \set{(r\overline{\partial}_r)^{\le n-1}S^{\le m}\varphi{},(r\overline{\partial}_r)^{\le n}S^{\le m-1}\varphi{}}.
\end{split}
\end{equation}
Now \cref{elliptic:mass-bound} follows from an induction argument and the \(n = 0\)
case.

\step{Step 1c: Completing the proof of \cref{elliptic:rS-box}.} Substitute
\cref{elliptic:mass-bound,elliptic:mass-bound-0,elliptic:gamma-bound-0,elliptic:gamma-bound}
into \cref{rdr-S-n-comm-equation} to obtain
\begin{equation}\label{elliptic:rS-comm-improved}
\begin{split}
\abs{r^2\Box{}(r\overline{\partial}_r)^nS^m\varphi{}} &\lesssim C(\mathfrak{D}_m,\abs{(r\overline{\partial}_r)^{\le n-1}S^{\le m}\varphi{}},\abs{(r\overline{\partial}_r)^{\le n}S^{\le m-1}\varphi{}})\\
&\Bigl[\mathbf{1}_{m\ge 1}\sum_{\substack{a+b+c\le m \\ i + j + k\le n \\ c\le m-1}}(1 + \abs{(r\overline{\partial}_r)^{\le i}S^{\le a}\varphi{}})(1 + \abs{(r\overline{\partial}_r)^{\le j}S^{\le b}\varphi{}})\abs{(r\overline{\partial}_r)^{\le k+2}S^c\varphi{}} \\
&\mathbf{1}_{n\ge 1}\sum_{\substack{a+b+c\le m \\ i+j+k+\ell{}\le n+1 \\ \ell{}\ge 1}} (1 + \abs{(r\overline{\partial}_r)^{\le i}\varphi{}})(1 + \abs{(r\overline{\partial}_r)^{\le j}S^{\le a}\varphi{}})(1 + \abs{(r\overline{\partial}_r)^{\le k}S^{\le b}\varphi{}})\abs{(r\overline{\partial}_r)^{\ell{}}S^c\varphi{}}\Bigr]. \\
\end{split}
\end{equation}
To obtain \cref{elliptic:rS-box}, put the top order terms in \(L^2\) and absorb the
lower order terms into the constant.

\step{Step 2: Basic estimates for the scalar field.} It is an immediate consequence of
\cref{spacetime-elliptic-estimate} that for \(n\ge 1\), we have
\begin{equation}\label{elliptic:rS-2}
\norm{(r\overline{\partial}_r)^nS^m\varphi{}}_{L^2(\mathcal{R}(\theta{}))}\lesssim \norm{(r\overline{\partial}_r)^{n-\min (2,n)}S^{\le m + 2}\varphi{}}_{L^2(\mathcal{R}(c\theta{}))} + \norm{r^2\Box{}(r\overline{\partial}_r)^{n-\min (2,n)}S^m\varphi{}}_{L^2(\mathcal{R}(c\theta{}))}.
\end{equation}
It follows from \cref{elliptic:rS-2,elliptic:rS-box} that, for \(n\ge 3\),
\begin{equation}\label{elliptic:L2-estimate}
\begin{split}
&\norm{(r\overline{\partial}_r)^nS^m\varphi{}}_{L^2(\mathcal{R}(\theta{}))} \\
&\lesssim  C(\mathfrak{D}_m,\norm{(r\overline{\partial}_r)^{\le n-3}S^{\le m}\varphi{}}_{L^\infty(\mathcal{R}(c\theta{}))},\norm{(r\overline{\partial}_r)^{\le n-2}S^{\le m-1}\varphi{}}_{L^\infty(\mathcal{R}(c\theta{}))}) \\
&\qquad [\norm{(r\overline{\partial}_r)^{n-2}S^{\le m + 2}\varphi{}}_{L^2(\mathcal{R}(c\theta{}))} + \norm{(r\overline{\partial}_r)^{\le n-1}S^{\le m}\varphi{}}_{L^2(\mathcal{R}(c\theta{}))} + \mathbf{1}_{m\ge 1}\norm{(r\overline{\partial}_r)^{\le n}S^{\le m-1}\varphi{}}_{L^2(\mathcal{R}(c\theta{}))}].
\end{split}
\end{equation}
Finally, it follows from \cref{klainerman-sobolev,elliptic:rS-box} that, for \(n,m\ge 0\),
\begin{equation}\label{elliptic:rS-infinity}
\begin{split}
&u_0^{1/2}r_0^{1/2}\norm{(r\overline{\partial}_r)^nS^m\varphi{}}_{L^\infty(\mathcal{R}(\theta{}))} \\
&\lesssim C(\mathfrak{D}_m,\norm{(r\overline{\partial}_r)^{\le n-1}S^{\le m}\varphi{}}_{L^\infty(\mathcal{R}(c\theta{}))},\norm{(r\overline{\partial}_r)^{\le n}S^{\le m-1}\varphi{}}_{L^\infty(\mathcal{R}(c\theta{}))})\\
&\qquad [\norm{(r\overline{\partial}_r)^nS^{\le m+2}\varphi{}}_{L^2(\mathcal{R}(c\theta{}))} + \norm{(r\overline{\partial}_r)^{\le n+1}S^{\le m}\varphi{}}_{L^2(\mathcal{R}(c\theta{}))} + \mathbf{1}_{m\ge 1}\norm{(r\overline{\partial}_r)^{\le n+2}S^{\le m-1}\varphi{}}_{L^2(\mathcal{R}(c\theta{}))}].
\end{split}
\end{equation}

\step{Step 3: Completing the proof.} For \(n,m\ge
0\), we will show the \(L^2\)--\(L^2\) estimate
\begin{equation}\label{elliptic:L2-induction}
\norm{(r\overline{\partial}_r)^nS^m\varphi{}}_{L^2(\mathcal{R}(\theta{}))}\lesssim C(\mathfrak{D}_{n+m})\norm{S^{\le n+m+1}\varphi{}}_{L^2(\mathcal{R}(c^{n+m+1}\theta{}))}
\end{equation}
and the \(L^2\)--\(L^\infty\) estimate
\begin{equation}\label{elliptic:Linf-induction}
u_0^{1/2}r_0^{1/2}\norm{(r\overline{\partial}_r)^nS^m\varphi{}}_{L^\infty(\mathcal{R}(\theta{}))}\lesssim C(\mathfrak{D}_{n+m+2})\norm{S^{\le n+m+3}\varphi{}}_{L^2(\mathcal{R}(c^{n+m+3}\theta{}))}.
\end{equation}
Note that \cref{elliptic:Linf-induction} implies the desired \(L^\infty\)--\(L^\infty\) estimate
\begin{equation}\label{elliptic:Linf-goal}
\norm{(r\overline{\partial}_r)^nS^m\varphi{}}_{L^\infty(\mathcal{R}(\theta{}))}\lesssim C(\mathfrak{D}_{n+m+2})\norm{S^{\le n+m+3}\varphi{}}_{L^\infty(\mathcal{R}(c^{n+m+3}\theta{}))}
\end{equation}
by the Hölder's inequality \(\norm{\psi{}}_{L^2(\mathcal{R}(\theta{}))}\lesssim u_0^{1/2}r_0^{1/2}\norm{\psi{}}_{L^\infty(\mathcal{R}(\theta{}))}\).

\step{Step 3a: The base case \cref{elliptic:L2-induction} for \(n\le 2\) and \(m\ge 0\).}
First, \cref{elliptic:L2-induction} is trivial when \(n = 0\) for all \(m\ge 0\). For
the case \(1\le n\le 2\), we claim that, for \(m\ge 0\),
\begin{equation}\label{elliptic:n12-2}
\norm{(r\overline{\partial}_r)^{\le 2}S^m\varphi{}}_{L^2(\mathcal{R}(\theta{}))} \lesssim C(\mathfrak{D}_m)\norm{S^{\le m+2}\varphi{}}_{L^2(\mathcal{R}(c^{m+1}\theta{}))}.
\end{equation}
An induction argument using \cref{spacetime-elliptic-estimate,S-n-comm} shows that,
for \(m\ge 0\),
\begin{equation}\label{elliptic:S-comm}
\norm{r^2\Box{}S^m\varphi{}}_{L^2(\mathcal{R}(\theta{}))}\lesssim C(\mathfrak{D}_m)\norm{S^{\le m+1}\varphi{}}_{L^2(\mathcal{R}(c^m\theta{}))}.
\end{equation}
Conclude the desired \cref{elliptic:n12-2} using \cref{elliptic:rS-2,elliptic:S-comm}.

\step{Step 3b: The base case \cref{elliptic:Linf-induction} for \(n = 0\) and \(m\ge 0\).}
\Cref{klainerman-sobolev,elliptic:S-comm} imply that
\begin{equation}
u_0^{1/2}r_0^{1/2}\norm{S^m\varphi{}}_{L^\infty(\mathcal{R}(\theta{}))}\lesssim \norm{S^{\le m + 2}\psi{}}_{L^2(\mathcal{R}(c\theta{}))} + \norm{r^2\Box{}S^m\varphi{}}_{L^2(\mathcal{R}(c\theta{}))}\lesssim C(\mathfrak{D}_m)\norm{S^{\le m + 2}\varphi{}}_{L^2(\mathcal{R}(c^{m+1}\theta{}))}.
\end{equation}

\step{Step 3c: The inductive step.} Let \(n\ge 1\) and \(m\ge 0\). Suppose that
\cref{elliptic:L2-induction} holds for pairs \((\le n+1,k)\) and
\cref{elliptic:Linf-induction} holds for pairs \((\le n-1,k)\) for all \(k\ge 0\)
as well as for pairs \((n,\le m-1)\) (this latter condition is vacuous when \(m
= 0\)). We will show that \cref{elliptic:L2-induction} holds for the pair \((n +
2,m)\) and \cref{elliptic:Linf-induction} holds for the pair \((n,m)\). In view of
the base cases established in Steps 3ab, this inductive step establishes
\cref{elliptic:L2-induction,elliptic:Linf-induction}.

We first claim that
\begin{equation}\label{elliptic:rn2-prep}
\begin{split}
&\norm{(r\overline{\partial}_r)^{n+2}S^m\varphi{}}_{L^2(\mathcal{R}(\theta{}))} \lesssim  C(\mathfrak{D}_m,\norm{(r\overline{\partial}_r)^{\le n-1}S^{\le m}\varphi{}}_{L^\infty(\mathcal{R}(c\theta{}))},\norm{(r\overline{\partial}_r)^{\le n}S^{\le m-1}\varphi{}}_{L^\infty(\mathcal{R}(c\theta{}))}) \\
&\qquad [\norm{(r\overline{\partial}_r)^{n}S^{\le m + 2}\varphi{}}_{L^2(\mathcal{R}(c\theta{}))} + \norm{(r\overline{\partial}_r)^{\le n+1}S^{\le m}\varphi{}}_{L^2(\mathcal{R}(c\theta{}))}].
\end{split}
\end{equation}
Indeed, this follows from \cref{elliptic:L2-estimate} (since the last term in
square brackets in that estimate can be removed by induction). Next,
use \cref{elliptic:rn2-prep} to control the last term in square brackets in the
estimate \cref{elliptic:rS-infinity} and obtain
\begin{equation}\label{elliptic:rS-infinity-prep}
\begin{split}
&u_0^{1/2}r_0^{1/2}\norm{(r\overline{\partial}_r)^nS^m\varphi{}}_{L^\infty(\mathcal{R}(\theta{}))}\lesssim \side{LHS}{elliptic:rn2-prep}.
\end{split}
\end{equation}
By the inductive hypotheses, we have
\begin{equation}\label{elliptic:goal-prep}
\begin{split}
\side{LHS}{elliptic:rn2-prep}&\lesssim C(\mathfrak{D}_{n+m+1},\norm{S^{n+m+2}\varphi{}}_{L^\infty(\mathcal{R}(c^{n+m+3}\theta{}))})\cdot C(\mathfrak{D}_{n+m+2})\norm{S^{\le n+m+3}\varphi{}}_{L^2(\mathcal{R}(c^{n+m+3}\theta{}))}.
\end{split}
\end{equation}
To conclude the proof, control \(\norm{S^{n + m + 2}\varphi{}}_{L^\infty(\mathcal{R}(c^{n + m +
3}\theta{}))}\le C(\mathfrak{D}_{n + m + 2})\) and combine
\cref{elliptic:goal-prep,elliptic:rS-infinity-prep,elliptic:rn2-prep}.
\end{proof}
\begin{lemma}
Let \(n,m\ge 0\) and let \(0\le k\le 2\). We have
\begin{equation}
\abs{(r^2\overline{\partial}_r)^k(r\overline{\partial}_r)^nS^m\varphi{}}\le C(n,m,\mathfrak{D}_{n+m+5})\quad \text{in }\set{r\ge 5\Rc}\cap \set{r\ge u}.
\end{equation}
\label{rsr-rdr-S-bound}
\end{lemma}
\begin{proof}
The cases \(k = 0,1\) follow from \cref{elliptic:pointwise-scalar-field}, so we
focus on the case \(k = 2\).

Suppose for the sake of induction that we have already shown \cref{rsr-rdr-S-bound}
for \(k = 2\) and pairs \((n',m')\) with \(m' < m\) or \(m=m'\) and \(n' < n\). Set \(\psi{}_0 =
(r\overline{\partial}_r)^{n + 1}S^m\varphi{}\). By
\cref{rdr-S-n-comm,elliptic:rS-gamma-bound,elliptic:rS-omega-bound} and the
induction hypothesis we have
\begin{equation}\label{psi-0-assumption-1}
\begin{split}
\abs{r^3\Box{}\psi{}_0}&\le C(n,m,\mathfrak{D}_{n+m+4})[r\abs{(r\overline{\partial}_r)^{\le n+3}S^{\le m-1}\varphi{}} + r\abs{(r\overline{\partial}_r)^{\le n+2}S^{\le m}\varphi{}}]\\
&\le C(n,m,\mathfrak{D}_{n+m+5})[r^{-1}\abs{(r^2\overline{\partial}_r)^2(r\overline{\partial}_r)^{\le n+1}S^{\le m-1}\varphi{}} + r^{-1}\abs{(r^2\overline{\partial}_r)^2(r\overline{\partial}_r)^{\le n}S^{\le m}\varphi{}}] \\
&\le C(n,m,\mathfrak{D}_{n+m+5})[r^{-1} + r^{-1}\abs{(r^2\overline{\partial}_r)^2(r\overline{\partial}_r)^{\le n}S^{\le m}\varphi{}}], \\
\end{split}
\end{equation}
and by \cref{elliptic:pointwise-scalar-field}, we have
\begin{equation}\label{psi-0-assumption-2}
\abs{(r\overline{\partial}_r)\psi{}_0}\le r^{-1}C(n,m,\mathfrak{D}_{n+m+4}).
\end{equation}
A computation (starting from \cref{dudv-rpsi}) for general \(\psi{}\) shows that
\begin{equation}
\partial{}_u((r^2\overline{\partial}_r)(r\psi{}))  = r^3(-\gamma{})\Box{}\psi{} + 2(\varpi{}-\mathbf{e}^2/r)(-\gamma{})(r\overline{\partial}_r)\psi{} - \frac{2}{r}(-\nu{})(r^2\overline{\partial}_r)(r\psi{}).
\end{equation}
Specialize this equation to \(\psi{} = \psi_0\), integrate to \(C^{\text{out}}\), where
the data for \(\psi_0\) vanishes, and use
\cref{psi-0-assumption-1,psi-0-assumption-2} and the inductive hypothesis to obtain
\begin{equation}
\begin{split}
&\abs{(r^2\overline{\partial}_r)^2(r\overline{\partial}_r)^{n}S^m\varphi{}(u,v)} = \abs{(r^2\overline{\partial}_r)(r(r\overline{\partial}_r)^{n+1}S^m\varphi{})(u,v)} \\
&\le C(n,m,\mathfrak{D}_{n+m+5})\Bigl[\int_1^u r^{-1}(u',v)\dd{}u' + \int_1^u r^{-1}\abs{(r^2\overline{\partial}_r)^2(r\overline{\partial}_r)^{n}S^{m}\varphi{}}(u',v)\dd{}u'\Bigr]. \\
\end{split}
\end{equation}
Grönwall's inequality, the monotonicity \(r^{-1}(u',v)\le r^{-1}(u,v)\) for
\(u'\in [1,u]\), and the estimate \(ur^{-1}\le 1\) in the region of interest give
\begin{equation}
\abs{(r^2\overline{\partial}_r)^2(r\overline{\partial}_r)^{n}S^m\varphi{}(u,v)}\le C(n,m,\mathfrak{D}_{n+m+5})ur^{-1}\exp (C(n,m,\mathfrak{D}_{n+m+5})ur^{-1})\le C(n,m,\mathfrak{D}_{n+m+5})ur^{-1},
\end{equation}
This concludes the proof of \cref{rsr-rdr-S-bound} for \(k=2\) and the pair
\((n,m)\).
\end{proof}
\begin{lemma}
We have
\begin{equation}
\abs{U(r\overline{\partial}_r)^nS^m\varphi{}}\lesssim C(n,m,\mathfrak{D}_{n+m+2})[\abs{US^m\varphi{}} + r^{-1}\abs{(r\overline{\partial}_r)^{\le n+1}S^{\le m-1}\varphi{}} +  r^{-1}\abs{(r\overline{\partial}_r)^{\le n}S^{\le m}\varphi{}}].
\end{equation}
\label{UrS-wave-comm}
\end{lemma}
\begin{proof}
Compute the wave equation (in the region \(\set{r\ge 2\Rc}\))
\begin{equation}
U(r\overline{\partial}_r) = \mathcal{O}(r^{-1}\varpi{},(-\gamma{}))[r\Box{} + U + r^{-1}(r\overline{\partial}_r)]
\end{equation}
It follows that, for \(n\ge 1\),
\begin{equation}\label{UrS-wave-comm-prep-2}
U(r\overline{\partial}_r)^nS^m = \mathcal{O}(r^{-1}\varpi{},(-\gamma{}))[r\Box{}(r\overline{\partial}_r)^{n-1}S^m + U(r\overline{\partial}_r)^{n-1}S^m + r^{-1}(r\overline{\partial}_r)^nS^m].
\end{equation}
It follows from \cref{elliptic:rS-comm-improved,elliptic:pointwise-scalar-field} that
(in the region \(\set{r\ge 5\Rc}\))
\begin{equation}\label{UrS-wave-comm-prep-1}
\begin{split}
\abs{r^2\Box{}(r\overline{\partial}_r)^nS^m\varphi{}} &\lesssim C(n,m,\mathfrak{D}_{n+m+3})[\abs{(r\overline{\partial}_r)^{\le n+2}S^{\le m-1}\varphi{}} +  \abs{(r\overline{\partial}_r)^{\le n+1}S^{\le m}\varphi{}}].
\end{split}
\end{equation}
Now \cref{UrS-wave-comm-prep-1,UrS-wave-comm-prep-2} imply
\begin{equation}\label{UrS-wave-comm-prep}
\abs{U(r\overline{\partial}_r)^nS^m\varphi{}}\lesssim C(n,m,\mathfrak{D}_{n+m+2})[r^{-1}\abs{(r\overline{\partial}_r)^{\le n+1}S^{\le m-1}\varphi{}} +  r^{-1}\abs{(r\overline{\partial}_r)^{\le n}S^{\le m}\varphi{}} + \abs{U(r\overline{\partial}_r)^{n-1}S^m\varphi{}}]
\end{equation}
We obtain \cref{UrS-wave-comm} from \cref{UrS-wave-comm-prep} by induction on \(n\).
\end{proof}
\subsubsection{Estimates for \((r\overline{\partial}_r)\)-derivatives of
  \texorpdfstring{\(\varpi{}\)}{the renormalized Hawking mass}}
\label{sec:org9159aa4}
\begin{lemma}
For \(n\ge 1\) and \(m\ge 0\), we have
\begin{equation}\label{omega-bound-1}
r\tau{}^{1-\epsilon{}}\abs{(r\overline{\partial}_r)^nS^m\varpi{}}\le C(\epsilon{},n,m,\mathfrak{D}_{n+m+3})\quad \text{in }\set{r\ge 5\Rc}.
\end{equation}
For \(\abs{\alpha{}}\ge 1\), we have
\begin{equation}\label{omega-bound-2}
\tau{}^{1-\epsilon{}}\abs{\Gamma{}^\alpha{}\varpi{}}\le C(\epsilon{},\abs{\alpha{}},\mathfrak{D}_{\abs{\alpha{}}})\quad \text{in }\set{r\le 5\Rc}.
\end{equation}
For \(m\ge 1\), we have
\begin{equation}\label{omega-bound-3}
\tau{}^{1-\epsilon{}}\abs{S^m\varpi{}}\le C(\epsilon{},n,)\quad \text{in }\set{r\ge 5\Rc}.
\end{equation}
\label{elliptic:rS-omega-bound}
\end{lemma}
\begin{proof}
The transport equation for \(\varpi{}\) is
\begin{equation}
r^{-1}(r\overline{\partial}_r)\varpi{} = \frac{1}{2}(1-\mu{})((r\overline{\partial}_r)\varphi{})^2 =_{\mathrm{s}} \mathcal{O}(r^{-1}\varpi{})((r\overline{\partial}_r)\varphi{})^2.
\end{equation}
An induction argument implies that for \(n\ge 1\), we have
\begin{equation}
r^{-1}(r\overline{\partial}_r)^nS^m\varpi{} =_{\mathrm{s}} \mathcal{O}(r^{-1}S^{\le m}\varpi{},r^{-1}(r\overline{\partial}_r)^{1\le n-1}S^{\le m}\varphi{})((r\overline{\partial}_r)^{\le n}S^{\le m}\varphi{})^2.
\end{equation}
Now \cref{omega-bound-1} follows from
\cref{Gamma-omega-bounds,elliptic:pointwise-scalar-field}.

For \cref{omega-bound-2}, specialize \cref{omega-bound-penultimate} to the region \(\set{r\le 5\Rc}\) for \(\abs{\alpha{}}\ge 1\),
\begin{equation}
\abs{\Gamma{}^\alpha{}\varpi{}} \le \tau{}C(\mathfrak{B}_{<\alpha{}})\abs{\Gamma{}^{\le \alpha{}}\varphi{}}^2\le C(\epsilon{},\mathfrak{D}_{\abs{\alpha{}}})\tau{}^{-1+\epsilon{}}.
\end{equation}

It follows that, in the region \(\set{r\ge 5\Rc}\), for \(m\ge 1\),
\begin{equation}
\begin{split}
\abs{S^m\varpi{}}(u,r)&\le \abs{S^m\varpi{}}(u,5\Rc) + \int_{5\Rc}^r \abs{\overline{\partial}_rS^m\varpi{}}(u,r')\dd{}r' \\
&\le C(\Rc)\abs{S^m\varpi{}}(u,5\Rc) + \int_{5\Rc}^r\mathcal{O}(r^{-1}S^{\le m}\varpi{}) \abs{(r\overline{\partial}_r)S^{\le m}\varphi{}}^2(u,r')\dd{}r' \\
&\le C(\mathfrak{D}_m)\tau{}^{-1+\epsilon{}} + C(\mathfrak{D}_{m+3})\int_{5\Rc}^r r^{-2}\tau{}^{-1+\epsilon{}}\dd{}r' \\
&\le C(\mathfrak{D}_{m+3})\tau{}^{-1+\epsilon{}},
\end{split}
\end{equation}
which is \cref{omega-bound-3}.
\end{proof}
\begin{lemma}[Asymptotic expansion of \(\varpi\)]
In \(\set{r\ge 5\Rc}\cap \set{r\ge u}\), we have
\begin{equation}\label{omega-expansion-equation}
\varpi{}(u,r) = \varpi{}|_{\mathcal{I}}(u) + \mathcal{E}_\varpi{}(u,r),
\end{equation}
where \(\varpi_{\mathcal{I}}\) satisfies
\begin{equation}\label{omega-expansion-2}
\abs{(u\overline{\partial}_{u})^m\varpi{}|_{\mathcal{I}}(u)}\lesssim C(\epsilon{},m,\mathfrak{D}_m)\min (r^\epsilon{},u^\epsilon{}),
\end{equation}
and the error \(\mathcal{E}_\varpi{}\) satisfies
\begin{equation}\label{omega-expansion-1}
\abs{(r\overline{\partial}_r)^n(u\overline{\partial}_u)^m\mathcal{E}_\varpi{}}\lesssim C(\epsilon{},n,m,\mathfrak{D}_{n+m+5})r^{-1}\min (r^\epsilon{},u^\epsilon{}).
\end{equation}
\label{omega-expansion}
\end{lemma}
\begin{proof}
\step{Step 1: Proof of \cref{omega-expansion-1}.} Define
\begin{equation}
\mathcal{E}_\varpi{}(u,r) \coloneqq{} -\int_r^\infty \overline{\partial}_r\varpi{}(u,r')\dd{}r' =  -\int_r^\infty r'^{-2}(r'^2\overline{\partial}_r)\varpi{}(u,r')\dd{}r' = -\int_r^\infty \frac{1}{2}r'^{-2}(1-\mu{})((r'^2\overline{\partial}_r)\varphi{})^2(u,r')\dd{}r',
\end{equation}
so that \cref{omega-expansion-equation} holds with \(\varpi_{\mathcal{I}}(u) \coloneqq{} \lim_{r\to
\infty}\varpi{}(u,r)\). It follows from \cref{elliptic:rS-omega-bound,rsr-rdr-S-bound}
(after replacing instances of \(u\overline{\partial}_u\) with \(S - r\overline{\partial}_r\))
that, for \(n\ge 1\),
\begin{equation}\label{omega-expansion-1-prep-1}
\begin{split}
\abs{(r\overline{\partial}_r)^n(u\overline{\partial}_u)^m\mathcal{E}_\varpi{}}&\lesssim r^{-2}\abs{(r\overline{\partial}_r)^{\le n-1}(u\overline{\partial}_u)^{\le m}(1-\mu{})}\abs{(r^2\overline{\partial}_r)(r\overline{\partial}_r)^{\le n-1}(u\overline{\partial}_u)^{\le m}\varphi{}}^2 \\
&\lesssim C(n,m,\mathfrak{D}_{n+m+4})r^{-2}.
\end{split}
\end{equation}
When \(n = 0\), we differentiate under the integral sign and use
\cref{elliptic:rS-omega-bound,Gamma-omega-bounds,rsr-rdr-S-bound} (after replacing instances of
\(u\overline{\partial}_u\) with \(S - r\overline{\partial}_r\)) to get
\begin{equation}\label{omega-expansion-1-prep-2}
\begin{split}
\abs{(u\overline{\partial}_u)^m\mathcal{E}_\varpi{}}&\lesssim \int_r^\infty r'^{-2}\abs{(u\partial{}_u)^{\le m}(1-\mu{})}\abs{(r'^2\overline{\partial}_r)(u\overline{\partial}_u)^{\le m}\varphi{}}^2(u,r')\dd{}r'\\
&\lesssim C(\epsilon{},n,m,\mathfrak{D}_{n+m+5}) \int_r^\infty r'^{-2}\min (r^\epsilon{},u^\epsilon{})\dd{}r'\lesssim C(\epsilon{},n,m,\mathfrak{D}_{n+m+5})r^{-1}\min (r^\epsilon{},u^\epsilon{}).
\end{split}
\end{equation}
Use \cref{omega-expansion-1-prep-1,omega-expansion-1-prep-2} to conclude
\cref{omega-expansion-1}.

\step{Step 2: Proof of \cref{omega-expansion-2}.} By similar arguments to those in
\cref{elliptic:rS-omega-bound}, one can use \cref{rsr-rdr-S-bound} to obtain
\(\abs{(r^2\overline{\partial}_r)(u\overline{\partial}_u)^m\varpi{}}\lesssim
C(m,\mathfrak{D}_{m+5})\). This shows that
\((u\overline{\partial}_u)^m\varpi{}\) is integrable in \(r\) towards
\(\mathcal{I}\), uniformly on compact subsets of \(u\). It follows from a fact
in real analysis (as in the proof of \cref{gamma-bound-step-3}) that
\((u\overline{\partial}_u)^m\varpi{}|_{\mathcal{I}}(u) = \lim_{r\to
\infty}((u\overline{\partial}_u)^m\varpi{})(u,r)\). Then \cref{omega-expansion-2}
follows from \cref{elliptic:rS-omega-bound,Gamma-omega-bounds} (after replacing
instances of \(u\overline{\partial}_u\) with \(S- r\overline{\partial}_r\)).
\end{proof}
\subsubsection{Estimates for \((r\overline{\partial}_r)\)-derivatives of \((-\gamma{})\)}
\label{sec:org1c7ff08}
\begin{lemma}
For \(n,m\ge 0\), we have
\begin{equation}\label{UrS-gamma-estimate}
r\tau{}\abs{U^{\le 1}(r\overline{\partial}_r)^nS^m\log (-\gamma{})}\le C(n,m,\mathfrak{D}_{n+m+3})\quad \text{in }\set{r\ge 5\Rc}.
\end{equation}
\label{elliptic:rS-gamma-bound}
\end{lemma}
\begin{proof}
When \(n = 0\), this follows from \cref{Gamma-gamma-bounds}. When \(n\ge 1\), the
transport equation for \(\log (-\gamma{})\) gives
\begin{equation}
\begin{split}
\abs{U^{\le 1}(r\overline{\partial}_r)^nS^m\log (-\gamma{})} &=  U^{\le 1}(r\overline{\partial}_r)^{n-1}S^m(r\overline{\partial}_r)\log (-\gamma{})  =  \abs{U^{\le 1}(r\overline{\partial}_r)^{n-1}S^m((r\overline{\partial}_r)\varphi{})^2} \\
&\lesssim \abs{(r\overline{\partial}_r)^{\le n}S^{\le m}\varphi{}}^2 + \abs{U(r\overline{\partial}_r)^{\le n}S^{\le n}\varphi{}}^2.
\end{split}
\end{equation}
Use \cref{UrS-wave-comm,elliptic:pointwise-scalar-field}, and the definition of the
pointwise norm \(\mathcal{P}_{\alpha{},1}\le C(\abs{\alpha{}},\mathfrak{D}_{\abs{\alpha{}}})\) to estimate
\begin{equation}
\begin{split}
\abs{U^{\le 1}(r\overline{\partial}_r)^nS^m\log (-\gamma{})} &\lesssim C(\mathfrak{D}_{n+m+2})[\abs{(r\overline{\partial}_r)^{\le n}S^{\le m}\varphi{}}^2 + \abs{(r\overline{\partial}_r)^{\le n+1}S^{\le m-1}\varphi{}}^2 + \abs{US^m\varphi{}}^2] \\
&\lesssim C(n,m,\mathfrak{D}_{n+m+3})r^{-1}\tau{}^{-1}.
\end{split}
\end{equation}
\end{proof}
\begin{lemma}[Asymptotic expansion of \((-\gamma{})\)]
In \(\set{r\ge 5\Rc}\cap \set{r\ge u}\), we have
\begin{equation}\label{gamma-expansion-equation}
(-\gamma{})(u,r) = 1 + r^{-2}\mathcal{F}_{(-\gamma{})}(u) + \mathcal{E}_{(-\gamma{})}(u,r),
\end{equation}
where \(\mathcal{F}_{(-\gamma{})}\) and \(\mathcal{E}_{(-\gamma{})}\) satisfy
\begin{equation}\label{gamma-expansion-1}
\abs{(u\overline{\partial}_u)^m\mathcal{F}_{(-\gamma{})}}\le C(m,\mathfrak{D}_{m+5})
\end{equation}
and
\begin{equation}\label{gamma-expansion-2}
\abs{(r\overline{\partial}_r)^n(u\overline{\partial}_u)^m\mathcal{E}_{(-\gamma{})}}\le C(n,m,\mathfrak{D}_{n+m+5})r^{-3}
\end{equation}
\label{gamma-expansion}
\end{lemma}
\begin{proof}
\step{Step 1: Preliminary computations and the derivation of
\cref{gamma-expansion-equation}.} Write \(K = r^2\overline{\partial}_r\). Compute
\begin{align}
K(-\gamma{}) &= r^{-1}(K\varphi{})^2(-\gamma{}) \label{K-gamma}\\
K^2(-\gamma{}) &= -(K\varphi{})^2(-\gamma{}) + r^{-1}K^2\varphi{}\cdot K\varphi{}(-\gamma{}) + r^{-2}(K\varphi{})^4(-\gamma{}) \label{K2-gamma}\\
K^3(-\gamma{}) &= -3K^2\varphi{}\cdot K\varphi{}(-\gamma{}) - 3r^{-1}(K\varphi{})^4(-\gamma{}) + r^{-1}(K^2\varphi{})^2(-\gamma{}) + 5r^{-2}K^2\varphi{}\cdot (K\varphi{})^3(-\gamma{})\notag  \\
&\qquad + r^{-3}(K\varphi{})^6(-\gamma{}) + (r\overline{\partial}_r)K^2\varphi{}\cdot K\varphi{}(-\gamma{}) \label{K3-gamma}
\end{align}
Next, define
\begin{equation}
\mathcal{E}_{(-\gamma{})}(u,r) \coloneqq{} \int_r^\infty \rho{}^{-2}\int_\rho^\infty \rho{}'^{-2}\int_{\rho{}'}^\infty \rho{}''^{-2}K^3(-\gamma{})(u,\rho{}'')\dd{}\rho{}''\dd{}\rho{}'\dd{}\rho{},
\end{equation}
so that
\begin{equation}
\begin{split}
1 - (-\gamma{})(u,r) &= \int_r^\infty \rho{}^{-2}K(-\gamma{})(u,\rho{})\dd{}\rho{} = -r^{-1}\lim_{r\to \infty}K(-\gamma{})(u,r) - \int_r^\infty \rho{}^{-2}\int_\rho^\infty \rho{}'^{-2}K^2(-\gamma{})(u,\rho{}')\dd{}\rho{}'\dd{}\rho{}  \\
&= -r^{-1}\lim_{r\to \infty}K(-\gamma{})(u,r) - 2r^{-2}\lim_{r\to \infty}K^2(-\gamma{})(u,r) -\mathcal{E}_{(-\gamma{})}(u,r).
\end{split}
\end{equation}
From \cref{K-gamma,K2-gamma}, the boundedness statement of \cref{rsr-rdr-S-bound}, and the gauge condition \((-\gamma{})|_{\mathcal{I}} = 1\), we see that
\begin{equation}
\lim_{r\to \infty}K(-\gamma{})(u,r) = 0, \qquad \lim_{r\to \infty}K^2(-\gamma{})(u,r) = -\lim_{r\to \infty}(K\varphi{})^2(u,r)
\end{equation}
It follows that \cref{gamma-expansion-equation} holds for \(\mathcal{F}_{(-\gamma{})}(u) =
-2\lim_{r\to \infty}(K\varphi{})^2(u,r)\).

\step{Step 2: Proof of \cref{gamma-expansion-1}.} This is proved with a similar method to
\cref{omega-expansion-2}. The important points are that, by \cref{rsr-rdr-S-bound}, \(K(K\varphi{})^2\) is bounded
(which provides the integrability in \(r\) towards infinity, uniformly on
compact subsets of \(u\)), and \(K(u\overline{\partial}_u)^{\le m}\varphi{}\) is bounded (since we
can replace instances of \(u\overline{\partial}_u\) by \(S - r\overline{\partial}_r\)).

\step{Step 3: Proof of \cref{gamma-expansion-2}.} All terms on the right side of
\cref{K3-gamma} are bounded by \(C(n,m,\mathfrak{D}_{n+m+5})\) after applications
of \((r\overline{\partial}_r)^n(u\overline{\partial}_u)^m\), by
\cref{rsr-rdr-S-bound} (since we can replace instances of
\(u\overline{\partial}_u\) by \(S- r\overline{\partial}_r\)). It follows that
all terms in
\((r\overline{\partial}_r)^n(u\overline{\partial}_u)^m\mathcal{E}_{(-\gamma{})}\)
are bounded by \(r^{-3}\) (with the constant just mentioned).
\end{proof}
\subsection{Retrieving the assumptions of the work of Luk--Oh on late-time tails}
\label{sec:org1c48a1c}
\subsubsection{Notation of the work of Luk--Oh}
\label{sec:org1f31c77}
We define notation as in \cite[Sec.~2.1.1--2.1.3]{luk-oh-tails}. Define the following parameters:
\begin{itemize}
\item \(R_{\text{far}} = 5\Rc\),
\item \(M_c\in \Z_{\ge 0}\),
\item \(\delta{}_c = 1/2\)
\item \(J_c = 3\),
\item \(K_c = 0\),
\item \(\eta_c = 1\),
\item \(M_0 \in \Z_{\ge 0}\),
\item \(A_0 = C(\varpi_i,c_{\mathcal{H}},r_{\text{min}},M_0,\mathfrak{D}_{M_0 + 5})\),
\item \(\alpha_0 = 3/4\),
\item \(\nu_0 = 1/4\),
\item \(D = C(\mathfrak{D}_{M_0})\),
\item \(\alpha_d = 4\),
\item \(\delta_d\in (0,1]\),
\item \(J_d = 3\),
\item \(K_d = 0\),
\item \(\eta_d = 1\),
\end{itemize}
Let \(\mathbf{g} = g\) be the spacetime metric on \(\mathcal{M}\). We now define
Cartesian coordinates on \(\mathcal{M}\). Let \((\theta{},\phi{})\) be the
standard local coordinate system on \(S^2\). Set
\begin{equation}
\begin{split}
x^0 &= \chi{}(r)(2v-r) + (1-\chi{}(r))(2u + r)
\end{split}
\end{equation}
and \(x^{1 }= r\sin \theta{}\cos \phi{}\), \(x^2 = r\sin \theta{}\sin \phi{}\), and \(x^3 = r\cos \theta{}\). Set
\begin{equation}
\mathbf{T} \coloneqq{} \partial{}_{x^0},
\end{equation}
where \(\partial_{x^0}\) is defined with respect to the coordinates
\((x^0,x^1,x^2,x^3)\). Define
\begin{equation}
\bar{u} \coloneqq{} x^0 - r + 3R_{\text{far}}\qquad \conj{\tau{}}\coloneqq{}x^0 - \chi{}_{>4R_{\text{far}}}(r)(r - 3R_{\text{far}})
\end{equation}
for a cutoff function \(\chi_{>4R_{\text{far}}}(r)\) that is \(0\) when \(r\le
2R_{\text{far}}\) and \(1\) when \(r\ge 4R_{\text{far}}\). In particular, observe
that
\begin{equation}
\bar{u} = 2u + 3{R_{\text{far}}}\quad \text{in }\set{r\ge R_{\text{far}}},
\end{equation}
and
\begin{equation}
\mathbf{T}\conj{\tau{}} = 1, \quad \text{and }\conj{\tau{}} = \begin{cases}
x^0, & \text{in }\set{r\le 2R_{\text{far}}}, \\
\bar{u}, & \text{in }\set{r\ge 4R_{\text{far}}}.
\end{cases}
\end{equation}
Write \((\partial_\tau^{(\tau{})},\partial_r^{(\tau{})})\) for the coordinate derivatives in the
\((\tau{},r)\) coordinate system and
\((\underline{\partial{}}_{\bar{u}},\overline{\partial}_r)\) for the coordinate
derivatives in \((\bar{u},r)\)-coordinates. We will write \(\mathcal{M} =
\mathcal{M}_{\text{near}}\cup \mathcal{M}_{\text{med}}\cup
\mathcal{M}_{\text{wave}}\) for
\begin{equation}
\begin{split}
\mathcal{M}_{\text{near}} &\coloneqq{} \set{(u,r)\in \mathcal{M} : r\le 2R_{\text{far}}, \conj{\tau{}}\ge 1}, \\
\mathcal{M}_{\text{med}} &\coloneqq{} \set{(u,r)\in \mathcal{M} : R_{\text{far}}\le r\le 400\bar{u}, \conj{\tau{}}\ge 1}, \\
\mathcal{M}_{\text{wave}} &\coloneqq{} \set{(u,r)\in \mathcal{M} : r\ge 4\bar{u}, \conj{\tau{}}\ge 1}. \\
\end{split}
\end{equation}
and write \(\mathcal{M}_{\text{far}} \coloneqq{} \mathcal{M}_{\text{med}}\cup
\mathcal{M}_{\text{wave}}\).
\subsubsection{Global assumptions on the spacetime}
\label{sec:org81c26ba}
Assumption \(\text{(G1)}\) is clearly satisfied. Assumption \(\text{(G2)}\) is satisfied, since the only boundary component of
\(\mathcal{M}\) is the null event horizon. Finally, it is easy to check the causality condition \(\text{(G3)}\) based on the
definition of \(x^0\).
\subsubsection{Assumptions on the metric}
\label{sec:org1ea8a47}
Define the absolute value of a \((2,0)\)-tensor as in
\cite[Sec.~2.1.4]{luk-oh-tails}. A tedious but straightforward computation
involving the expression for a Lie derivative in local coordinates, the chain
rule, the definitions of our schematic geometric quantities, and the results of
\cref{sec:putting-it-all-together} yields the following expression for these
absolute values in terms of components in \((u,r)\) and \((r,v)\)-coordinates.
\begin{lemma}
Fix a function \(m(\conj{\tau{}},r)\). Let \(\mathbf{a}\) be \((2,0)\)-tensor
\begin{equation}
\mathbf{a} = \mathbf{a}^{ur}(\overline{\partial}_u\otimes \overline{\partial}_r + \overline{\partial}_r\otimes \overline{\partial}_u) + \mathbf{a}^{rr}\overline{\partial}_r\otimes \overline{\partial}_r + r^{-2} \slashed{\mathbf{a}}
\end{equation}
for \(\slashed{\mathbf{a}}^{AB}\) a \((2,0)\)-tensor on \(S^2\). Then in
\(\mathcal{M}_{\text{far}}\) we have
\begin{equation}
\abs{\mathcal{L}_{\tau{}\mathbf{T}}^n\mathcal{L}_{\langle{}r\rangle{}\partial{}_r^{(\tau{})}}^m\mathbf{a}}\lesssim_{n,m} C(\mathfrak{D}_{n+m})[\abs{(u\overline{\partial}_u)^{\le m}(r\overline{\partial}_r)^{\le n}\mathbf{a}^{ur}} + \abs{(u\overline{\partial}_u)^{\le m}(r\overline{\partial}_r)^{\le n}\mathbf{a}^{rr}} + r^{-2}\abs{\slashed{\mathbf{a}}^{AB}}].
\end{equation}
Similarly, if \(\mathbf{b}\) takes the form
\begin{equation}
\mathbf{b} = \mathbf{b}^{vr}(\underline{\partial}_v\otimes \underline{\partial}_r + \underline{\partial}_r\otimes \underline{\partial}_v) + \mathbf{b}^{rr}\underline{\partial}_r\otimes \underline{\partial}_r + r^{-2} \slashed{\mathbf{b}},
\end{equation}
then in \(\mathcal{M}_{\mathrm{near}}\) we have
\begin{equation}
\abs{\mathcal{L}_{\tau{}\mathbf{T}}^n\mathcal{L}_{\langle{}r\rangle{}\partial{}_r^{(\tau{})}}^m\mathbf{b}}\lesssim_{n,m} C(\mathfrak{D}_{n+m})\sum_{\substack{n',m'\ge 0 \\ n'+m' = n+m}}[\abs{S^{\le m'}D^{\le n'}\mathbf{b}^{vr}} + \abs{S^{\le m'}D^{\le n'}\mathbf{b}^{rr}} + r^{-2}\abs{\slashed{\mathbf{b}}^{AB}}.
\end{equation}
Finally, if \(a\) is a spherically symmetric smooth function, such that
\begin{equation}
\sum_{m+n\le N}\abs{S^m(r\overline{\partial}_r)^na}\lesssim m(\conj{\tau{}},r)\quad \text{in }\mathcal{M}_{\mathrm{far}}, \quad \text{and}\quad  \sum_{m+n\le N}\abs{S^mD^na}\lesssim m(\tau{},r)\quad \text{in }\mathcal{M}_{\mathrm{near}},
\end{equation}
then
\begin{equation}
a = O_{\mathbf{\Gamma{}}}^N(m(\tau{},r)),
\end{equation}
where the \(O_{\mathbf{\Gamma{}}}\) notation is defined as in \cite[Sec.~2.1.4]{luk-oh-tails}.
\label{tensor-absolute-value}
\end{lemma}
One computes
\begin{equation}\label{metric-far-expansion}
\mathbf{g}^{-1} = -\frac{1}{(-\gamma{})}(\overline{\partial}_{\bar{u}}\otimes \overline{\partial}_r + \overline{\partial}_r\otimes \overline{\partial}_{\bar{u}}) + (1-\mu{})\overline{\partial}_r\otimes \overline{\partial}_r + r^{-2}\mathring{\slashed{\gamma{}}}^{-1}
\end{equation}
and
\begin{equation}\label{metric-near-expansion}
\mathbf{g}^{-1} = \frac{1}{\kappa{}}(\underline{\partial{}}_r\otimes \underline{\partial{}}_v + \underline{\partial{}}_v\otimes \underline{\partial{}}_r) + (1-\mu{})\underline{\partial{}}_r\otimes \underline{\partial{}}_r + r^{-2}\mathring{\slashed{\gamma{}}}^{-1}.
\end{equation}
for \(\mathring{\slashed{\gamma{}}}\) the round metric on the unit sphere.

In \(\mathcal{M}_{\text{near}}\), by
\cref{Gamma-kappa-bounds,metric-near-expansion,tensor-absolute-value}, we have
\begin{equation}
\mathbf{g}^{-1} = O_{\mathbf{\Gamma{}}}^{M_c}(1),
\end{equation}
which satisfies assumption \((\mathbf{g}\mathbf{B}V1)\).

For \(\mathbf{m}^{-1}\) the inverse Minkowski metric, we have
\begin{equation}
\mathbf{g}^{-1} - \mathbf{m}^{-1} = (1 - (-\gamma{})^{-1})(\overline{\partial}_{\bar{u}}\otimes \overline{\partial}_r + \overline{\partial}_r\otimes \overline{\partial}_{\bar{u}}) - \Bigl(\frac{2\varpi{}}{r} - \frac{\mathbf{e}^2}{r^2}\Bigr)\overline{\partial}_r\otimes \overline{\partial}_r
\end{equation}
By \cref{gamma-expansion,omega-expansion,metric-far-expansion,tensor-absolute-value}, and a Taylor
expansion, we have
\begin{equation}
\mathbf{g}^{-1} - \mathbf{m}^{-1} = O_{\mathbf{\Gamma{}}}^{M_c}(r^{-1+\epsilon{}})\quad \text{in }\mathcal{M}_{\text{med}},
\end{equation}
where the implicit constant depends on \(\epsilon{}\), \(M_c\), and \(\mathfrak{D}_{M_c +
5}\). This satisfies assumption \((\mathbf{gB}V2)\) (with \(\delta_c = 1/2\)).

By \cref{gamma-expansion,omega-expansion} and Taylor expansion, in
\(\mathcal{M}_{\text{wave}}\) we have the asymptotics
\begin{equation}
\begin{split}
\mathbf{g}^{\conj{u}r} &= -1 + r^{-2}\mathring{\mathbf{h}}^{\conj{u}r}_{2,0}(u) + \varrho{}_{3}[\mathbf{h}^{\conj{u}r}], \\
\mathbf{g}^{rr} &= r^{-1}\mathring{\mathbf{h}}^{rr}_{1,0}(u) + \varrho{}_{2}[\mathbf{h}^{rr}], \\
r^2\mathbf{g}^{AB} &= \mathring{\slashed{\gamma{}}}^{AB},
\end{split}
\end{equation}
where the unlisted components vanish, and where, in the notation of \cite[Sec.~2.1.4]{luk-oh-tails},
\begin{equation}
\begin{split}
\mathring{\mathbf{h}}^{\conj{u}r}_{2,0} &=  O_{\mathbf{\Gamma{}}^{M_c}}(1), \quad \text{and}\quad \mathring{\mathbf{h}}^{rr}_{1,0} = 2\varpi{}|_{\mathcal{I}}(u) = O_{\mathbf{\Gamma{}}}(u^\epsilon{}),
\end{split}
\end{equation}
and
\begin{equation}
\varrho_{3}[\mathbf{h}^{\conj{u}r}] = O_{\mathbf{\Gamma{}}}^{M_c}(r^{-3}),\quad \text{and}\quad \varrho_{2}[\mathbf{h}^{rr}] =
\mathcal{E}_\varpi{} + r^{-2}\mathbf{e}^2 = O_{\mathbf{\Gamma{}}}^{M_c}(r^{-2}u^\epsilon{}).
\end{equation}
This satisfies assumption \((\mathbf{gB}V3)\) (with \(\eta_c = 1\) and \(\delta_c =
1/2\)).
\subsubsection{Stationary estimate}
\label{sec:org1daea2f}
The assumption \(\text{(SE1)}\) is trivially satisfied, since \(\mathcal{M}\)
has no timelike boundary components.

Define the inverse metric
\begin{equation}
^{(\infty)}\mathbf{g}^{-1} \coloneqq{} -(\overline{\partial}_{\bar{u}}\otimes \overline{\partial}_r + \overline{\partial}_r\otimes \overline{\partial}_{\bar{u}}) + \Bigl(1-\frac{\varpi{}_f}{r} + \frac{\mathbf{e}^2}{r^2}\Bigr)\overline{\partial}_r\otimes \overline{\partial}_r + r^{-2}\mathring{\slashed{\gamma{}}}^{-1}
\end{equation}
in \((\bar{u},r)\)-coordinates. Since \(\mathbf{T}r = 0\), we have the
stationarity property \(\mathcal{L}_{\mathbf{T}}(^{(\infty)}\mathbf{g}) = 0\).
The stationary estimate for the Reissner--Nordström metric
\(^{(\infty)}\mathbf{g}\) (with mass \(\varpi_f\) and charge \(\mathbf{e}\))
follows from the methods in \cite[Sec.~7]{ANGELOPOULOS2023108939} (see also the discussion in
\cite[Ex.~3.6]{luk-oh-tails}). As for the convergence of \(\mathbf{g}^{-1}\) to
\(^{(\infty)}\mathbf{g}^{-1}\), we have
\begin{equation}
\mathbf{g} - \,^{(\infty)}\mathbf{g}^{-1} \coloneqq{} (1 - (-\gamma{})^{-1})(\overline{\partial}_{\bar{u}}\otimes \overline{\partial}_r + \overline{\partial}_r\otimes \overline{\partial}_{\bar{u}}) + \frac{\varpi{}-\varpi{}_f}{r}\overline{\partial}_r\otimes \overline{\partial}_r.
\end{equation}
On \(\Sigma{}_{\bar{\tau{}}}\cap \set{r\ge R_{\text{far}}}\) (where \(\Sigma{}_{\bar{\tau{}}}\) is defined as
in \cite[Sec.~2.1.3]{luk-oh-tails}), the assumption \(\text{(SE2)}\) (with
\(\delta_c = 1/2\)) follows from
\cref{elliptic:rS-gamma-bound,elliptic:rS-omega-bound} (and an easy decay estimate
for \(\varpi{}-\varpi_f\) that follows from the energy decay established in
\cref{sec:energy-estimates}). The verification of assumption \(\text{(SE2)}\) in
the region \(\Sigma{}_{\tau{}}\cap \set{r\le R_{\text{far}}}\) is done in
\((v,r)\)-coordinates, and it is similar (but uses \cref{Gamma-kappa-bounds}). Thus
assumption \(\text{(SE2)}\) holds.
\subsubsection{Assumptions on the initial data}
\label{sec:org6a2987a}
The assumption \(\text{(D}_{\Sigma_1}\text{)}\) on the initial data is satisfied for
our compactly supported data (with \(D = C(\mathfrak{D}_{M_0})\)), for any
\(\alpha_d\in \R\) and \(\delta{}_d\in (0,1]\). Note that, although the data hypersurface \(\Sigma_1\)
leaves the region \(\mathcal{R}_{\text{char}}\), the scalar field vanishes in the region \(\Sigma_1 \cap \mathcal{R}_{\text{char}}^c\).
\subsubsection{Assumptions on the scalar field}
\label{sec:org3e165ee}
By the results of \cref{sec:putting-it-all-together} and \cref{tensor-absolute-value},
we have
\begin{equation}
\abs{S^mD^n\varphi{}}\lesssim v^{-1+\epsilon{}}\implies \varphi{} = O_{\mathbf{\Gamma{}}}^{M_0}(A_0\bar{\tau{}}^{-1+\epsilon{}})\quad \text{in }\mathcal{M}_{\text{near}}.
\end{equation}
By \cref{elliptic:pointwise-scalar-field,tensor-absolute-value}, we have
\begin{equation}
\abs{S^m(r\overline{\partial}_r)^n\varphi{}}\lesssim u^{-1+\epsilon{}}\implies \varphi{} = O_{\mathbf{\Gamma{}}}^{M_0}(A_0\bar{u}^{-1+\epsilon{}})\quad \text{in }\mathcal{M}_{\text{med}}.
\end{equation}
By \cref{elliptic:pointwise-scalar-field,tensor-absolute-value}, we have
\begin{equation}
\abs{S^m(r\overline{\partial}_r)^n\varphi{}}\lesssim r^{-1/2-\epsilon{}}u^{-1+2\epsilon{}}\implies \varphi{} = O_{\mathbf{\Gamma{}}}^{M_0}(A_0r^{-1/2-\epsilon{}}\bar{u}^{-1+2\epsilon{}})\quad \text{in }\mathcal{M}_{\text{wave}}.
\end{equation}
Taking \(\epsilon{} = 1/4\), the above three implications show that assumption
\(\text{(S)}\) is satisfied with \(\alpha_0 = 3/4\) and \(\nu_0 = 1/4\).
\subsection{Obtaining a late-time tails result from the work of Luk--Oh}
\label{sec:org5bf1f2a}
We begin by showing that the first higher radiation field vanishes (i.e. the
radiation field \(\Phi{}\) has no \(r^{-1}\) term in its expansion).
\begin{lemma}
In the notation of \cite[Sec.~2.3]{luk-oh-tails}, we have
\begin{equation}
\mathring{\Phi{}}_{1,0}(u) \equiv  0.
\end{equation}
\label{vanishing-radiation-field}
\end{lemma}
\begin{proof}
Define \(\Phi{}\coloneqq{}r\varphi{}\). From \cref{dudv-rpsi}, we have
\begin{equation}
\partial{}_u\partial{}_v\Phi{} = - \frac{2(\varpi{}-\mathbf{e}^2/r)}{r^3}\kappa{}(-\nu{})\Phi{}\implies \partial{}_u(r^3\partial{}_v\Phi{}) + \frac{3}{r}(-\nu{})r^3\partial{}_v\Phi{} = -2(\varpi{}-\mathbf{e}^2/r)\kappa{}(-\nu{})\Phi{} = O(u^{-1/2+\epsilon{}}).
\end{equation}
By Grönwall's inequality and the vanishing of the data on \(C^{\text{out}}\), we have
\begin{equation}
\abs{r^3\partial{}_v\Phi{}}\lesssim u^{1/2+\epsilon{}}e^{Cu/r}.
\end{equation}
In particular, if we define \(\mathring{\Phi{}}_0(u) \coloneqq{} \lim_{v\to \infty}\Phi{}(u,v)\), then
(for \(r(u,v)\ge u\)) we have
\begin{equation}
\abs{\Phi{}(u,v) - \mathring{\Phi{}}_0(u)}\le \int_v^\infty \abs{\partial{}_v\Phi{}(u,v')}\dd{}v'\lesssim r^{-2}u^{1/2+\epsilon{}}.
\end{equation}
In particular, it follows that
\begin{equation}
\mathring{\Phi{}}_{1,0}(u) \coloneqq{} \lim_{v\to \infty} r(u,v)(\Phi{}(u,v) - \mathring{\Phi{}}_0(u))\equiv 0,
\end{equation}
as desired.
\end{proof}
\begin{theorem}[Sharp Price's law result]
There exists a universal large constant \(N\ge 0\) and small constants \(\epsilon{},\delta{} > 0\)
such that for every \(k\ge 0\), we have the following asymptotic along the event
horizon
\begin{equation}\label{tails-theorem-1}
\abs{(v\underline{\partial{}}_v)^k\varphi{}|_{\mathcal{H}}(v) - C_k\mathfrak{L}v^{-3}} = O(v^{-3-\delta{}}),
\end{equation}
where the implicit constant depends on \(\varpi_i\), \(c_{\mathcal{H}}\),
\(r_{\mathrm{min}}\), \(k\), and \(\mathfrak{D}_{\max(N,\epsilon^{-1}k)}\). Moreover,
for \(0\le k\le 2\), we have
\begin{equation}\label{tails-theorem-2}
\abs{(\bar{v}\partial{}_{\bar{v}})^k\varphi{}|_{\mathcal{H}}(v) - C_k\mathfrak{L}\bar{v}^{-3}} = O(\bar{v}^{-3-\delta{}}),
\end{equation}
where \(\bar{v}\) and \(\partial_{\bar{v}}\) are the coordinate and associated
derivative in the Eddington--Finkelstein-type gauge (see
\cref{sec:coordinate-systems}), and where the implied constant depends on the same
parameters as in \cref{tails-theorem-1}.
\label{tails-theorem-precise}
\end{theorem}
\begin{proof}
Define \(J_{\mathfrak{f}}\coloneqq{}2\) and \(\alpha_{\mathfrak{f}} \coloneqq{} 3\). Then \cite[Eq.~(2.83)]{luk-oh-tails} holds. Evidently, the asymptotic spatial profile
defined in \cite[Eq.~(2.90)]{luk-oh-tails} is identically equal to \(1\). In
spherical symmetry and with \(K_1 = 0\), we have
\begin{equation}
\stackrel{\mathclap{\bullet{}}}{\Phi{}}_{(0)1,0}\!(u) = \,\stackrel{\mathclap{\bullet{}}}{\Phi{}}_{1,0}\!(u) = \mathring{\Phi{}}_{1,0}(u),
\end{equation}
so \cref{vanishing-radiation-field} implies \cite[Eq.~(2.73)]{luk-oh-tails}. In the
previous section we showed that all the main assumptions of \cite{luk-oh-tails}
are satisfied. An appeal to \cite[Main Theorem 4]{luk-oh-tails} shows that there
exists \(\epsilon{} > 0\) small, \(M\ge 0\) large, and an explicit constant \(C\neq{}0\) such
that for \(M_0,M_c\ge M\) and \(\delta{}>0\) sufficiently small, we have
\begin{equation}
\varphi{}(\conj{\tau{}},r) = C\mathfrak{L}\conj{\tau{}}^{-3} + O_{\mathbf{\Gamma{}}}^{\floor{\epsilon{}M}}((D+A_0)\conj{\tau{}}^{-3-\delta{}})\quad \text{in }\set{r\le \conj{\tau{}}^{1-2\delta{}}},
\end{equation}
where \(\mathfrak{L} = \int _{\mathcal{I}} 2\varpi{}|_{\mathcal{I}}(u)\varphi{}(u)\dd{}u\) is non-zero
for solutions arising from generic Cauchy data, by \cite[Thm.~4.3]{luk-oh-scc2} (to
relate this expression for \(\mathfrak{L}\) to the notion of final asymptotic
charge in \cite[Sec.~2.r5]{luk-oh-tails}, one integrates the recurrence equations
in \cite[Sec.~2.3]{luk-oh-tails}). Specializing to a region \(\set{r\le R_0}\), where
\(\conj{\tau{}} = 2v - r\), and performing a Taylor expansion, we get
\begin{equation}\label{phi-asymptotic}
\varphi{}(r,v) = C\mathfrak{L}v^{-3} + O_{\mathbf{\Gamma{}}}^{\floor{\epsilon{}M}}(v^{-3-\delta{}}),
\end{equation}
where the implicit constant depends on \(\varpi_i\), \(c_{\mathcal{H}}\),
\(r_{\text{min}}\), and \(\mathfrak{D}_M\). Differentiating this asymptotic
yields \cref{tails-theorem-1}.

For the result in the Eddington--Finkelstein-type gauge, one notes that the
argument of \cref{S-vdv-comparison} (see also
\cref{S-lambda-horizon-decay,lambda-horizon-bound}) provides, on the event horizon,
\begin{equation}
\begin{split}
S - \bar{v}\partial_{\bar{v}} &=  \rho{}S + (1-\rho{})v\lambda{}U \\
S^2 - (\bar{v}\partial{}_{\bar{v}})^2 &=_{\mathrm{s}} \sum_{\substack{a,b,c\le 1 \\ a + b + c\ge 1}}(S^{\le 1}\rho{})^{a}(S^{\le 1}(v\lambda{}))^{b}(v\lambda{}U(v\lambda{}))^{c}\set{U,S}^{1\le 2}, \\
\end{split}
\end{equation}
for some quantity \(\rho{}\), together with the estimates
\begin{equation}
\abs{S^{\le 1}\rho{}}\lesssim v^{-1},\qquad \abs{S^{\le 1}(v\lambda{})}\lesssim v^{-1},\qquad \abs{v\lambda{}U(v\lambda{})}\lesssim v^{-1}.
\end{equation}
The upshot is that for \(0\le k\le 2\),
\begin{equation}
\abs{S^k\varphi{} - (\bar{v}\partial{}_{\bar{v}})^k\varphi{}}\lesssim v^{-1}\abs{\set{U,S}^{1\le 2}\varphi{}}\lesssim v^{-4},
\end{equation}
where the final estimate is due to \cref{phi-asymptotic}. Noting that \(S =
v\underline{\partial{}}_v\) on the horizon, \cref{tails-theorem-2} now follows from
the above equation and \cref{tails-theorem-1}.
\end{proof}
\appendix
\section{Proofs of vector field commutator calculations}
\label{sec:org5969797}
\label{sec:vf-commutator-proofs}
\begin{proof}[Proof of \cref{V-dv-coordinate-change}]
From
\begin{equation}\label{V-dv-U-expression}
V = (\chi{} + \lambda{}^{-1}(1-\chi{}))\partial{}_v + \chi{}\lambda{}U,
\end{equation}
we obtain \cref{V-dv-coordinate-change-1} and
\begin{equation}
V - \overline{\partial}_r = \chi{}(\lambda{}-1)\overline{\partial}_r + \chi{}\lambda{}U,
\end{equation}
which gives \cref{V-dr-coordinate-change}. From \cref{V-dv-U-expression}, one obtains \cref{dv-in-terms-of-V-and-U}:
\begin{equation}
\partial{}_v = F^{-1}V - F^{-1}\chi{}\lambda{}U\qquad \text{for }F = (\chi{} + (1-\chi{})/\lambda{}) =_{\mathrm{s}} \mathfrak{b}_0.
\end{equation}
\end{proof}
\begin{proof}[Proof of \cref{S-in-terms-of-U-V}]
One obtains \cref{S-in-terms-of-U-V} by computing
\begin{equation}
S = \chi{}(v\partial{}_v + v\lambda{}U) + (1-\chi{})((r + u(-\nu{}))\overline{\partial}_r + u(-\nu{})U)
\end{equation}
and using \cref{dv-in-terms-of-V-and-U}.
\end{proof}
\begin{proof}[Proof of \cref{dr-comm-1}]
Use the transport equations for \((-\nu{})\) and \(\lambda{}\) to compute
\begin{equation}\label{app:dr-U-comm}
[\overline{\partial}_r,U] =_{\mathrm{s}} r^{-2}\set{1,\varpi{},(1-\mu{})^{-1}}[\overline{\partial}_r + U],
\end{equation}
and then use \(\mathbf{1}_{r\ge R_0}(1-\mu{})^{-1}=_{\mathrm{s}}\mathfrak{b}_{0}\) and \(\varpi{}=_{\mathrm{s}}
\mathfrak{g}_0\) to get \cref{dr-U-commutation}.

We have
\begin{equation}\label{dr-dv-commutation}
[\overline{\partial}_r,\underline{\partial{}}_v] = [\overline{\partial}_r,\partial_v] + [\overline{\partial}_r,\lambda{}U] = \lambda^{-2}\partial{}_v\lambda{}\partial{}_v + \lambda{}^{-1}\partial{}_v\lambda{}U + \lambda{}[\partial{}_r,U].
\end{equation}
It follows that
\begin{equation}
[\overline{\partial}_r,V] = \chi{}'(\underline{\partial{}}_v - \overline{\partial}_r) + \chi{}\lambda{}^{-2}\partial{}_v\lambda{}\partial{}_v + \chi{}\lambda{}^{-1}\partial{}_v\lambda{}U + \chi{}\lambda{}[\overline{\partial}_r,U].
\end{equation}
Now use \cref{dv-in-terms-of-V-and-U}, write \(\underline{\partial{}}_v = \lambda{}[\overline{\partial}_r +
U]\), and use \cref{dr-U-commutation} to obtain \cref{dr-comm-1-V}.

Now we compute \([\overline{\partial}_r,S]\). First, compute
\begin{equation}\label{dr-S-commutation-prep}
\begin{split}
[\overline{\partial}_r,S] &= \chi{}'(v\underline{\partial{}}_v - (r\overline{\partial}_r + u\overline{\partial}_u)) + \chi{}[\overline{\partial}_r,v\underline{\partial{}}_v] + (1-\chi{})\overline{\partial}_r \\
&= \underbrace{\chi{}'(v\underline{\partial{}}_v - (r\overline{\partial}_r + u\overline{\partial}_u))}_{\coloneqq{}\text{(I)}} + \underbrace{\chi{}v\overline{\partial}_r\Bigl(\frac{\kappa{}}{(-\gamma{})}\Bigr)\underline{\partial{}}_v}_{\coloneqq{}\text{(II)}} + (1-\chi{})\overline{\partial}_r \\
\end{split}
\end{equation}
We now treat term \(\text{(I)}\). Observe that
\begin{equation}\label{dover-in-terms-of-dunder}
\overline{\partial}_r = \lambda{}^{-1}\underline{\partial{}}_v - U\qquad \overline{\partial}_u = \frac{(-\gamma{})}{\kappa{}}\underline{\partial{}}_v.
\end{equation}
From this we claim that
\begin{equation}\label{v-u-r-vector-field}
\text{(I)} =_{\mathrm{s}} \mathbf{1}_{\Rc\le r\le 2\Rc}\mathfrak{B}_0D
\end{equation}
Indeed, we have
\begin{equation}
\begin{split}
\text{(I)} &= \chi{}'\Bigl(v - u \frac{(-\gamma{})}{\kappa{}} - \frac{r}{\lambda{}}\Bigr)\underline{\partial{}}_v + \chi{}'rU \\
&= \chi{}'(v - u - r)\underline{\partial{}}_v + \chi{}'(\frac{u}{\kappa{}}(\kappa{} - (-\gamma{})) +r(1-\lambda{}^{-1})) + \chi{}'rU \\
&=_{\mathrm{s}} \mathbf{1}_{\Rc\le r\le 2\Rc}\mathfrak{B}_0\underline{\partial{}}_v + \mathbf{1}_{\Rc\le r\le 2\Rc}U,
\end{split}
\end{equation}
and using \(\underline{\partial{}}_v = \partial_v + \lambda{}U\) and \cref{dv-in-terms-of-V-and-U} gives
\cref{v-u-r-vector-field}. For term \(\text{(II)}\), write
\begin{equation}
\text{(II)} = \chi{}v\overline{\partial}_r\Bigl(\frac{\kappa{}}{(-\gamma{})}\Bigr)\underline{\partial{}}_v =_{\mathrm{s}} \mathfrak{b}_0(\chi{}(-\gamma{})^{-1}vD\kappa{} + \chi{}\kappa{}(-\gamma{})^{-2}vD(-\gamma{}))D,
\end{equation}
so that
\begin{equation}\label{dr-S-comm-II}
\mathbf{1}_{r\ge R_0}\text{(II)}=_{\mathrm{s}} \mathbf{1}_{R_0\le r\le 2\Rc}\mathfrak{B}_VD.
\end{equation}
Combine \cref{v-u-r-vector-field,dr-S-comm-II} with \cref{dr-S-commutation-prep} to
conclude \cref{dr-S-comm}.
\end{proof}
\begin{proof}[Proof of \cref{D-comm-1}]
First we compute \([U,V]\). Since \(U\) and \(\underline{\partial{}}_v\) commute, we have
\begin{equation}\label{UV-comm-prep}
[U,V] = \chi{}'(\overline{\partial}_r-\underline{\partial{}}_v) + (1-\chi{})[U,\overline{\partial}_r].
\end{equation}
Use \cref{dr-U-commutation,dv-in-terms-of-V-and-U} to conclude.

Next, we compute \([U,S]\). Since \(U\) and \(v\underline{\partial{}}_v\) commute, we have
\begin{equation}\label{U-S-commutation-prep}
[U,S] = -\chi{}'(v\underline{\partial{}}_v - (r\overline{\partial}_r + u\overline{\partial}_{u})) + (1-\chi{})[U,r\overline{\partial}_r + u\overline{\partial}_u].
\end{equation}
As shown in \cref{v-u-r-vector-field}, the first term is of the schematic form
\(\mathbf{1}_{\Rc\le r\le 2\Rc}\mathfrak{B}_0D\). To treat the second term, we
compute using \cref{dover-in-terms-of-dunder} that
\begin{equation}
[U,\overline{\partial}_r] = -\frac{U\lambda{}}{\lambda{}^2}\underline{\partial{}}_v =_{\mathrm{s}} r^{-2}\set{1,\varpi{},(1-\mu{})^{-1}}\underline{\partial{}}_v\qquad [U,\overline{\partial}_u] = U\Bigl(\frac{(-\gamma{})}{\kappa{}}\Bigr)\underline{\partial{}}_v,
\end{equation}
which implies
\begin{equation}\label{U-rdr-udu-comm}
\begin{split}
[U,r\overline{\partial}_r + u\overline{\partial}_u] &= -\overline{\partial}_r + \frac{1}{(-\nu{})}\overline{\partial}_u +  \Bigl(- r\frac{U\lambda{}}{\lambda{}^2} + uU\Bigl(\frac{(-\gamma{})}{\kappa{}}\Bigr)\Bigr)\underline{\partial{}}_v.
\end{split}
\end{equation}
Since \(-\overline{\partial}_r + (-\nu{})^{-1}\overline{\partial}_u = U\) and the coefficient of the
final term consists has the structure \(r^{-1}\mathfrak{g}_0 +
r^{-1}\mathfrak{G}_U\), the desired \cref{U-S-comm-1} follows from \cref{U-S-commutation-prep}.

Now we compute \([V,S]\). Use \cref{dover-in-terms-of-dunder} to obtain the
following identities:
\begin{align}
[\underline{\partial{}}_v,v\underline{\partial{}}_v] &= \underline{\partial{}}_v, \\
[\underline{\partial{}}_v,r\overline{\partial}_r + u\overline{\partial}_u] &= \Bigl(r\underline{\partial{}}_v\lambda{}^{-1} + u\underline{\partial{}}_v\Bigl(\frac{(-\gamma{})}{\kappa{}}\Bigr) + 1\Bigr)\underline{\partial{}}_v, \\
[\overline{\partial}_r,v\underline{\partial{}}_v] &= v\overline{\partial}_r\Bigl(\frac{\kappa{}}{(-\gamma{})}\Bigr)\overline{\partial}_u,\\
[\overline{\partial}_r,r\overline{\partial}_r + u\overline{\partial}_u] &= \overline{\partial}_r.
\end{align}
It follows that
\begin{equation}
\chi{}[\underline{\partial{}}_v,S] - \chi{}^2\underline{\partial{}}_v = \chi{}(1-\chi{})\Bigl(r\underline{\partial{}}_v\lambda{}^{-1} + u\underline{\partial{}}_v\Bigl(\frac{(-\gamma{})}{\kappa{}}\Bigr) + 1\Bigr)\underline{\partial{}}_v =_{\mathrm{s}} \mathbf{1}_{\Rc\le r\le 2\Rc}\mathfrak{B}_VD.
\end{equation}
and
\begin{equation}
(1-\chi{})[\overline{\partial}_r,S] - (1-\chi{})^2\overline{\partial}_r = \chi{}'(1-\chi{})(v\underline{\partial{}}_v - (r\overline{\partial}_r + u\overline{\partial}_{u})) + \chi{}(1-\chi{})v\overline{\partial}_{r}\Bigl(\frac{\kappa{}}{(-\gamma{})}\Bigr)\overline{\partial}_u =_{\mathrm{s}} \mathbf{1}_{\Rc\le r\le 2\Rc}\mathfrak{B}_VD
\end{equation}
To conclude, note that
\begin{equation}
[V,S] - \chi{}^2\underline{\partial{}}_v - (1-\chi{})^2\overline{\partial}_r =_{\mathrm{s}} \mathbf{1}_{\Rc\le r\le 2\Rc}\mathfrak{b}_0D.
\end{equation}
and
\begin{equation}
\chi{}^2\underline{\partial{}}_v + (1-\chi{})^2\overline{\partial}_r - V =_{\mathrm{s}} \mathbf{1}_{\Rc\le r\le 2\Rc}\mathfrak{b}_0D.
\end{equation}
\end{proof}
\begin{proof}[Proof of \cref{D-comm}]
The strategy is to use an induction argument to simultaneously prove the
following two estimates
\begin{align}
[U,L] &=_{\mathrm{s}} \mathcal{O}(\mathfrak{B}_{\alpha{}-S+V},r^{-1}\mathfrak{G}_{\alpha{}-S+U})[U\Gamma{}^{\le \alpha{}-S} + r^{-2}\mathfrak{g}_{\alpha{}-V}D\Gamma{}^{\le \alpha{}-V} + r^{-1}\mathfrak{G}_{\alpha{}-S+U}D\Gamma{}^{\le \alpha{}-S}] \label{U-comm} \\
[V,L] &=_{\mathrm{s}} V\Gamma{}^{\le \alpha{}-S} + \mathcal{O}(\mathfrak{B}_{\alpha{}-S+V},r^{-1}\mathfrak{G}_{\alpha{}-S+U})r^{-2}D\Gamma{}^{\le \alpha{}-U}.  \label{V-comm}
\end{align}
We then prove \cref{V2-comm} assuming \cref{U-comm,V-comm}. Observe that
\cref{U-comm-weak,V-comm-weak,D-comm-weak} follow immediately from
\cref{U-comm,V-comm} and the definition of \(\mathcal{C}_\alpha{}\) (see
\cref{C-alpha-def}). To streamline the notation, we schematically write
``\(\text{good}^U_\beta{}\)'' (resp. ``\(\text{good}^V_\beta{}\)'' and
``\(\text{good}^{V^2}_\beta{}\)'') for terms appearing on the right side of
\cref{U-comm} (resp. \cref{V-comm,V2-comm}) for the multi-index \(\beta{}\). Observe
that \(\text{good}^U_\beta{} =_{\mathrm{s}}\text{good}^U_\alpha{}\) whenever \(\beta{}\le
\alpha{}\) (and the analogous statements hold for the schematic quantities
``\(\text{good}^V_\beta{}\)'' and ``\(\text{good}^{V^2}_\beta{}\)'').

\step{Step 1: Proof of \cref{U-comm,V-comm}.} Both sides of both equations vanish when
\(\abs{\alpha{}} = 0\). when \(\abs{\alpha{}} = 1\), both statements follow from
\cref{D-comm-1}.

Now suppose inductively that \cref{U-comm,V-comm} hold for multi-indices \(\le
\alpha{}\) and \(\le \alpha{}'\) (with \(\abs{\alpha{}},\abs{\alpha{}'}\ge 1\)).
We will show that they hold for multi-indices \(\le \alpha{} + \alpha{}'\).
Before proceeding, we note that \cref{U-comm,V-comm} for multi-indices \(\le
\alpha{}\) imply
\begin{equation}\label{D-comm-ind-consequence}
[D,\Gamma{}^{\le \alpha{}}] =_{\mathrm{s}} \mathcal{O}(\mathfrak{B}_{\alpha{}-S+V},r^{-1}\mathfrak{G}_{\alpha{}-S+U})D\Gamma{}^{\le \alpha{}-U}.
\end{equation}

\step{Step 1a: Proof of \cref{U-comm}.} The induction hypothesis for \cref{U-comm} with multi-indices \(\le \alpha{}\) implies
\begin{equation}\label{U-comm-ind-1}
\mathcal{O}(\mathfrak{B}_{\alpha{}+\alpha{}'-S+V},r^{-1}\mathfrak{G}_{\alpha{}+\alpha{}'-S+U})[U,\Gamma^{\le \alpha{}}]\Gamma^{\le \alpha{}'} =_{\mathrm{s}} \text{good}^U_{\alpha{}+\alpha{}'}
\end{equation}
We now use \cref{U-comm} for multi-indices \(\le \alpha{}'\) to compute
\begin{equation}\label{U-comm-ind-2}
\begin{split}
\Gamma{}^{\le \alpha{}}[U,\Gamma{}^{\le \alpha{}'}] &=_{\mathrm{s}} \mathcal{O}(\mathfrak{B}_{\alpha{}+\alpha{}'-S+V},r^{-1}\mathfrak{G}_{\alpha{}+\alpha{}'-S+U})[\Gamma{}^{\le \alpha{}}U\Gamma{}^{\le \alpha{}'-S} + r^{-2}\mathfrak{g}_{\alpha{}+\alpha{}'-V}\Gamma{}^{\le \alpha{}}D\Gamma{}^{\le \alpha{}'-V}  \\
&\qquad + r^{-1}\mathfrak{g}_{\alpha{}+\alpha{}'-S+U}\Gamma{}^{\le \alpha{}}D\Gamma{}^{\le \alpha{}'-S}] \\
&=_{\mathrm{s}} \mathcal{O}(\mathfrak{B}_{\alpha{}+\alpha{}'-S+V},r^{-1}\mathfrak{G}_{\alpha{}+\alpha{}'-S+U})[[U,\Gamma{}^{\le \alpha{}}]\Gamma{}^{\le \alpha{}'-S} + r^{-2}\mathfrak{g}_{\alpha{}+\alpha{}'-V}[D,\Gamma{}^{\le \alpha{}}]\Gamma{}^{\le \alpha{}'-V} \\
&\qquad + r^{-1}\mathfrak{g}_{\alpha{}+\alpha{}'-S+U}[D,\Gamma{}^{\le \alpha{}}]\Gamma{}^{\le \alpha{}'-S}] + \text{good}^U_{\alpha{}+\alpha{}'} \\
&=_{\mathrm{s}} \text{good}^U_{\alpha{}+\alpha{}'}.
\end{split}
\end{equation}
In passing to the last line, we noted that the first term in square brackets on
the second line is good by \cref{U-comm-ind-1}, and the other two terms are good by
\cref{D-comm-ind-consequence}. Summing \cref{U-comm-ind-1,U-comm-ind-2} concludes the
proof of \cref{U-comm} for multi-indices \(\le \alpha{} + \alpha{}'\).

\step{Step 1b: Proof of \cref{V-comm}.} The induction hypothesis for \cref{V-comm} with
multi-indices \(\le \alpha{}\) implies
\begin{equation}\label{V-comm-ind-1}
[V,\Gamma^{\le \alpha{}}]\Gamma^{\le \alpha{}'} =_{\mathrm{s}} \text{good}^V_{\alpha{}+\alpha{}'}.
\end{equation}
We now use \cref{V-comm} for multi-indices \(\le \alpha{}'\) and
\cref{V-comm-ind-1,D-comm-ind-consequence} for multi-indices \(\le \alpha{}\) to compute
\begin{equation}\label{V-comm-ind-2}
\begin{split}
\Gamma^{\le \alpha{}}[V,\Gamma^{\le \alpha{}'}] &=_{\mathrm{s}} \Gamma{}^{\le \alpha{}}V\Gamma{}^{\le \alpha{}-S} + \mathcal{O}(\mathfrak{B}_{\alpha{}+\alpha{}'-S+V},r^{-1}\mathfrak{G}_{\alpha{}+\alpha{}'-S+U})r^{-2}\Gamma{}^{\le \alpha{}}D\Gamma{}^{\le \alpha{}'-U}  =_{\mathrm{s}} \text{good}^V_{\alpha{}+\alpha{}'}.
\end{split}
\end{equation}
In passing to the final equality we used
\cref{V-comm-ind-1,D-comm-ind-consequence}. Add \cref{V-comm-ind-1,V-comm-ind-2} to
establish \cref{V-comm} for multi-indices \(\le \alpha{} + \alpha{}'\).

\step{Step 2: Proof of \cref{V2-comm} given \cref{U-comm,V-comm}.} We now show that
\cref{V2-comm} holds for multi-indices \(\le \alpha{}\) if \cref{U-comm,V-comm} do.
\begin{equation}\label{V2-comm-ind-1}
\begin{split}
V[V,\Gamma{}^{\le \alpha{}}] &=_{\mathrm{s}} V^2\Gamma{}^{\le \alpha{}-S} + \mathcal{O}(\mathfrak{B}_{\alpha{}-S+V+V},r^{-1}\mathfrak{G}_{\alpha{}-S+U+V})r^{-2}VD\Gamma{}^{\le \alpha{}-U} \\
&=_{\mathrm{s}} \text{good}_{\alpha{}}^{V^2} + \mathcal{O}(\mathfrak{B}_{\alpha{}-S+V+V},r^{-1}\mathfrak{G}_{\alpha{}-S+U+V})r^{-2}[D,V]\Gamma{}^{\le \alpha{}-U} \\
&=_{\mathrm{s}} \text{good}_{\alpha{}}^{V^2}.
\end{split}
\end{equation}
To pass to the last line we used \cref{U-V-comm} and absorbed the
\(r^{-2}\mathfrak{g}_0\) term into the \(\mathcal{O}\) term. Use
\cref{V2-comm-ind-1} to establish
\begin{equation}\label{V2-comm-ind-2}
[V,\Gamma^{\le \alpha{}}]V =_{\mathrm{s}} V\Gamma^{\le \alpha{}-S}V + \text{good}_\alpha^{V^2} =_{\mathrm{s}} V[V,\Gamma^{\le \alpha{}-S}] + \text{good}_\alpha^{V^2}
\end{equation}
Sum \cref{V2-comm-ind-1,V2-comm-ind-2} to conclude \cref{V2-comm} for multi-indices
\(\le \alpha{}\).
\end{proof}
\begin{proof}[Proof of \cref{U-rearrangement-formula}]
Write \(L = L_1UL_2\) for \(L_i\in \Gamma^{\alpha_i}\) and \(\alpha{} = \alpha_1 + \alpha_2 + U\), where
\(\alpha_i\ge 0\). Set \(L' = L_1L_2\), and use \cref{U-comm} to compute \(L - UL'
= -[U,L_1]L_2\), noting that \(L_1\in \Gamma{}^{\le \alpha{}-U}\):

\begin{equation}
\begin{split}
L-UL' &=_{\mathrm{s}} \mathcal{O}(\mathfrak{B}_{\alpha{}-U-S+V},r^{-1}\mathfrak{G}_{\alpha{}-S})[U\Gamma{}^{\le \alpha{}-U-S} + r^{-2}\mathfrak{g}_{\alpha{}-U-V}D\Gamma{}^{\le \alpha{}-U-V} + r^{-1}\mathfrak{G}_{\alpha{}-S}D\Gamma{}^{\le \alpha{}-U-S}] \\
&= \mathcal{O}(\mathfrak{B}_{\alpha{}-U-S+V},r^{-1}\mathfrak{G}_{\alpha{}-S})[U\Gamma{}^{\le \alpha{}-U-V} + r^{-2}\mathfrak{g}_{\alpha{}-U-V}V\Gamma{}^{\le \alpha{}-U-V} + r^{-1}\mathfrak{G}_{\alpha{}-S}V\Gamma{}^{\le \alpha{}-U-S}]
\end{split}
\end{equation}
\end{proof}
\begin{proof}[Proof of \cref{rearrangement-formula}]
If \(\abs{\alpha{}} \le 1\), then this is immediate because \(L = L'\). The operators \(L\)
and \(L'\) are made up of the same vector fields, but in possibly different
orders. We can reorder \(L\) into \(L'\) by performing finitely many swaps of
two adjacent vector fields. It is therefore enough to consider \(L =
L_1\Gamma_1\Gamma_2L_2\) and \(L' = L_1\Gamma_2\Gamma_1L_2\), where \(L_i\in
\Gamma^{\alpha_i}\) and \(\Gamma_i\in \Gamma^{\beta_i}\) for \(\abs{\beta_i} =
1\) and \(\alpha_i\ge 0\) and \(\alpha{} = \alpha_1 + \alpha_2 + \beta_1 +
\beta_2\). Inspecting \cref{D-comm-1} shows that
\begin{equation}
[\Gamma{}_1,\Gamma{}_2] =_{\mathrm{s}} [1 + r^{-2}\mathfrak{g}_{\beta{}_1+\beta{}_2-U-V} + r^{-1}\mathfrak{G}_{\beta{}_1+\beta{}_2-S} + \mathfrak{B}_{\beta{}_1+\beta{}_2-S}]D =_{\mathrm{s}} \mathcal{C}_{\le \beta{}_1+\beta{}_2-S}D,
\end{equation}
and so
\begin{equation}\label{rearrangement-step-1}
\begin{split}
L - L' &= L_1[\Gamma{}_1,\Gamma{}_2]L_2 =_{\mathrm{s}} \mathcal{C}_{\le \alpha{}_1+\beta{}_1+\beta{}_2-S}\Gamma{}^{\le \alpha{}_1}D\Gamma{}^{\le \alpha{}_2} =_{\mathrm{s}} \mathcal{C}_{\le \alpha{}-S}D\Gamma{}^{\le \alpha{}_1+\alpha{}_2} =_{\mathrm{s}} \mathcal{C}_{\le \alpha{}-S}D\Gamma{}^{\le \alpha{}-\beta{}_1-\beta{}_2}.
\end{split}
\end{equation}
In the final equality we used the consequence \([D,\Gamma^\beta{}] =_{\mathrm{s}}
\mathcal{C}_{<\beta{}}D\Gamma^{\le \beta{}-1}\) of \cref{D-comm}. If \(\beta_1 =
\beta_2\), then \(\Gamma_1 = \Gamma_2\), so \(L - L' = 0\). We can therefore
suppose \(\beta_1\neq{}\beta_2\), in which case \(\beta_1 + \beta_2 \ge U + V\). This concludes the
proof.
\end{proof}
\begin{proof}[Proof of \cref{UV-box-lemma}]
We first prove \cref{vU-box}. Use \cref{box-uv-gauge} to compute
\begin{equation}\label{box-uv-gauge-prep}
\partial_vU = \partial_v(-\nu{})^{-1}\partial{}_u + (-\nu{})^{-1}\partial{}_u\partial{}_v =  \kappa{}\Box{} + \frac{2(\varpi{}-\mathbf{e}^2/r)}{r^2}\kappa{}U - r^{-1}(\lambda{}U - \partial{}_v)
\end{equation}
and use \cref{dv-in-terms-of-V-and-U} to conclude.

Now we prove \cref{UV-box}. We first prove the formula with \(VU\) on the left
side. By \cref{box-H-gauge} and \(U = -\underline{\partial{}}_r\), we have
\begin{equation}\label{dv-U-box-close}
\underline{\partial{}}_vU = \kappa{}\Box{} + \kappa{}(1-\mu{})U^2 -  r^{-1}\kappa{}(2-2r^{-1}\varpi{})U + r^{-1}\underline{\partial{}}_v =_{\mathrm{s}} \mathfrak{b}_0[\Box{} + r^{-1}D + U^2].
\end{equation}
By \cref{vU-box}, we have
\begin{equation}\label{dv-U-box-far}
\mathbf{1}_{r\ge R_0}\overline{\partial}_rU =_{\mathrm{s}} \mathfrak{b}_0[\Box{} + r^{-1}D].
\end{equation}
The claim with \(UV\) on the left now follows from
\cref{dv-U-box-close,dv-U-box-far} and the expression \(V =
\chi{}\underline{\partial{}}_v + (1-\chi{})\overline{\partial}_r\). To complete
the proof, note that \([U,V] =_{\mathrm{s}} \mathfrak{b}_0r^{-1}D\) by \cref{UV-comm-prep}.
\end{proof}
\begin{proof}[Proof of \cref{wave-comm-formula-UVS}]
A direct computation shows that
\begin{equation}
\begin{split}
[\Box{},U] + \frac{2(\varpi{}-\mathbf{e}^2/r)}{r^2}U^2 &=
\frac{U\kappa{}}{\kappa{}}\Box{}  + \frac{2}{r^2}\Bigl[-1 + \frac{2\varpi{}}{r} + \frac{1}{2}\frac{rU\kappa{}}{\kappa{}}\bigl(1 - \frac{\mathbf{e}^2}{r^2}\bigr)\Bigr]U + \frac{1}{r^2}\frac{1}{\kappa{}}\underline{\partial}_v \\
&=_{\mathrm{s}} \mathfrak{b}_U[\Box{} + r^{-2}V + r^{-2}U],
\end{split}
\end{equation}
which is \cref{U-comm-formula}.

A similar computation shows that
\begin{equation}\label{V-comm-prep-1}
\begin{split}
\chi{}[\Box{},\underline{\partial{}}_v] &= \frac{\chi{}\underline{\partial}_v\kappa{}}{\kappa{}}\Box{} + \frac{\chi{}\underline{\partial{}}_v\lambda{}}{\kappa{}}U^2 - \frac{2}{r}\Bigl[-r^{-1}\chi{}\underline{\partial{}}_v\varpi{} + \frac{\chi{}\underline{\partial}_v\kappa{}}{\kappa{}}(1 - \varpi{}/r)\Bigr]U \\
&=_{\mathrm{s}} \mathbf{1}_{r\le 2\Rc}\mathfrak{b}_V[\Box{} + r^{-2}V + r^{-2}U + U^2]
\end{split}
\end{equation}
where we used the fact that \(\mathfrak{b}_V\) can depend on \(\Rc\). Similarly,
one computes
\begin{equation}\label{V-comm-prep-2}
\begin{split}
(1-\chi{})[\Box{},\overline{\partial}_r] - (1-\chi{})\frac{2(\varpi{}-\mathbf{e}^2/r)}{r^2}\overline{\partial}_r^2 &= \frac{(1-\chi{})\overline{\partial}_r(-\gamma{})}{(-\gamma{})}\Box{} + \frac{2}{r^2}\Bigl[-1 + \frac{2\varpi{}}{r} + \frac{1}{2}\frac{r\overline{\partial}_r(-\gamma{})}{(-\gamma{})}\bigl(1 - \frac{\mathbf{e}^2}{r^2}\bigr)\Bigr](1-\chi{})\overline{\partial}_r\\
&\qquad  - \frac{1}{r^2}\frac{1}{(-\gamma{})}(1-\chi{})\overline{\partial}_u \\
&=_{\mathrm{s}} \mathfrak{b}_V[\Box{} + r^{-2}\mathfrak{g}_VV + r^{-2}U].
\end{split}
\end{equation}
A direct computation shows that
\begin{equation}\label{V-comm-prep-3}
G\overline{\partial}_r^2  =V^2 + \mathcal{E}
\end{equation}
for \(G = (1 + 2\chi{}(\lambda{}-1) + \chi^2(\lambda{}-1)^2)\) and \(\mathcal{E}=_{\mathrm{s}} \mathbf{1}_{r\le
2\Rc}\mathfrak{b}_V[D + U^2]\). Since \(r(\lambda{}-1) =_{\mathrm{s}} \mathfrak{B}_0^{\circ
}\), we find that \(G\) does not vanish and \((1-\chi{})G^{-1} =_{\mathrm{s}}
\mathbf{1}_{r\ge \Rc}\mathfrak{B}_0^{\circ }\) when \(\Rc\) is large enough
depending on \(\mathfrak{B}_0^{\circ }\). We next compute
\begin{equation}\label{V-comm-prep-4}
[\Box{},V] = \chi{}[\Box{},\underline{\partial{}}_v] + (1-\chi{})[\Box{},\overline{\partial}_r] + [\Box{},\chi{}](\underline{\partial{}}_v - \overline{\partial}_r)
\end{equation}
and
\begin{equation}\label{V-comm-prep-5}
[\Box{},\chi{}] =_{\mathrm{s}} \mathbf{1}_{\Rc\le r\le 2\Rc}\mathfrak{B}_0^{\circ }[1 + D]
\end{equation}
for \(\Rc\) large enough depending on \(\mathfrak{B}_0^{\circ }\) (due to the
occurrence of \(F^{-1}\) for \(F =_{\mathrm{s}} \mathbf{1}_{r\ge \Rc}(1 +
\chi{}(1-\lambda{})^{-1})\)). Since \(\mathbf{1}_{\Rc\le r\le
2\Rc}(\underline{\partial{}}_v - \overline{\partial}_r) =_{\mathrm{s}} \mathfrak{b}_0D\), we can use
\cref{UV-box} to compute
\begin{equation}\label{V-comm-prep-6}
[\Box{},\chi{}](\underline{\partial{}}_v - \overline{\partial}_r) =_{\mathrm{s}} \mathbf{1}_{\Rc\le r\le 2\Rc}\mathfrak{B}_0^{\circ }[V^2 + D + U^2]
\end{equation}
Combine
\cref{V-comm-prep-1,V-comm-prep-2,V-comm-prep-3,V-comm-prep-4,V-comm-prep-5,V-comm-prep-6}
to obtain \cref{V-comm-formula}.

Write \(\underline{S} \coloneqq{} v\underline{\partial{}}_v\) and \(\overline{S} \coloneqq{}
r\overline{\partial}_r + u\overline{\partial}_u\). One computes explicitly
\begin{equation}
[\Box{},\underline{S}] = \Bigl(1 + \frac{\underline{S}\kappa{}}{\kappa{}}\Bigr)\Box{} - \frac{2}{r^2}\Bigl[-\underline{S}\varpi{} + \Bigl(1 + \frac{\underline{S}\kappa{}}{\kappa{}}\Bigr)(r-\varpi{})\Bigr]U + \frac{1}{\kappa{}}(\underline{S}\lambda{} + \lambda{})U^2.
\end{equation}
and
\begin{equation}\label{S-far-commutation}
[\Box{},\overline{S}] = \Bigl(2 + \frac{\overline{S}(-\gamma{})}{(-\gamma{})}\Bigr)\Box{} + \frac{\overline{S}(-\nu{})}{(-\gamma{})}\overline{\partial}_r^2 + \frac{2}{r^2}\Bigl[\varpi{} - \overline{S}\varpi{} + \frac{r\overline{S}(-\gamma{})}{(-\gamma{})}(1-r^{-1}\varpi{})\Bigr]\overline{\partial}_r.
\end{equation}
Since
\begin{equation}
[\Box{},S] = \chi{}[\Box{},\underline{S}] + (1-\chi{})[\Box{},\overline{S}] + [\Box{},\chi{}](\underline{S} - \overline{S}),
\end{equation}
and
\begin{equation}
\underline{S} - \overline{S} =_{\mathrm{s}} \mathbf{1}_{\Rc\le r\le 2\Rc}\mathfrak{b}_SDD^{\le 1},
\end{equation}
we conclude \cref{S-comm-formula} using similar arguments to those used in the computation of
\([\Box{},V]\).
\end{proof}
\begin{proof}[Proof of \cref{wave-comm-formula-UVS}]
Schematically write ``\(\text{good}_\alpha{}\)'' for terms appearing on the right side of
\cref{wave-comm-formula-equation} when \(L\in \Gamma^\alpha{}\). With this
notation, \cref{wave-comm-formula-equation} reads
\begin{equation}
[\Box{},L] + \alpha{}_Uf_UUL + \alpha{}_Vf_VVL =_{\mathrm{s}} \text{good}_\alpha{}.
\end{equation}
Let \(\alpha{}\) and \(\alpha{}'\) be multi-indices with \(\abs{\alpha{}},\abs{\alpha{}'} \ge 1\). Suppose that
\cref{wave-comm-formula-equation} holds for all operators in \(\Gamma^{\le
\alpha{}}\) and in \(\Gamma^{\le \alpha{}'}\). We will show that if \(L\in
\Gamma^\alpha{}\) and \(L'\in \Gamma^{\alpha{}'}\), then
\cref{wave-comm-formula-equation} holds for \(L L'\in \Gamma^{\alpha{} +
\alpha{}'}\).

\step{Step 0: The base cases \(\abs{\alpha{}}\le 1\).} If \(\abs{\alpha{}} = 0\), then \(L = 1\), so both sides
of \cref{wave-comm-formula-equation} vanish. If \(\abs{\alpha{}} = 1\), then \(L
\in \set{U,V,S}\), and so \cref{wave-comm-formula-equation} follows from inspection
of \cref{wave-comm-formula-UVS}. From now on we can suppose that \(\abs{\alpha{}},\abs{\alpha{}'}\ge 1\).

\step{Step 1: Differentiating the right side of \cref{wave-comm-formula-equation}.} We
will use \cref{D-comm} to show that
\begin{equation}\label{wc-ind-step-1}
\Gamma^{\le \alpha{}}\text{good}_{\alpha{}'} =_{\mathrm{s}} \text{good}_{\alpha{}+\alpha{}'}.
\end{equation}
To start, we have
\begin{equation}\label{wc-ind-1-0}
\begin{split}
\Gamma^{\le \alpha{}}\text{good}_{\alpha{}'} &=_{\mathrm{s}} \mathfrak{b}_{\alpha{}+\alpha{}'}[\underbrace{\Gamma{}^{\le \alpha{}}\Gamma{}^{\le \alpha{}'-1}\Box{}}_{\text{(I)}} + \underbrace{r^{-1}\mathfrak{g}_{\alpha{}+\alpha{}'}\Gamma{}^{\le \alpha{}}V^2\Gamma{}^{\le \alpha{}'-S}}_{\text{(II)}} \\
&\qquad + \underbrace{r^{-2}\mathfrak{g}_{\alpha{}+\alpha{}'}\Gamma{}^{\le \alpha{}}V\Gamma{}^{\le \alpha{}'-1}}_{\text{(III)}} + \underbrace{r^{-2}\mathfrak{g}_{\alpha{}+\alpha{}'}\Gamma{}^{\le \alpha{}}U\Gamma{}^{<\alpha{}'}}_{\text{(IV)}}].
\end{split}
\end{equation}
We handle the terms in the square brackets one by one. Clearly
\begin{equation}\label{wc-ind-I}
\text{(I)} =_{\mathrm{s}} \Gamma{}^{\le \alpha{}+\alpha{}'-1}\Box{} =_{\mathrm{s}} \text{good}_{\alpha{}+\alpha{}'}.
\end{equation}
We now develop term \(\text{(II)}\) using \cref{D-comm}:
\begin{equation}\label{wc-ind-II}
\begin{split}
\text{(II)} &=_{\mathrm{s}} r^{-1}\mathfrak{g}_{\alpha{}+\alpha{}'}[V^2\Gamma{}^{\le \alpha{}+\alpha{}'-S} + \mathcal{C}_{<\alpha{}+V}r^{-2}D\Gamma{}^{\le \alpha{}+\alpha{}'-U-S+V}] \\
&=_{\mathrm{s}} \text{good}_{\alpha{}+\alpha{}'} + r^{-1}\mathfrak{b}_{\alpha{}+\alpha{}'}\mathfrak{g}_{\alpha{}+\alpha{}'}r^{-2}D\Gamma{}^{\le \alpha{}+\alpha{}'-1} =_{\mathrm{s}} \text{good}_{\alpha{}+\alpha{}'}\\
\end{split}
\end{equation}
To pass to the second line we used the fact that term \(\text{(II)}\) vanishes
unless \(\alpha{}'_S > 0\), so
\begin{equation}
\mathcal{C}_{<\alpha{} + V}=_{\mathrm{s}} \mathcal{C}_{<\alpha{} + \alpha{}'-S + V} = \mathcal{C}_{<\alpha{} + \alpha{}'} =_{\mathrm{s}} \mathfrak{b}_{\alpha{}+\alpha{}'}.
\end{equation}
From \cref{D-comm}, we handle term \(\text{(III)}\):
\begin{equation}\label{wc-ind-III}
\begin{split}
\text{(III)} &=_{\mathrm{s}} r^{-2}\mathfrak{g}_{\alpha{}+\alpha{}'}[V\Gamma{}^{\le \alpha{}+\alpha{}'-1} + \mathcal{C}_{<\alpha{}}r^{-2}D\Gamma{}^{\le \alpha{}+\alpha{}'-U-1}] \\
&=_{\mathrm{s}} \text{good}_{\alpha{}+\alpha{}'} + r^{-2}\mathfrak{b}_\alpha{}\mathfrak{g}_{\alpha{}+\alpha{}'}D\Gamma{}^{\le \alpha{}+\alpha{}'-U-1} =_{\mathrm{s}} \text{good}_{\alpha{}+\alpha{}'}, \\
\end{split}
\end{equation}
where in passing to the second line we recalled that \(\mathcal{C}_{<\alpha{}}=_{\mathrm{s}}
\mathfrak{b}_\alpha{}\) (see \cref{b-alpha-def,C-alpha-def}). We handle term
\(\text{(IV)}\) similarly using \cref{D-comm}:
\begin{equation}\label{wc-ind-IV}
\begin{split}
\text{(IV)} &=_{\mathrm{s}} r^{-2}\mathfrak{g}_{\alpha{}+\alpha{}'}[U\Gamma{}^{<\alpha{}+\alpha{}'} + \mathcal{C}_{<\alpha{}}[U\Gamma{}^{<\alpha{}+\alpha{}'-S} + r^{-2}\mathfrak{G}_{<\alpha{}}D\Gamma{}^{<\alpha{}+\alpha{}'-V} + r^{-1}\mathfrak{G}_{<\alpha{}}D\Gamma{}^{<\alpha{}+\alpha{}'-S}]] \\
&=_{\mathrm{s}} \text{good}_{\alpha{}+\alpha{}'} + r^{-2}\mathfrak{b}_\alpha{}\mathfrak{g}_{\alpha{}+\alpha{}'}[U\Gamma{}^{<\alpha{}+\alpha{}'} + V\Gamma{}^{\le \alpha{}+\alpha{}'-1}] =_{\mathrm{s}} \text{good}_{\alpha{}+\alpha{}'}.
\end{split}
\end{equation}
Substitute \cref{wc-ind-I,wc-ind-II,wc-ind-III,wc-ind-IV} into \cref{wc-ind-1-0} to
obtain \cref{wc-ind-step-1}.

\step{Step 2: Differentiating \cref{wave-comm-formula-equation}.} In
this step we show that
\begin{equation}\label{wc-ind-step-2}
L[\Box{},L'] + \alpha{}_{U}'f_U ULL' + \alpha{}_{V}'f_V V L L' =_{\mathrm{s}} \text{good}_{\alpha{}+\alpha{}'}.
\end{equation}
Let \(X\in \set{U,V}\). Since \(f_U,f_V\in r^{-2}\mathfrak{B}_0^{\circ }\mathfrak{g}_0\), we have \(\Gamma^{\le
\alpha{}}f_{X}=_{\mathrm{s}}r^{-2}\mathfrak{B}_{\alpha{}}\mathfrak{g}_{\alpha{}}=_{\mathrm{s}} r^{-2}\mathfrak{G}_\alpha{}\). Use
this to compute
\begin{equation}\label{wc-ind-2}
\begin{split}
[\Gamma{}^{\le \alpha{}},f_{X}X] &=_{\mathrm{s}} [\Gamma{}^{\le \alpha{}},f_{X}]X + f_{X}[\Gamma{}^{\le \alpha{}},X] =_{\mathrm{s}} \sum_{\substack{\alpha{}_1+\alpha{}_2\le \alpha{} \\ \abs{\alpha{}_1}\ge 1}}(\Gamma{}^{\alpha{}_1}f_{X})\Gamma{}^{\alpha{}_2}X + f_{X}[\Gamma{}^{\le \alpha{}},X] \\
&=_{\mathrm{s}} \sum_{\substack{\alpha{}_1+\alpha{}_2 \le  \alpha{}}}(\Gamma{}^{\alpha{}_1}f_{X})[X,\Gamma{}^{\alpha{}_2}] + \sum_{\substack{\alpha{}_1+\alpha{}_2 \le  \alpha{} \\ \abs{\alpha{}_1}\ge 1}}(\Gamma{}^{\alpha{}_1}f_{X})X\Gamma{}^{\alpha{}_2} \\
&=_{\mathrm{s}} (\Gamma{}^{\le \alpha{}}f_{X})[X\Gamma{}^{\le \alpha{}-1} + [X,\Gamma{}^{\le \alpha{}}]] =_{\mathrm{s}} r^{-2}\mathfrak{G}_\alpha{}[X\Gamma{}^{\le \alpha{}-1} + [X,\Gamma{}^{\le \alpha{}}]].
\end{split}
\end{equation}
Using \cref{wc-ind-2,D-comm} and \(\mathfrak{G}_\alpha{} =_{\mathrm{s}} \mathfrak{g}_{\alpha{} + \alpha{}'}\)
(since \(\abs{\alpha{}'} > 1\) implies \(\alpha{} < \alpha{} + \alpha{}'\)), we compute
\begin{equation}\label{wc-ind-3}
\begin{split}
&L(\alpha{}_U'f_UUL' + \alpha{}_V'f_V V L') - \alpha{}_{U}'f_U ULL' - \alpha{}_{V}'f_V V L L' \\
&=_{\mathrm{s}} \alpha{}_U'[L,f_UU]L' + \alpha{}_{V}'[L,f_V V]L' =_{s} [L,f_DD]L' \\
&=_{\mathrm{s}} r^{-2}\mathfrak{G}_\alpha{}[X\Gamma{}^{\le \alpha{}+\alpha{}'-1} + [X,\Gamma{}^{\le \alpha{}}]L'] \\
&=_{\mathrm{s}} r^{-2}\mathfrak{g}_{\alpha{}+\alpha{}'}[X\Gamma{}^{\le \alpha{}+\alpha{}'-1} + [X,\Gamma{}^{\le \alpha{}}]L'] \\
&=_{\mathrm{s}} \text{good}_{\alpha{}+\alpha{}'} + r^{-2}\mathfrak{g}_{\alpha{}+\alpha{}'}\mathcal{C}_{<\alpha{}}D\Gamma{}^{\le \alpha{}+\alpha{}'-U} =_{\mathrm{s}} \text{good}_{\alpha{}+\alpha{}'}. \\
\end{split}
\end{equation}
To complete the proof of \cref{wc-ind-step-2}, combine
\cref{wave-comm-formula-equation} for the multi-index \(\alpha{}\) with
\cref{wc-ind-3}, and then apply \cref{wc-ind-step-1}, recalling that \(L\in \Gamma^\alpha{}\):
\begin{equation}
\begin{split}
L[\Box{},L'] + \alpha{}_{U}'f_U ULL' + \alpha{}_{V}'f_V V L L' &= L([\Box{},L'] + \alpha{}_U'f_UUL' + \alpha{}_V'f_V V L') \\
&\qquad - \bigl(L(\alpha{}_U'f_UUL' + \alpha{}_V'f_V V L') - \alpha{}_{U}'f_U ULL' - \alpha{}_{V}'f_V V L L'\bigr) \\
&=_{\mathrm{s}} L(\text{good}_{\alpha{}'}) + \text{good}_{\alpha{}+\alpha{}'} =_{\mathrm{s}} \text{good}_{\alpha{}+\alpha{}'}.
\end{split}
\end{equation}

\step{Step 3: Acting on \cref{wave-comm-formula-equation} from the right.} We now show
that
\begin{equation}\label{wc-ind-step-3}
[\Box{},L]L' + \alpha{}_Uf_UUL L' + \alpha{}_Vf_VVLL' =_{\mathrm{s}} \text{good}_{\alpha{}+\alpha{}'}.
\end{equation}
Act on both sides of \cref{wave-comm-formula-equation} for \(L\in \Gamma{}^\alpha{}\) with \(L'\in
\Gamma{}^{\alpha{}'}\) on the right to get
\begin{equation}\label{wc-ind-4}
\begin{split}
[\Box{},L]L' + \alpha{}_Uf_UUL L' + \alpha{}_Vf_VVLL' &=_{\mathrm{s}} \mathfrak{b}_\alpha{}\Gamma{}^{\le \alpha{}-1}\Box{}L' + \text{good}_{\alpha{}+\alpha{}'} =_{\mathrm{s}} \mathfrak{b}_\alpha{}\Gamma{}^{\le \alpha{}-1}[\Box{},L']  + \text{good}_{\alpha{}+\alpha{}'}.
\end{split}
\end{equation}
By \cref{wc-ind-step-2,D-comm} and the inductive hypothesis, we have \(\Gamma{}^{\le \alpha{}-1}[\Box{},L'] =
_s \text{good}_{\alpha{}+\alpha{}'}\), and substituting this into \cref{wc-ind-4}
proves \cref{wc-ind-step-3}.

\step{Step 4: Completing the induction.} Add \cref{wc-ind-step-2,wc-ind-step-3} to get
\begin{equation}
\begin{split}
[\Box{},LL'] + (\alpha{}_U + \alpha{}_U')f_U UL L ' +  (\alpha{}_V + \alpha{}_V')f_V VL L' &= L[\Box{},L'] + \alpha{}_{U}'f_U ULL' + \alpha{}_{V}'f_V V L L'\\
&\qquad + [\Box{},L]L' + \alpha{}_Uf_UUL L' + \alpha{}_Vf_VVLL' \\
&=_{\mathrm{s}} \text{good}_{\alpha{}+\alpha{}'}.
\end{split}
\end{equation}
We have now shown that \cref{wave-comm-formula-equation} for \(L\in
\Gamma^\alpha{}\) and \(L'\in \Gamma^\alpha{}\) implies
\cref{wave-comm-formula-equation} for \(L L'\in \Gamma^{\alpha{} + \alpha{}'}\).
Since we established \cref{wave-comm-formula-equation} for \(\abs{\alpha{}}\le 1\)
in Step 0, an induction argument concludes \cref{wave-comm-formula-equation} for
all \(\alpha{}\).
\end{proof}
\begin{proof}[Proof of \cref{wave-pointwise-inequalities}]
We can assume \(\abs{\alpha{}}\ge 1\), since the case \(\alpha{}=0\) is trivial.

\step{Step 1: Proof of \cref{wave-pointwise-1}.} Observe that \cref{wave-pointwise-1} is a
consequence of \cref{wave-comm-formula,V-dv-coordinate-change} and the following
two identities:
\begin{align}
V &=_{\mathrm{s}} \mathfrak{b}_0[\partial{}_v + \mathbf{1}_{r\le 2\Rc}U], \label{wave-pointwise-inequalities-step-2} \\
rV^2\Gamma{}^{\le \alpha{}-S} &=_{\mathrm{s}} \mathfrak{B}_0[\mathbf{1}_{r\ge \Rc}\partial{}_v(r\Gamma{}^{<\alpha{}}\varphi{}) + \partial{}_v\Gamma{}^{<\alpha{}}\varphi{} + \mathbf{1}_{r\le 2\Rc}U\Gamma{}^{<\alpha{}}\varphi{}] \label{wave-pointwise-inequalities-step-3}.
\end{align}
First, recall that \cref{wave-pointwise-inequalities-step-2} was proved in
\cref{V-dv-coordinate-change}. To prove \cref{wave-pointwise-inequalities-step-3}, it
is enough to consider the quantity \(\mathbf{1}_{r\ge \Rc}rV^2\Gamma{}^{\le
\alpha{}-S}\), since \(\mathbf{1}_{r\le 2\Rc}rV^2\Gamma^{\le \alpha{}-S} =_{\mathrm{s}}
\mathfrak{B}_0V\Gamma^{\le \alpha{}-S + V} =_{\mathrm{s}}
\mathfrak{B}_0V\Gamma^{<\alpha{}}\) is handled by
\cref{wave-pointwise-inequalities-step-2}. Use
\cref{wave-pointwise-inequalities-step-2} to rewrite the first \(V\) in terms of
\(\partial_v\) and \(U\), then use \(\partial{}_vr = \lambda{} =_{\mathrm{s}}
\mathfrak{b}_0\) to commute \(r\) past \(\partial_v\), and note that the term
beginning with \(U\) is supported in \(\set{r\le 2\Rc}\).

\step{Step 2: Proof of \cref{wave-pointwise-2}.} By the good sign \(f_U\ge 0\) and the
triangle inequality, we have
\begin{equation}
\begin{split}
r^{2-s}UL\varphi{}\Box{}L\varphi{} &= -\alpha_Ur^{2-s}f_U(UL\varphi{})^2 - \alpha_Vr^{2-s}f_VUL\varphi{}VL\varphi{} \\
&\qquad + r^{2-s}(UL\varphi{})(\Box{}L\varphi{} + \alpha{}_uf_U UL\varphi{} + \alpha{}_Vf_V VL\varphi{}) \\
&\le \underbrace{\alpha{}_V\abs{r^{2-s}f_V}\abs{UL\varphi{}}\abs{VL\varphi{}}}_{\coloneqq{}\text{(I)}} + \underbrace{r^{2-s}\abs{UL\varphi{}}\abs{\Box{}L\varphi{} + \alpha{}_Uf_U UL\varphi{} + \alpha{}_Vf_VVL\varphi{}}}_{\coloneqq{}\text{(II)}}.
\end{split}
\end{equation}
Term \(\text{(I)}\) vanishes unless \(\alpha_V > 0\), namely unless \(L\) contains a
 \(V\). In this case, use \cref{rearrangement-formula} and the wave equation
 \cref{UV-box} to rewrite
\begin{equation}\label{wave-pointwise-inequalities-step-4}
\mathbf{1}_{\alpha{}_{V}>0}UL\varphi{} =_{\mathrm{s}} \mathfrak{b}_\alpha{}[\Box{}\Gamma{}^{<\alpha{}}\varphi{} + D\Gamma{}^{<\alpha{}}].
\end{equation}
Control term \(\text{(I)}\) using
\cref{wave-comm-formula,wave-pointwise-inequalities-step-2,wave-pointwise-inequalities-step-4}
and \(\abs{r^{2-s}f_V} \le C(\mathcal{G}_{\alpha,s})\). For term
\(\text{(II)}\), use
\cref{wave-comm-formula,wave-pointwise-1,wave-pointwise-inequalities-step-4}.

\step{Step 3: Proof of \cref{wave-pointwise-3}.} By the triangle inequality, we have
\begin{equation}
\begin{split}
r^{2-s}\abs{\partial{}_vL\varphi{}}\abs{\Box{}L\varphi{}}&\le \underbrace{\alpha{}_U\abs{r^{2-s}f_U}\abs{UL\varphi{}}\abs{\partial{}_vL\varphi{}}}_{\coloneqq{}\text{(I)}} + \underbrace{r^{2-s}\abs{\partial{}_vL\varphi{}}\abs{\Box{}L\varphi{} + \alpha{}_Uf_U UL\varphi{} + \alpha{}_Vf_V VL\varphi{}}}_{\coloneqq{}\text{(II)}}  \\
&\qquad + \underbrace{\alpha{}_V\abs{r^{2-s}f_V}\abs{\partial{}_vL\varphi{}}\abs{VL\varphi{}}}_{\coloneqq{}{\text{(III)}}}. \\
\end{split}
\end{equation}
Handle terms \(\text{(I)}\) and \(\text{(II)}\) as in Step 2. In particular, for
term \(\text{(I)}\) we bring the \(U\) that \(L\) contains to the front using
\cref{rearrangement-formula}, and then use the wave equation \(\partial_vU =_{\mathrm{s}}
\mathfrak{b}_0[\Box{} + r^{-1}D]\) from \cref{vU-box}. It is important here that
\(\partial_vU\) produces a wave equation with no \(U^2\)-term; such a term would
produce \(U^2\Gamma{}^{\le \alpha{}-U} =_{\mathrm{s}} U\Gamma{}^{\le \alpha{}}\), and this
top-order term is not admissible. The analogous term in Step 2 was
\(U^2\Gamma^{\le \alpha{}-V} =_{\mathrm{s}} U\Gamma^{<\alpha{}}\), which is a lower order
term. For term \(\text{(III)}\), use \cref{wave-pointwise-inequalities-step-2} and
\(\abs{r^{2-s}f_{V}} \le C(\mathfrak{B}_0^{\circ },\mathfrak{g}_0)\) to get
\begin{equation}\label{wave-pointwise-inequalities-step-7}
\text{(III)}\le \mathbf{1}_{r\ge \Rc}C((\mathfrak{B}_0^{\circ },\mathfrak{g}_0,\alpha{})\abs{\partial_vL\varphi{}}(\abs{UL\varphi{}} + \abs{\partial_vL\varphi{}}).
\end{equation}
\end{proof}
\printbibliography
\par \medskip \begin{tabular}{@{}l@{}} \textsc{Department of Mathematics, Princeton University}\\ \textit{Email address}: \texttt{gm9165@princeton.edu} \end{tabular}
\end{document}

%% file: rn-penrose.pdf_tex
%% Creator: Inkscape 1.3.2 (091e20ef0f, 2023-11-25, custom), www.inkscape.org
%% Accompanies image file 'rn-penrose.pdf' (pdf, eps, ps)
\begingroup%
  \makeatletter%
  \providecommand\color[2][]{%
    \errmessage{(Inkscape) Color is used for the text in Inkscape, but the package 'color.sty' is not loaded}%
    \renewcommand\color[2][]{}%
  }%
  \providecommand\transparent[1]{%
    \errmessage{(Inkscape) Transparency is used (non-zero) for the text in Inkscape, but the package 'transparent.sty' is not loaded}%
    \renewcommand\transparent[1]{}%
  }%
  \providecommand\rotatebox[2]{#2}%
  \newcommand*\fsize{\dimexpr\f@size pt\relax}%
  \newcommand*\lineheight[1]{\fontsize{\fsize}{#1\fsize}\selectfont}%
  \ifx\svgwidth\undefined%
    \setlength{\unitlength}{311.60502997bp}%
    \ifx\svgscale\undefined%
      \relax%
    \else%
      \setlength{\unitlength}{\unitlength * \real{\svgscale}}%
    \fi%
  \else%
    \setlength{\unitlength}{\svgwidth}%
  \fi%
  \global\let\svgwidth\undefined%
  \global\let\svgscale\undefined%
  \makeatother%
  \begin{picture}(1,0.55975162)%
    \lineheight{1}%
    \setlength\tabcolsep{0pt}%
    \put(0,0){\includegraphics[width=\unitlength,page=1]{rn-penrose.pdf}}%
    \put(0.40962771,0.26575992){\makebox(0,0)[lt]{\lineheight{1.25}\smash{\begin{tabular}[t]{l}$\mathcal{BH}$\end{tabular}}}}%
    \put(0.63320078,0.3216844){\makebox(0,0)[lt]{\lineheight{1.25}\smash{\begin{tabular}[t]{l}$\mathcal{CH}$\end{tabular}}}}%
    \put(0,0){\includegraphics[width=\unitlength,page=2]{rn-penrose.pdf}}%
    \put(0.46246351,0.41905715){\makebox(0,0)[lt]{\lineheight{1.25}\smash{\begin{tabular}[t]{l}$\mathcal{S}$\end{tabular}}}}%
    \put(0.22878939,0.33160105){\makebox(0,0)[lt]{\lineheight{1.25}\smash{\begin{tabular}[t]{l}$\mathcal{CH}'$\end{tabular}}}}%
    \put(0,0){\includegraphics[width=\unitlength,page=3]{rn-penrose.pdf}}%
    \put(0.84911886,0.11282617){\makebox(0,0)[lt]{\lineheight{1.25}\smash{\begin{tabular}[t]{l}$\mathcal{I}$\end{tabular}}}}%
    \put(0.55641371,0.1656182){\makebox(0,0)[lt]{\lineheight{1.25}\smash{\begin{tabular}[t]{l}$\mathcal{H}$\end{tabular}}}}%
    \put(0.12041803,0.19446999){\makebox(0,0)[lt]{\lineheight{1.25}\smash{\begin{tabular}[t]{l}$\mathcal{I}'$\end{tabular}}}}%
    \put(0.29637725,0.17163615){\makebox(0,0)[lt]{\lineheight{1.25}\smash{\begin{tabular}[t]{l}$\mathcal{H}'$\end{tabular}}}}%
    \put(0.69094864,0.27287184){\makebox(0,0)[lt]{\lineheight{1.25}\smash{\begin{tabular}[t]{l}$i^+$\end{tabular}}}}%
    \put(0.18570472,0.26485552){\makebox(0,0)[lt]{\lineheight{1.25}\smash{\begin{tabular}[t]{l}$i^{+}{'}$\end{tabular}}}}%
    \put(0.8815283,0.05091428){\makebox(0,0)[lt]{\lineheight{1.25}\smash{\begin{tabular}[t]{l}$i_0$\end{tabular}}}}%
    \put(-0.00719346,0.06309141){\makebox(0,0)[lt]{\lineheight{1.25}\smash{\begin{tabular}[t]{l}$i_0'$\end{tabular}}}}%
  \end{picture}%
\endgroup%

%% file: penrose.pdf_tex
%% Creator: Inkscape 1.3.2 (091e20ef0f, 2023-11-25, custom), www.inkscape.org
%% Accompanies image file 'penrose.pdf' (pdf, eps, ps)
\begingroup%
  \makeatletter%
  \providecommand\color[2][]{%
    \errmessage{(Inkscape) Color is used for the text in Inkscape, but the package 'color.sty' is not loaded}%
    \renewcommand\color[2][]{}%
  }%
  \providecommand\transparent[1]{%
    \errmessage{(Inkscape) Transparency is used (non-zero) for the text in Inkscape, but the package 'transparent.sty' is not loaded}%
    \renewcommand\transparent[1]{}%
  }%
  \providecommand\rotatebox[2]{#2}%
  \newcommand*\fsize{\dimexpr\f@size pt\relax}%
  \newcommand*\lineheight[1]{\fontsize{\fsize}{#1\fsize}\selectfont}%
  \ifx\svgwidth\undefined%
    \setlength{\unitlength}{311.60502997bp}%
    \ifx\svgscale\undefined%
      \relax%
    \else%
      \setlength{\unitlength}{\unitlength * \real{\svgscale}}%
    \fi%
  \else%
    \setlength{\unitlength}{\svgwidth}%
  \fi%
  \global\let\svgwidth\undefined%
  \global\let\svgscale\undefined%
  \makeatother%
  \begin{picture}(1,0.55975162)%
    \lineheight{1}%
    \setlength\tabcolsep{0pt}%
    \put(0,0){\includegraphics[width=\unitlength,page=1]{penrose.pdf}}%
    \put(0.40962771,0.26575992){\makebox(0,0)[lt]{\lineheight{1.25}\smash{\begin{tabular}[t]{l}$\mathcal{BH}$\end{tabular}}}}%
    \put(0.63320078,0.3216844){\makebox(0,0)[lt]{\lineheight{1.25}\smash{\begin{tabular}[t]{l}$\mathcal{CH}$\end{tabular}}}}%
    \put(0,0){\includegraphics[width=\unitlength,page=2]{penrose.pdf}}%
    \put(0.46246351,0.41905715){\makebox(0,0)[lt]{\lineheight{1.25}\smash{\begin{tabular}[t]{l}$\mathcal{S}$\end{tabular}}}}%
    \put(0.40091092,0.0626695){\makebox(0,0)[lt]{\lineheight{1.25}\smash{\begin{tabular}[t]{l}$\Sigma$\end{tabular}}}}%
    \put(0.23360318,0.33160105){\makebox(0,0)[lt]{\lineheight{1.25}\smash{\begin{tabular}[t]{l}$\mathcal{CH}'$\end{tabular}}}}%
    \put(0,0){\includegraphics[width=\unitlength,page=3]{penrose.pdf}}%
    \put(0.34352168,0.07861077){\makebox(0,0)[lt]{\lineheight{1.25}\smash{\begin{tabular}[t]{l}$\Sigma_0$\end{tabular}}}}%
    \put(0.84911886,0.11282617){\makebox(0,0)[lt]{\lineheight{1.25}\smash{\begin{tabular}[t]{l}$\mathcal{I}$\end{tabular}}}}%
    \put(0.55641371,0.1656182){\makebox(0,0)[lt]{\lineheight{1.25}\smash{\begin{tabular}[t]{l}$\mathcal{H}$\end{tabular}}}}%
    \put(0.12041803,0.19446999){\makebox(0,0)[lt]{\lineheight{1.25}\smash{\begin{tabular}[t]{l}$\mathcal{I}'$\end{tabular}}}}%
    \put(0.29637725,0.17163615){\makebox(0,0)[lt]{\lineheight{1.25}\smash{\begin{tabular}[t]{l}$\mathcal{H}'$\end{tabular}}}}%
    \put(0.69094864,0.27287184){\makebox(0,0)[lt]{\lineheight{1.25}\smash{\begin{tabular}[t]{l}$i^+$\end{tabular}}}}%
    \put(0.18570472,0.26485552){\makebox(0,0)[lt]{\lineheight{1.25}\smash{\begin{tabular}[t]{l}$i^{+}{'}$\end{tabular}}}}%
    \put(0.8815283,0.05091428){\makebox(0,0)[lt]{\lineheight{1.25}\smash{\begin{tabular}[t]{l}$i_0$\end{tabular}}}}%
    \put(-0.00719346,0.06309141){\makebox(0,0)[lt]{\lineheight{1.25}\smash{\begin{tabular}[t]{l}$i_0'$\end{tabular}}}}%
    \put(0,0){\includegraphics[width=\unitlength,page=4]{penrose.pdf}}%
    \put(0.76960619,0.22316327){\makebox(0,0)[lt]{\lineheight{1.25}\smash{\begin{tabular}[t]{l}$C_\text{out}$\end{tabular}}}}%
    \put(0,0){\includegraphics[width=\unitlength,page=5]{penrose.pdf}}%
    \put(0.78829127,0.175406){\makebox(0,0)[lt]{\lineheight{1.25}\smash{\begin{tabular}[t]{l}$C_\text{in}$\end{tabular}}}}%
    \put(0,0){\includegraphics[width=\unitlength,page=6]{penrose.pdf}}%
    \put(0.55133525,0.25328591){\makebox(0,0)[lt]{\lineheight{1.25}\smash{\begin{tabular}[t]{l}$\mathcal{R}_\text{char}$\end{tabular}}}}%
  \end{picture}%
\endgroup%

%% file: rp.pdf_tex
%% Creator: Inkscape 1.3.2 (091e20ef0f, 2023-11-25, custom), www.inkscape.org
%% Accompanies image file 'rp.pdf' (pdf, eps, ps)
\begingroup%
  \makeatletter%
  \providecommand\color[2][]{%
    \errmessage{(Inkscape) Color is used for the text in Inkscape, but the package 'color.sty' is not loaded}%
    \renewcommand\color[2][]{}%
  }%
  \providecommand\transparent[1]{%
    \errmessage{(Inkscape) Transparency is used (non-zero) for the text in Inkscape, but the package 'transparent.sty' is not loaded}%
    \renewcommand\transparent[1]{}%
  }%
  \providecommand\rotatebox[2]{#2}%
  \newcommand*\fsize{\dimexpr\f@size pt\relax}%
  \newcommand*\lineheight[1]{\fontsize{\fsize}{#1\fsize}\selectfont}%
  \ifx\svgwidth\undefined%
    \setlength{\unitlength}{158.40000343bp}%
    \ifx\svgscale\undefined%
      \relax%
    \else%
      \setlength{\unitlength}{\unitlength * \real{\svgscale}}%
    \fi%
  \else%
    \setlength{\unitlength}{\svgwidth}%
  \fi%
  \global\let\svgwidth\undefined%
  \global\let\svgscale\undefined%
  \makeatother%
  \begin{picture}(1,1.0454545)%
    \lineheight{1}%
    \setlength\tabcolsep{0pt}%
    \put(0,0){\includegraphics[width=\unitlength,page=1]{rp.pdf}}%
    \put(0.34932987,0.49904631){\makebox(0,0)[lt]{\lineheight{1.25}\smash{\begin{tabular}[t]{l}$(\tau,v_{R_0}(\tau))$\end{tabular}}}}%
    \put(0.36402902,0.37068495){\makebox(0,0)[lt]{\lineheight{1.25}\smash{\begin{tabular}[t]{l}$\set{r=R_0}$\end{tabular}}}}%
    \put(0.47640298,0.61806898){\makebox(0,0)[lt]{\lineheight{1.25}\smash{\begin{tabular}[t]{l}$\Sigma_\tau$\end{tabular}}}}%
  \end{picture}%
\endgroup%